\documentclass[a4paper,11pt]{article}

\usepackage{jheppub}
\usepackage{amsmath}
\usepackage{amsthm}
\usepackage{amssymb}
\allowdisplaybreaks
\usepackage{graphicx}
\usepackage{epstopdf}
\usepackage{xcolor}
\usepackage{dsfont}
\usepackage{hyperref}
\usepackage{rotating}
\usepackage{afterpage}
\usepackage{arydshln}
\usepackage{cancel}

\newtheorem{prop}{Proposition}

\newtheorem*{prop*}{Proposition}
\newtheorem{cor}[prop]{Corollary}
\newtheorem*{cor*}{Corollary}

\newcommand{\be}{\begin{equation}}
	\newcommand{\ee}{\end{equation}}
\newcommand{\bea}{\begin{eqnarray}}
	\newcommand{\eea}{\end{eqnarray}}
\newcommand{\nn}{\nonumber}
\newcommand{\bdm}{\begin{displaymath}}
	\newcommand{\edm}{\end{displaymath}}

\newcommand{\dd}{\text{d}}
\newcommand{\diag}{\text{diag}}
\DeclareMathOperator{\Ker}{Ker}
\DeclareMathOperator{\image}{Im}

\newcommand{\Rea}{\text{Re}}

\title{
Timelike Supersymmetric Geometries 
in Seven-Dimensional Supergravity} 

\author{Mich\`ele Levi}%

\affiliation{Mathematical Institute, University of Oxford, \\
	Radcliffe Observatory Quarter, Woodstock Road, \\
	Oxford OX2 6GG, United Kingdom}

\emailAdd{michele.levi@maths.ox.ac.uk}

\abstract{
Seven-dimensional supergravity poses a uniquely intricate supersymmetry 
problem, combining topologically massive higher-form gauge dynamics, 
distinct Clifford algebra properties, 
and rich geometric structures.
We first establish the complete formulation of the maximal Cartan-gauged theory 
in terms of the gauge-invariant higher-form flux, 
including both the bosonic action and supersymmetry variations.
This formulation, which also makes the smooth ungauged limit manifest, is 
achieved by requiring a consistent truncation to the minimal Cartan-gauged
theory in the equal-charge limit.
We construct the complete bilinear system for a timelike Killing 
vector, where additional internal orientation of the symplectic-Majorana spinor 
arises already in the scalar sector, 
uniquely among pure minimal ungauged theories.
In the equal-charge sector, we decouple the algebraic constraints into 
SU$(3)$-isotypic components, and analyse the differential bilinear relations, 
determining the permitted dilaton gradient, fluxes, fibration twist 
and intrinsic torsion of the base.
We prove that with gauging and active Cartan flux, 
the complex scalar bilinear vanishes identically. 
We also construct a Killing gauge for the higher-form potential 
compatible with the self-duality condition.
We show that with gauging but no active Cartan flux,
the Killing vector also acts as a generator of the gauged R-symmetry,
and the base structure is enlarged from SU$(3)$ to U$(3)$. 
We find that BPS-flow solutions realize this setting. 
With supersymmetry, we isolate a single independent component of the 3-form 
equation across an exhaustive four-class partition for both theories 
within this setup.
Applied to known and proposed solutions in both theories, this classification 
exposes unexplored sectors and provides an accessible toolbox for 
constructing and testing solutions and developing new geometric 
and holographic applications.	
}


\begin{document}

\maketitle

\flushbottom

\pagenumbering{arabic}
\setcounter{page}{1}

\section{Introduction}

Six-dimensional superconformal field theories (SCFTs) occupy an exceptional place
among quantum field theories (QFTs).
Their ${\cal N}=(2,0)$ theory associated with coincident M5-branes is 
maximally supersymmetric and interacting.
It has no known conventional Lagrangian description, yet its compactifications
generate a large range of lower-dimensional strongly-coupled supersymmetric
QFTs.
In the holographic regime, seven-dimensional gauged supergravity provides 
a natural effective description on the gravitational side of the
six-dimensional physics.
Maximal SO$(5)$-gauged seven-dimensional supergravity arises 
from the consistent reduction of eleven-dimensional supergravity on $S^4$, 
and its maximally-supersymmetric AdS$_7$ vacuum 
is dual to the six-dimensional $(2,0)$ theory at large $N$
\cite{Maldacena:1997re,Henningson:1998gx,Pernici:1984xx,
	Nastase:1999cb,Nastase:2000tu}.
The reduction retains the R-symmetry gauge fields, scalars and higher-form
degrees of freedom (DOFs) needed to describe charged deformations of the AdS$_7$
vacuum and wrapped-brane backgrounds, while admitting an uplift back to the
full eleven-dimensional theory.
Seven-dimensional supergravity is thus more than an effective reduction of
string theory: it is a common gravitational arena for studying six-dimensional 
SCFTs and the associated lower-dimensional QFTs, 
together with defects and states arising across these dimensions.

This convergence of diverse physics in seven dimensions occurs across more than 
a single supergravity theory.
The admissible gauge deformations of the maximal seven-dimensional theory 
are organized by an embedding tensor; 
the original compact SO$(5)$ gauging obtained from the $S^4$ reduction 
is the one considered here 
\cite{Sezgin:1982gi,Pernici:1984xx,Samtleben:2005bp}.
A different reduction of type IIA supergravity on $S^3$ gives maximal
SO$(4)$-gauged supergravity \cite{Cvetic:2000ah}.
The minimal theory is half-maximal, with SU$(2)$ as its full R-symmetry gauging 
\cite{Townsend:1983kk}. 
Its gravity multiplet contains a metric, a scalar, an SU$(2)$ triplet of vector 
fields and a $3$-form potential. 
The theory admits two independent deformations: the SU$(2)$ gauging 
and a topological mass for the $3$-form. 
The minimal theory has several higher-dimensional origins. 
Its ungauged form is obtained by truncating the half-maximal reduction of 
ten-dimensional type I theory on $T^3$ \cite{Chamseddine:1999cr}.
When both deformations are present, the minimal theory both embeds into the
eleven-dimensional theory on $S^4$, and arises as a consistent 
truncation of the ten-dimensional massive type IIA theory
\cite{Lu:1999bc,Apruzzi:2013yva,Passias:2015gya}.
\textit{The common Cartan truncation provides a strategically useful   
window onto this wider network of gauged seven-dimensional theories 
and their string-theory origins.}

Gauging is indispensable for the holographic applications of these
theories, yet the corresponding ungauged limits remain part of the underlying
supergravity problem.
In seven dimensions, several features combine to make this problem highly 
distinctive.
The vector and 3-form gauge fields are always coupled through a Chern--Simons
term, yet notably,
when the theory is gauged the $3$-form can be topologically massive. 
It is then subject to an odd-dimensional self-duality relation, which couples 
the $3$-form potential to its field strength, a distinctive feature of 
seven-dimensional gauged supergravity \cite{Townsend:1983xs,Pernici:1984xx}.
In Lorentzian dimensions $d=4k-1$, this first-order massive relation for
real $(2k-1)$-form potentials takes the schematic form
\be
\qquad \qquad \qquad A_{2k-1}\propto\star_d\dd A_{2k-1}\, ,\qquad d=4k-1\, .
\ee
Its eleven-dimensional analogue would therefore require a $5$-form potential, 
which eleven-dimensional supergravity does not contain.
Indeed, the latter theory has only the metric and a $3$-form potential in its
bosonic sector 
\cite{Gauntlett:2002fz}.

Furthermore, a Hodge self-dual decomposition of real forms at a single middle 
degree requires a positive Hodge square, which can be identified 
from the convention-independent rule
\be
\star_{2k{\rm \bf E}}^2
= -\star_{2k{\rm \bf L}}^2
= (-1)^k\, ,
\ee
for middle-degree $k$-forms in Euclidean or Lorentzian even dimension $2k$.
For a spacetime of dimension $d$, the relevant manifold is therefore
Lorentzian spacetime itself when $d$ is even, or its Euclidean base
of dimension $d-1$ when $d$ is odd.
Therefore, among the supergravity dimensions $4\leq d\leq11$,
only $d=4,7,8,11$ lack this real middle-degree splitting.
In pure minimal ungauged theories in the other dimensions, 
the relevant Clifford-algebra properties then
constrain the middle-degree bilinears to definite Hodge duality.
For spinors in five and nine dimensions 
in the case with a timelike Killing vector, 
the induced spinors on the Euclidean bases are chiral.
Their respective $2$- and $4$-form bilinears therefore have definite
Hodge duality \cite{Gauntlett:2002nw,Tsimpis:2005kj}.
In six and ten dimensions, spinor chirality instead fixes the duality
of the Lorentzian $3$- and $5$-form bilinears, respectively
\cite{Gutowski:2003rg,Tomasiello:2011eb}.

Nevertheless, even the Lorentzian Clifford algebras in $d=4,8,11$ admit 
ordinary Majorana spinors, and their bilinears vanish identically 
at several degrees by the symmetry of their Clifford algebras 
\cite{VanProeyen:1999ni}:
a single spinor yields only two, five, and three non-vanishing
real bilinear forms in these dimensions, respectively.
In seven dimensions, the properties of Cliff$(1,6)$ instead require 
a symplectic-Majorana (SM) doublet, yielding eight non-vanishing 
real bilinears.
Moreover, its scalar and $3$-form bilinears form SM SU$(2)$ 
triplets, whose internal orientation therefore additionally 
entails two SM angles.
Their variation enters the coupled geometric and flux constraints,
introducing an angular dependence that is also absent in $d=4,8,11$.
Seven dimensions is thus singled out as uniquely intricate 
even among pure minimal ungauged theories: 
the confluence of its Clifford-algebra properties and the absence 
of the above real Hodge splitting 
yields a larger, more complex spinor bilinear system.

The underlying geometry of the six-dimensional base further compounds the 
distinctive complexities already identified in seven dimensions.
A timelike Killing vector in seven-dimensional spacetime determines a fibration 
over a base of odd complex dimension three, whose generic SU$(3)$ structure 
is specified by a real $2$-form and a complex $3$-form.
Unlike the hyper-K\"ahler or K\"ahler four-dimensional bases familiar from
five-dimensional supergravity \cite{Gauntlett:2002nw,Gauntlett:2003fk},
its middle-degree structure is intrinsically complex, and its intrinsic
torsion is distributed among five classes 
\cite{chiossi2002intrinsic,Cabrera:2005}.
Moreover, the seven-dimensional fluxes consist of numerous
irreducible representations of the SU$(3)$ structure group.
Together with the interacting higher-form gauge dynamics, these algebraic
and geometric complications compound to make 
the seven-dimensional supergravity problem uniquely challenging.
This helps explain why progress has relied largely on rather specialized
ans\"atze.

The solutions program has therefore developed along three
complementary directions, all motivated by holographic investigation. 
The first direction is the search for black holes.
Charged and rotating asymptotically AdS$_7$ black holes encode 
finite-temperature, finite-density and angular-momentum sectors of 
six-dimensional SCFTs.
Their BPS limits provide precision tests through the 
connection of black-hole thermodynamics to the superconformal index and 
microscopic entropy 
\cite{Cvetic:1999ne,Liu:1999ai,Chong:2004dy,Cvetic:2005zi,Chow:2007ts,
	Kantor:2019lfo}.
Despite the early construction of several restricted families, the general 
solution carrying all three independent
angular momenta and both Cartan charges was obtained only recently
\cite{Bobev:2023bxl}.
Subsequent work developed its holographic interpretation and constructed
a distinct extension of the equal-rotation black holes of
\cite{Chong:2004dy,Cvetic:2005zi}, with horizon geometries based on the
toric Sasaki--Einstein spaces $L^{p,q,r}$ \cite{Bobev:2025xan}.
That this general charge and rotation completion remained unknown for almost
four decades after the gauged theory was constructed is a telling measure 
of how much of the seven-dimensional solution space can remain hidden by 
a tractable ansatz.

The second direction in this solutions program concerns wrapped-brane 
backgrounds in the maximal theory.
Compactifying the M5-brane world-volume theory on a curved cycle,
with its R-symmetry background coupled to the geometry,
produces lower-dimensional QFTs.
Their fixed points and interpolating geometries are naturally described
in seven-dimensional gauged supergravity
\cite{Maldacena:2000mw,Gauntlett:2000ng}.
The magnetic AdS$_5$ flux-branes of \cite{Lu:2003iv} were obtained
by analytic continuation of the static electric family of 
\cite{Cvetic:1999ne,Liu:1999ai}.
Their smooth noncompact equal-charge completion interpolates between
an AdS$_5\times H^2$ region and an asymptotically locally AdS$_7$ region.
The two warped AdS$_3$ families of \cite{Cheung:2022ilc} extend this
direction to M5-branes wrapped on four-dimensional orbifolds,
with two independent Cartan charges.

The third direction concerns holographic RG flows and defects.
In the minimal theory, a dyonic $3$-form supports BPS domain walls and
AdS$_3$-sliced flows \cite{Dibitetto:2017tve}, while smooth half-BPS Janus
solutions then describe two-dimensional conformal defects in six-dimensional
SCFTs \cite{Conti:2024rwd}.
The equal-charge AdS$_5$ flux-brane of \cite{Lu:2003iv} has also been used
in the Cartan truncation of the minimal theory to construct monodromy defects
in six-dimensional $(1,0)$ SCFTs through its uplift to massive type IIA
supergravity \cite{Conti:2025wyj}.
Most recently, classifying AdS$_3$ warped-product solutions within the same
U$(1)$ truncation led to new eleven-dimensional backgrounds describing M5-branes
wrapped on weighted projective spaces and on spindle fibrations
\cite{Conti:2026row}.
These monodromy and wrapped-orbifold constructions further demonstrate a
similar lesson:
even a small Abelian truncation of the seven-dimensional theory retains
enough information to connect different six-dimensional parent theories,
distinct higher-dimensional origins, and holographic observables in lower
dimensions.

The present breadth of known solutions in seven dimensions makes a systematic 
treatment even more necessary.
An explicit solution reveals what one ansatz permits, but not which of its 
features are imposed by the theory itself, 
and which are merely imposed by its ansatz.
Nor does it reveal which flux components of its geometric structure 
group, or geometries remain available beyond that solution. 
A classification of solutions reverses this logic: 
it first determines the fluxes and geometry that are compatible with the 
complete set of constraints imposed by the Killing spinor equations (KSEs) and 
the bosonic field equations. 
Only then does one impose additional symmetries or coordinates in order to 
construct solutions.
In five dimensions, such classifications became a common language for organizing 
supersymmetric geometries and for discovering black holes, rings and other black 
objects \cite{Gauntlett:2002nw,Gauntlett:2003fk,Gutowski:2004ez,Elvang:2004rt}.
Seven dimensions have not acquired a comparably general working framework. 
An early timelike analysis in minimal supergravity identified a generic SU$(3)$ 
structure \cite{Cariglia:2004qi}. 
A comprehensive timelike classification framework that explicitly resolves the 
supersymmetry constraints and the geometric structure has remained missing, 
as has a unified treatment of the maximal and minimal theories 
in a common Cartan sector.

An additional timely motivation to establish such a framework for organizing 
seven-dimensional information geometrically is its usefulness for the novel 
program of localization in supergravity. 
This new methodology has recently made it possible to derive on-shell actions, 
entropies and central charges from spinor bilinears and fixed-point data, 
without requiring the construction of full solutions
\cite{BenettiGenolini:2023kxp,Martelli:2023oqk,
	BenettiGenolini:2023ndb,Colombo:2023fhu}.
In \cite{Suh:2024asy}, the warped AdS$_5\times M_2$ ansatz with
a vanishing $3$-form in the seven-dimensional maximal Cartan truncation 
reduces the action to an integral over $M_2$, to which ordinary two-dimensional
equivariant localization is applied.
Recent developments have extended localization to odd-dimensional theories
with boundaries and Chern--Simons interactions
\cite{Cassani:2024kjn,Colombo:2025ihp,Cassani:2025iix,BenettiGenolini:2025icr,
	Colombo:2025yqy,Cassani:2026teb} in $d=5$, and also to $d=11$ in 
\cite{BenettiGenolini:2026cdw}.
Extending the localization program in seven dimensions for the spacetime action 
with these interactions requires precisely the data exposed by a geometric 
classification:
the supersymmetric Killing vector and other spinor bilinears,
a complete set of bosonic action and KSEs, and the gauge potentials,
where the field-strength formulation established here separates the
gauge-invariant part of the action from its Chern--Simons terms.
This separation provides the essential starting point for the relative
equivariant localization of seven-dimensional supergravity with all gauge fields 
active in \cite{Gaar:2026gmo}.
Developing the seven-dimensional analogue of the mature lower-dimensional 
toolbox therefore has value both for discovering new solutions and for 
extracting observables without solving for them explicitly.

The purpose of the present work is to provide this framework in the common
Abelian sector of maximal and minimal seven-dimensional supergravity. 
We begin with the U$(1)^2$ truncation of the maximal SO$(5)$-gauged
theory, and subsequently restrict to the equal-charge sector that incorporates
the minimal U$(1)$-gauged theory. 
The Cartan-gauged sector is a strategically-chosen tractable common window. 
It retains a non-trivial scalar, an Abelian vector field
and the topologically massive $3$-form, admits both maximal- and minimal-theory 
interpretations, and encompasses representatives of all three directions 
in the solutions program described above.
Moreover, while the gauged theory enables the AdS$_7$ holographic applications, 
treating the gauged and ungauged theories together clearly distinguishes the 
constraints due to gauging from those intrinsic to seven-dimensional 
supergravity.

We start here by establishing the field-strength formulation
that yields a complete setup of the maximal U$(1)^2$-gauged theory,
including both its bosonic action and supersymmetry variations,
in terms of the gauge-invariant $4$-form field strength.
The earlier Cartan KSEs of \cite{Liu:1999ai} and the matching bosonic
action provided later in \cite{Cheung:2022ilc} retain the explicit $3$-form
potential, as detailed in Appendix~\ref{addendumtov1}.
We derive these two critical components, which specify the supersymmetric 
solutions, directly from the original maximal SO$(5)$-gauged theory, 
yet their complete form cannot be fixed within the maximal theory alone. 
We explicitly fix their form by requiring that the maximal theory admits 
a consistent truncation in the equal-charge limit to the minimal supergravity 
at both the bosonic and fermionic levels.
This requirement uniquely determines, in particular, the topological mass term 
and the higher-form mapping from potential to field strength.
The resulting action then yields the self-duality condition, up to its 
associated gauge freedom, and admits a smooth ungauged limit, while the 
maximal-supersymmetry coupling relation emerges via the supersymmetry 
variations. The fermionic truncation also identifies the additional 
maximal-theory content that still remains in the equal-charge sector. 
This formulation provides a single complete starting point for the two theories 
before any ansatz is imposed.

Starting from a single SM Killing spinor, 
we next construct its complete set of non-vanishing bilinears assuming the 
timelike geometry as a fibration over a six-dimensional base, 
which turns out to generically admit an SU$(3)$ structure.
Keeping the internal SM spinor orientation general, we decompose the algebraic 
KSE constraints into independent SU$(3)$ isotypic systems 
and analyse the differential bilinear relations.
This determines the permitted components of the dilaton gradient and the
Cartan and $4$-form fluxes, the fibration twist, 
and the intrinsic torsion of the base, including all five torsion
forms, whenever an SU$(3)$ structure is preserved.
Gauging forces the almost-complex structure of the base to be integrable.
In gauged configurations where the Abelian flux acts trivially on
the preserved spinor, the holomorphic $3$-form rotates in phase along
the Killing direction, and the base instead carries a U$(3)$ structure,
with no intrinsic fifth torsion form.
We further construct compatible Killing gauges for both gauge potentials,
including the $3$-form potential in the presence of gauging, namely in 
conjunction with the self-duality condition.

Combining these supersymmetry results with the bosonic field equations yields 
four solution classes, distinguished independently by the presence of gauging 
and by whether the Abelian flux acts non-trivially on the Killing spinor.
We further establish that, given the KSEs and Bianchi identities, 
the independent bosonic equations of motion reduce 
to the transverse $(2,2)$-form component of the $3$-form equation 
in every class. This specific equation should be supplemented by the scalar 
one only when the Abelian flux acts non-trivially on the Killing spinor, 
and the complex scalar bilinear vanishes.
All remaining equations of motion then follow.
Tables~\ref{tab:preliminary-classes-ug} and \ref{tab:preliminary-classes-g} 
summarize the ungauged and gauged classifications, respectively, including the 
constraints on the fields, internal SM orientation, and the resulting 
intrinsic torsion and geometric structure.
The classification applies to both the equal-charge maximal theory 
and the minimal U$(1)$-gauged theory with the coupling relation required 
by the maximal embedding. 
The equal-charge analysis provided here also captures the core of its extension 
to two independent Cartan charges, as shown in Section~4 of \cite{Gaar:2026gmo}, 
where an efficient reorganization of the scalar and flux fields confines the 
additional unequal-charge input.
This economy is demonstrated in the unequal-charge bilinear analysis therein.
Finally, we use the classification framework as a diagnostic for the known and 
proposed solutions landscape.

Within the equal-charge Lorentzian sector, the gauged active class
contains two distinct general black-hole families: the solutions with three
independent angular momenta of \cite{Chow:2007ts}, also recovered
from \cite{Bobev:2023bxl}, and the $L^{p,q,r}$ black holes of
\cite{Bobev:2025xan}.
The same class also contains the two warped AdS$_3$ families of
\cite{Cheung:2022ilc}, whose local equal-charge specializations 
we develop in order to include them in our survey.
We therefore present smooth noncompact global completions for both families
and further establish their admissibility in our survey.
One of these families shares the fully activated SU$(3)$ content
of the generic black hole of \cite{Chow:2007ts}.
The complementary gauged null class with vanishing Cartan field strength
likewise contains two geometrically distinct BPS-flow families: the
$\mathbb R^{1,2}\times S^3$ background of \cite{Dibitetto:2017tve}
and the Janus family of \cite{Conti:2024rwd},
which contains the earlier AdS$_3\times S^3$ background solution of
\cite{Dibitetto:2017tve}.
These black holes, warped AdS$_3$ geometries and BPS flows constitute
three complementary solution frontiers.
As such, they provide complementary tests for the utility of our classification.
Their SU$(3)$ content illustrates the reach of the common Cartan window,
as well as the geometric and functional freedom that remains within the
classification and solution space.
Our review also rules out early proposals of solutions in 
\cite{Cariglia:2004qi}. 
Notably, to the best of our knowledge, the ungauged solution space 
remains essentially unexplored.
Tables~\ref{tab:known-gactive-solutions} and \ref{tab:minimal-solutions}
summarize the presently known solution landscape.
This framework therefore also serves as an accessible toolbox for studying 
seven-dimensional theories and geometries, providing a common starting point for 
constructing and analysing further solutions and developing new geometric 
applications.

The paper is organized as follows. 
Section~\ref{7dsugra} derives the complete formulation of
the bosonic action and KSEs of the maximal U$(1)^2$-gauged theory 
in terms of the 4-form field strength,
and establishes its bosonic and fermionic equal-charge truncation to the
minimal theory. 
Section~\ref{bilinearsandstructure} constructs the spinor bilinears, and 
expresses them in terms of the timelike fibration, the internal SM orientation, 
and the SU$(3)$ structure of the six-dimensional base. 
In Section~\ref{Killingresults} we decouple the algebraic KSE constraints by
SU$(3)$ isotypic components, 
analyse the differential bilinear relations, construct compatible
Killing gauges for the vector and 3-form potentials, 
and determine the intrinsic torsion of the base.
Section~\ref{classifyingbps} identifies the independent bosonic
equations through KSE integrability, organizes the resulting
systems into four solution classes, and applies this
classification framework to the known and proposed maximal and
minimal solutions that fall within this setup.
We summarize the results and end with future directions in Section~\ref{lafin}. 
Appendices~\ref{7Cliff}--\ref{altcompletion} collect necessary technical 
material, as well as intermediate results used throughout.
Appendix~\ref{addendumtov1} was added after the first version of this work appeared on arXiv, and clarifies the relation of our formulation 
in Section~\ref{7dsugra} to \cite{Cheung:2022ilc} 
and summarizes the additions to our solution survey in Section~\ref{maxsol}.

\section{Supergravity in Seven Dimensions}
\label{7dsugra}

In this work we focus on the consistent reduction of eleven-dimensional 
supergravity on $S^4$, associated with a stack of M5-branes, 
which yields maximal ${\cal{N}}=4$ gauged supergravity in seven
dimensions \cite{Pernici:1984xx}. 
The latter theory has a local SO$(5)$ Yang-Mills gauge symmetry, 
which is the isometry group of the internal space $S^4$, 
related with a real gauge coupling $g$, to be denoted as SO$(5)_g$. 
There is also a local composite USp$(4)$ $\cong$ Spin$(5)$ R-symmetry, 
acting on fermion indices, which is also denoted by SO$(5)_c$ 
when acting on vector indices.
This R-symmetry is not an independent Yang--Mills gauge factor 
with a separate set of gauge fields.  
Thus the physical Yang--Mills gauge group is SO$(5)_g$,
whose maximal torus is U$(1)^2$.

The bosonic field content of the theory consists of the graviton,
ten gauge fields in the adjoint representation of SO$(5)_g$, 
$A_{\mu}^{IJ} = - A_{\mu}^{JI}$, with SO$(5)_g$ indices $I,J$,  
and spacetime coordinate indices, $\mu,\nu,\ldots$ 
(note that in \cite{Pernici:1984xx} spacetime is considered in
a Euclidean signature),
and five $3$-form potentials in the fundamental multiplet, 
$B_{\mu\nu\sigma}^I$.
In addition, there are fourteen scalars, which parametrize 
the coset space SL$(5,{\mathbb R})/\text{SO}(5)_c$, $\Pi_I^i$ 
with $i$ denoting vector indices of SO$(5)_c$. 
The fermionic field content consists of four gravitini, $\psi_\mu^A$, 
where $A$ denotes spinorial indices of Spin$(5)$, and sixteen spin-$1/2$ fields, 
$\lambda_i^A$, in a spinor-vector representation of Spin$(5)$-SO$(5)_c$, 
satisfying the gamma-tracelessness condition $\Gamma_i\lambda_i=0$,
which selects the irreducible representation $16$ from
$5\otimes4=16\oplus4$,
with the Clifford matrices $\Gamma_i\in \text{Cliff}(5)$, 
as further detailed in Section~\ref{spinors}. 

The topological mass term for the $3$-form gauge field requires an
odd-dimensional self-duality condition, a distinctive feature of
seven-dimensional gauged supergravity.
The theory must therefore be formulated carefully,
in particular so that the limit $g\to0$ to the ungauged theory is smooth.
In the field-strength formulation we employ below, the integrated 
self-duality condition introduces an auxiliary $2$-form. 
This auxiliary field may be gauged away locally,
together with a corresponding choice of gauge for the $3$-form. 
These subtle points will be properly discussed below.

\subsection{Bosonic Theory}
\label{bostheory}

We now further focus on the truncation of the gauged theory 
to the U$(1)^2$ Cartan of the gauge group SO$(5)_g$.
The local SO$(5)_c$ freedom is then used to rotate the composite index 
$i$, such that the scalar coset $\Pi_I^i$ is specialized to diagonal form, 
which for the Cartan of the gauge group reduces to 
\be
\label{Cartansector}
\Pi_{I=i} = \diag \left( (X_1)^{-1/2}, (X_1)^{-1/2}, 
(X_2)^{-1/2}, (X_2)^{-1/2}, X_1 X_2 \right)\, , 
\ee 
that now consists of only two independent real scalars, relabelled as 
$X_{I=1,2}$. 
In addition, the bosonic sector now contains the graviton, two Abelian 
gauge fields, relabelled as $A^{12} \to A_{1}$, $A^{34}\to A_2$, and a 
single $3$-form gauge potential, relabelled as $B^5 \to B$. 
Solutions to the truncated sector will necessarily extend 
to those of the full gauged theory.

The bosonic action of the truncated U$(1)^2$-gauged sector
of the maximal theory reads:
\begin{align}
\label{bosaction}
S = \frac{1}{16\pi G_N} \int & \bigg[ \star R 
-\frac{1}{2} \sum_{I=1}^{2}  \dd\varphi_I \wedge \star \dd\varphi_I
- \frac{1}{2} \sum_{I=1}^{2} X_I^{-2} F_I \wedge \star F_I
- \frac{1}{2}  (X_1 X_2)^{2} G \wedge \star G \nn \\*
& - 2g^2 \star {\cal{V}}(X_I) 
- m B \wedge G
+ B \wedge F_1 \wedge F_2  \bigg]\, ,
\end{align} 
where the scalars $X_I$ are defined via the dilatons, 
$\varphi_I$, as
\begin{subequations}
\begin{align}	
X_1 \equiv & \exp \left( - \frac{1}{\sqrt{10}}\varphi_1 
- \frac{1}{\sqrt{2}}\varphi_2 \right)
\, , \\*
X_2 \equiv & \exp \left( -\frac{1}{\sqrt{10}}\varphi_1 
+ \frac{1}{\sqrt{2}}\varphi_2 \right)
\, ,
\end{align}
\end{subequations}
the field strengths read
\be 
I=1,2\, ,\quad F_I = \dd A_I\, ,
\qquad G=\dd B\, , 
\ee
the scalar potential reads
\be
\label{scapot}
{\cal{V}}(X_I) = (X_1 X_2)^{-4} - 4 (X_1 X_2)^{-1} 
\left(X_1^{-1}+X_2^{-1}\right) - 8 X_1X_2\, ,
\ee
and lastly there are two $B$-dependent Chern-Simons terms.

Note in particular the topological mass term with a real St\"uckelberg coupling, 
$m$, which makes the $3$-form $B$ massive.
The latter is therefore further constrained by a self-duality 
condition, which is unique to seven-dimensional theories 
\cite{Townsend:1983xs}.
This condition must be imposed in order to ensure that the $3$-form DOFs fit 
into the massless supermultiplet with the graviton, as we will see below, 
when we present the resulting equations of motion (EOMs). 

At this point it should be highlighted that in our formulation of the action 
in eq.~\eqref{bosaction}, as opposed to the original one in 
\cite{Pernici:1984xx}, the $3$-form potential, $B$, appears directly 
only in the Chern-Simons terms, 
whereas its kinetic term is given in terms of its field strength, $G$. 
The original formulation in \cite{Pernici:1984xx} features only the explicit 
$3$-form potential in the action, which is unfavourable both because it results 
in an apparent discontinuity in the limit $g\to0$ and because the form of its 
kinetic term is not gauge-invariant.
We return to this critical point in the subsequent derivation 
of the KSEs in Section~\ref{spinors}. 
We provide the relation to the Cartan-gauged potential-based action of 
\cite{Cheung:2022ilc} in Appendix~\ref{addendumtov1}.

Furthermore, the topological mass term appears in our action
in eq.~\eqref{bosaction} with the negative coefficient $-m$, as in the
action of \cite{Gauntlett:2000ng,Chong:2004dy}, 
and in the minimal-theory action of \cite{Dibitetto:2017tve,Conti:2024rwd}. 
By contrast, the bosonic works
\cite{Chow:2007ts,Bobev:2023bxl,Bobev:2025xan} use an alternative
coefficient $+g$ in their maximal-theory action. 
Appendix~B of \cite{Bobev:2023bxl} relates the bosonic conventions used in 
several black-hole papers through either $g\to-g$ or a simultaneous sign 
reversal of the $3$-form.
These transformations establish bosonic maps, but they do not provide 
consistent supersymmetry variations.
We therefore tested whether this alternative bosonic action could be 
incorporated into a complete supersymmetric setup, 
that is also compatible with an equal-charge truncation to the minimal theory, 
which is defined below.
Unlike the case for our formulation in eq.~\eqref{bosaction}, we found this 
mapping to be non-unique, as it leaves two possible completions  
of the maximal theory, which turn out to be inconsistent with one another.
This can be seen only once the bosonic action and the KSEs are treated 
together, as we explain in Section~\ref{spinors}. 
In Appendix~\ref{altcompletion}, we constructed the candidate supersymmetric 
completions of this alternative bosonic action and examined their validity via 
an integrability analysis. 
We therefore adhere to our formulation in eq.~\eqref{bosaction} throughout.

As noted in \cite{Nastase:2000tu}, there exists a consistent truncation 
of the maximal SO$(5)$-gauged theory to the minimal SU$(2)$-gauged one 
from \cite{Townsend:1983kk,Mezincescu:1984ta}, 
where the relation between the mass, $m$, 
and the gauge coupling, $g$, is fixed.
In this work we focus on the truncation of the maximal SO$(5)$-gauged theory 
to the U$(1)^2$-gauged sector, which can thus be consistently truncated to the 
minimal U$(1)$-gauged (and ungauged) theory, by specializing to the so-called 
equal-charge case of equal Abelian gauge fields
\be  
\label{eqchargelim}
A \equiv A_1 = A_2 \, ,
\ee
where also 
\be 
\label{eqchargescalar}
X \equiv X_1 = X_2\, .
\ee
The scalar potential in eq.~\eqref{scapot} then simplifies to
\be 
\label{scapoteqcharge}
{\cal{V}} (X) = X^{-8} - 8 X^{-3} -8 X^2\, ,
\ee
which has a stationary point at $X=1$, 
as its derivative reads
\be 
\label{derscapoteqcharge}
{\cal{V}}' (X) = - 8 X^{-9} 
\left( X^5 - 1 \right) \left( 2X^{5} - 1 \right)\, ,
\ee
with the second root, $X = \sqrt[5]{\frac12}$, 
as an unstable stationary point.
Beginning in Section~\ref{Killingresultsalg}, we will specialize to 
this equal-charge sector, which incorporates the minimal theory.

The EOMs of the bosonic fields are derived directly from the 
variation of the action in eq.~\eqref{bosaction}. 
The Einstein field equations take the form:
\begin{align} 
\label{EinsEOM}
R_{\mu\nu} - \frac{1}{2} g_{\mu\nu} R & = 
\frac{1}{2} \sum_{I=1}^{2} 
\left[ \partial_\mu\varphi_I \partial_\nu\varphi_I	
- \frac{1}{2} g_{\mu\nu} \partial^\alpha\varphi_I 
\partial_\alpha\varphi_I 
\right] 
\nn \\*
& 
+ \frac{1}{2} \sum_{I=1}^{2} X_I^{-2} 
\left[ F_{I\mu\alpha} F_{I\nu}^{\,\,\,\,\alpha}
- \frac{1}{4} g_{\mu\nu} F_{I\alpha\beta} F_I^{\,\,\alpha\beta} 
\right] 
\nn \\*
&
+ \frac{1}{12} (X_1 X_2)^2 
\left[ G_{\mu\alpha\beta\gamma} G_{\nu}^{\,\,\alpha\beta\gamma} 
- \frac{1}{8} g_{\mu\nu} G_{\alpha\beta\gamma\delta} 
G^{\alpha\beta\gamma\delta} \right] -  g^2 \,g_{\mu\nu} 
\,{\cal{V}}(X_I) \, .
\end{align}	
The EOMs of the scalars read:
\begin{subequations}
\label{scalarEOM}	
\begin{align}
\nabla^{\mu}\nabla_{\mu} \varphi_1 
& = \frac{1}{2\sqrt{10}} \sum_{I=1}^{2} X_I^{-2} F_{I\mu\nu} 
F_{I}^{\,\mu\nu}
- \frac{1}{12\sqrt{10}} (X_1 X_2)^{2} 
G_{\mu\nu\alpha\beta} G^{\mu\nu\alpha\beta} \nn \\
& + \frac{8}{\sqrt{10}} g^2 
\Big( 2 (X_1 X_2)^{-4} - 3 (X_1 X_2)^{-1} \left(X_1^{-1} + 
X_2^{-1}\right) 
+4 X_1X_2 \Big) \, , \\*
\nabla^{\mu}\nabla_{\mu} \varphi_2 
& = \frac{1}{2\sqrt{2}} \Big( X_1^{-2} F_{1\mu\nu} F_1^{\,\mu\nu} 
- X_2^{-2} F_{2\mu\nu} F_2^{\,\mu\nu} \Big) 
+ 4\sqrt{2} g^2 (X_1 X_2)^{-1} \left(X_2^{-1} - X_1^{-1}\right)\, .
\end{align}	
\end{subequations}
The EOMs of the two Abelian gauge fields read:
\begin{subequations}
\label{AEOM}	
\begin{align}	
\dd\left(X_1^{-2} \star F_1 \right) = F_2 \wedge G \, ,\\
\dd\left(X_2^{-2} \star F_2 \right) = F_1 \wedge G \, ,
\end{align}
\end{subequations}
and the EOM of the $3$-form gauge field reads:
\be 
\label{EOM3-form}
\dd\left[ (X_1 X_2)^{2} \star G \right] 
= - 2 m G + F_1 \wedge F_2\, .
\ee
We see here that the ``magnetic'' components of $G$ serve as its 
electric source, hence the so-called self-duality of $B$.
Finally, the Bianchi identities for the gauge field strengths read:
\be 
\label{Bianchi}
\dd F_I = \dd G = 0. 
\ee

As noted above, the $3$-form is also subject to a self-duality 
condition \cite{Townsend:1983xs}
\be
\label{selfdual}
(X_1 X_2)^{2} \star G 
= - 2 m B 
+ \omega_{(3)} - \dd C_{(2)}\, ,
\ee
where we defined 
\be
\label{someCS}
\omega_{(3)} \equiv \frac{1}{2} \left( 
A_1 \wedge F_2 +  A_2 \wedge F_1 \right)\, , 
\ee
which is a Chern-Simons form, and we introduced 
an auxiliary $2$-form potential $C_{(2)}$. 
Such a $2$-form is ensured to exist locally due to Poincar\'e's lemma 
applied to the integrated EOM in eq.~\eqref{EOM3-form}, which is a 
second-order differential equation in $B$. 
The self-duality condition thus imposes a first-order differential 
equation in $B$, which reduces the number of DOFs by half from those of a 
massive 3-form to those of a massless one. 
Again we see that the self-dual condition in eq.~\eqref{selfdual} 
relates the $3$-form gauge potential to its curl, hence 
we refer to it as self-dual. 

Unlike the gauge-invariant EOM in eq.~\eqref{EOM3-form},
the self-dual condition in eq.~\eqref{selfdual} is an integrated
first-order relation, which effectively constitutes a gauge condition,
and therefore depends on the choice of gauge representatives for
$B$ and $C_{(2)}$.
Indeed, its form can be preserved under the following combined gauge 
transformations of $B$ and $C_{(2)}$
\be
\label{BCgauge}
B \longrightarrow B + \dd\Lambda_{(2)}\, ,
\qquad
C_{(2)} \longrightarrow C_{(2)} - 2 m \Lambda_{(2)} + \dd\Lambda_{(1)}\, ,
\ee
under which both $G$ and the combination $-2 m B - \dd C_{(2)}$ remain invariant.
Thus, $C_{(2)}$ may locally be gauged away,
but this simultaneously selects a gauge representative for $B$.
We shall therefore not impose $\dd C_{(2)}=0$ generally.
Instead, we retain $\dd C_{(2)}$ in eq.~\eqref{selfdual},
and later choose it consistently with the so-called Killing gauge for $B$ 
when $g\neq0$, see Section~\ref{diff1-2forms}.

It is important to highlight that when the bosonic theory is 
formulated in terms of the field strength, as in our eq.~\eqref{bosaction}, 
the transition to the ungauged theory, as in \cite{Sezgin:1982gi}, 
see also \cite{Samtleben:2005bp}, 
is smooth in the limit $g \to 0$. 
The self-dual condition on the $3$-form is no longer relevant 
as $g \sim m \to 0$, and the 3-form field is no longer massive. 
The ungauged theory can then be formulated instead in terms of
a $2$-form field, as in $C_{(2)}$, whose field strength is dual 
in seven dimensions to that of the 3-form.

\subsection{Killing Spinors}
\label{spinors}

In this work we generally use the conventions of 
\cite{Freedman:2012zz,VanProeyen:1999ni}, unless otherwise 
explicitly noted, regarding irreducible spinors, and the 
corresponding representation of the Clifford algebra.

In seven-dimensional spacetime the Dirac spinors carry 
$2^{\lfloor7/2\rfloor} = 8$ complex components. 
The $8\times8$ gamma matrices, $\{\gamma^{\mu}\}_{\mu=0,\ldots,6}$, 
span the Clifford algebra with a Lorentzian signature, Cliff$(1,6)$, 
whose details are given in Appendix~\ref{7Cliff}.
The DOFs of the Dirac spinor can be rearranged 
into a doublet of SM spinors, $\chi^i$, ${i=1,2}$, 
which satisfies the SM condition:
\be
\label{sympMajorana}
\chi^i = \epsilon^{ij} (\chi^j)^c\, ,
\ee
where $\epsilon_{ij}$, the $2$-dimensional Levi-Civita symbol, is the 
symplectic metric on the SM pair indices, with 
\be 
\epsilon_{12} = 1\, , 
\qquad \epsilon_{ij}\epsilon^{jk} = - \delta_{i}{}^k\, ,
\ee 
and each such pair of SM spinors carries $16$ real DOFs as required.

Due to the USp$(4) \cong \text{Spin}(5)$ R-symmetry, both the
gravitini and the spin-$\tfrac{1}{2}$ fermions, as well as the
supersymmetry parameters or Killing spinors, carry a Spin$(5)$
spinor index $A=1,\ldots,4$, in addition to their Spin$(1,6)$
spinor index $\alpha=1,\ldots,8$.
The associated Euclidean Clifford algebra, Cliff$(5)$, is
detailed in Appendix~\ref{5ECliff}.
In the U$(1)^2$-gauged sector, all internal matrices entering
the KSEs are diagonal in the common eigenbasis of the
Spin$(5)$ Cartan generators.
The four USp$(4)$ component equations therefore decouple.
We may thus analyse a single Cartan component of the Killing spinor, 
denoted by $\epsilon^{\alpha A}$, writing it as the direct product 
of the associated spacetime spinor $\chi^\alpha$ 
and the chosen internal eigenvector $\chi^A$:
\be
\label{spinorasdirectproduct}
\epsilon^{\alpha A}
= \chi^\alpha \otimes \chi^A
\in \mathbb{C}^8 \otimes \mathbb{C}^4\, .
\ee

Although truncation of the maximal SO$(5)$-gauged theory to its 
U$(1)^2$-gauged sector was tackled in \cite{Liu:1999ai}, 
that work did not include a complete treatment of the $3$-form, 
with a consistent action and a self-duality condition. 
Here we instead derived the KSEs directly from the supersymmetry variations 
of the SO$(5)$-gauged theory in \cite{Pernici:1984xx}. 
First, we restricted the SO$(5)$-gauged variations, also in 
\cite{Gauntlett:2000ng}, to the U$(1)^2$-gauged sector. 
Yet, both \cite{Gauntlett:2000ng} and \cite{Liu:1999ai} followed the original 
explicit $3$-form formulation of \cite{Pernici:1984xx}, whereas our action
in eq.~\eqref{bosaction} incorporates its field strength,
as discussed in Section~\ref{bostheory}.
Before our higher-form mapping from potential to field strength, 
the variations we obtained from 
\cite{Gauntlett:2000ng} agree with those of \cite{Liu:1999ai}.
A consistent bosonic action for this Cartan truncation, matching the KSEs
of \cite{Liu:1999ai}, was subsequently provided in \cite{Cheung:2022ilc}
still in the $3$-form potential formulation, as discussed in
Appendix~\ref{addendumtov1}.

Thus, the relation between our formulation and that of
\cite{Pernici:1984xx,Gauntlett:2000ng} requires a joint treatment 
of the bosonic action, and of the fermion variations.
The mapping is further complicated by the absence of explicitly fixed 
conventions therein for the spacetime-orientation, as well as the
odd-dimensional irreducible Cliff$(1,6)$ and internal Cliff$(5)$ 
representations. 
Consequently, mapping the action initially only fixes absolute normalizations, 
while leaving different correlated signs possible, 
associated with the missing convention data.
The latter are then constrained by the KSEs, 
and most importantly by our subsequent additional requirement 
of a consistent truncation to the minimal Cartan-gauged theory.

Denoting the fields and couplings of \cite{Gauntlett:2000ng}
by a subscript ${\rm G}$, the combination of bosonic action and KSEs 
initially leaves the correlated sign freedom
\be
g_{\rm G}=4\kappa_{\rm G}g\, ,
\quad
m_{\rm G}=2\omega_{\rm G}m\, ,
\quad
\left(gA_I,F_I,m,G\right)
\longrightarrow
\left(
\kappa_{\rm G}\eta_{\rm G}gA_I,
\eta_{\rm G}F_I,
\omega_{\rm G}m,
\xi_{\rm G}G
\right)\, ,
\ee
where
$\kappa_{\rm G},\omega_{\rm G},\eta_{\rm G},\xi_{\rm G}=\pm1$.
We then additionally require a consistent equal-charge truncation to the 
minimal Cartan-reduced bosonic action and supersymmetric variations of 
\cite{Townsend:1983kk,Mezincescu:1984ta}, 
following \cite{Dibitetto:2017tve} which is discussed below. 
This then uniquely fixes
\be
\kappa_{\rm G}
=\omega_{\rm G}
=\eta_{\rm G}
=\xi_{\rm G}
=+1\, ,
\ee
up to a simultaneous reversal of Cartan generators and 1-form gauge fields. 
The resulting two-coupling, non-equal charges, general map reads
\begin{align}
g_{\rm G}&=4g\, ,
\qquad m_{\rm G}=2m\, ,
\qquad e^{2\lambda_I}=X_I\, ,
\nn\\
2B_{\rm G}^{12}&=A_1\, ,
\qquad 2B_{\rm G}^{34}=A_2\, ,
\qquad 2F_{\rm G}^{12}=F_1\, ,
\qquad 2F_{\rm G}^{34}=F_2\, ,
\end{align}
with the two scalars from the diagonal of $\Pi^i_A \to \Pi_I^i|_{I=i}$. 
The corresponding and most critical higher-form mapping from potential to field 
strength is given by
\be
\label{BtoGmap}
2\sqrt{3}\,m_{\rm G}\,(S_{(3)})_{G}
\longrightarrow
+\left(X_1X_2\right)^2 \star G\, ,
\ee
consistent with the linear map of the $3$-form,  
$ 2\sqrt{3}\,(S_{(3)})_{G}\to - B$, and the self-duality condition in 
eq.~\eqref{selfdual}.
We reiterate that the sign in eq.~\eqref{BtoGmap} is uniquely fixed only when
the equal-charge action and KSEs are required to match those
of the minimal Cartan-gauged theory, as detailed below.

The theory in \cite{Gauntlett:2000ng} obeys 
the maximal-supersymmetry coupling relation $g_{\rm G}=2m_{\rm G}$. 
The above mapping therefore yields the maximal-supersymmetry relation 
of couplings in our formulation as
\be
\label{maxsusy}
m=g\, .
\ee
Nevertheless, for the purpose of later constructing 
the minimal-theory truncation, we retain $g$ and $m$ as independent parameters 
in the displayed Cartan-reduced KSEs. 
This bookkeeping permits the gauge and St\"uckelberg deformations 
of the minimal theory to be mapped independently; 
the maximal-theory relation in eq.~\eqref{maxsusy} will then be imposed 
only after that mapping has been fixed.
The residual simultaneous reversal of the Cartan generators and 
$1$-form gauge fields noted above is a remaining convention freedom 
that should be distinguished from the foundational aspects anticipated 
in Section~\ref{bostheory}. 

Passing from the original formulation in terms of the $3$-form potential 
to the field-strength formulation in eq.~\eqref{bosaction},
the topological mass term cannot be treated independently of the KSEs.
Our joint treatment of the action and the Cartan-reduced fermionic variations, 
which further requires a consistent truncation to the minimal theory, 
then fixes the complete formulation of this supersymmetric system. 
In particular, it selects the higher-form mapping to field strength in 
eq.~\eqref{BtoGmap}, the coefficient $-m$ in eq.~\eqref{bosaction}, 
and the maximal-supersymmetry coupling relation in eq.~\eqref{maxsusy}.
The explicit equal-charge analysis below shows how our formulation 
uniquely admits a consistent truncation to the minimal theory.
By contrast, the coefficient $+g$ used in 
\cite{Chow:2007ts,Bobev:2023bxl,Bobev:2025xan} is presented only within a 
bosonic action, without the corresponding supersymmetry variations. 
We tested whether this alternative bosonic action could be incorporated into a 
complete supersymmetric setup. Unlike the case for our formulation, 
we found its mapping to be non-unique, as it leaves two 
possible completions of the maximal theory, which turn out to be inconsistent 
with one another. In Appendix~\ref{altcompletion}, we constructed the two 
candidate supersymmetric completions of this alternative bosonic action and 
investigated them via an integrability analysis.

In the truncated U$(1)^2$-gauged sector of \cite{Pernici:1984xx}, 
as noted in Section~\ref{bostheory}, we are left with the USp$(4)$ quadruplet 
of gravitini, $\psi^A$, and two independent spin-$\tfrac{1}{2}$ fermion 
USp$(4)$ quadruplets, $\lambda_i^A$, so that two SO$(5)_c$ vector indices can 
be specified.
We then simplify the KSEs by rewriting them in 
terms of shifted gravitini defined as
\be 
\hat{\psi}_\mu^{\alpha A} \equiv \psi_\mu^{\alpha A} 
+ \frac{1}{2} \gamma_\mu \Gamma_5\, \lambda_5^{\alpha A}\, ,
\ee
and two independent shifted spin-$\tfrac{1}{2}$ fermions defined as 
\be 
\label{shiftedfermion}
\hat{\lambda}_{1}^{\alpha A} 
\equiv \frac{3}{2} \,\Gamma_1 \lambda_1^{\alpha A} 
+ \Gamma_3 \lambda_3^{\alpha A}\, , 
\qquad
\hat{\lambda}_{2}^{\alpha A} 
\equiv \Gamma_1 \lambda_1^{\alpha A} 
+ \frac{3}{2}\,\Gamma_3 \lambda_3^{\alpha A}\, .
\ee

Using the properties of Cliff$(5)$ from Appendix~\ref{5ECliff}, 
and our conventions from Appendix~\ref{bilinearformulae} for dual 
Clifford matrices of Cliff$(1,6)$ in eq.~\eqref{dualclifford7} 
and for the Hodge dual in eq.~\eqref{stardef} for terms with 
the 3-form, the KSEs from the gravitini 
take the form:
\begin{align}
\label{Killinggravitin}
0 = \delta\hat{\psi}_\mu^{\alpha A}
&= \bigg[
\nabla_\mu
+ig\left(
A_{1\mu}\Gamma_{\hat{1}}
+A_{2\mu}\Gamma_{\hat{2}}
\right)
+\frac{i}{4}\gamma^\nu
\left(
X_1^{-1}F_{1\mu\nu}\Gamma_{\hat{1}}
+X_2^{-1}F_{2\mu\nu}\Gamma_{\hat{2}}
\right)
\nn\\*
&\qquad
+\frac{m}{2}\left(X_1X_2\right)^{-2}\gamma_\mu
+\frac{1}{4}\gamma_\mu\gamma^\nu
\left(
X_1^{-1}\partial_\nu X_1
+X_2^{-1}\partial_\nu X_2
\right)
\nn\\
&\qquad
- \frac{X_1X_2}{96}
\gamma_{\mu}^{\,\,\,\nu\lambda\sigma\rho}
G_{\nu\lambda\sigma\rho}\Gamma_5
\bigg]\epsilon^{\alpha A}\, ,
\end{align}
where we use here the Levi-Civita connection, 
$\nabla_\mu \equiv \left( 
\partial_\mu + \tfrac{1}{4} \omega_{\mu}^{ab} \gamma_{ab} \right)$. 
The KSEs from the two spin-$\tfrac{1}{2}$ fermions read:
\begin{subequations}
\label{Killingspinhalf}
\begin{align}
0 = \delta\hat{\lambda}_1^{\alpha A}
&= \bigg[
\frac{m}{2}
\left(
X_1-\left(X_1X_2\right)^{-2}
\right)
-\frac{1}{4}\gamma^\mu
\left(
\frac{3}{2}X_1^{-1}\partial_\mu X_1
+X_2^{-1}\partial_\mu X_2
\right)
\nn\\*
&\qquad
-\frac{i}{16}\gamma^{\mu\nu}
X_1^{-1}F_{1\mu\nu}\Gamma_{\hat{1}}
+ \frac{X_1X_2}{192}
\gamma^{\mu\nu\lambda\sigma}
G_{\mu\nu\lambda\sigma}\Gamma_5
\bigg]\epsilon^{\alpha A}\, ,
\\*
0 = \delta \hat{\lambda}_2^{\alpha A} 
& = \delta \hat{\lambda}_1^{\alpha A} (1 \leftrightarrow 2)\, ,
\end{align}
\end{subequations}
where the supersymmetry variations of the two spin-$\tfrac{1}{2}$ 
fermion quadruplets are related by an exchange of field labels, 
$1 \leftrightarrow 2$. 
Note that the differential KSE from the gravitini variation 
in eq.~\eqref{Killinggravitin} is symmetric under an exchange of field labels, 
$1 \leftrightarrow 2$. 
This symmetry together with the simultaneous exchange of 
bosonic and fermionic field labels in eqs.~\eqref{Killingspinhalf} 
is a symmetry of the U$(1)^2$-Cartan truncation.

Furthermore, as detailed in Appendix~\ref{5ECliff}, 
the USp$(4)$ gamma matrices which are featured in these equations, 
are defined as
\begin{subequations}
\label{SO5cartan}	
\begin{align}	
i \Gamma_{\hat{1}} & \equiv \Gamma_{12} = i \, \diag (+1,+1,-1,-1) \, ,\\*
i \Gamma_{\hat{2}} & \equiv \Gamma_{34} = i \, \diag (+1,-1,+1,-1)\, ,\\*
\Gamma_5 = -\Gamma_{1234} & = \Gamma_{\hat{1}} \Gamma_{\hat{2}} 
= \diag (+1,-1,-1,+1)\, ,
\end{align}
\end{subequations}
see eqs.~\eqref{cliff5}.
Hence we have a USp$(4)$ quadruplet of Spin$(1,6)$ Killing spinors to solve for. 
With this understanding, both spinorial indices of the Killing spinors 
will be suppressed for the remainder of this section.

As noted already, the KSE in eq.~\eqref{Killinggravitin} 
due to the gravitini variation is symmetric under an exchange of field 
labels, $1 \leftrightarrow 2$. As will be seen shortly, it is useful 
to transform the KSEs in eqs.~\eqref{Killingspinhalf}, 
due to the spin-$\tfrac{1}{2}$ fermions, into equations that also have 
definite-symmetry properties under 
the exchange of field labels. 
We can then use instead the following pair of algebraic KSEs, which are 
symmetric, and antisymmetric, respectively, under the exchange of labels, 
$1 \leftrightarrow 2$:
\begin{subequations}
\label{Killinghalfspinexchange}
\begin{align}
\label{Killingfermionsym}
0
&= \bigg[
\frac{m}{2}
\left(
X_1+X_2
-2\left(X_1X_2\right)^{-2}
\right)
-\frac{5}{8}\gamma^\mu
\left(
X_1^{-1}\partial_\mu X_1
+X_2^{-1}\partial_\mu X_2
\right)
\nn\\*
&\qquad
-\frac{i}{16}\gamma^{\mu\nu}
\left(
X_1^{-1}F_{1\mu\nu}\Gamma_{\hat{1}}
+X_2^{-1}F_{2\mu\nu}\Gamma_{\hat{2}}
\right)
+ \frac{X_1X_2}{96}
\gamma^{\mu\nu\lambda\sigma}
G_{\mu\nu\lambda\sigma}\Gamma_5
\bigg]\epsilon\, ,
\\*
\label{Killingfermionantisym}
0
&= \bigg[
\frac{m}{2}\left(X_1-X_2\right)
-\frac{1}{8}\gamma^\mu
\left(
X_1^{-1}\partial_\mu X_1
-X_2^{-1}\partial_\mu X_2
\right)
\nn\\*
&\qquad
-\frac{i}{16}\gamma^{\mu\nu}
\left(
X_1^{-1}F_{1\mu\nu}\Gamma_{\hat{1}}
-X_2^{-1}F_{2\mu\nu}\Gamma_{\hat{2}}
\right)
\bigg]\epsilon\, .
\end{align}
\end{subequations}

Considering these definite-symmetry KSEs in the equal-charge case, 
as discussed above in eqs.~\eqref{eqchargelim}, 
\eqref{eqchargescalar}, yields the following further simplified form: 
\begin{align}
\label{Killinggravitineqcharge}
\nabla_\mu\epsilon
= -\bigg[
&2igA_\mu\Gamma_+
+\frac{i}{2}X^{-1}\gamma^\nu
F_{\mu\nu}\Gamma_+
+\frac{m}{2}X^{-4}\gamma_\mu
\nn\\*
&+\frac{1}{2}X^{-1}\gamma_\mu\gamma^\nu
\partial_\nu X
- \frac{X^2}{96}
\gamma_{\mu}{}^{\nu\lambda\sigma\rho}
G_{\nu\lambda\sigma\rho}\Gamma_5
\bigg]\epsilon\, ,
\end{align}
and 
\begin{subequations}
\label{Killinghalfspinequal}
\begin{align}
\label{Killingfermionsymeq}
0
= \bigg[
&m\left(X-X^{-4}\right)
-\frac{5}{4}X^{-1}\gamma^\mu\partial_\mu X
-\frac{i}{8}X^{-1}\gamma^{\mu\nu}
F_{\mu\nu}\Gamma_+
+ \frac{X^2}{96}
\gamma^{\mu\nu\lambda\sigma}
G_{\mu\nu\lambda\sigma}\Gamma_5
\bigg]\epsilon\, ,
\\
\label{Killingfermionantisymeq}
0
&= \bigg[
\gamma^{\mu\nu}F_{\mu\nu}\Gamma_-
\bigg]\epsilon\, ,
\end{align}
\end{subequations}
where from the Cliff$(5)$ Cartan matrices in eqs.~\eqref{SO5cartan}, 
we defined $\Gamma_\pm$ matrices as 
\begin{align}
\label{cartanbasis}	
\Gamma_{+} \equiv \frac{1}{2} 
\left( \Gamma_{\hat{1}} + \Gamma_{\hat{2}} \right)
= \diag (1,0,0,-1)\, , 
\qquad
\Gamma_{-} \equiv \frac{1}{2} 
\left( \Gamma_{\hat{1}} -\Gamma_{\hat{2}} \right)
= \diag (0,1,-1,0)\, ,
\end{align}
taken as the two Cartan generators, that are orthogonal, 
and have non-trivial kernels, see Appendix~\ref{5ECliff} for further details.

The benefit of switching to equations of definite symmetry under
$1\leftrightarrow2$ now becomes apparent. 
In the equal-charge sector, the symmetric and antisymmetric equations involve 
$\Gamma_+$ and $\Gamma_-$, respectively. In particular,
eq.~\eqref{Killingfermionantisymeq} reduces to a simple $\Gamma_-$ projection, 
which acts only on the Cartan field strength $F$. 
This equation vanishes identically on $\Ker\Gamma_-$, whereas for
spinor components in $\Ker\Gamma_+$ it constrains $F$ to be primitive,
as we will show in Section~\ref{uniqueF}. 
This symmetry-definite, equal-charge form of the KSEs thus considerably 
disentangles the problem, while retaining an illuminating view of the essential 
features of the general U$(1)^2$-gauged sector.
This organization also underlies the efficient extension to independent
Cartan charges developed in Section~4 of \cite{Gaar:2026gmo}.
On the gauged active spinor branch $\Ker\Gamma_-$, using the geometric
mean of the two scalars and a symmetric scalar-weighted combination
of the Cartan field strengths isolates the additional supersymmetric
input into a single algebraic gauging correction 
and two new algebraic bilinear relations from the antisymmetric KSE.
The underlying KSE reorganization applies in both Lorentzian 
and Euclidean signatures.

At this stage, eqs.~\eqref{eqchargelim}, \eqref{eqchargescalar}
restrict only the bosonic fields within the maximal theory.
The two shifted spin-$\tfrac12$ quadruplets therefore remain
independent. It is useful to denote their symmetric and antisymmetric
combinations by
\begin{align}
\label{shiftedfermionexchange}
\hat{\lambda}_+^A
&\equiv\frac12\left(
\hat{\lambda}_1^A+\hat{\lambda}_2^A
\right),
&
\hat{\lambda}_-^A
&\equiv\frac12\left(
\hat{\lambda}_1^A-\hat{\lambda}_2^A
\right).
\end{align}
Their supersymmetry variations give, respectively, the symmetric and
antisymmetric equations in eqs.~\eqref{Killinghalfspinequal}.
Both equations must be imposed when determining supersymmetry of a
bosonic configuration in the equal-charge maximal theory.

Let us now comment on the comparison within maximal theory with 
the bosonic work of \cite{Chong:2004dy}. 
Their bosonic action is also written directly in terms of the $4$-form field 
strength, and agrees up to unfixed signs with eq.~\eqref{bosaction}, though they 
have $m$ fixed to $g$ to begin with. 
Its bosonic mapping may be written as
\begin{align}
\label{Chongmap}
X_{I{\rm Ch}}&=X_I\, ,
&
A_{\rm Ch}^{I}&=\eta_I A_I\, ,
&
F_{\rm Ch}^{I}&=\eta_I F_I\, ,
\nn\\
g_{\rm Ch}&=\omega_{\rm Ch}g\, ,
&
A_{(3){\rm Ch}}&=\omega_{\rm Ch}\eta_1\eta_2 B\, ,
&
F_{(4){\rm Ch}}&=\omega_{\rm Ch}\eta_1\eta_2 G\, ,
\end{align}
where $\eta_I=\pm1$, and $\omega_{\rm Ch}=\pm1$ records the relative
orientation convention. In the equal-charge sector
$\eta_1=\eta_2$, so the common sign of the two Abelian fields remains
a charge-convention choice.
However, a direct comparison with the single algebraic supersymmetric variation 
displayed in eq.~(19) of \cite{Chong:2004dy} cannot be established. 
That equation suppressed the internal Cliff$(5)$ matrices, and did not identify 
either the USp$(4)$ component or the linear combination of maximal 
spin-$\tfrac12$ fermions whose variation is displayed. 
Moreover, \cite{Chong:2004dy} used this equation only as 
a bosonic eigenvalue test, and selected an unphysical branch, 
which differs from the regular supersymmetric black-hole branch that was 
identified only subsequently in \cite{Cvetic:2005zi}. 
We therefore retain here only the bosonic mapping in eq.~\eqref{Chongmap}.

We now proceed to connect our equal-charge maximal Cartan-gauged  truncation 
with the U$(1)$-Cartan truncation of the minimal SU$(2)$-gauged theory of 
\cite{Dibitetto:2017tve}. 
This means setting in the latter:
\be
\label{Cartanrestrict}
A_{\rm D}^{1}=A_{\rm D}^{2}=0\, ,
\qquad
F_{\rm D}^{3}=\dd A_{\rm D}^{3}\, .
\ee
The minimal-theory spinor $\zeta$ is then obtained by restricting 
the maximal-theory spinor:
\begin{align}
\label{maxtominspinor}	
\epsilon\in\Ker\Gamma_-
&=\operatorname{span}\{\epsilon^1,\epsilon^4\}\, ,
\qquad
\Gamma_+\epsilon^{1,4}=\pm\epsilon^{1,4}\, ,
\nn\\
\zeta^1\leftrightarrow\epsilon^1\, ,&
\quad
\zeta^2\leftrightarrow\epsilon^4
\, ,
\qquad \qquad
\left.\Gamma_5\right|_{\Ker\Gamma_-}=1\, .
\end{align}
The complete fermionic truncation to the minimal theory additionally
imposes
\begin{subequations}
\label{maxtominfermions}
\begin{align}
\hat{\lambda}^{1,4}& 
\equiv\hat{\lambda}_1^{1,4}=\hat{\lambda}_2^{1,4}\, ,
\\
\label{maxtominremove}
\hat{\psi}_\mu^{2,3}
&=\hat{\lambda}_1^{2,3}
=\hat{\lambda}_2^{2,3}
=0\, .
\end{align}
\end{subequations}
The two retained components of the maximal USp$(4)$ quadruplet
therefore form the minimal SM doublet, with their opposite U$(1)$ charges 
represented by the two non-zero eigenvalues of $\Gamma_+$.
The symmetric fermion $\hat{\lambda}_+$ is retained, whereas
$\hat{\lambda}_-$ and the complementary USp$(4)$ components are
removed by eqs.~\eqref{maxtominfermions}. The variation of
$\hat{\lambda}_-$ in eq.~\eqref{Killingfermionantisymeq}
vanishes identically on the retained spinor sector because
$\Gamma_-\epsilon=0$. It therefore imposes no additional constraint
on $F$ in the minimal theory, whose spin-$\tfrac12$ KSE is instead 
the symmetric equation \eqref{Killingfermionsymeq}.

The complementary maximal spinors are not retained in the
minimal-theory fermionic truncation. 
When $F=0$, however, they obey the alternative identification
\begin{align}
\label{maxtominspinorcomplement}
\epsilon\in\Ker\Gamma_+
&=\operatorname{span}\{\epsilon^2,\epsilon^3\}\, ,
\qquad
\Gamma_-\epsilon^{2,3}=\pm\epsilon^{2,3}\, ,
\nn\\
\zeta^1\leftrightarrow\epsilon^2\, ,&
\quad
\zeta^2\leftrightarrow\epsilon^3\, ,
\qquad \qquad
\left.\Gamma_5\right|_{\Ker\Gamma_+}
=-1\, ,
\end{align}
after choosing $A=0$ locally and reversing the maximal $3$-form
potential, $B\rightarrow-B$.
The map between minimal and maximal theories, 
which we will present below in eq.~\eqref{Dibitettomap}, then applies. 
This complementary spinor branch describes additional supersymmetry of the 
$F=0$ configuration embedded in the maximal theory.

We thus determine the maximal-to-minimal truncation to \cite{Dibitetto:2017tve} 
first from the bosonic action, and then from the KSEs, 
keeping the two couplings of minimal theory independent to begin with. %
Again, the action fixes the displayed magnitudes and sign correlations, but
does not fix the independent Abelian charge sign $\eta$, or
$\omega_{\rm D}$ that records the relative spacetime-orientation
convention, which is not specified in \cite{Dibitetto:2017tve}.
At the action level, the most general map to \cite{Dibitetto:2017tve} 
reads
\begin{align}
\label{Dibitettobosmap}
X_{\rm D}=X\, ,
\qquad
A_{\rm D}^{3}&=\eta\sqrt{2}A\, ,
\quad
F_{\rm D}^{3}=\eta\sqrt{2}F\, ,
\nn\\
g_{\rm D}=2\sqrt{2}\omega_{\rm D}g\, ,
\quad
h_{\rm D}=\omega_{\rm D}m\, ,
\qquad
B_{\rm D}&=\omega_{\rm D}B\, ,
\quad
F_{(4){\rm D}}=\omega_{\rm D}G\, ,
\qquad
\eta,\omega_{\rm D}=\pm1\, .
\end{align}

For the mapping of KSEs, we identify the minimal-theory Cartan generator as
\be
\sigma^3=q\Gamma_+\, ,
\qquad q=\pm1\, ,
\ee
where $q$ records the remaining overall sign in the identification of
the two Cartan generators. 
From the KSEs, we then fix 
\be
\omega_{\rm D}=-1\, ,
\qquad
\eta q=-1\, .
\ee
Choosing the convention $q=+1$ then gives the final map from 
\cite{Dibitetto:2017tve} to our work as
\begin{align}
\label{Dibitettomap}
X_{\rm D}=X\, ,
\qquad
A_{\rm D}^{3}&=-\sqrt{2}A\, ,
\quad
F_{\rm D}^{3}=-\sqrt{2}F\, ,
\nn\\
g_{\rm D}=-2\sqrt{2}g\, ,
\quad
h_{\rm D}=-m\, ,
\qquad
B_{\rm D}&=-B\, ,
\quad
F_{(4){\rm D}}=-G\, .
\end{align}
The alternative choice $q=-1$, $\eta=+1$ is equivalent under the simultaneous 
reversal of the minimal Cartan generator and 1-form gauge potential.

Using eq.~\eqref{Dibitettomap} and $\sigma^3=\Gamma_+$, the gravitino variation 
in \cite{Dibitetto:2017tve} agrees with eq.~\eqref{Killinggravitineqcharge} on 
$\Ker\Gamma_-$. 
For the spin-$\tfrac{1}{2}$ variation, the gauging contribution printed in
eq.~(2.7) of \cite{Dibitetto:2017tve} contains a redundant factor
of $-1/\sqrt{5}$. Removing it gives the correct coefficient
confirmed independently by, e.g., eq.~(2.11) of \cite{Conti:2026row}.
With this correction, keeping $m$ and $g$ independent, 
the scalar shifts in the KSEs compare as
\begin{align}
\left.X\partial_Xf_{\rm D}\right|_{\eqref{Dibitettomap}}
&=2mX^{-4}-2gX%
\quad \longleftrightarrow \quad
2m\left(X^{-4}-X\right)\, .
\end{align}
They coincide precisely upon imposing the maximal-supersymmetry coupling relation 
in eq.~\eqref{maxsusy}. The scalar potentials of the two actions provide
a similar consistency check.
Finally, eqs.~\eqref{maxsusy}, \eqref{Dibitettomap}, then imply the relation of 
couplings in the minimal theory as
\be
\label{mintomax}
h_{\rm D}=\frac{g_{\rm D}}{2\sqrt{2}}\, .
\ee

Thus, the bosonic actions, the gravitino KSE, and the retained
spin-$\tfrac12$ KSE select a consistent maximal-to-minimal truncation.
This completes the consistency requirement announced in Section~\ref{bostheory}, 
and confirms the formulation in eq.~\eqref{bosaction} as part of the complete 
supersymmetric system.

\section{Spinor Bilinears and G-Structure}
\label{bilinearsandstructure}

\subsection{Spinor Bilinears}
\label{bilinear}

Spinor bilinears are constructed from Killing spinors with their 
Majorana conjugate.
We fix a single commuting Killing spinor throughout the analysis, 
and construct all bilinears from it. 
In SM notation, the two members of its SM pair are related by the reality 
condition in eq.~\eqref{sympMajorana}, 
and do not represent independent Killing spinors.
We will define $r$-form bilinears with matrices of the Clifford 
algebra Cliff$(1,6)$ up to rank $r = \lfloor d/2\rfloor=3$, 
where $d$ is the spacetime dimension, that take the generic form:
\be
\gamma_{(r)} \equiv \frac{1}{r!} \gamma_{\mu_1\cdots\mu_r} 
\dd x^{\mu_1} \wedge \cdots \dd x^{\mu_r}\, ,  \qquad 0\le r \le3\, ,
\ee
and with the following diagonal SO$(5)$ Clifford matrices up to rank 
$s = \lfloor 5/2\rfloor=2$:
\be 
\label{Gammas}
\Gamma_{a_1\cdots a_s} 
= \{ \mathds{1}_4, \Gamma_5, \Gamma_{12}, \Gamma_{34} \} 
\subset M_4(\mathbb{C})\, ,
\ee 
as detailed in eq.~\eqref{SO5cartan} and Appendix~\ref{5ECliff}.

In order to keep the reality of all bilinears on par, 
we take instead the following set:
\be
\label{realGammas}
\Gamma_I = \{\Gamma_{\pm},P_{\pm}\} \subset M_4(\mathbb{R})\, ,
\ee
where $\Gamma_+$ and $\Gamma_-$ are the two real Cartan generators,
symmetric and antisymmetric under the exchange of labels, $1 \leftrightarrow 2$, 
respectively, that we defined in eq.~\eqref{cartanbasis}, and $P_\pm$ 
are the two $\tfrac{1}{2}$-supersymmetry projections that we defined as 
\begin{align}
\label{susyprojections}	
P_{\pm} \equiv  \frac{1}{2} \left( \mathds{1}_4 \pm \Gamma_{5} \right)
= \Gamma_{\pm}^2\, , 
\end{align}
and explicitly read
\be
P_+ = \diag (1,0,0,1) \, , 
\qquad 
P_- = \diag (0,1,1,0)\, ,
\ee
see Appendix~\ref{5ECliff}. 
Thus, it is actually sufficient to consider the two basis generators 
$\Gamma_{\pm}$.

First, let us consider a generic term in the KSEs. 
Since the Cartan matrices $\Gamma_I$ are diagonal, and recalling
eq.~\eqref{spinorasdirectproduct}, such a term can be written as
\be 
\label{minmaxdecompose}
f_{\mu_1\cdots \mu_r} \gamma^{\,\beta}_{_{(r)}\,\alpha} 
\Gamma^{\,B}_{I\,A} \left( \chi^\alpha \otimes \chi^A \right) = 
f_{\mu_1\cdots \mu_r} 
\left( \gamma^{\,\beta}_{_{(r)}\,\alpha}\, \chi^\alpha \right)
\Gamma^{\,B}_{I} \chi^B\, , 
\ee
where $f_{\mu_1\cdots\mu_r}$ is some generic tensor that contains the 
fields, $\alpha,\beta$, are the Spin$(1,6)$ spinor indices,
$A,B$, are the USp$(4)$ spinor indices, and the notation on 
$\Gamma_I$ in the last equality denotes its $B$th component on 
the diagonal, with no summation on the USp$(4)$ indices. 
Thus the USp$(4)$ dependence reduces to four decoupled scalar
eigenvalues, labelled by $B$.
Let us then contract such a generic term with the proper 
Spin$(1,6)$ spinor from the left, 
so that a generic building block can be written as
\be 
\left (\overline{\chi^j}_{\beta}
\,\gamma^{\,\beta}_{_{(r)}\,\alpha}\, \left(\chi^i\right)^\alpha \right)
\Gamma^{\,A}_{I} \chi^A 
\equiv b_{(r)} \Gamma^{\,A}_{I} \chi^A\, , 
\ee
where we denoted by $b_{(r)}$ what is just a standard $r$-form bilinear 
in the minimal seven-dimensional theory. 
Therefore, we can now focus on the generic seven-dimensional bilinears 
$b_{(r)}$, where in the following we suppress again the spinorial indices. 

In order to determine which seven-dimensional bilinears are non-vanishing,
and analyse them, some bilinear formulae given
in Appendix~\ref{bilinearformulae} are required,
starting with the Majorana flip relation
in eq.~\eqref{Majoranaflip}.
Next, in order to determine the reality properties of these bilinears,
we use the formula for complex conjugation of spinor bilinears with
an arbitrary matrix $M$, which is given in eq.~\eqref{Lorentzcc}.
The bilinear forms of Cliff$(1,6)$ are conveniently packaged and
manipulated when the spinors are considered in the SM form in
eq.~\eqref{sympMajorana}. An analysis similar to that in
\cite{Gauntlett:2002nw,Gauntlett:2003fk},
see also \cite{Cariglia:2004qi},
and \cite{Ortin:2015hya} for a comprehensive pedagogic overview,
then yields the following non-vanishing bilinears:
\begin{subequations}
\label{bilineardef}
\begin{align}
\label{scalarbil}
s^{(ij)} & \equiv \overline{\chi^i} \,\chi^j\, ,\\*
\epsilon^{ij} K_\mu 
& \equiv i\,\overline{\chi^i}\,\gamma_\mu\, \chi^j\, ,\\*
\epsilon^{ij} T_{\mu\nu} &
\equiv i\, \overline{\chi^i} \, \gamma_{\mu\nu}	\, \chi^j\, ,\\*
\label{3formsdef}
u^{(ij)}_{\mu\nu\rho} 
& \equiv \overline{\chi^i} \,\gamma_{\mu\nu\rho}\, \chi^j\, .
\end{align}
\end{subequations}
Here $i,j=1,2$, are the SM pair indices. 
Since the spinors $\chi^i$ form an SU$(2)$ doublet, with indices raised 
and lowered by the invariant antisymmetric tensor $\epsilon_{ij}$, 
bilinears carrying two SM indices decompose as $2\otimes 2 = 1 \oplus 3$, 
where the antisymmetric part in the SM indices is an SU$(2)$ singlet,
while the symmetric part is an SU$(2)$ triplet. 
Thus the 1- and 2-form bilinears $K$ and $T$, which are written 
with an explicit factor of $\epsilon^{ij}$, are real SU$(2)$ singlet forms. 
On the other hand, the scalars $s^{(ij)}=s^{ij}$ and 3-forms
$u^{(ij)}_{\mu\nu\rho}=u^{ij}_{\mu\nu\rho}$ are symmetric in the SM
indices, and therefore form SU$(2)$ triplets.

To make this triplet structure explicit, we lower one SM index with
$\epsilon_{ij}$,
\be
s^i{}_{j}\equiv \epsilon_{jk}s^{ik}\, ,
\qquad
(u_{\mu\nu\rho})^i{}_{j}
\equiv \epsilon_{jk}(u_{\mu\nu\rho})^{ik}\, .
\ee
The resulting traceless $2\times2$ matrices can then be expanded in the
basis of Pauli matrices 
\begin{subequations}
\begin{align}
\label{matrixscalars}
s^{i}{}_{j} & \equiv \vec{S} \cdot (\vec{\sigma})^{i}{}_{j}\, ,\\
\label{matrix3forms}
(u_{\mu\nu\rho})^{i}{}_{j} &
\equiv \vec{U}_{\mu\nu\rho} \cdot (\vec{\sigma})^{i}{}_{j}\, ,
\end{align}
\end{subequations}
where $\sigma_a/2$, $a=1,2,3$, are the generators of SU$(2)$ with the
Pauli matrices $\sigma_a$ defined in eq.~\eqref{pauli}, and $S_a$,
$\left(U_a\right)_{\mu\nu\rho}$, are real bilinears. 
The scalar triplet then explicitly reads
\be
s^i{}_{j}
=
\begin{pmatrix}
s^{12}\, & \,-s^{11}\\
s^{22}\, & \,-s^{12}
\end{pmatrix}
=
\begin{pmatrix}
S_3\, & \,S_1-iS_2\\
S_1+iS_2\, & \,-S_3
\end{pmatrix}.
\ee
It is useful to record the relations of general to real scalar bilinears
\begin{subequations}
\label{bilinscalars}
\begin{align}
\label{comscalarstoreal}
s^{11} & = - (s^{22})^* = -S_1 + i S_2 \, ,\\
s^{12} & = s^{21} = S_3\, ,
\end{align}
\end{subequations}
with similar relations for the 3-forms, namely
\begin{subequations}
\begin{align}
\label{c3bilineartoreal}
u^{11}_{\mu\nu\rho} & = - (u^{22}_{\mu\nu\rho})^*
= -\left(U_1\right)_{\mu\nu\rho} + i \left(U_2\right)_{\mu\nu\rho} \, ,\\
u^{12}_{\mu\nu\rho} & = u^{21}_{\mu\nu\rho}
= \left(U_3\right)_{\mu\nu\rho}\, .
\end{align}
\end{subequations}
Thus the components $s^{11}$ and $u^{11}_{\mu\nu\rho}$ are not 
independent complex bilinears. They are complex combinations of the real SU$(2)$ 
triplet components $S_1,S_2$, and 
$\left(U_1\right)_{\mu\nu\rho},\left(U_2\right)_{\mu\nu\rho}$, while $s^{12}$ 
and $u^{12}_{\mu\nu\rho}$ give the third real components of the triplets.

In order to get geometric bilinear relations, we apply the Fierz 
identity in $d=7$ in eq.~\eqref{Fierz} first with 
$\bar{\epsilon}_1=\overline{\chi^i}$, 
$\bar{\epsilon}_3=\overline{\chi^k}$,
$\epsilon_2=\chi^l$, $\epsilon_4=\chi^j$, 
and in turn switching to 
$\bar{\epsilon}_1=\overline{\chi^i} \gamma^\mu$, 
$\epsilon_4=\gamma_\mu\chi^j$, and then to
$\bar{\epsilon}_1=\overline{\chi^i} \gamma^{\mu\nu}$, 
$\epsilon_4=\gamma_{\mu\nu}\chi^j$.
This yields the following relations among the bilinear norms:
\begin{align}
\label{Fierznorms}	
\sum_{a=1}^{3} S_a^2 & = - K^2
= \frac{1}{6} T^2 = \frac{1}{30} \sum_{a=1}^{3} U_a^2  \, , %
\end{align}
where we see that the Killing vector must 
be either timelike or null, and the scalar contractions of 
the 2- and 3-form bilinears are later used to normalize 
the defining structure forms on the base manifold, 
see Section~\ref{gstructure}. 
In the above applications of the Fierz identity, we also needed the identity in 
eq.~\eqref{gammasandwich}, and to derive its extension in 
eq.~\eqref{gammadblsandwich}. 

Next, applying the Fierz identity with 
$\bar{\epsilon}_1=\overline{\chi^i}$, 
$\bar{\epsilon}_3=\overline{\chi^k}$,
$\epsilon_2=\chi^l$, $\epsilon_4=\gamma_\mu\chi^j$, 
and in turn switching to $\epsilon_4=\gamma_{\mu\nu}\chi^j$, 
and then to
$\bar{\epsilon}_1=\overline{\chi^i} \gamma_{\mu\nu\rho}$, 
$\epsilon_4=\gamma^{\mu}\chi^j$,
yields the following geometric relations of the interior products with the 
Killing vector:
\begin{subequations}
\label{Fierzintprods}	
\begin{align}
\label{Fierzintprod1}	
\iota_K T & = 0\, ,\\      
\label{Fierzintprod2}
\forall a\, ,\quad \iota_K U_a & = S_a T\, .  
\end{align}
\end{subequations}
In the above applications of the Fierz identity we used the duality 
relation between Clifford matrices of ranks $r$ and $d-r$ when the dimension is 
odd, which in $d=7$ dimensions is specified in eq.~\eqref{dualclifford7}.
We also used the interior product of a vector $V$ 
and a form $\omega$, as defined in eq.~\eqref{intprod}, and in what follows 
also the more general contraction with 
a tensor $T$, as defined in eq.~\eqref{formcontract}.

\subsection{G-Structure}
\label{gstructure}

Let us consider the G-structure associated with the spinor bilinears 
of the seven-dimensional theory, when the Killing vector is 
timelike. 
We choose the time coordinate $t$ accordingly, such that 
\be
\label{Killingtime}
K = \frac{\partial}{\partial t }\, . 
\ee
The general metric admitting such a Killing vector reads:
\be
\label{genmetric}
ds^2 = -h^2 ( dt + \omega )^2 + g_{ij} dx^idx^j\, ,
\ee
where $h$, $\omega$ and $g_{ij}$ are independent of $t$,  due to the 
isometry, and $\omega$ is the $1$-form connection of the timelike fibration.
This Killing vector need not exhaust the isometries of a supersymmetric
solution; rather, it is the minimal isometry captured by the bilinear
of the chosen Killing spinor.
Using the same notation for its metric-dual $1$-form, we have
\be 
\label{Killingdual}
K = - h^2 ( dt + \omega ) \equiv - h e^0\, . 
\ee

The chirality matrix on the Riemannian six-dimensional base space 
with the metric $g_{ij}$ is defined by
\be 
\label{chiral6}
\gamma_* \equiv \gamma_1 \ldots \gamma_6 = \gamma_0 
= h^{-1} K^{\mu} \gamma_\mu\, ,
\ee
where in the second equality we used eqs.~\eqref{Cliff1+6}, 
\eqref{7dchirality}. 
With this convention for the 6d chirality 
$\gamma_*^2=-\mathds{1}_8$,
so its eigenvalues are $\pm i$. 
Thus we define positive chirality by
\be
\label{chiraleigen}
\gamma_*\psi = i \psi\, .
\ee
The Fierz identity in eq.~\eqref{Fierznorms} then implies 
the following projection by the 6d chirality matrix of the 
seven-dimensional SM spinor:
\be
\label{chiralproj}
\gamma_*\chi^i = - i h^{-1} s^{i}{}_{j} \chi^j\, ,
\ee
which is obtained by hitting with $\overline{\chi^k}$ from the left, 
using the 0- and 1-form bilinear definitions, 
and recalling the product rule of the Pauli matrices in 
eq.~\eqref{Pauliprod}.

Let us consider then the G-structure on the six-dimensional base 
manifold. 
Note that on this base we switch to considering 6-dimensional 
spinors within the corresponding Euclidean Clifford algebra, 
Cliff$(6)$, as detailed in Appendix~\ref{6ECliff}.
Now, let us entertain the assumption that there are two orthogonal 
unit-norm spinors of positive chirality, $\psi$, $\hat{\psi}$, 
namely 
\be
\label{funnyassumption}
\psi^\dag\hat{\psi} 
=0\, .
\ee
As in \cite{Gauntlett:2004zh}, we can then always take two 
non-orthogonal unit-norm spinors of positive chirality 
to be $\psi$, and 
\be 
\label{tildepsi}
\tilde{\psi} \equiv a\psi+b\hat{\psi}\, , 
\qquad |a|^2+|b|^2=1\, .
\ee
Therefore, in seven dimensions a general spinor of no definite chirality, 
such as the SM Killing spinors, can be written as 
\be
\label{7to6ansatz}
\chi^1 
= \chi^1_+ + \chi^1_- = c_1 \psi + c_2 \left( a\psi+b\hat{\psi} 
\right)^*\, ,
\ee
which is a chiral decomposition over the six-dimensional base 
\cite{Gauntlett:2004zh}. 
$\psi$ and $\hat{\psi}$ are uniquely defined up to phases, which are 
fixed such that $c_{1}$, $c_{2}$, are taken to be real.

Using the SM condition in eq.~\eqref{sympMajorana} requires some
care in distinguishing the seven-dimensional Lorentzian charge
conjugation from its six-dimensional Euclidean counterpart, denoting them 
by $c_{(7)}$ and $c_{(6)}$, respectively.
The conventions in Appendices~\ref{7Cliff}, \ref{6ECliff}, together
with $C=\mathds{1}_8$ and $\gamma^0=-\gamma_*$, give $B=-i\gamma^0=i\gamma_*$.
Consequently, using the chiral projections of $\psi$ and
$\hat\psi$, we find
\be
\psi^{c_{(7)}}=\psi^*\, ,
\qquad
(\psi^*)^{c_{(7)}}=-\psi\, ,
\ee
with identical relations for $\hat\psi$. By contrast, Euclidean
charge conjugation on the base satisfies
\be
\psi^{c_{(6)}}=\psi^*\, ,
\qquad
(\psi^*)^{c_{(6)}}=\psi\, ,
\ee
with analogous relations for $\hat\psi$. The relative minus sign
in the second seven-dimensional relation is therefore essential
when constructing the SM counterpart of $\chi^1$.

Using $\chi^2=-(\chi^1)^{c_{(7)}}$, which follows from
eq.~\eqref{sympMajorana}, we obtain the SM counterpart spinor as
\be
\label{smcompanion}
\chi^2
= -\left(\chi^1\right)^{c_{(7)}}
= c_2\left(a\psi+b\hat\psi\right)-c_1\psi^*\, .
\ee
The scalar bilinears are seven-dimensional Majorana bilinears,
as defined in eq.~\eqref{scalarbil}, and must therefore be
evaluated using the seven-dimensional charge conjugation. In
particular, the unit Hermitian normalization
$\psi^\dagger\psi=1$ gives
\be
\overline{\psi^{c_{(7)}}}\psi
=\psi^\dagger\psi=1\, ,
\qquad
\overline{(\psi^*)^{c_{(7)}}}\psi^*
=-\psi^T\psi^*=-1\, ,
\ee
whereas the corresponding Euclidean charge-conjugation norms of
both $\psi$ and $\psi^*$ are positive. Thus the norm assigned to
$\psi^*$ in the seven-dimensional bilinears should not be interchanged
with its six-dimensional Euclidean counterpart.
Equivalently, using $C=\mathds{1}_8$, the contractions required
below are
\be
\psi^T\psi^*=1\, ,
\qquad
\tilde\psi^{*T}\tilde\psi=1\, ,
\qquad
\psi^T\tilde\psi^*=a^*\, ,
\ee
while scalar contractions between spinors of the same chirality vanish. 

Using the orthogonality condition in
eq.~\eqref{funnyassumption}, together with
eqs.~\eqref{matrixscalars}, \eqref{bilinscalars}, we then obtain
\begin{subequations}
\label{scalarsfromspinors}
\begin{align}
s^1{}_{1} & = s^{12}=c_2^2-c_1^2\, ,\\
\label{s11to6dunsolved}
s^1{}_{2} & =-s^{11}=-2c_1c_2a^*\, .
\end{align}
\end{subequations}
Using the projection in eq.~\eqref{chiralproj}, and the chirality 
conventions in eqs.~\eqref{chiraleigen}, \eqref{chiralfix}, 
we then find the following system of equations to solve for 
the coefficients $c_1$, $c_2$, $a$ and $b$:
\begin{subequations}
\label{syseq}	
\begin{align}	
c_1 \left( c_1^2 - c_2^2 \left( 1 - 2|a|^2 \right) - h \right) 
& = 0\, ,\\*
c_2 a \left( c_1^2 + c_2^2 - h \right) & = 0\, ,\\*
\label{syseq3}
c_1 c_2 a b & = 0\, ,\\*
c_2 b \left( c_1^2 - c_2^2 + h \right) & = 0\, . 
\end{align}	
\end{subequations}
A critical equation here is that in eq.~\eqref{syseq3}, as it is 
already telling that at least one of the coefficients in our general 
ansatz in eq.~\eqref{7to6ansatz} must vanish, and thus a Killing 
spinor in seven dimensions is actually more constrained than what we 
assumed. 

As we will show in Section~\ref{Killingresults}, 
an important class of solutions constrains the complex scalar bilinear to vanish,
\be
\label{cscalarto6d}
s^{11}=S_1=S_2=0\, .
\ee
Within the general system of eqs.~\eqref{syseq}, we will then show 
that this implies
\be
\label{reallygeneral}
c_1c_2=0\, , 
\ee
which directly satisfies eq.~\eqref{s11to6dunsolved}.
As we will show in Section~\ref{classifyingbps}, this result 
is consequential for the geometry of the base manifold, 
and hence for the permitted supersymmetric geometries. 
Nevertheless, here we analyse the general system of eqs.~\eqref{syseq}, 
without assuming eq.~\eqref{cscalarto6d} holds.

Assume first that $c_1\neq0$. Equation~\eqref{syseq3}
then separates the solutions into the branches $c_2=0$,
$a=0$, or $b=0$. On the first branch, the first equation
in eqs.~\eqref{syseq} gives
\be
c_2=0\, ,\qquad c_1^2=h\, .
\ee
If instead $a=0$, then $|b|=1$. Assuming also $c_2\neq0$,
the first and fourth equations in eqs.~\eqref{syseq} would give,
respectively, $h=c_1^2-c_2^2$, $h=c_2^2-c_1^2$,
and hence $h=0$, in contradiction with the timelike case.
Thus the branch $a=0$ also necessarily reduces to
$c_2=0$, $c_1^2=h$. The remaining possibility $b=0$,
with both $c_1$ and $c_2$ non-vanishing, will be considered
below.

For $c_1^2=h$, $c_2=0$, we then get for the SM Killing spinor in 
eqs.~\eqref{7to6ansatz}, \eqref{smcompanion}, the following form:
\be
\label{simpleisdeep}
\chi^1 = \sqrt{h} \psi\, , \qquad \chi^2 = - \sqrt{h}\psi^*\, , 
\ee
where an optional flip in the choice of root sign on both counterparts 
can be absorbed into an overall phase redefinition. 
On this branch the Killing spinor depends on only one
positive-chirality spinor and therefore defines an SU$(3)$
structure, as defined in Appendix~\ref{gstructure6d}. 
The second orthogonal chiral spinor assumed in
eq.~\eqref{funnyassumption} does not enter the Killing spinor.
Notably, this branch also satisfies eqs.~\eqref{reallygeneral}, 
\eqref{cscalarto6d}. 

Assuming instead that $c_1=0$, which also satisfies 
eqs.~\eqref{reallygeneral}, \eqref{cscalarto6d}, 
necessarily leads to $c_2^2=h$, and yields a similar solution 
for the SM Killing spinor of the form:
\be
\label{simpleisdeep2}
\chi^1 = \sqrt{h} (\tilde{\psi})^*\, , 
\qquad 
\chi^2= \sqrt{h} \tilde{\psi}\, , 
\ee
where here too an optional sign flip on both counterparts can 
be absorbed into an overall phase redefinition. 
On this branch the Killing spinor also depends only on the single
positive-chirality spinor $\tilde{\psi}$ from eq.~\eqref{tildepsi}, 
giving rise to an SU$(3)$ structure. 

Finally, taking $b=0$, assuming $c_1\neq0$ and $c_2\neq0$, 
which allows us to evade eq.~\eqref{cscalarto6d}, 
leads to a solution where the Killing spinor has non-vanishing 
coefficients $c_1$, $c_2$, $a$, 
which can be parametrized by two angles specifying the internal
SU$(2)$ orientation of the SM doublet, or equivalently the direction
of its normalized scalar-bilinear triplet on $\mathbb S^2$, as we will show now. 
Yet, on this branch the Killing spinor also depends only on a single
positive-chirality spinor, $\tilde{\psi} = a \psi$. 
In this case $|a|=1$, and we get from eqs.~\eqref{syseq} that $c_1^2+c_2^2 = h$, 
which can be parametrized as 
\be
c_1 = \sqrt{h} \cos \theta\, , 
\qquad
c_2 = \sqrt{h} \sin \theta\, ,  
\qquad 
a = \exp (-i\varphi)\, ,
\ee
where here too an optional sign flip on both counterparts can be 
absorbed into an overall phase redefinition. 
The angles $\theta$ and $\varphi$ are real $0$-forms on the
six-dimensional base, valued in the internal $\mathbb S^2$ of
SM orientations; they are not angular coordinates of the base.
Thus the domain of parameters in this case is
$\theta\in(0,\pi)$, $\theta \neq \pi/2$, and $\varphi\in[0,2\pi)$, 
where $\theta = 0,\pi/2$, are excluded due to 
the assumption on $c_1$, $c_2$.
The two limiting values $\theta=0,\pi/2$ correspond to the two cases 
above with $c_2=0$, or $c_1=0$, respectively.

From our ansatz for the SM spinor in eqs.~\eqref{7to6ansatz}, 
\eqref{smcompanion}, we then get in this case
\begin{subequations}
\label{fullsmspinors}	
\begin{align}
\chi^1 & = 
\sqrt{h} \left( 
\cos \theta \, \psi + e^{i\varphi} \sin \theta \psi^* \right)\, , 
\\*
\chi^2 & = 
\sqrt{h} \left( 
e^{-i\varphi} \sin \theta \psi - \cos \theta \,\psi^* \right)\, ,
\end{align}	
\end{subequations}
and we see that indeed the simple solutions 
in eqs.~\eqref{simpleisdeep}, \eqref{simpleisdeep2}, 
are obtained when $\theta=0,\pi/2$, up to an overall phase, 
and then again $c_2=0$, or $c_1=0$, respectively.
It is easy to see that this case is similar to those 
in eqs.~\eqref{simpleisdeep}, \eqref{simpleisdeep2}, 
as the general SM doublet can be transformed into either of the
simple forms by an SU$(2)$ rotation, which induces the
corresponding SO$(3)$ rotation of the scalar and 3-form
bilinear triplets.

The analysis above shows that this Killing spinor always depends on a single
positive-chirality spinor on the base, and hence defines a single SU$(3)$
structure.
This does not exclude solutions with enhanced supersymmetry. An
additional independent Killing spinor may define a different
timelike bilinear Killing vector and hence a different timelike
decomposition and base geometry. The underlying bosonic fields
remain fixed, while their projected components relative to the chosen
spinor-induced structure may differ.
Using eqs.~\eqref{fullsmspinors} for the generic SM Killing spinor in
terms of the six-dimensional chiral spinors, one can now relate all
of its seven-dimensional bilinears to the defining forms of an
SU$(3)$ structure, as defined in Appendix~\ref{gstructure6d}.

First, for the scalar bilinears in eqs.~\eqref{scalarsfromspinors} we get
\begin{subequations}
\label{genericscalars}	
\begin{align}
\label{complicates11}	
s^{11} %
& = h \sin 2\theta e^{i\varphi} 
\equiv h s p \, ,
\qquad s\equiv \sin 2 \theta\, ,\quad p\equiv e^{i\varphi}\, ,\\ 
\label{complicateS3}
s^{12} %
& = - h \cos 2 \theta 
\equiv - h c\, ,
\qquad c\equiv \cos 2 \theta\, , 
\end{align}
\end{subequations}
and we see that indeed eq.~\eqref{cscalarto6d} 
is obtained when $\theta=0,\pi/2$.
Therefore, the two additional seven-dimensional real scalar bilinears, 
$S_1$, $S_2$, are encoded by the two additional angles
$\theta$, $\varphi$, defined on the six-dimensional base,
which parametrize the internal SU$(2)$ orientation of the SM doublet,
or equivalently of its scalar-bilinear triplet.
It is easy to verify consistency with taking the SM Killing spinors 
in eqs.~\eqref{simpleisdeep}, \eqref{simpleisdeep2}, respectively,
which fixes the complex scalar bilinear to eq.~\eqref{cscalarto6d}, 
and the real scalar bilinear to
\be
\label{realscalarto6d}
S_3 = \mp h\, . 
\ee

The values $\theta=\pi/4, 3\pi/4 $ are special points of the general
parametrisation. At these points $c=\cos2\theta=s^{12}=S_3=0$, 
so that the seven-dimensional real scalar bilinear vanishes.
This case is also distinguished as the SM pair can be reorganized into spinors 
that are individually Majorana on the base,
which arises in 6d/7d Euclidean Clifford algebras, as we noted 
in Appendix~\ref{6ECliff}.
For example, for $\theta=\pi/4$ and $\varphi=0$, the SM pair in 
eqs.~\eqref{fullsmspinors} becomes
\be
\chi^1
=
\sqrt{\frac{h}{2}}\left(\psi+\psi^*\right)\, ,
\qquad
\chi^2
=
\sqrt{\frac{h}{2}}\left(\psi-\psi^*\right)\, .
\ee
Since charge conjugation on the six-dimensional base simply 
exchanges $\psi$ and $\psi^*$, one finds
\be
(\chi^1)^{c_{(6)}}=\chi^1\, ,
\qquad
(\chi^2)^{c_{(6)}}=-\chi^2\, , 
\ee
such that $\chi^1$ and $i\chi^2$ are individually Majorana spinors with
respect to six-dimensional Euclidean charge conjugation. 

Next, for the 2-form bilinear in terms of the fundamental structure 
form $J$, from eq.~\eqref{Fierzintprod1} we infer that $T_{0i}=0$, so 
$T_{\mu\nu}=T_{ij}$, with $i,j,\ldots$, indices over base space. 
Then from eq.~\eqref{defJ}, it is easy to get that
\be
\label{TtoJ}
T = h J\, ,
\ee
independent of the angle parameters $\theta$, $\varphi$.

Next, we turn to the fundamental 3-form of the base structure, 
defined in eq.~\eqref{defPsi}, which along with the fundamental 
2-form, satisfies by definition that
\be 
\label{transversesu3}
\iota_K J = \iota_K \Psi = 0\, .
\ee
We thus proceed to express the 3-form bilinears in terms of 
the structure forms, recalling the geometric relations 
in eqs.~\eqref{Fierznorms}, \eqref{Fierzintprods}, 
so that generally $u^{11}_{0ij}, u^{11}_{ijk}\neq0$, 
and $(U_3)_{0ij},(U_3)_{ijk}\neq0$. 
Using our definitions of the bilinear forms in eqs.~\eqref{bilineardef}, 
the structure forms in eqs.~\eqref{structuredef}, and the angle parameters 
in eqs.~\eqref{genericscalars}, it is then straightforward to get 
\begin{subequations}
\label{generic3forms}
\begin{align}
\label{complicateu11}
u^{11}
& = - s p K \wedge J
-\frac{h}{2}\Bigl(
\left(c-1\right)p^2\Psi
+\left(c+1\right)\overline{\Psi}
\Bigr)\, ,\\
\label{complicateU3}
U_3
& = c K \wedge J
-hs\,\Rea\left(p\Psi\right)\, .
\end{align}
\end{subequations}

It is easy to verify consistency when we take eq.~\eqref{cscalarto6d} 
to hold, taking the SM Killing spinors in eqs.~\eqref{simpleisdeep}, 
\eqref{simpleisdeep2}, respectively. We then get that 
\be
\label{intprodKU11}
\iota_K u^{11} = 0\, , 
\ee
and thus $u^{11}_{0ij}=0$, so 
$u^{11}_{\mu\nu\rho}=u^{11}_{ijk}$. 
Then from eq.~\eqref{defPsi}, it is easy to get that

\be
\label{u11to6d}
u^{11}=-h\overline{\Psi}
\quad\text{or}\quad
u^{11}=h p^{2}\Psi\, ,
\ee
or if we write the first option in terms of the real and imaginary parts
\be
U_1
=\frac{h}{2}\left(\Psi+\overline{\Psi}\right)
=h\Psi_+\, ,
\qquad
U_2
=-i\frac{h}{2}\left(\Psi-\overline{\Psi}\right)
=h\Psi_-\, ,
\ee
where the real 3-forms $\Psi_{\pm}$ are the real and imaginary parts 
of $\Psi$, as in eq.~\eqref{gettingreal}.
Therefore, when considering the general case, 
in which $s^{11}=0$ is relaxed, 
we simply use the generic scalar, and 3-form bilinears 
in eqs.~\eqref{genericscalars}, \eqref{generic3forms}, respectively, 
in the full domain 
$\theta\in[0,\pi)$, and $\varphi\in[0,2\pi)$.
Finally, when $\theta=0,\pi/2$, from  eqs.~\eqref{Fierzintprod2}, 
\eqref{realscalarto6d}, \eqref{TtoJ}, we also get that 
$\iota_KU_3=S_3 T= \mp h^2 J$, 
and it is easy to infer that
\be
\label{u3to6d}
U_3= \pm K \wedge J\, ,  
\ee
consistent with eq.~\eqref{complicateU3}. 

To recap, considering a general SM Killing spinor in 
eqs.~\eqref{fullsmspinors}, eqs.~\eqref{genericscalars}, 
\eqref{TtoJ}, \eqref{generic3forms}, together with 
eq.~\eqref{Killingdual}, then provide the relations 
of the spinor bilinears to the geometric  structure on the base. 
When the former are recast in terms of the latter, we see that 
five real forms, $\{h,K,J,\Psi_+,\Psi_-\}$, together with the
two additional real base $0$-forms, $\{\theta,\varphi\}$,
encode the spinorial bilinears and the associated G-structure data, 
in terms of which the real bosonic fields $\{X, A, B, g_{ij}\}$ 
are subsequently constrained.

Finally, a generic SU$(3)$ structure is classified according to its 
intrinsic torsion, which decomposes into five torsion forms, 
$\{{\mathcal{W}}_i\}_{i=1,\ldots,5}$, defined through $\dd J$ and $\dd \Psi$
in Appendix~\ref{torsiondef}, see also 
Appendices~\ref{gstructure6dwg}, \ref{almostcalculus}.
Appendix~\ref{gstructure6dwg} distinguishes this SU$(3)$ data from 
the intrinsic torsion of the underlying U$(3)$ structure \cite{Gray:1980}, 
while Appendix~\ref{almostcalculus} fixes the single-index and all-index actions 
of $J$ used in the corresponding decompositions of forms.
The algebraic and differential bilinear relations, which are derived 
from the KSEs in Section~\ref{Killingresults}, will then enable us to 
study the intrinsic torsion of the base, and classify the relevant 
structures and solutions in Section~\ref{classifyingbps}.

\section{Bilinear Relations and Base Torsion}
\label{Killingresults}

We now turn to the relations among the spinor bilinears, with algebraic ones
arising from the variations of spin-$\tfrac{1}{2}$ fermions in 
eqs.~\eqref{Killinghalfspinequal}, 
and differential ones from the gravitini variations in 
eq.~\eqref{Killinggravitineqcharge}.

We recall that in the equal-charge sector algebraic equations
contain either $\Gamma_+$, or $\Gamma_-$, 
which are the matrices taken as the two Cartan generators 
that have non-trivial orthogonal kernels, 
see Appendix~\ref{5ECliff}. 
Recalling the general form of a Killing spinor in 
eq.~\eqref{spinorasdirectproduct}, 
it thus holds for the USp$(4)$ spinor that 
\be
\label{directsumofkers}
\chi^A\in
\mathbb C^4
=
\Ker\Gamma_+\oplus\Ker\Gamma_-
=
\image\Gamma_+\oplus\image\Gamma_-,
\qquad
\Ker\Gamma_\pm=\image\Gamma_\mp.
\ee
This direct-sum decomposition considerably simplifies the analysis
of bilinear relations. We keep the factor of USp$(4)$ spinor $\chi^A$
explicit and resolve it through the two kernels, while the
Spin$(1,6)$ SM spinor $\chi^\alpha$ remains encoded in the
seven-dimensional bilinears defined in Section~\ref{bilinear}.

The algebraic relations below retain the two Cartan operators until
a USp$(4)$ spinor component is selected. The relations containing
$\Gamma_-$ will then distinguish the complementary spinor branches
according to their Cartan flux, as shown in Section~\ref{uniqueF}.

\subsection{Algebraic Relations}
\label{Killingresultsalg}

Using the spin-$\tfrac12$ variations in
eqs.~\eqref{Killinghalfspinequal}, together with their conjugate
modifications in Appendix~\ref{KSEforMCCC}, we contract the
equations from the left with the Majorana conjugate of the
Spin$(1,6)$ SM spinor, with insertions of Cliff$(1,6)$ matrices
of ranks $0$ through $3$. This covers all potentially independent
relations in seven dimensions. The resulting equations are split into
parts that are symmetric and antisymmetric in the SM indices, with
the latter further separated into their real and imaginary parts.
Thus, at each of the four ranks these contractions produce 
one complex and two real equations.

All relations containing $\Gamma_+$ are consequential for the
analysis of possible solutions, and survive in the 
truncation to the minimal theory. 
The $\Gamma_-$ system of equations is used primarily to distinguish the
spinor branches according to their Cartan flux in Section~\ref{uniqueF}, 
for which relations through rank $2$ suffice. 
We therefore display here the $\Gamma_+$ relations through
rank $3$ and the $\Gamma_-$ relations through rank $2$. A rank-$4$
relation containing $\Gamma_+$, used to simplify the differential
relations, and the rank-$3$ relations containing $\Gamma_-$ are
collected in Appendix~\ref{algbilrel}. 
It is important to highlight that the restriction in eq.~\eqref{cscalarto6d}, 
which fixes the internal spinor orientation, 
is not imposed in deriving this algebraic system.

First, let us define the dilaton-potential dependent term, 
using eq.~\eqref{derscapoteqcharge}, as
\be
\label{DeltaX}
\Delta(X)\equiv X^2-X^{-3}\, ,
\qquad
\mathcal V'(X)
=
-8X^{-6}(2X^5-1)\Delta(X)\, .
\ee
Then, stemming from the spin-$\tfrac{1}{2}$ fermion variation 
with the generator $\Gamma_+$ in eqs.~\eqref{Killingfermionsymeq}, 
\eqref{Killingfermionsymcharge}, we obtain from the part that 
is symmetric in the SM indices the following complex relation at rank 0: 
\begin{align}
\label{br0cg}	
0 & = \left[ 4 g s^{11} \Delta (X)
+ X^3 \left( u^{11} \lrcorner \star G \right)
\Gamma_5 \right] \chi\, ,
\end{align}
where here $\chi$ stands only for the USp$(4)$ spinor, 
$\chi^A$, discussed above, 
and the USp$(4)$ index is suppressed here and henceforth.
Moreover, the identity operator $\mathds{1}_4$ is also suppressed here 
and henceforth, such that it is implicit in terms with no explicit Cartan 
operator. 
From the part that is antisymmetric in the SM indices we extract 
the following real and imaginary parts at rank 0:
\begin{subequations}
\label{br0rrg}	
\begin{align}
\label{br0rg}
0 & = \left[ 4 g S_3 \Delta (X) 
+ \left( T \lrcorner F \right) \Gamma_{+} 
+ X^3 \left( U_3 \lrcorner \star G \right) \Gamma_5 \right] 
\chi\, ,\\ %
\label{br0r}
0 & = \left[ \iota_K \dd X \right]\chi\, .
\end{align} 
\end{subequations}
Similarly, we obtain the following relations at rank 1:
\begin{subequations}
\label{br1}
\begin{align} 
\label{br1c}	
0 & = \left[ 5 s^{11} \dd X_{\mu} 
+ i \left( F \lrcorner u^{11} \right)_\mu \Gamma_+ 
+ X^3 \left( u^{11} \lrcorner G \right)_{\mu} 
\Gamma_5 \right] \chi\, ,\\
\label{br1r}
0 & = \left[ 5 S_3 \dd X_\mu + \left( \iota_K F \right)_\mu \Gamma_+ 
+ X^3 \left( U_3 \lrcorner G \right)_{\mu} 
\Gamma_5 \right] \chi\, ,\\
\label{br1rg}
0 & = 
\Big[ 4 g \Delta (X)  K_\mu  
+ 5 \left( \dd X \lrcorner T \right)_\mu  
-\left( F \lrcorner U_3 \right)_\mu \Gamma_+ 
+ X^3 \left( T \lrcorner \star G \right)_{\mu} 
\Gamma_5 \Big] \chi\, ,
\end{align}	
\end{subequations}
the following relations at rank 2:
\begin{subequations}
\label{br2}	
\begin{align}
\label{br2c}
0 & = 
\left[ 15 \left( \dd X \lrcorner u^{11} \right)_{\mu\nu}
- 3 i \left( s^{11} F - F \lrcorner \star u^{11} \right)_{\mu\nu} \Gamma_+ 
+ X^3 \left( \left( \star u^{11} \right)^{\alpha\beta\gamma}{}_{[\mu}
G_{\nu]\alpha\beta\gamma} \right) \Gamma_5 \right] \chi\, ,\\*	
\label{br2r}
0 & = 
\left[ 15 \left( \dd X \lrcorner U_3 \right)_{\mu\nu} 
+ 6 \left( T^{\alpha}{}_{[\mu} F_{\nu]\alpha} \right) \Gamma_+
+ X^3 \left( (\star U_3)^{\alpha\beta\gamma}{}_{[\mu} 
G_{\nu]\alpha\beta\gamma}\right) \Gamma_5 \right] \chi\, , \\*
\label{br2rg}
0 & = 
\bigg[ 4 g \Delta(X) T_{\mu\nu} 
+ 10 \dd X_{[\mu} K_{\nu]}  
+ \Bigl( S_3 F - F \lrcorner \star U_3 \Bigr)_{\mu\nu} 
\Gamma_+ 
- X^3 \Bigl( \iota_K \star G + T \lrcorner G \Bigr)_{\mu\nu} 
\Gamma_5 \bigg] \chi\, ,
\end{align}	
\end{subequations}
and the following relations at rank 3:
\begin{subequations}
\label{br3} 	
\begin{align} 
\label{br3cg}	
0 & = \bigg[ 8 g \Delta(X) u^{11}_{\mu\nu\sigma}
+ 10 \left( \dd X \lrcorner \star u^{11} \right)_{\mu\nu\sigma}  
+ 6 i \Bigl( F^{\alpha}{}_{[\mu} (u^{11})_{\nu\sigma]\alpha} \Bigr) 
\Gamma_+ 
\nn \\*
& \qquad \qquad \qquad \qquad \qquad \qquad \qquad 
- X^3 \left( 2 s^{11} \left( \star G \right)_{\mu\nu\sigma} 
+ 3 \left( u^{11} \right)^{\alpha\beta}{}_{[\mu} 
G_{\nu\sigma]\alpha\beta} \right) 
\Gamma_5 \bigg] \chi\, ,
\\* 
\label{br3rg}
0 & = \bigg[ 8 g \Delta(X) \left(U_3\right)_{\mu\nu\sigma}  
+ 10 \left( \dd X \lrcorner \star U_3 \right)_{\mu\nu\sigma} 
- 2 \left( 3 K_{[\mu} F_{\nu\sigma]} 
- \left( F \lrcorner \star T \right)_{\mu\nu\sigma} \right) \Gamma_+  
\nn\\*
& \qquad \qquad \qquad \qquad \qquad \qquad	\qquad 
- X^3 \left( 2 S_3 \left( \star G \right)_{\mu\nu\sigma} 
+ 3 (U_3)^{\alpha\beta}{}_{[\mu} G_{\nu\sigma]\alpha\beta} \right) 
\Gamma_5 \bigg] \chi\, ,
\\* 
\label{br3r}
0 & = \left[ 15 \dd X_{[\mu} T_{\nu\sigma]}  
- 3 \left( F^{\alpha}{}_{[\mu} 
\left(U_3\right)_{\nu\sigma]\alpha} \right) \Gamma_+
- X^3 \left( \left( \iota_K G \right)_{\mu\nu\sigma} 
+3 T^{\alpha}{}_{[\mu} ( \star G )_{\nu\sigma]\alpha} 
\right) \Gamma_5 \right] \chi\, . 
\end{align}	
\end{subequations}

Next, we turn to the spin-$1/2$ fermion variation with 
the generator $\Gamma_-$ in eqs.~\eqref{Killingfermionantisymeq},
\eqref{Killingfermionantisymcharge}, 
which as noted at the end of Section~\ref{spinors}, transformed in the 
equal-charge sector into a decoupled projection of $\Gamma_-$. 
From this variation we obtain only a single non-trivial relation 
at rank 0:
\begin{align}
\label{antibr0}
0 & = \left( T \lrcorner F \right) \Gamma_{-} \chi\, . 
\end{align}
Yet, we obtain the following complete set of non-trivial relations 
at rank 1:
\begin{subequations}
\label{antibr1}	
\begin{align}
\label{antibr1c}	
0 & = \left( F \lrcorner u^{11} \right)_\mu \Gamma_- \chi\, ,\\
\label{antibr1r1}
0 & = \left( F \lrcorner U_3 \right)_\mu \Gamma_- \chi\, ,\\
\label{antibr1r2}
0 & = \left( \iota_K F \right)_\mu \Gamma_- \chi\, ,
\end{align}	
\end{subequations}
and at rank 2: 
\begin{subequations}
\label{antibr2}
\begin{align}
0 & = \left[ s^{11} F - F \lrcorner \star u^{11} \right]_{\mu\nu} 
\Gamma_- \chi\, ,\\
\label{antibr2r1}
0 & = \left[ S_3 F - F \lrcorner \star U_3 \right]_{\mu\nu}  
\Gamma_- \chi\, ,\\
0 & = \left[ T^{\alpha}{}_{[\mu} F_{\nu]\alpha} \right] \Gamma_- \chi\, .
\end{align}	
\end{subequations}
Together with eqs.~\eqref{antibr3}, 
these are the only relations containing $\Gamma_-$. 
They constrain $F$ only for $\chi^A\in \Ker\Gamma_+$, independently of gauging. 

The $\Gamma_+$ and $\Gamma_-$ relations displayed here, together with the 
3-rank relations in Appendix~\ref{algbilrel}, 
form the complete algebraic system used below. 
We first record several immediate consequences.
\begin{prop}[Dilaton Gradient is Transverse]
\label{dxtransverse}
The dilaton gradient lies in the transverse base. Consequently,
$X$ is preserved by the Killing isometry.
\end{prop}

\begin{proof}
The algebraic relation in eq.~\eqref{br0r} gives
\be
\label{LieofX}
0=\iota_K\dd X={\cal L}_K X\, .
\ee
\end{proof}

\begin{prop}[Running Dilaton Requires Non-Trivial Gauge Flux]
\label{codilaton}	
Timelike supersymmetric solutions with a running dilaton 
and no active gauge fields are ruled out. 
Thus, regardless of gauging, it holds that
\be
\label{nolonedilaton}
F = G = 0 \qquad \implies \qquad \dd X = 0\, .
\ee	
\end{prop}

\begin{proof}
When $F=G=0$, eqs.~\eqref{br1c}, \eqref{br1r} reduce to
\be
s^{11}\dd X=0\, ,
\qquad
S_3\dd X=0\, .
\ee
Since the Killing vector is timelike, $K^2=-h^2\neq0$, while
eqs.~\eqref{genericscalars} show that $s^{11}$ and $S_3$
cannot vanish simultaneously for any value of the internal spinor angle $\theta$.
Consequently, when the gauge fluxes are trivial, $\dd X=0$.
\end{proof}

\begin{cor}[Unique VEV with Gauging]
\label{gaugedvev}
When the theory is gauged and in the vacuum, the dilaton 
is uniquely fixed to its value at the supersymmetric stationary point 
in eq.~\eqref{derscapoteqcharge} of the scalar potential 
in eq.~\eqref{scapoteqcharge}, namely 
\be
\label{gneq0vev}
g \neq 0\, , \qquad F = G = 0 
\qquad  \implies  \qquad 
X=1\, .
\ee 
\end{cor}

\begin{proof}
When $F=G=0$, eqs.~\eqref{br0cg}, \eqref{br0rg}, reduce to
\be
g s^{11}\Delta(X)=0\, ,
\qquad
g S_3\Delta(X)=0\, .
\ee
For a timelike Killing vector, eqs.~\eqref{genericscalars} imply
that $s^{11}$ and $S_3$ cannot vanish simultaneously. Hence, for
$g\neq0$, these relations require
\be
\Delta(X)=X^{-3}(X^5-1)=0\, ,
\ee
and therefore $X=1$. Although the scalar potential in
eq.~\eqref{scapoteqcharge} has a second stationary point at
$X=2^{-1/5}$, that point does not satisfy $\Delta(X)=0$, and is
therefore not a supersymmetric vacuum in the present sector.
\end{proof}

To further uncover the rich information that is encapsulated in the 
algebraic bilinear relations, we must invoke at this point the 
representation-theoretic content of solutions in terms of the 
transverse SU$(3)$ structure. 
In Appendix~\ref{decomposeformstosu3irreps} we spell out the 
algebraic decomposition of the seven-dimensional real forms 
$\{\dd X, F, G\}$ into finite-dimensional SU$(3)$ irreducible 
representations. 
Their SU$(3)$-irreducible components are obtained via 
projection maps into transverse forms that are built from elementary
contractions with the three invariant bilinear forms, 
$\{K, J, \Psi\}$, as detailed in Appendix~\ref{su3projections} 
(note that $K$ is a scalar singlet with respect to the transverse 
SU$(3)$ structure). 
Moreover, all of the operators that appear in the algebraic 
bilinear relations, $\{\lrcorner, \wedge, \star\}$, are 
SU$(3)$-equivariant. 
Consequently, every algebraic bilinear relation defines 
a pointwise SU$(3)$-equivariant linear map between 
finite-dimensional SU$(3)$ representation spaces.

By \textit{Schur's lemma}, an SU$(3)$-equivariant linear map
between two inequivalent irreducible representations must vanish,
whereas an equivariant endomorphism of a single irreducible
representation acts as multiplication by a scalar. Consequently,
components belonging to inequivalent SU$(3)$ representations cannot
mix in the algebraic bilinear relations.
When an irreducible representation occurs more than once, its
equivalent copies may mix with one another. Their mixing is described
by an ordinary finite-dimensional matrix acting on the corresponding
multiplicity space. We may therefore suppress the tensor indices
internal to each irreducible representation and represent every copy
by a real or complex coefficient, as appropriate. In this way, the
tensorial algebraic bilinear relations reduce to linear scalar
equations among the projected components within each
\textit{isotypic component}.

The projected components of $\{\dd X,F,G\}$, together with an
additional gauging-dependent scalar, are thus organized
into SU$(3)$ isotypic components in the extended algebraic 
representation space $\mathcal V$:
\begin{enumerate}

\item \textbf{Fundamental Isotypic Component} 
$\mathbf{3} \oplus \bar{\mathbf{3}}$: \textit{Multiplicity 5}. 
This sector contains the complex 1-forms of type $(1,0)$, 
together with their conjugates:
\be
\mathcal{I}_{\mathbf{3} \oplus \bar{\mathbf{3}}} 
= \bigl\{ D, \, f, \, \alpha, \, q,  \, \beta \bigr\}\, ,
\ee
which are defined in Appendices~\ref{decomposeformstosu3irreps}, 
\ref{su3projections}. 

Note that the two irreducible components, $\alpha$, $\beta$, 
are defined as contractions with $\Psi$. Thus they are automatically 
of type $(1,0)$, and therefore require no further application of 
the projector in eq.~\eqref{projector}, namely
\be
\alpha,\beta\in \mathbf{3}_{\mathbb C}
\, .
\ee

\item \textbf{Singlet Isotypic Component} 
$\mathbf{1} \oplus (\mathbf{1} \oplus \bar{\mathbf{1}})$: 
\textit{Multiplicity 4+1}. 
This sector contains the four real scalar components represented by the
two real flux singlets $\lambda,\sigma$, and the complex flux
singlet $\zeta$, together with its conjugate:
\be
\mathcal{I}_{\text{Singlet}} 
= \bigl\{ \lambda, \, \sigma, \, \zeta , \bar\zeta \bigr\}\, ,
\ee
which are projected flux components defined in 
Appendices~\ref{decomposeformstosu3irreps}, \ref{su3projections}.
In the gauged theory we also include the additional real scalar singlet
\be
\label{Deltag}
\Delta_g \equiv g \Delta(X)\, ,
\ee
where $\Delta(X)$ is defined in eq.~\eqref{DeltaX}. 
Unlike the projected components, the additional scalar $\Delta_g$
does not arise from the SU$(3)$ decomposition of a differential
form, and is not an independent DOF.
It enters the gauged algebraic relations as an additional singlet source, 
and vanishes identically in the ungauged theory.

\item \textbf{Adjoint Isotypic Component} 
$\mathbf{8}$: \textit{Multiplicity 2}. 
This sector contains the real primitive $(1,1)$ forms:
\be
\mathcal{I}_{\mathbf{8}} = \bigl\{ \tau, \, \kappa \bigr\}\, ,
\ee
which are defined in Appendices~\ref{decomposeformstosu3irreps}, 
\ref{su3projections}.

\item \textbf{Symmetric Isotypic Component} 
$\mathbf{6} \oplus \bar{\mathbf{6}}$: \textit{Multiplicity 1}. 
This sector contains the single primitive $(2,1)$ form:
\be
\mathcal{I}_{\mathbf{6} \oplus \bar{\mathbf{6}}} 
= \bigl\{ \rho \bigr\}\, ,
\ee
which is defined in Appendices~\ref{decomposeformstosu3irreps}, 
\ref{su3projections}.
\end{enumerate}
The extended representation-theoretic content in 
algebraic \textit{representation space} thus reads:
\begin{equation}
\mathcal{V} \cong  
(\mathbf{3} \oplus \bar{\mathbf{3}})^{\oplus 5} \oplus
\bigl( \mathbf{1}^{\oplus 2} \oplus (\mathbf{1} \oplus \bar{\mathbf{1}}) 
\oplus \mathbf{1}_g 
\bigr) 
\oplus \mathbf{8}^{\oplus 2} 
\oplus (\mathbf{6} \oplus \bar{\mathbf{6}})\, .
\end{equation}

Accordingly, the full set of algebraic bilinear relations yields 
a linear system, which has the following schematic block-diagonal form:
\begin{align}
\label{irrepmatrix}
\begin{pmatrix}
\mathcal M^{(1,0)} & 0 & 0 & 0\\
0 & \mathcal M^{(0,0)} & 0 & 0\\
0 & 0 & \mathcal M^{(1,1)_0} & 0\\
0 & 0 & 0 & \mathcal M^{(2,1)_0}
\end{pmatrix}
\begin{pmatrix}
(D,f,\alpha,q,\beta)\\
(\lambda,\sigma,\zeta,\bar\zeta)\\
(\tau,\kappa)\\
(\rho)
\end{pmatrix}
+
\begin{pmatrix}
0\\
\mathcal S^{(0,0)}\Delta_g\\
0\\
0
\end{pmatrix}
=0\, .
\end{align}
Here each block $\mathcal M^{(r)}$ maps the copies of
the SU$(3)$ representation $r$ into equations of the
same representation, and need not be square, while
$\mathcal S^{(0,0)}$ inserts the gauging-dependent scalar only
into the singlet equations.
This block structure is the direct consequence of 
\textit{SU$(3)$-equivariance} and \textit{Schur's lemma}, as discussed above.

Even though when there is gauging, the additional scalar, $\Delta_g$, 
enters the singlet equations, it is related to the fundamental component
$D=\Pi^{(1,0)}(\dd X)$ because both arise from the same scalar field, 
i.e.~the dilaton $X$:
\be
\dd \Delta_g
= g\left(2X+3X^{-4}\right)(D+\overline D)\, .
\ee
Consequently, 
\be
g\neq0\, ,\quad \Delta_g\equiv0
\quad \implies \quad
X\equiv1
\quad \implies \quad
D = \dd X = 0\, .
\ee
Equivalently, a running dilaton implies that $\Delta_g$ cannot
vanish identically when there is gauging:
\be
\label{singletcoupling}
g\neq0\, ,\qquad \dd X \not\equiv 0
\qquad\implies\qquad
\Delta_g \not\equiv 0\, .
\ee
Thus, the singlet and fundamental linear systems remain separate
under the SU$(3)$ decomposition, but for $g\neq0$ their common
dependence on $X$ allows a constraint in either sector to restrict
the other.

We now project the complete algebraic system onto the singlet,
fundamental, adjoint and symmetric isotypic components, 
in that order, before using the $\Gamma_-$ relations 
to distinguish the two spinor branches.
All projections below are pointwise and do not assume integrability
of the base. We retain the general scalar and $3$-form bilinears in
eqs.~\eqref{genericscalars}, \eqref{generic3forms}, without imposing
$s^{11}=0$. 
In each isotypic sector below, form degree and SU$(3)$ type determine
which projected relations can contain that sector. 
SU$(3)$ equivariance then separates its projection from all inequivalent
components, so that it must vanish independently.

The $\Gamma_+$ projections in the four isotypic analyses below similarly
hold in the minimal seven-dimensional theory, with or without
U$(1)$ gauging. For the explicit maximal-theory projections we take
the first USp$(4)$ component, on which $\Gamma_5$ and $\Gamma_+$
act as $+1$. The other component in $\Ker\Gamma_-$ is obtained by
reversing the signs of the $\Gamma_+$-dependent $F$ terms. 
On $\Ker\Gamma_+$, the $\Gamma_+$-dependent terms vanish and the
$\Gamma_5$-dependent $G$ terms reverse sign; 
the $\Gamma_-$ relations separately restrict $F$ to its primitive component, 
as shown in Section~\ref{uniqueF}.

\subsubsection{Singlet Isotypic Component}
\label{su3singlet}

\begin{prop}[Constrained Singlet Sector]
\label{constrainedsingletsector}		
The singlet isotypic component is subject to the following constraints:
\begin{enumerate}
\item The two real singlets of $F$ and $G$ are subject 
to a single constraint.
\item If either the complex scalar bilinear,
or the real singlet in $F$, is identically vanishing,
then the complex singlet in $G$ vanishes identically.
\end{enumerate}
\end{prop}

\begin{proof}

The singlet isotypic component is populated by two real singlets
that are defined in eqs.~\eqref{projlambda}, \eqref{projsigma}, as 
\begin{subequations}
\begin{align}
\lambda &\equiv J\lrcorner F\, ,\\*
\sigma &\equiv J\lrcorner J\lrcorner G\, ,
\end{align}
\end{subequations}
and a complex holomorphic singlet defined in eq.~\eqref{projchi} as
\be 
\zeta \equiv \Psi \lrcorner \iota_K G\, .
\ee
In addition, we consider the real scalar $\Delta_g \equiv g \Delta(X)$, 
defined in eq.~\eqref{Deltag}, as an additional singlet, 
that is switched off when $g=0$. 
It is useful to exchange the complex singlet $\zeta$ 
with two real singlets, $u,v$, by defining
\begin{align}
\label{chitoreal}	
\zeta \equiv p^{-1} \left( u - i v \right)\, , 
\quad
\overline \zeta \equiv p \left( u + i v \right)	
\quad\longleftrightarrow\quad 
u \equiv \frac12 \left( p \zeta + p^{-1} \overline \zeta \right)\, ,
\quad  v \equiv \frac1{2i} \left( p^{-1} \overline \zeta - p \zeta \right)
\, ,
\end{align}
where we used the parameter of internal spinor orientation defined in 
eqs.~\eqref{genericscalars}.

The singlet irreducible component can occur in all eleven $\Gamma_+$
relations (other than eq.~\eqref{br0r}): in the rank-$0$ relations;
in the longitudinal $[0]$ projections of the rank-$1$ relations,
or equivalently after contraction with $h^{-1}\iota_K$; in the
transverse rank-$2$ relations contracted with $J$; and in the
longitudinal $[0ij]$ and transverse $[ijk]$ projections of the
rank-$3$ relations, contracted with $J$ and $\Psi$, respectively.

Using the contraction and Hodge-dual identities provided in 
Appendices~\ref{contractidentities}, \ref{hodgeidentities}, we extract 
the singlets from all the algebraic relations as described,
yielding the following independent equations from rank 0:
\begin{subequations}
\label{projsingrank0}	
\begin{align}
\text{Eq.~\eqref{br0cg}:} 
\quad 0 & = 8 h s g \Delta - X^3 h s \sigma 
- 2 X^3 \left( i u - c  v \right)\, , 
\\ %
\text{Eq.~\eqref{br0rg}:} 
\quad 0 & =  8 h c g \Delta - 2 h \lambda - X^3 h c \sigma 
- 2 X^3 s v\, .
\end{align}
\end{subequations}
All higher-rank singlet projections are collected in
Appendix~\ref{higherproj}. They provide consistency checks but no
additional independent constraints. Real projections by $K$ or
$J$ of a complex relation each yield two independent real
equations, while contractions of the complex rank-$3$ relation
with $\Psi$ and $\overline\Psi$ yield independent complex
equations, since they are not related by complex conjugation.
Thus, in total, we examine $23$ potentially independent real
equations for $4+1$ real scalar quantities.

As noted, the two rank-0 projections in eqs.~\eqref{projsingrank0} 
provide all three independent generic constraints:
\be
\label{gensingletconstraint}	
\begin{cases}	
u = 0\, ,\\
h s \lambda = - X^3 v\, ,\\
2 c \lambda = 8 g \Delta - X^3 \sigma \, .
\end{cases}
\ee
It is useful to exchange these generic constraints by separating
them into the following three branches:
\be
\label{specsingletconstraint}
\begin{cases}
s = 0\, , c = \pm1 
\quad \qquad \qquad \qquad \implies \qquad
\begin{cases}
u = v = 0\, ,\\
2 c \lambda = 8 g \Delta - X^3 \sigma\, ;
\end{cases}
\\
s \neq 0\, ,\quad \lambda = 0\, ,\quad  (g \neq0) 
\qquad \implies \qquad
\begin{cases}
u = v = 0\, ,\\
0 =	8 g \Delta - X^3 \sigma\, ;
\end{cases}
\\
s \neq 0\, ,\quad \lambda \neq 0 \implies g = 0
\qquad \implies \qquad
\begin{cases}
	u=0\, ,\\
	(v,\sigma)=-X^{-3}\lambda\,(hs,2c)\, .
\end{cases}	
\end{cases}	
\ee
The second branch remains valid at $s=0$, where it is precisely
the $\lambda=0$ specialization of the first. 
The third branch is restricted to $g=0$, as follows later from
Propositions~\ref{uniquespinor}, \ref{vanishingcscalar}.
Here $\lambda$ determines $v$ and $\sigma$, which belong to
the longitudinal and transverse parts of $G$, respectively;
in particular, $\sigma=0$ at $c=0$.

Notably, from the first two branches we see that 
\be 
s = 0 \quad \mathtt{or} \quad \lambda = 0
\quad 
\implies  
\qquad 
\zeta = 0\, ,
\ee
so that the holomorphic singlet in the longitudinal part of $G$ 
from eq.~\eqref{f4su3} therefore vanishes.

\end{proof}

\subsubsection{Fundamental Isotypic Component}
\label{su3fundamental}

\begin{prop}[Constrained 1-Form Sector]
\label{constrained1form}
The fundamental isotypic component is subject to the following constraints:
\begin{enumerate}
\item The fundamental isotypic component is generically subject
to two constraints.
\item When the complex scalar bilinear vanishes identically, the five 
fundamental 1-form components decouple into a three-component subset, 
and a two-component one, each subject to one independent constraint.
\end{enumerate}	
\end{prop}

\begin{proof}
	
The fundamental isotypic component is the most populated one,
consisting of five irreducible components defined in
eqs.~\eqref{8nonprim} as follows:
\begin{subequations}
\label{fundamentals}
\begin{align}
\label{D}
D &\equiv \Pi^{(1,0)}(\dd X)\, ,\\
\label{f}
f &\equiv \Pi^{(1,0)}(\iota_K F)\, ,\\
\label{alpha}
\alpha &\equiv F\lrcorner\Psi\, ,\\
\label{q}
q &\equiv \Pi^{(1,0)}
\bigl(J\lrcorner\iota_KG\bigr)\, ,\\
\label{beta}
\beta &\equiv \Psi\lrcorner G_\perp\, .
\end{align}
\end{subequations}

The fundamental component can occur in the rank-$1$ through
rank-$3$ relations in eqs.~\eqref{br1}--\eqref{br3}, through
their appropriate transverse or longitudinal $(1,0)$ projections.
Using the contraction and Hodge-dual identities provided in
Appendices~\ref{contractidentities}, \ref{hodgeidentities},
together with the holomorphic projector $\Pi^{(1,0)}$ in
eq.~\eqref{1formprojector}, we extract the $(1,0)$-type
components from the above noted algebraic relations.
A complete independent set of constraints, valid 
throughout the full domain of internal spinor angles, 
is given by the following two projections from ranks 1 and 2:
\begin{subequations}
\label{projfundind}
\begin{align}
\text{Eq.~\eqref{br1r}$_{[i]}$:} \qquad
0 & =
5 h c D - f - X^3 c q + \frac12 X^3 h s p \beta\, ,
\\
\text{Eq.~\eqref{br2r}$_{[ij]}\lrcorner\Psi$:} \qquad
0 & =
5 h s D - X^3 s q
+ \frac{i}{2} h p \left( \alpha + i c X^3 \beta \right)\, .
\end{align}
\end{subequations}
All other fundamental projections from ranks $1$ through $3$ are
collected in Appendix~\ref{higherproj}. They provide consistency
checks but no additional independent constraints. Thus, in total,
we examine $15$ potentially independent complex equations for five
complex $1$-form components.
	
The two independent generic constraints from eqs.~\eqref{projfundind} then read
\be
\label{gen1vectorconstraint}
\begin{cases}
f = 5 h c D - X^3 c q + \frac12 X^3 h s p \beta\, ,\\
5 h s D - X^3 s q = - \frac{i}{2} h p
\left( \alpha + i c X^3 \beta \right)\, .
\end{cases}
\ee
Similarly to the singlet sector, it is useful to separate these
constraints into the following three branches:
\be
\label{spec1vectorconstraint}
\begin{cases}
s = 0\, , c = \pm1
\qquad \qquad \implies \qquad
\begin{cases}
c f = 5 h D - X^3 q\, ,\\
0 = \alpha + i c X^3 \beta\, ;
\end{cases}
\\
s \neq 0\, ,\quad f = \alpha = 0
\qquad \implies \qquad
\begin{cases}
\beta = 0\, ,\\
0 = 5 h D - X^3 q\, ;
\end{cases}
\\
s \neq 0\, ,\quad (f,\alpha) \neq (0,0)
\quad \implies \quad
\begin{cases}
f = - \dfrac{hp}{2s}
\left( i c \alpha - X^3 \beta \right)\, ,\\
5 h D - X^3 q = -\dfrac{hp}{2s}
\left( i \alpha - c X^3 \beta \right)\, .
\end{cases}
\end{cases}
\ee
The second branch remains valid at $s=0$, where it is precisely
the $f=\alpha=0$ specialization of the first.
As we shall see later from Propositions \ref{uniquespinor} and 
\ref{vanishingcscalar}, the third branch is restricted to $g=0$.
\end{proof}

\subsubsection{Primitive Isotypic Components}
\label{su3primitive}

\paragraph{Adjoint isotypic component.}

\begin{prop}[Unconstrained Primitive Component $(1,1)_0$]
\label{unconst2formprim}
The primitive $(1,1)$ irreducible components of 
$F$ and $G$ are unconstrained, regardless of gauging. 	
\end{prop}

\begin{proof}

The isotypic component $\mathbf{8}$ is populated by two primitive 
$(1,1)$ forms defined in eqs.~\eqref{tau}, \eqref{kappa}, as
\begin{subequations}
\begin{align}
\tau & \equiv \Pi^{(1,1)_0} (F)\, , 
\qquad \quad
\iota_K \tau = J \lrcorner \tau = \tau \lrcorner \Psi = 0\, ,\\
\kappa & \equiv \Pi^{(1,1)_0} \bigl( J \lrcorner G \bigr)\, ,
\qquad 
\iota_K \kappa = J \lrcorner \kappa = \kappa \lrcorner \Psi = 0\, .
\end{align}
\end{subequations}

The adjoint component can occur in the transverse rank-$2$ relations in 
eqs.~\eqref{br2} and in the longitudinal projections of the rank-$3$ relations 
in eqs.~\eqref{br3}.
Using eqs.~\eqref{primids}, \eqref{primHodgeP}, \eqref{primHodgeJP}, 
and the Hodge-split formulae in Appendix~\ref{hodgeidentities},
we extract from eqs.~\eqref{br2}, \eqref{br3}, 
the primitive 2-form irreducible components, yielding only tautologies, namely
\be
\text{Eqs.~\eqref{br2}$_{[ij]}$:} 
\quad 0 = 0\, , 
\qquad	
\text{Eqs.~\eqref{br3}$_{[0ij]}$:} 
\quad 0 = 0\, , 
\ee
leaving both components in $F$ and $G$ unconstrained,
regardless of gauging.
\end{proof}

\paragraph{Isotypic component 
	$\mathbf{6} \oplus \bar{\mathbf{6}}$.}

\begin{prop}[Unconstrained Primitive Component $(2,1)_0$]
\label{unconst3formprim}	
The primitive component of $G$, $((2,1)+(1,2))_0$, is
unconstrained, regardless of gauging.
\end{prop}

\begin{proof}
	
The isotypic component $\mathbf{6}\oplus\bar{\mathbf{6}}$
comprises the complex primitive $3$-form, defined in eq.~\eqref{rho} as
\be
\rho \equiv \Pi^{(2,1)_0}\bigl(\iota_KG\bigr)\, ,
\qquad
\iota_K\rho=J\lrcorner\rho=\rho\lrcorner\Psi=0\, ,
\ee
together with its conjugate. 

This component can occur only in
the transverse rank-$3$ relations in eqs.~\eqref{br3}.	
Using eqs.~\eqref{rhoJinsert}, \eqref{primHodgerho}, and 
the Hodge split identities of Appendix~\ref{hodgeidentities},
we extract from eqs.~\eqref{br3} the primitive 3-form irreducible component, 
yielding only tautologies, namely
\be
\text{Eqs.~\eqref{br3}$_{[ijk]}$:} 
\quad 0 = 0\, , 
\ee	
leaving this irreducible component in $G$ unconstrained, regardless of gauging.

\end{proof}

\subsubsection{Branch of Primitive Cartan Flux 
}
\label{uniqueF}

We now use the algebraic bilinear relations containing $\Gamma_-$
to distinguish the two complementary USp$(4)$ spinor branches.
Since $\Gamma_-$ acts non-trivially precisely for
\be
\chi^A\in(\Ker\Gamma_-)^\perp=\Ker\Gamma_+\, ,
\ee
these relations constrain $F$ only on that branch.

\begin{prop}[Primitive Cartan Flux on USp$(4)$ Spinor Branch]
\label{uniquespinor}	
\be 
\label{PrimF}
\chi^A \in \Ker\Gamma_+ 
\qquad \implies \qquad 
F = F^{(1,1)_0}\, .  \qquad 
\ee
\end{prop}

\begin{proof}
	
We use here the rank-0 and rank-1 relations 
in eqs.~\eqref{antibr0}, \eqref{antibr1}, 
from which it is easy to infer
\be
\iota_K F = 0\, , \qquad 
J \lrcorner F = 0
\quad \implies \quad 
f = 0\, , \qquad \lambda = 0\, ,
\ee
and 
\begin{align}
0 & =  
F \lrcorner u^{11} 
= - \frac{h}{2} \Bigl( (c-1) p^2 \alpha 
+ (c+1) \bar{\alpha} \Bigr)\, . 
\end{align}
Since $\alpha$ and $\bar\alpha$ have types $(1,0)$ and $(0,1)$, respectively, 
the two terms vanish separately.
As $p\neq0$ and $c-1$ and $c+1$ cannot both vanish, 
complex conjugacy implies that for all $c$
\be
\alpha = \bar \alpha = 0\, .
\ee
Using the decomposition of $F$ into SU$(3)$-irreducible components
in eq.~\eqref{f2su3}, these results imply
\be
F=F^{(1,1)_0}\equiv\tau\, .
\ee

As to the higher-rank relations with $\Gamma_-$, 
the transverse projections of eqs.~\eqref{antibr2}, and the longitudinal 
projections of eqs.~\eqref{antibr3}, yield only tautologies for the remaining 
primitive component, exactly as in the adjoint analysis of 
Section~\ref{su3primitive}. 
Hence $\tau$ is not further constrained also on this branch.

\end{proof}

It follows that
\begin{cor}[Non-Primitive Cartan Flux on Complementary USp$(4)$ Spinor Branch]
\label{nonprimspinor}	
A non-primitive Cartan flux is compatible only with USp$(4)$ spinor components
in $\Ker\Gamma_-$, whereas a primitive Cartan flux is generally not excluded 
on either complementary spinor branch, namely
\be
\label{nonprimflux}
F^{\rm np} \equiv F - \tau \neq0
\quad\implies\quad
\chi^A\in\Ker\Gamma_-\, .
\ee	
\end{cor}

\begin{proof}
This is immediate from eq.~\eqref{directsumofkers} 
and Proposition \ref{uniquespinor}.	
\end{proof}

\subsection{Differential Relations}
\label{Killingresultsdiff}

We turn to the differential relations among the bilinear forms from 
the gravitini variations in eq.~\eqref{Killinggravitineqcharge}, and 
the related equations for the Majorana and charge conjugates of 
the Killing spinors in Appendix~\ref{KSEforMCCC}. 
We recall that from the two Cartan generators $\Gamma_\pm$, 
the differential bilinear relations contain only $\Gamma_+$.  
Therefore we consider $\chi^A \in \Ker \Gamma_-$, 
or $\chi^A \in \Ker\Gamma_+$, 
in which case terms containing $A$ and $F$ simply drop out of the 
differential relations, effectively as if $F=0$. 
It is thus the first component of the USp$(4)$ spinor which is generally 
considered below, effectively replacing the Cartan operators $\Gamma_5$, 
$\mathds{1}_4$, and $\Gamma_+$, with their first diagonal-matrix entries, 
i.e.~by the scalar 1. Then the gauge form $A$ is allowed to flip sign, 
as in eq.~\eqref{gauged1form} below, and when $A$ is turned off, 
then the gauge form $B$ is allowed to flip sign, as 
in eq.~\eqref{gauged3formnoF} below, essentially due to R-symmetry.

\subsubsection{Scalar Bilinears}
\label{diffscalars}

As to the scalar bilinears we first consider the complex scalar $s^{11}$, 
for which the following differential relation is obtained:
\begin{align}	
\label{diffs11}	
\left[ \dd \left(X s^{11} \right) \right] \chi
& = - 4i g \left[ \left(X s^{11}\right) A \right] \Gamma_{+} \chi\, ,
\end{align}
and we recall the link of $s^{11}$ to the real forms $S_1$, $S_2$, 
in eq.~\eqref{comscalarstoreal}.
Note that the gauge potential itself appears in this differential relation. 
Yet, as we will show shortly, this differential relation 
also implies that the complex scalar bilinear must vanish 
when $g,F\neq0$, and $\chi \in \Ker \Gamma_-$, as in eq.~\eqref{cscalarto6d}.

First, when $g=0$, or $F=0$ with the gauge choice $A=0$, 
or $\chi^A \in \Ker\Gamma_+$,
the RHS of eq.~\eqref{diffs11} drops out, which implies 
\be 
\label{diffs11nognoF}
\dd \left( X s^{11} \right) = 0\, , 
\ee
so that $Xs^{11}=C$ with $C$ some complex constant, 
and $C=0$ when $s^{11}=0$.
In any case, it is readily inferred %
from eq.~\eqref{diffs11nognoF} that 
\be
\label{Lieofs11}
{\cal{L}}_{K} s^{11} 
= \iota_K \dd s^{11} 
= 0\, , 
\ee
where we used eq.~\eqref{br0r}, %
which holds independently of gauging.
Thus the complex scalar bilinear $s^{11}$ is preserved 
by the Killing isometry, even if it is non-vanishing.

Furthermore, in the latter case, using eqs.~\eqref{diffs11nognoF}, 
\eqref{complicates11}, it is also easy to infer 
\be
s \neq 0  
\quad \implies \quad 
\dd \log (Xh) + \dd \log |s| + \dd \log p = 0 
\quad \implies \quad 
i \dd \varphi = 0\, , 
\ee
where we used that $X$, $h$, $s$, are all real, so that
\be
\label{nosLieofvarphi}
s \neq 0 %
\qquad \implies \qquad 
\dd \varphi = 0 
\quad \implies \quad 
{\cal{L}}_{K} \varphi = \iota_K \dd \varphi=0\, .
\ee
Moreover, due to Proposition \ref{dxtransverse}, 
and since $h$ is independent of time 
due to the isometry, namely
\be
\label{Lieofh}
{\cal{L}}_{K} h = \iota_K \dd h = 0\, ,
\ee 
then eq.~\eqref{diffs11nognoF} also yields 
\be 
s \neq 0 \qquad \implies \qquad 
\iota_K \left( \dd s + i s \dd\varphi \right) = 0\, ,
\ee
so that when $g=0$, or $F=0$ with the gauge choice $A=0$, 
or $\chi^A \in \Ker\Gamma_+$, it generally holds
\be 
\label{Lieoftheta} 
{\cal{L}}_{K} s = \iota_K \dd s = 0
\qquad \implies \qquad 
{\cal{L}}_{K} \theta = \iota_K \dd \theta 
= 0\, .
\ee

\begin{prop}[Vanishing Complex Scalar Bilinear]
\label{vanishingcscalar}	
The complex scalar bilinear vanishes identically 
when the gauging and Cartan flux are non-zero, 
for a USp$(4)$ spinor $\chi^A \in \Ker \Gamma_-$. 
Thus
\be 
\label{nocscalar}
g,F\neq0\, , \quad \chi^A \in \Ker \Gamma_- 
\qquad \implies \qquad s^{11}=0\, .
\ee	
\end{prop}

\begin{proof}
Applying an exterior derivative on the differential relation 
of the complex scalar bilinear in eq.~\eqref{diffs11}, 
we obtain 
\be 
\label{integrableproof}
0 = g \Big[ \dd \left( X s^{11} \right) \wedge A 
+ \left( X s^{11} \right) F \Big] \,\Gamma_{+} \chi
= g \left( X s^{11} F \right) \Gamma_{+} \chi\, ,
\ee
where in the second equality we used eq.~\eqref{diffs11} 
again in order to substitute for $\dd (Xs^{11})$. 
Thus, either $\chi^A \in \Ker \Gamma_+$, 
in which case Proposition \ref{uniquespinor} implies $F=F^{(1,1)_0}$, 
or $\chi^A \in \Ker \Gamma_-$.
When the latter holds, we return to eq.~\eqref{integrableproof}, and since 
$X,g,F\neq0$, it is implied that $s^{11} = 0$. 
\end{proof}
This result similarly holds in the minimal U$(1)$-gauged theory,
\textit{independently of any gauge choice}. 
As noted in Section~\ref{gstructure}, a vanishing complex
scalar bilinear simplifies the geometric bilinear relations due to 
eq.~\eqref{Fierzintprod2}: 
The complex 3-form bilinear is then simply proportional 
to the holomorphic 3-form of the structure, as in eq.~\eqref{u11to6d}, 
whereas the real 3-form contains then 
only lower-degree forms, as in eq.~\eqref{u3to6d}. 
Moreover, this result considerably simplifies the intrinsic torsion 
of the base, most critically, as it implies that the base is integrable, 
which will be discussed in Section~\ref{intrinsictorsion}. 

We thus deduce from eq.~\eqref{diffs11} and from 
Proposition \ref{vanishingcscalar} that generally it holds that
\be
\dd \left( Xs^{11} \right) = 0\, , 
\ee
where we have taken $A=0$ when $F=0$.

Next, for the real scalar bilinear $S_3$ 
we obtain the differential relation
\begin{align}
\label{diffS3}	
0 = \left[ \dd \left( X S_3 \right) 
+ \iota_{K} F \,\Gamma_{+} \right]\chi 
= \left[ - \dd \left( X h c \right) 
+ \iota_{K} F \,\Gamma_{+} \right]\chi\, ,
\end{align}
where in the second equality we used eq.~\eqref{complicateS3}. 
Using eq.~\eqref{br0r}, 
which holds regardless of gauging as already noted, 
it is readily inferred that 
\be
\label{LieofS3}
{\cal{L}}_{K} S_3 = \iota_K \dd S_3 = 0\, ,
\ee
so that $S_3$ is also preserved by the Killing isometry. 

Finally, using the Bianchi identity for $F$ in eq.~\eqref{Bianchi}, 
it is also readily inferred from eq.~\eqref{diffS3}, and from  
Proposition~\ref{uniquespinor} in Section~\ref{uniqueF}, 
that for either spinor branch it holds that
\be
\label{LieofF}
{\cal{L}}_{K} F = 0\, . 
\ee

Let us highlight that \textit{no gauge choice is required 
when there is non-zero gauging and a Cartan flux, in order to ensure 
that the scalar bilinears are preserved by the Killing
isometry}, i.e.~the results in eqs.~\eqref{Lieofs11}, 
\eqref{nocscalar}, \eqref{LieofS3}.

Rather, if one would like to additionally ensure that for 
$F\neq0$, $\chi^A \in \Ker \Gamma_-$, 
the gauge potential itself is preserved by the Killing isometry -- 
to be referred to as Killing-invariant henceforth --
then eq.~\eqref{diffS3} suggests to choose 
the Killing gauge for $A$, namely
\be 
\label{gaugechoice}
\pm \iota_K A = XS_3\, , 
\ee 
which then indeed yields 
\be 
{\cal{L}}_{K} A=0\, .
\ee 
Moreover, as will be discussed below, this gauge additionally
ensures that in this case when $g\neq0$, the 3-form bilinears
are also Killing-invariant. 

Note that the Killing gauge in eq.~\eqref{gaugechoice} is an 
incomplete gauge, 
as using eqs.~\eqref{Fierznorms}, \eqref{realscalarto6d}, it is 
easily inferred that the 1-form gauge field is of the form:
\be
\label{gauged1form}
\pm A = c h^{-1} X K + A_{\perp}\, ,
\ee
where $A_{\perp}$ is the component of $A$ that is 
transverse to the Killing vector, namely
\be
\label{remaindergauge}
\iota_K A_{\perp} = 0\, ,
\ee
which thus still allows for time-independent gauge transformations, 
as it is easy to verify that it also holds that 
\be
\label{perpAnotime}
\iota_K \dd A_{\perp} = 0\, .  
\ee
Note that Proposition \ref{vanishingcscalar} implies $c=\pm1$ 
in the gauge potential in eq.~\eqref{gauged1form}.
We also recall that due to the generator $\Gamma_+$, 
or equivalently R-symmetry, 
$A$ is allowed to flip sign on the LHS of eq.~\eqref{gauged1form}.

\subsubsection{One- and Two-Form Bilinears}
\label{diff1-2forms}

For the single 1-form bilinear, which is dual to the Killing vector, 
we obtain the differential relation:
\begin{align}
\label{diffKraw}
\left( X^2 \nabla_\mu K_\nu \right) \chi 
= \Big[ 
& -2 X \dd X_{[\mu} K_{\nu]} 
- X S_3 F_{\mu\nu} \,\Gamma_{+} 
+ \frac{1}{2} X^4 \left( \iota_K \star G \right)_{\mu\nu} \,\Gamma_5
+ g X^{-2} T_{\mu\nu} \Big] \chi\, ,
\end{align}
where thanks to eq.~\eqref{LieofX} a term that is symmetric in the 
indices 
$(\mu\nu)$ dropped out, so that $\nabla_{(\mu} K_{\nu)} = 0$, and $K$ 
is indeed a Killing vector satisfying
\be
{\cal{L}}_{K} g_{\alpha\beta} = 0\, . 
\ee
Rearranging the terms in eq.~\eqref{diffKraw} yields the 
following simpler form:
\be
\label{diffK}	
\left[ \dd \left( X^2 K \right) \right] \chi
= \Big[ 2 c h X F \,\Gamma_{+} 
+ X^4 \iota_K \star G \,\Gamma_5
+ 2 g h X^{-2} J \Big] \chi\, , 
\ee
where we also used eqs.~\eqref{complicateS3}, \eqref{TtoJ}, 
in the first and last terms, respectively. 

For the single 2-form bilinear we get the following differential relation:
\begin{align}
\left( X^{-2} \nabla_{\mu} T_{\nu\sigma} \right) \chi = \Big[ 
& -3 X^{-3} \dd X_{[\mu} T_{\nu\sigma]} 
- 2 X^{-3} g_{\mu[\nu} \left( \dd X \lrcorner T \right)_{\sigma]}
+ X^{-3} F^{\alpha}{}_{\mu} \left( U_3 \right)_{\nu\sigma\alpha} \Gamma_{+} 
\nn\\ & 
+ 2 g X^{-6} g_{\mu[\nu}K_{\sigma]}
+ T^{\alpha}{}_{[\nu} \left( \star G \right)_{\sigma]\mu\alpha} \Gamma_5 \Big] 
\chi\, .
\end{align}
Using the 3-rank algebraic relation in eq.~\eqref{br3r} and 
further rearranging terms, we then obtain the following simpler form 
for the differential relation:
\be
\label{diffT}
\left[ \dd \left( X^{-2} T \right) \right] \chi 
= - \left( \iota_K G \right) \Gamma_5 \,\chi\, .
\ee
Using the Bianchi identity for $G$ in eq.~\eqref{Bianchi}, 
it is then readily inferred that
\be 
{\cal{L}}_{K} G = 0.
\ee
It is also readily inferred from eq.~\eqref{diffT}, using 
eqs.~\eqref{LieofX}, 
\eqref{Fierzintprod1}, \eqref{TtoJ}, \eqref{Lieofh}, that
\be
\label{LieofT}
{\cal{L}}_{K} T  = \iota_K \dd T = 0\, , 
\qquad {\cal{L}}_{K} J = \iota_K \dd J = 0\, ,
\ee
thus $dJ$ is transverse, as expected. 
We can thus conclude at this point that 
\textit{the Killing vector generates a symmetry 
of the full solution for $\{X, F, G,g_{\alpha\beta}\}$}. 

The differential relation in eq.~\eqref{diffT} further 
suggests to choose a Killing gauge also for the 3-form gauge field. 
When there is no gauging there is no self-duality condition imposed on $B$, 
and the Killing gauge simply reads 
\be
\label{Bgaugenog}
\pm \iota_K B = X^{-2} T\, .
\ee
It is easy to see that this gauge choice ensures 
that the 3-form gauge field itself 
is also Killing-invariant, namely
\be 
\label{LieofB}
{\cal{L}}_{K} B = 0\, .
\ee 

Similar to the 1-form in eq.~\eqref{gaugechoice}, the Killing gauge 
for the 3-form in eq.~\eqref{Bgaugenog} is incomplete, as using it 
together with eqs.~\eqref{Fierznorms}, \eqref{TtoJ}, \eqref{Fierzintprod1}, 
yields the 3-form gauge field as
\be
\label{gauged3formnog}
\pm B = - h^{-1} X^{-2} K \wedge J 
+ B_{\perp}\, ,
\ee 
where $B_{\perp}$ is the component of $B$ transverse to the Killing vector, 
namely
\be 
\label{Bremaindergauge}
\iota_K B_{\perp} = 0\, ,
\ee
which thus still allows for time-independent gauge transformations, 
as it is easy to confirm that it also holds that 
\be
\label{perpBnotime}
\iota_K \dd B_{\perp} = 0\, . 
\ee

Yet, when there is gauging, the consistency of such a Killing gauge
with the self-duality condition on $B$ in eq.~\eqref{selfdual} 
must be maintained.
We first consider $F\neq0$, $\chi^A\in\Ker\Gamma_-$, 
as in Proposition~\ref{vanishingcscalar}.
For a fixed spinor component in this branch, $\Gamma_+\chi = \pm \chi$, 
$\Gamma_5\chi = \chi$.
From eqs.~\eqref{eqchargelim}, \eqref{someCS}, in the equal-charge sector 
we have $\omega_{(3)}=A\wedge F$.
Using eq.~\eqref{diffS3}, the Killing gauge of $A$ in eqs.~\eqref{gaugechoice},
\eqref{gauged1form}, and $c^2=1$, we then find
\be
\label{KCScontraction}
\iota_K\omega_{(3)}
=
c\,\dd\left(XhA_\perp\right)
\mp2chXF
+\dd\left(X^2K\right)\, , 
\ee
where the upper and lower signs correspond to the two eigenvalues of $\Gamma_+$, 
independently of the value $c = \pm 1$, implied by 
Proposition~\ref{vanishingcscalar}.

We now use the gauge freedom in the auxiliary 2-form, 
that was introduced in the self-duality condition, to choose
\be
\label{C2Killinggauge}
\iota_K\dd C_{(2)}=0\, .
\ee
This fixes only the longitudinal component of $\dd C_{(2)}$, 
while leaving its transverse component subject to the self-duality condition.
With this choice, and using eq.~\eqref{diffK}, we find that the self-duality 
condition is consistent with the following simultaneous Killing gauge of $B$:
\be
\label{Bgaugechoice}
\iota_KB
= X^{-2} T + \frac{c}{2g}\,
\dd\left( X h A_\perp \right)\, .
\ee
Indeed, using eq.~\eqref{diffT}, and $c=\pm1$, it is easy to see 
that this Killing gauge ensures that the 3-form is Killing-invariant, 
namely that eq.~\eqref{LieofB} is satisfied. 
Note that the $g \to 0$ limit in eq.~\eqref{Bgaugechoice} should not be a 
concern, as to begin with, this gauge choice relies on 
the self-duality condition, which is irrelevant without gauging. 

Using eqs.~\eqref{perpAnotime}, \eqref{remaindergauge}, we find that
the corresponding incomplete gauge representative of the 3-form 
can be written as 
\be
\label{gauged3form}
B
=
- h^{-1}X^{-2}K\wedge J
-\frac{c}{2g}K\wedge
\left[
h^{-1}X\,\dd A_\perp
+h^{-2}\dd(hX)\wedge A_\perp
\right]
+B_\perp\, ,
\ee
where $B_\perp$ satisfies eqs.~\eqref{Bremaindergauge}, \eqref{perpBnotime}, 
in this case as well.
Its transverse component $B_{\perp}$ is not free, but is fixed by 
the transverse part of the self-duality condition.

When there is gauging, but $F=0$,  $\chi^A \in \Ker \Gamma_-$, 
or $\chi^A \in \Ker \Gamma_+$, 
we can still look for a Killing gauge that is consistent with the self-duality 
condition, following the same considerations noted above. 
For $F=0$ we take $A=0$, whereas for
$F=F^{(1,1)_0}\equiv\tau$ we take $A=A_\perp$ and
$\iota_KA_\perp=0$. 
Since $\iota_K\tau=0$, in both cases
\be
\iota_K\omega_{(3)} = 0\, .
\ee
Accordingly, the same auxiliary gauge choice in
eq.~\eqref{C2Killinggauge} applies.
Therefore, contracting eq.~\eqref{selfdual} with $K$, 
and using eq.~\eqref{diffK} yields
\be
\label{BgaugechoicenoF}
\pm \iota_K B 
= X^{-2} T - \frac{1}{2g}\,
\dd \left( X^2 K \right)\, ,
\ee
which is just the differential relation of the 1-form dual 
to the Killing vector, that satisfies eq.~\eqref{LieofB}.

In this case, the 3-form gauge field takes the form
\be
\label{gauged3formnoF}
\pm B = - h^{-1} X^{-2} K \wedge J 
+ \frac{1}{2g} 
h^{-2} X^2 K \wedge \dd K
+ B_{\perp}\, ,
\ee
with $B_{\perp}$ satisfying 
eqs.~\eqref{Bremaindergauge}, \eqref{perpBnotime} also in this case.
Also here $B_{\perp}$ is not free: 
it is fixed by the transverse self-duality condition in eq.~\eqref{selfdual}.
Note the unique term
\be 
K \wedge \dd K\, ,
\ee 
which shows up in this case, that stands for the twist of the Killing vector. 
We will return to it when we consider the Killing invariance 
of the 3-form bilinears, where there is gauging and 
$F=0$, $\chi^A \in \Ker \Gamma_-$, 
or $\chi^A \in \Ker \Gamma_+$.
Finally, recall that in this case, $B$ can flip sign on the LHS of 
eq.~\eqref{gauged3formnoF}, due to the two possible spinor branches 
and the operator $\Gamma_5$, or equivalently R-symmetry, 
similarly to what was noted for $A$.

\textit{The aspects noted above, regarding the Killing invariance
of the $3$-form gauge field, the relation of this invariance to the self-duality
condition with gauging, and the form of the $3$-form gauge field in the three
distinct cases above, similarly hold in the minimal theory.}

\subsubsection{Three-Form Bilinears}
\label{diff3forms}

As to the 3-form bilinears, we can focus on the complex 
one, as it always contains the complex 3-form bilinear 
of the SU$(3)$ structure on the base, %
as in eq.~\eqref{u11to6d}, whereas the real one can consist of
only lower-form bilinears that we already analysed, as in eq.~\eqref{u3to6d}.
For the complex 3-form bilinear $u^{11}$
we then get the differential relation
\begin{align}
\label{diffu11}	
\left( X^{-1} \nabla_{\mu} u^{11}_{\nu\sigma\rho} \right) \chi 
& = 
\Big[ - 4 X^{-2} \dd X_{[\mu} u^{11}_{\nu\sigma\rho]} 
- 4 i g X^{-1} A_{\mu} u^{11}_{\nu\sigma\rho} \,\Gamma_+
+ g X^{-5} \left(\star u^{11}\right)_{\mu\nu\sigma\rho}
\nn\\
& \qquad
+ i X^{-2} F^{\alpha}_{\,\,\,\mu} 
\left(\star u^{11}\right)_{\nu\sigma\rho\alpha} \Gamma_{+} 
- \frac{3}{2} X \left( u^{11} \right)^\alpha{}_{[\nu\sigma}
\left( \star G \right)_{\rho]\mu\alpha} \Gamma_5
\Big] \chi\, .
\end{align}
Totally antisymmetrizing this equation 
gives
\begin{align}
\label{diffu11exterior}
\left[ \dd \left( X^{-1} u^{11} \right)_{\mu\nu\sigma\rho}\right] \chi
= {} 
\Big[ &
- 20 X^{-2}\dd X_{[\mu} u^{11}_{\nu\sigma\rho]}
- 16 i g X^{-1} A_{[\mu}u^{11}_{\nu\sigma\rho]} \Gamma_+
+ 4 g X^{-5} \left( \star u^{11} \right)_{\mu\nu\sigma\rho}
\nn\\
& 
+ 4 i X^{-2} F^\alpha{}_{[\mu} 
\left( \star u^{11} \right)_{\nu\sigma\rho]\alpha} \Gamma_+
+ 6 X \left( u^{11} \right)^\alpha{}_{[\mu\nu} 
\left(\star G \right)_{\sigma\rho]\alpha} \Gamma_5 \Big] \chi\, .
\end{align}
Using the algebraic relation in eq.~\eqref{br4cg}, and further 
rearranging terms, we then obtain the following simpler form for the 
differential relation:
\begin{align}
\label{diffU}	
\left[ \dd \left( X^{-1} u^{11} \right) \right] \chi
& = \Big[ X s^{11} G \,\Gamma_5 
- 4 g \Big( \left[ i X^{-1} A \wedge u^{11} \right] \Gamma_+ 
- \star u^{11} \Big) \Big] \chi \, .
\end{align}
Note the term with the complex scalar $s^{11}$ in this equation, 
which arises from the algebraic relation in eq.~\eqref{br4cg}, 
as the result in Proposition \ref{vanishingcscalar} is relaxed
in deriving this general relation. 
Also notice the similarity to the differential relation of 
the complex scalar bilinear in eq.~\eqref{diffs11}, 
featuring an imaginary coefficient. 
Yet, in contrast to eq.~\eqref{diffs11}, 
when $g,F\neq0$, this equation already indicates 
that $\dd u^{11}$, and recalling eq.~\eqref{u11to6d}, 
also the exterior derivative of the complex 3-form 
bilinear of the structure, involve some gauge-dependence.

We proceed to examine whether the complex 3-form bilinears
are also Killing-invariant. 
Let us first consider eq.~\eqref{diffU} without gauging, 
where as noted, the result in Proposition~\ref{vanishingcscalar} is relaxed. 
Then using eqs.~\eqref{br0r}, \eqref{Fierzintprod2}, 
and due to eq.~\eqref{LieofX}, we find that
\be
{\cal{L}}_{K} \left( X^{-1} u^{11} \right) \chi 
= X s^{11} \left[ \dd \left (X^{-2} T \right) 
+ \iota_K G \Gamma_5 \right] \chi
= 0\, ,
\ee
where we also used eqs.~\eqref{diffs11}, \eqref{diffU}, 
in the second equality, and eq.~\eqref{diffT} in the last, 
and there occurs a non-trivial cancellation.
Using eq.~\eqref{LieofX}, 
it is thus easily inferred that
\be 
\label{Lieofu11}
0 = {\cal{L}}_K u^{11} 
= {\cal{L}}_K \Bigl[- s p K \wedge J 
- \frac{h}{2} \Bigl( \left( c-1 \right) p^2 \Psi
+ \left( c+1 \right) \overline{\Psi} \Bigr) \Bigr]\, , 
\ee
where we substituted in the general form of the complex 3-form 
bilinear from eq.~\eqref{complicateu11}. 
Due to eq.~\eqref{LieofT}, and since the Lie derivative 
of the scalar angle-dependent coefficients maintains the 
longitudinal/transverse and the complex form-types, we infer that 
eq.~\eqref{Lieofu11} yields independent equations for each of the types 
$K \wedge (1,1)$, $(3,0)$, or $(0,3)$.
Thus, using eqs.~\eqref{Lieoftheta}, \eqref{transversesu3}, 
we find
\begin{subequations}
\begin{align}
\label{Lieofangles}
0 = & s \, {\cal{L}}_K \varphi\, ,\\*
\label{rawLieofPsi}
0 = & \left( c + 1 \right) {\cal{L}}_K \overline{\Psi}\, ,
\end{align}  
\end{subequations}
which implies
\be
\label{LieofPsi}
{\cal{L}}_K \Psi 
= {\cal{L}}_K \bar \Psi 
= \iota_K \dd \Psi = \iota_K \dd\bar\Psi = 0\, .
\ee
For $\theta=\pi/2$, the positive-chirality spinor defining the
SU$(3)$ structure is phase-redefined as
$\widetilde{\psi}=p^{-1}\psi$, as in
eq.~\eqref{simpleisdeep2}. Using eq.~\eqref{defPsi}, its fundamental
complex 3-form is therefore
$\widetilde{\Psi} = p^2 \Psi$.
The $(3,0)$ component of eq.~\eqref{Lieofu11} then gives directly
$\mathcal{L}_K\widetilde{\Psi} = 0$.
Relabelling the structure form $\widetilde{\Psi}$ as $\Psi$ 
gives the final result in eq.~\eqref{LieofPsi}.
Therefore $\dd \Psi$ is also transverse, as expected.

Next, similarly using the general form of 
the real 3-form bilinear in eq.~\eqref{complicateU3}, 
and using eqs.~\eqref{Lieofangles}, \eqref{LieofT}, 
\eqref{Lieofh}, \eqref{LieofPsi},
it is also easily inferred that
\be 
\label{LieofU3}
{\cal{L}}_{K} U_3 = 0\, .
\ee
At this point we conclude that \textit{when there is no gauging,
the Killing vector generates a symmetry not 
only of the full solution of the bosonic fields, 
but also of all the bilinears, and the structure}.

Let us now turn to consider first that $g,F\neq0$, 
$\chi^A \in \Ker \Gamma_{-}$. 
The distinct case of $F=0$, $\chi^A \in \Ker \Gamma_{-}$, 
or $\chi^A \in \Ker \Gamma_{+}$, 
will be treated subsequently.
We recall the Killing gauge of the 1-form gauge field in 
eqs.~\eqref{gaugechoice}, \eqref{gauged1form}, \eqref{remaindergauge}, 
eq.~\eqref{u11to6d}, which holds in this case, 
and eq.~\eqref{hodgeu11gen} for the Hodge dual 
of the complex 3-form bilinear that reads
\be
\label{dualofu11}
\star u^{11} = - i K \wedge \overline \Psi 
\quad \mathtt{or} \quad 
\star u^{11} = - i p^2 K \wedge \Psi \, .
\ee
Then applying $\iota_K$, and using eqs.~\eqref{intprodKU11}, 
\eqref{LieofX}, \eqref{Lieofh}, \eqref{transversesu3}, 
it is easily implied that eqs.~\eqref{Lieofu11}, \eqref{LieofPsi}, 
and thus also eq.~\eqref{LieofU3}, 
hold in the Killing gauge.
Note that in this case, due to Proposition \ref{vanishingcscalar}, 
and eq.~\eqref{Fierzintprod2}, we also readily get
\be
\label{diffrentialu11}
{\cal{L}}_{K} u^{11} = \iota_K \dd u^{11} = 0\, .
\ee 
Let us highlight that these are gauge-dependent results, so that 
\textit{when $g,F\neq0$, $\chi^A \in \Ker \Gamma_{-}$, 
and the Killing gauge is applied, the 3-form bilinears 
and the structure are also Killing-invariant}. 

Moreover, we can deduce
\begin{prop}[No Holomorphic Component in $\dd A_\perp$ with Gauging]
\label{noholo1formwg}	
When there is gauging and $F\neq0$, the differential of 
the transverse component of the 1-form gauge field has 
no holomorphic component, namely
\be
g,F\neq0
\quad \implies \quad 
\dd A_\perp = \left( \dd A_\perp\right)^{(1,1)}\, .
\ee
\end{prop}

\begin{proof}
For the spinor branch of $\chi^A \in \Ker \Gamma_+$, the result is trivial as
\be 
\iota_K A = 0\, ,
\qquad
A = A_\perp\, ,
\qquad
F = \dd A_\perp = F^{(1,1)_0}\, ,
\ee 
so we shall focus on $\chi^A \in \Ker \Gamma_-$.	
When $g,F\neq0$ in the Killing gauge, 
using the 1-form in eq.~\eqref{gauged1form}, we get 
\begin{align}
\label{diffUgauged}		
\dd \left( X^{-1} u^{11} \right) 
& = - 4 i g 
\left[ A_\perp \wedge \left(X^{-1} u^{11}\right) \right] %
\, , 
\end{align}
similarly to eq.~\eqref{diffs11} for 
the complex scalar bilinear.
Thus, similar to what we did with the complex scalar bilinear 
in eq.~\eqref{integrableproof}, and using our result there 
in Proposition \ref{vanishingcscalar}, and eq.~\eqref{u11to6d}, 
we take an exterior derivative of eq.~\eqref{diffUgauged}, 
and then substitute eq.~\eqref{diffUgauged} into the resulting equation. 
This derivation then yields 
\be
\label{dAperpconstraint}
0 = \dd A_\perp \wedge \Psi 
\quad \iff \quad 
0 = \dd A_\perp \lrcorner \Psi 
\qquad \iff \qquad 
\dd A_\perp = \left( \dd A_\perp\right)^{(1,1)}\, ,
\ee
where we also used eqs.~\eqref{perpAnotime}, 
\eqref{Hodgeidentity}, \eqref{su3hodge}.
\end{proof}

We proceed to the final distinct case where $F=0$, 
$\chi^A \in \Ker \Gamma_{-}$, or $\chi^A \in \Ker \Gamma_{+}$. 
\begin{prop}[Gauging and Null Cartan Flux yield U$(3)$ Structure]
\label{enlargedstructure}	
When there is gauging, and $F=0$, $\chi^A \in \Ker \Gamma_{-}$, 
or $\chi^A \in \Ker \Gamma_{+}$, 
the complex $3$-form bilinear and the complex $3$-form of the structure 
are no longer Killing-invariant, as the latter is
charged under the gauged U$(1)$ R-symmetry. 
The geometric structure is enlarged then from SU$(3)$ to U$(3)$.
\end{prop}

\begin{proof}
When $g\neq0$, and $F=0$, $\chi^A \in \Ker \Gamma_{-}$, 
or $\chi^A \in \Ker \Gamma_{+}$, 
then we find that the differential relation in eq.~\eqref{diffU} takes the form:
\begin{align}
\label{diffUnoF}
\left[\dd\left(X^{-1}u^{11}\right)\right]\chi
&=
\left[
Xs^{11}G\,\Gamma_5
+4g\star u^{11}
\right]\chi\, .
\end{align}
Then using eqs.~\eqref{diffs11nognoF}, \eqref{Fierzintprod2},
\eqref{diffT}, \eqref{hodgeu11gen}, \eqref{Fierznorms},
\eqref{transversesu3}, \eqref{LieofX}, we infer that
\be
\label{noFLieofu11} 
{\cal{L}}_{K} u^{11} = 2 i g h^2 X 
\Bigl[ \left( 1 - c \right) p^2 \Psi 
+ \left( 1 + c \right) \overline \Psi \Bigr] \neq 0\, , 
\ee
which is always non-vanishing. 
In particular, if we specify $s=0$, $c=1$, we get
\be
\label{nosnoFLieofu11}
{\cal{L}}_{K} u^{11} = - 4 i g h X u^{11}\, . 
\ee
From eq.~\eqref{noFLieofu11}, 
using eqs.~\eqref{LieofT}, \eqref{nosLieofvarphi}, 
\eqref{Lieoftheta}, \eqref{Lieofh}, \eqref{LieofX}, 
we then get
\be 
\label{masternoFLieofpsi}
\left( c - 1 \right) p^2 \mathcal L_K\Psi
+
\left( c + 1 \right) \mathcal L_K\overline\Psi
=
4 i g h X 
\Big[ \left( c - 1 \right) p^2 \Psi 
- \left( c + 1 \right)  \overline\Psi \Big]\, .
\ee
Since $K$ preserves the metric and $J$, its Lie derivative
preserves form type. The $(3,0)$ and $(0,3)$ components of
eq.~\eqref{masternoFLieofpsi} therefore give separate equations.
As $p\neq0$ and $c-1$ and $c+1$ cannot both vanish,
complex conjugacy implies, including at $c=0$,
\be 
\label{LieofpsiwgnoF}
{\cal{L}}_{K} \Psi = 4 i g h X \Psi\, .  
\ee
Note that this result is independent of the internal spinor angle, 
so an identical result is obtained 
if we simply use eq.~\eqref{u11to6d} 
in eq.~\eqref{nosnoFLieofu11}.  
This result describes a constant phase rotation of the complex $3$-form 
on the base along the Killing direction.
Thus, when there is gauging, and $F=0$, $\chi^A \in \Ker \Gamma_{-}$, 
or $\chi^A \in \Ker \Gamma_{+}$, 
the complex seven- and six-dimensional 3-form bilinears 
are not Killing-invariant, and the geometric structure is enlarged to U$(3)$. 	
\end{proof}

In these unique configurations, where the Cartan gauge field is not active, 
or its flux is primitive, the timelike Killing vector 
is more than a pure coordinate translation, 
as it also acts as a generator of the gauged R-symmetry, 
under which the 3-form $\Psi$ is charged. 
Thus geometrically, translation along the Killing direction is accompanied 
by an internal U$(1)$ rotation of the transverse complex $3$-form, 
and thus the structure is enlarged to U$(3)$, 
see Appendix~\ref{gstructure6dwg}.
This is geometrically encoded in the twist of the Killing vector, 
which is naturally sourced by the active 3-form gauge potential, with 
the particular twist term that we found in eq.~\eqref{gauged3formnoF}, 
where the Killing gauge and the self-duality condition were considered in 
conjunction. 
The presence of this twist there is a necessary consequence 
of the geometrization of the R-symmetry rotation 
in the absence of a Cartan gauge field.
The realization of these configurations includes the
BPS flows discussed in Section~\ref{revsol} below.

Finally, still considering $g\neq0$, and $F=0$, $\chi^A \in \Ker \Gamma_{-}$, 
or $\chi^A \in \Ker \Gamma_{+}$, 
we repeat similar steps to what we did in eq.~\eqref{integrableproof}. 
We take an exterior derivative of eq.~\eqref{diffUnoF}, using 
eq.~\eqref{diffs11}, and the Bianchi identity for $G$ in eq.~\eqref{Bianchi}, 
which yields
\be
\label{noFKconstraint}
\dd\left(\star u^{11}\right)=0\, .
\ee
When the complex scalar bilinear is vanishing, 
similar to eq.~\eqref{dAperpconstraint}, the result reads
\be
\label{nosnoFKconstraint}
0 = \dd \left( h^{-1} X K \right) \wedge \Psi 
\qquad \implies \qquad \dd \left( h^{-1} X K \right) 
= \left[ \dd \left( h^{-1} X K \right)\right]^{(1,1)}\, , 
\ee
where we also used eq.~\eqref{diffS3}, and that the form $K$ is dual 
to a Killing vector, so that
\be 
\iota_K \dd \left( h^{-1} X K \right) = \dd (h X) = 0\, .
\ee

\subsection{Intrinsic Torsion of Base Manifold}
\label{intrinsictorsion}

With the differential relations at hand, together with the additional
algebraic constraints, we can read off the intrinsic torsion of the
six-dimensional base. We use the SU$(3)$ decomposition in
eqs.~\eqref{torsionforms}, while treating separately the cases for
which only a U$(3)$ structure is preserved.

For a Killing spinor in either spinor branch, we introduce
\be
\label{etadef}
\begin{cases}
\eta=+1
& \quad \chi^A\in\Ker\Gamma_-\, ,
\\
\eta=-1
& \quad \chi^A\in\Ker\Gamma_+\, .
\end{cases}
\ee
By eq.~\eqref{projections}, this is equivalently the eigenvalue of
$\Gamma_5$,
\be
\label{eigenGamma5}
\Gamma_5\chi^A=\eta\chi^A\, .
\ee
Thus $\eta$ labels the $\Gamma_5$ eigenvalue of the chosen spinor branch,
rather than an eigenvalue of $\Gamma_\pm$, whose kernels also distinguish
the two branches. 

Equation~\eqref{diffT}, together
with eq.~\eqref{TtoJ}, then gives
\be
\label{dJraw}
\dd J
=
- \eta h^{-1}X^2\iota_KG
+\dd\log\left(h^{-1}X^2\right)\wedge J\, .
\ee
Thus terms originating from $G\Gamma_5$ acquire the branch factor
$-\eta$, while the bosonic field strength $G$ itself is fixed.
We can now identify the torsion forms defined in eq.~\eqref{dJ}.

First, using eqs.~\eqref{equivw1}, \eqref{normalize},
\eqref{projchi}, we find
\be
\label{w1fromdj}
\mathcal W_1
=
\frac{i\eta}{6}h^{-1}X^2
\left(\Psi\lrcorner\iota_KG\right)
=
\frac{i\eta}{6}h^{-1}X^2\zeta\, .
\ee
On $\Ker\Gamma_-$, where $\eta=+1$, Proposition
\ref{constrainedsingletsector} and
eqs.~\eqref{chitoreal}, \eqref{gensingletconstraint} give
\be
\label{w1}
\mathcal W_1
=
-\frac16X^{-1}p^{-1}s\lambda\, ,
\qquad
s\equiv\sin2\theta\, ,
\quad
p\equiv e^{i\varphi}\, ,
\ee
where dependence on the spinor internal angles shows up.
Whenever it is non-vanishing, the phase of $\mathcal W_1$ 
is fixed by the spinor phase $\varphi$.
On the complementary branch $\Ker\Gamma_+$,
Proposition~\ref{uniquespinor} gives $\lambda=0$, while
Proposition~\ref{constrainedsingletsector} then gives $\zeta=0$.
Hence eq.~\eqref{w1} holds on either spinor branch. 
From our results in eqs.~\eqref{specsingletconstraint}, 
and Propositions \ref{uniquespinor}, \ref{vanishingcscalar}, 
we then deduce that in the gauged theory it holds
\be
\label{w1gone}
g \neq 0 
\quad \implies \quad 
\mathcal{W}_1 = 0\, . \qquad
\ee
In light of eqs.~\eqref{gensingletconstraint},\eqref{specsingletconstraint}, 
a non-vanishing first torsion form could only arise in the ungauged theory as
\be
\label{w1notgone}
\mathcal W_1\neq0
\quad \implies \quad
g = 0\, ,
\quad %
s, \lambda, \zeta \neq 0\, .
\ee
A non-vanishing $\mathcal W_1$ therefore produces a non-vanishing
Nijenhuis tensor, and hence a non-integrable base, see
Appendix~\ref{almostcalculus}.

Next, projecting the $G$-dependent term in eq.~\eqref{dJraw} onto
the primitive $(2,1)+(1,2)$ component gives
\begin{align}
\label{w3}
\mathcal W_3
&=
- \eta h^{-1}X^2
\left(\iota_KG\right)^{(2,1)_0+(1,2)_0}
= - \eta h^{-1}X^2
\left(\rho+\overline\rho\right)\, .
\end{align}
The component $\rho$ remains unconstrained by the bilinear relations,
see Proposition~\ref{unconst3formprim}. 
Although eq.~\eqref{w3} contains an explicit factor of $\eta$, 
this does not represent a sign ambiguity in the geometric torsion. 
Both $\rho$, as the primitive projection of $\iota_KG$, 
and $\mathcal W_3$, as the primitive projection of $\dd J$, 
are defined with respect to the SU$(3)$ structure induced 
by the Killing spinor. 
Only the combination
$\eta(\rho+\overline\rho)$ is therefore geometrically meaningful.
For a fixed bosonic configuration admitting both complementary
spinor branches with the same $K$ and $J$, eq.~\eqref{dJraw} instead
forces $\iota_KG=0$, and hence $\mathcal W_3=0$.

As the last torsion form featured in $\dd J$, the fourth one is
\be
\label{w4}
\mathcal W_4
=
- \frac{\eta}{2}h^{-1}X^2
J\lrcorner\left(\iota_KG\right)
+\dd\log\left(h^{-1}X^2\right)\, .
\ee
On the $\Ker\Gamma_-$ branch, where $\eta=+1$, the algebraic
constraints in eqs.~\eqref{gen1vectorconstraint}, together with
eqs.~\eqref{diffs11nognoF}, \eqref{diffS3},
\eqref{complicates11}, give
\be
X^3\left(q+\overline q\right)
=
- X\dd h + 4h\dd X
- hs\,\text{Im}(p\alpha)\, .
\ee
Substitution into eq.~\eqref{w4} yields
\be
\label{w4gen}
\mathcal W_4
=
-\frac12\dd\log h
+\frac{s}{2X}\text{Im}(p\alpha)\, .
\ee
On the complementary branch $\Ker\Gamma_+$, the Cartan flux is
necessarily primitive and $\alpha=0$. The corresponding algebraic
projection gives
\be
q=5\eta hX^{-3}D\, ,
\ee
while the differential relations give $\dd(Xh)=0$. The explicit
factor of $\eta$ in eq.~\eqref{w4} therefore cancels that in $q$,
again yielding
\be
\mathcal W_4=-\frac12\dd\log h\, .
\ee
Thus eq.~\eqref{w4gen}, with $\alpha=0$, also gives the result on
the complementary spinor branch.

For $g\neq0$, Propositions~\ref{uniquespinor}, \ref{vanishingcscalar}, 
and the algebraic constraints in eqs.~\eqref{spec1vectorconstraint} 
leave either $s=0$, or $f=\alpha=0$. 
Thus the fourth torsion form equals
\be
\label{w4exactornot}
\mathcal W_4
=
\begin{cases}
-\frac12\dd\log h
+\frac{s}{2X}\text{Im}(p\alpha)
& g=0\, ,
\\
-\frac12\dd\log h
& g\neq0\, .
\end{cases}
\ee
Hence $\mathcal W_4$ is necessarily exact in the gauged theory,
whereas in the ungauged theory it may not be exact.
Moreover, in the gauged theory explicit dependence on the dilaton can also be 
ruled out, and only that on the running warp factor $h$ remains.

Analysing the exterior derivative of $\Psi$ is considerably more involved.
We begin by extracting its $(2,2)$-type component, which contains the
Nijenhuis tensor, see eqs.~\eqref{dpsi}, \eqref{dpsi22},
\eqref{nijenhuis}. As we shall confirm later, when $s=0$, 
the $(2,2)$ projection of eq.~\eqref{diffU}
immediately gives $(\dd\Psi)^{(2,2)}=0$. 
We therefore focus on the sectors where $s\neq0$ remains possible.
By Propositions \ref{uniquespinor}, \ref{vanishingcscalar}, these consist
of $\chi^A\in\Ker\Gamma_-$ with either $g=0$ or $F=0$, and the
complementary branch $\chi^A\in\Ker\Gamma_+$. 

We first consider the first component of the former branch, 
for which $\Gamma_+\chi=\Gamma_5\chi=\chi$; 
the complementary branch will be treated separately below.
We thus return to the general differential relation in eq.~\eqref{diffU},
allowing for $s^{11}\neq0$, and substitute the general angle-dependent
bilinears and their Hodge duals from eqs.~\eqref{genericscalars},
\eqref{generic3forms}, \eqref{hodgeu11gen}. 
The gauge-dependent term, $A\wedge u^{11}$, has no transverse $(2,2)$ component, 
independently of the gauge choice. 
Taking the transverse $(2,2)$ projection and combining the resulting equation 
with its complex conjugate gives a linear system for $(\dd\Psi)^{(2,2)}$ and
$(\dd\overline\Psi)^{(2,2)}$, whose determinant is $-4c\neq0$.
Solving this system yields
\be
\label{rawdpsi22}
( \dd \Psi )^{(2,2)} 
= p^{-1} \,\frac{s}{c} \left( - X^2 G_{\perp} 
- h^{-1} (\dd K)_\perp \wedge J + 2 g X J \wedge J \right)^{(2,2)} \, .
\ee
This equation therefore applies on the $c\neq0$ angle domain, 
precisely where this linear system is invertible.
We now reduce eq.~\eqref{rawdpsi22} on this spinor branch. We use the
transverse projection of the differential relation for $K$ in
eq.~\eqref{diffK}, the Hodge-dual split in
eq.~\eqref{hodgedualsplit}, the SU$(3)$ decompositions of $F$ and
$G$ in eqs.~\eqref{f2su3}, \eqref{f4su3}, respectively, 
and the SU$(3)$ Hodge duals in eqs.~\eqref{su3hodge}, \eqref{primHodgeJP}.
We then obtain
\begin{align}
\label{ourdpsi22}	
(\dd \Psi)^{(2,2)} = - X^{-1} p^{-1} \, \frac{s}{c}  
\biggl( & \left[ \frac23 c \lambda
+ \frac{1}{4} X^3 \sigma
- 2 g \left( X^2 - X^{-3} \right) \right] J \wedge J 
+ 2 c \tau \wedge J \biggr)\, .
\end{align}

On the complementary spinor branch $\chi^A\in\Ker\Gamma_+$, 
due to Proposition \ref{uniquespinor},
all $F$-dependent terms vanish, whereas the $G$-dependent terms reverse
sign. Using the corresponding algebraic constraints 
on the singlet sector, as in eq.~\eqref{gensingletconstraint}, namely
\be
8 g \Delta + X^3 \sigma = 0\, ,
\ee
we therefore obtain
\be
\label{ourdpsi22kerplus}
\chi^A \in \Ker\Gamma_+ 
\qquad \implies \qquad 
(\dd\Psi)^{(2,2)} = 0\, , 
\ee
so that on this spinor branch a $(2,2)$-component in the torsion 
identically vanishes, thus the Nijenhuis tensor vanishes, 
and the base is integrable. 

Returning to the spinor branch $\chi^A\in\Ker\Gamma_-$, 
and recalling eq.~\eqref{equivw1}, where the first torsion form 
is defined equivalently to eq.~\eqref{w1fromdj}, 
via its extraction from eq.~\eqref{ourdpsi22}, we also have
\begin{align}
\mathcal{W}_1 & 
= \frac{1}{12} J \lrcorner J \lrcorner \dd \Psi \nn\\
& = - X^{-1} p^{-1} \frac{s}{c}
\left( \frac23 c \lambda
+ \frac{1}{4} X^3 \sigma
- 2 g \left( X^2 - X^{-3} \right) \right)
= - \frac16 X^{-1} p^{-1} s \lambda\, , 
\end{align}
which is consistent with eq.~\eqref{w1}, 
and where in the last equality we used the algebraic constraints 
on the singlet sector in eq.~\eqref{gensingletconstraint}.

Next, extracting the second torsion form on this spinor branch, we get
\be
\label{w2}
\chi^A\in\Ker\Gamma_-
\qquad \implies \qquad
\mathcal W_2
= \left(J\lrcorner\dd\Psi\right)^{(1,1)_0}
= -2X^{-1}p^{-1}s\tau\, ,
\ee
where the phase of this second complex torsion form is also fixed 
by the spinor phase $\varphi$.
On the complementary spinor branch
$\chi^A\in\Ker\Gamma_+$, the second torsion form vanishes
identically, namely $\mathcal W_2=0$, 
as follows from eq.~\eqref{ourdpsi22kerplus}.

To complete the derivation of $\mathcal W_2$ for the case $c=0$,
considering first the branch $\chi^A\in\Ker\Gamma_-$, we contract $\mu,\nu$
in eq.~\eqref{diffu11} and take the primitive $(1,1)$ projection:
\be
\left(\nabla^\mu u^{11}_{\mu ij}\right)^{(1,1)_0}
=-\left(\star\dd\star u^{11}\right)^{(1,1)_0}_{ij}
=-2ihspX^{-1}\tau_{ij}\, .
\ee
The first equality follows from eq.~\eqref{Hodgeidentity}
with $A=e^a$, $B=\nabla_a u^{11}$, using
$\nabla_a\star=\star\nabla_a$ and
$\dd=e^a\wedge\nabla_a$ for the Levi--Civita connection.
The second equality follows from eq.~\eqref{diffu11} 
using antisymmetry and type of terms, 
eq.~\eqref{primitiveJPantisym} to drop the $G$-dependent term, and
eq.~\eqref{primitivePJJsq} to reduce the surviving $F$-term.
At $c=0$, eq.~\eqref{hodgeu11gen} gives
$\iota_K\star u^{11}=ih^2p\,\operatorname{Re}(p\Psi)$, and
from \eqref{Lieofs11}, \eqref{LieofT}, we also 
have $(\mathcal L_K\star u^{11})^{(2,2)}_\perp=0$, so
Cartan's formula and eq.~\eqref{dpsi22} therefore give
\begin{align}
	-ih^2p\,\operatorname{Re}(p\mathcal W_2)
	&=\left[J\lrcorner\iota_K\dd\star u^{11}\right]^{(1,1)_0}
	\nn\\
	&=h\left(\star\dd\star u^{11}\right)_\perp^{(1,1)_0}
	=2ih^2spX^{-1}\tau\, ,
\end{align}
where for the first equality differentiation precedes setting $c=0$: 
derivatives of scalar coefficients have types $(3,1)+(1,3)$ and drop out.
The second equality uses eq.~\eqref{Hodgeidentity}
with $A=K$, $B=\dd\star u^{11}$, followed by
eqs.~\eqref{hodgedualsplit}, \eqref{primHodgeP},
\eqref{primitiveJJP}.
Together with $\operatorname{Im}(p\mathcal W_2)=0$
from eqs.~\eqref{complicateu11}, \eqref{diffU}, 
this coincides with eq.~\eqref{w2} at $c=0$.
On the complementary branch $\Ker\Gamma_+$, the $F$-term 
in eq.~\eqref{diffu11} is absent, giving $\mathcal W_2=0$.

Using Proposition~\ref{vanishingcscalar} on the branch $\chi^A\in\Ker\Gamma_-$,  
we then generally get from eq.~\eqref{w2} that
\be
\label{w2gone}
g\neq0
\quad \implies \quad
\mathcal W_2=0\, .
\ee
Together with eq.~\eqref{w1gone}, this implies that the Nijenhuis
tensor vanishes and the base is integrable. 
By contrast, a non-vanishing second torsion form could only arise 
in the ungauged theory as
\be
\label{w2notgone}
\mathcal W_2\neq0
\quad \implies \quad
g=0\, ,
\quad
s, \tau \neq 0\, ,
\quad
\chi^A\in\Ker\Gamma_-\, .
\ee
Together with eq.~\eqref{w1notgone}, we deduce that the base 
is non-integrable as  
\be
\label{noninteg}
N_J\neq0
\quad \iff \quad
g=0\, ,
\quad
s\neq0\, ,
\quad
(\lambda,\tau)\neq(0,0)\, ,
\quad
\chi^A\in\Ker\Gamma_-\, .
\ee

As established in Appendix~\ref{almostcalculus}, the transverse
exterior derivative of a $(3,0)$ form in complex dimension
three has only $(3,1)$ and $(2,2)$ components. Therefore,
independently of integrability,
\be
\label{gendpsi13}
(\dd\Psi)^{(1,3)}=0\, .
\ee
Still on the spinor branch $\chi^A\in\Ker\Gamma_-$,
a seemingly forbidden $(1,3)$ component arises in $\dd\Psi$.
Projecting eq.~\eqref{diffU} and its complex conjugate onto this
component, and combining them without dividing by $c$, gives
\begin{align}
\label{ourdpsi13}	
c (\dd \Psi)^{(1,3)}&%
= - \frac{1}{4} X^{-1} p^{-1} \, s
\Bigl[ 2 h^{-1} s p^{-1} f + i c \alpha - X^3 \beta \Bigr] 
\wedge \overline \Psi
=0 \, ,
\end{align}	
where the last equality follows from
eq.~\eqref{prim2formbr2rgproj}. 
Thus a generic algebraic constraint from the fundamental isotypic component
nullifies this forbidden term.
On the complementary spinor branch
$\chi^A\in\Ker\Gamma_+$, all $F$-dependent terms vanish,
while the corresponding fundamental constraint, with the
$G$-dependent terms reversed, as in eqs.~\eqref{spec1vectorconstraint}, 
implies $\beta=0$.
Hence the forbidden $(1,3)$ component vanishes on that spinor branch as well.

In the gauged configurations with
$F=0$, $\chi^A\in\Ker\Gamma_-$, or $\chi^A\in\Ker\Gamma_+$, 
$\dd\Psi$ also contains the longitudinal component
\be
\label{dpsilong}
\dd\Psi_{||}
=
-4igh^{-1}XK\wedge\Psi
\neq0\, .
\ee
This component was anticipated by Proposition
\ref{enlargedstructure}: in these cases the SU$(3)$ structure is
enlarged to U$(3)$, and $\Psi$ is no longer Killing-invariant, as
shown in eq.~\eqref{LieofpsiwgnoF}, 
independent of the spinor internal orientation.

Finally, for $s\neq0$ in the same sectors, 
we extract the transverse $(3,1)$ projection of eq.~\eqref{diffU}, 
and using eq.~\eqref{gendpsi13} for $(\dd\overline\Psi)^{(3,1)}=0$,
we find on both spinor branches
\be
(\dd\Psi)^{(3,1)}
= \dd \log \left( h^{-1}X \right) \wedge \Psi\, .
\ee

Considering now the complementary sector with $s^{11}=0$, 
we can extract the complete exterior derivative of $\Psi$ 
from the two distinct differential bilinear relations 
in eqs.~\eqref{diffUgauged}, \eqref{diffUnoF}.
On $c=-1$, $\Psi$ denotes the structure form introduced below 
eq.~\eqref{LieofPsi}.
We can now use eqs.~\eqref{u11to6d}, \eqref{dualofu11}, 
and get that in the gauged theory
the complete exterior derivative of $\Psi$ reads:
\be 
g \neq 0 \quad 
\implies %
\dd \Psi = 
\begin{cases}
\bigl[ \dd \log \left( h^{-1} X \right) - 4 i g A_\perp \bigr]
\wedge \Psi & \quad F\neq 0\, , 
\quad \chi^A \in \Ker \Gamma_-\, , \\
\bigl[ \dd \log \left( h^{-1} X \right) - 4 i g	h^{-1} X K \bigr]
\wedge \Psi & \quad F=0 \quad \mathtt{or} \quad \chi^A \in \Ker \Gamma_+\, . 
\end{cases}
\ee
When $s=0$, we now readily see, similar to eqs.~\eqref{w1}, \eqref{w2}, that
$\mathcal W_1=\mathcal W_2=0$, and therefore 
the almost-complex structure is integrable.
We can then identify the fifth torsion form 
from eqs.~\eqref{dpsi}, \eqref{dpsi31},
\eqref{w5complexifiedconj}, 
for the cases $g=0$, or $F\neq0$, $\chi^A \in \Ker\Gamma_-$,
when an SU$(3)$ structure is maintained.

Together with the result obtained above for $s\neq0$, we then get
for the last torsion form in $\dd \Psi$ that
\be
\label{w5wF}
g = 0
\quad \mathtt{or} \quad
\left( F \neq 0\, ,\quad \chi^A \in \Ker\Gamma_- \right)
\quad \implies \quad 
\mathcal{W}_5 = \dd \log \left( h^{-1} X \right) 
- 4 g A_\perp \lrcorner J\, ,
\ee
where we used the $J$ action in eq.~\eqref{1formprojector}, 
as in eq.~\eqref{projector}. 
Note that when $g,F\neq0$, 
the Killing gauge in eqs.~\eqref{gaugechoice}, \eqref{gauged1form}, 
\eqref{remaindergauge}, \eqref{perpAnotime}, is used, yielding 
the second term, whereas when $g = 0$, or $A_\perp=0$, 
the second term drops out, 
and this torsion form becomes exact.
In particular, when $g=0$, whether $s=0$ or not, eq.~\eqref{w5wF} holds as
\be
\mathcal W_5=\dd\log\left(h^{-1}X\right)\, .
\ee
Using eqs.~\eqref{diffs11nognoF}, \eqref{complicates11},
this ungauged expression may be rewritten, when $s\neq0$, as
\be
\mathcal W_5
=
-2\dd\log h
-2\cot(2\theta)\dd\theta\, .
\ee

Let us highlight the independence of intrinsic torsion on the spinor branch. 
Whenever $F\Gamma_+\chi^A\neq0$, only the branch $\chi^A\in\Ker\Gamma_-$ occurs, 
so that no branch ambiguity arises, and any $F\neq0$ is permitted. 
The complementary configurations for which $F\Gamma_+\chi^A=0$ are
\be
\label{nullcartanflux}
\begin{cases}
F=0\, ,& \chi^A\in\Ker\Gamma_-\, ,
\\
F=\tau\, ,& \chi^A\in\Ker\Gamma_+\, .
\end{cases}
\ee
On both configurations one has
\be
\mathcal W_1=\mathcal W_2=0\, ,
\qquad
\mathcal W_4=-\frac12\dd\log h\, ,
\ee
whereas the third torsion form only retains a relative branch sign,
\be
\mathcal W_3
=
- \eta h^{-1}X^2
\left(\rho+\overline\rho\right)\, .
\ee
In the ungauged theory the fifth torsion form is likewise identical
on the two configurations,
\be
\mathcal W_5=-2\dd\log h=4\mathcal W_4\, .
\ee
With gauging, both configurations instead admit a U$(3)$
structure and no intrinsic $\mathcal W_5$. 
Thus the two complementary $F\Gamma_+\chi^A=0$ configurations obey the same
intrinsic-torsion restrictions. The explicit $\eta$ in
eq.~\eqref{w3} only records the branch-adapted SU$(3)$ decomposition,
and does not define an additional bosonic torsion class.

The restrictions obtained in eqs.~\eqref{w1gone}, \eqref{w2gone}, \eqref{w3},
\eqref{w4exactornot}, \eqref{w5wF},
define special classes of SU$(3)$ structures.
In particular, for gauged configurations that retain an SU$(3)$
structure, if we additionally impose $\rho=0$, and hence
$\mathcal W_3=0$, the base is globally conformally K\"ahler.
The remaining intrinsic torsion lies in the class
$\mathcal W_4\oplus\mathcal W_5$, with the holomorphic volume
form twisted by the gauge connection when $A_\perp\neq0$.
In the gauged configurations of Proposition~\ref{enlargedstructure}, 
by contrast, the geometric structure is enlarged from SU$(3)$ to U$(3)$, 
as reviewed in Appendix~\ref{gstructure6dwg}.
These geometric restrictions are independent of the spinor branch,
survive on the complete solution space, and reproduce those of the
minimal theory, as we will discuss in Section~\ref{solanalysis}.

\section{Classifying Supersymmetric Solutions}
\label{classifyingbps}

We begin by collecting general results that apply across 
all supersymmetric solutions in the equal-charge sector of 
the maximal gauged theory considered here.

The timelike Killing vector and metric retain the forms fixed in
eqs.~\eqref{Killingtime}--\eqref{Killingdual}. The algebraic and
differential bilinear relations, together with the Bianchi
identities, imply
\be
\label{commonKillingdata}
\iota_K\dd X=\mathcal L_KX=0\, ,
\quad
\mathcal L_Kh=\mathcal L_K\theta=0\, ,
\quad
\mathcal L_KF=\mathcal L_KG=0\, ,
\quad
\mathcal L_Kg_{\mu\nu}=\mathcal L_KJ=0\, ,
\ee
with the spinor angles and geometric structure as defined in 
Section~\ref{gstructure}. Thus the dilaton gradient always has the SU$(3)$ 
decomposition in eq.~\eqref{dxsu3}. In particular, 
$K$ also preserves the spinor angle, 
and the almost-Hermitian structure of the base. 
Killing invariance of the gauge potentials themselves requires 
a choice of gauge representative.

Proposition~\ref{codilaton} and Corollary~\ref{gaugedvev} give
\begin{align}
\label{commonvacuumdata}
F = G = 0
&\qquad\implies\qquad
\dd X=0\, ,
\\
g\neq0\, ,\quad F=G=0
&\qquad\implies\qquad
X=1\, .
\end{align}
Thus, a running dilaton must be supported by at least one non-vanishing flux.

The internal-orientation data separate into the two complementary cases
\be
\begin{cases}
s=0
& 
\qquad
\theta=0, \frac{\pi}{2}\, ,
\qquad
\varphi \ \text{unconstrained}\, ,
\\
s\neq0
& 
\qquad
2\cot(2\theta)\dd\theta = - \dd\log(Xh)\, ,
\quad
\dd\varphi=0\, .
\end{cases}
\ee
On a connected regular region, the complex scalar bilinear that is
proportional to $s$, as in eq.~\eqref{complicates11}, 
therefore either vanishes identically or is nowhere zero. 
The remaining class-dependent restrictions on the spinor internal angles are
derived below from the complete KSE and bosonic systems.

The primitive components of the field strengths, $\tau$, $\rho$ and $\kappa$, 
defined in Appendix~\ref{su3projections}, are left
unconstrained by the algebraic KSEs, as shown in 
Propositions~\ref{unconst2formprim}, \ref{unconst3formprim}. 
After incorporating the differential bilinear relations, 
the components $\rho$ and $\kappa$ of $G$ remain undetermined, 
respectively sourcing the torsion form $\mathcal W_3$ 
and the primitive twist $(\dd \omega)^{(1,1)_0}$. 
The component $\tau$ of the Cartan flux, however, 
is further restricted in the gauged 
$\Ker\Gamma_-$ spinor branch.
By Proposition~\ref{vanishingcscalar},
\be
\chi^A\in\Ker\Gamma_-\, ,\quad
g\neq0\, ,\quad
s^{11}\neq0
\qquad\Longrightarrow\qquad
F=0
\quad\Longrightarrow\quad
\tau=0\, .
\ee 
The survival and further differential restrictions of the primitive components 
must then be determined from the bosonic field equations.

We organize these general results into a classification of bosonic
supersymmetric solutions, rather than of the Killing spinors themselves.
The decomposition in eq.~\eqref{directsumofkers} into two USp$(4)$ spinor 
components is retained only insofar as it determines the allowed bosonic
configurations, and the possible fraction of preserved supersymmetry.
As noted in Section~\ref{spinors}, the equal-charge sector of the
maximal theory defined in eqs.~\eqref{eqchargelim},
\eqref{eqchargescalar}, truncates to the minimal U$(1)$-gauged theory
under the restrictions in eqs.~\eqref{Cartanrestrict}, \eqref{maxtominspinor}, 
\eqref{maxtominfermions}, and the map in eq.~\eqref{Dibitettomap}. 
The classification is therefore naturally organized by the two independent 
bosonic distinctions of the presence or absence of gauging, and whether the 
Cartan flux operator, $F\Gamma_+$, vanishes on the Killing spinor.

The Cartan flux operator gives the exhaustive distinction
\be
\label{Cartanfluxclasses}
\begin{cases}
F\Gamma_+\chi^A\neq0
& \qquad \text{Active Cartan flux}\, ,
\\
F\Gamma_+\chi^A=0
& \qquad \text{Null Cartan flux}\, .
\end{cases}
\ee
Combining this criterion with the presence or absence of gauging
gives the four classes
\be
\mathcal C_{\rm ug}^{\rm active}\, ,
\qquad
\mathcal C_{\rm ug}^{\rm null}\, ,
\qquad
\mathcal C_{\rm g}^{\rm null}\, ,
\qquad
\mathcal C_{\rm g}^{\rm active}\, ,
\ee
where the labels ${\rm ug}$ and ${\rm g}$ denote $g=0$ and
$g\neq0$, respectively. These solution classes describe the distinct bosonic 
sectors selected by the complete KSEs and bosonic field equations. 
They are summarized in Tables~\ref{tab:preliminary-classes-ug},
\ref{tab:preliminary-classes-g}, and analysed in the same order in
the following sections.

\afterpage{
\begin{sidewaystable}[ht]
\centering
\footnotesize
\setlength{\tabcolsep}{4.0pt}
\renewcommand{\arraystretch}{1.5}
\begin{tabular}
{|
p{0.04\linewidth}| %
p{0.08\linewidth}| %
p{0.23\linewidth}|
p{0.36\linewidth}|
p{0.20\linewidth}| %
@{}}
\hline
\textbf{Class}
&
\textbf{Definition}
&
\textbf{Constraints on $F$, $A$ and $G$, $B$}
& 
\textbf{Constraints on $\theta$, $\varphi$, $X$, $h$, $\omega$ and $g_{\mu\nu}$}
&
\textbf{Structure and Torsion}
\\
\hline
$\mathcal C_{\rm ug}^{\rm active}$
&
$g=0$

$F \Gamma_+ \chi^A \neq 0$
&
$u=0,\qquad \zeta=-ip^{-1}v$

$hs\lambda=-X^3v,\qquad 2c\lambda=-X^3\sigma$

$f=5hcD-X^3cq+\tfrac12X^3hsp\,\beta$

$5hsD-X^3sq
=-\tfrac{i}{2}hp(\alpha+icX^3\beta)$

$F=-h^{-2}K\wedge\dd(Xhc)+F_\perp$

$A=ch^{-1}XK+A_\perp$

$\dd A_\perp=F_\perp+Xhc\,\dd\omega$

$G=h^{-2}K\wedge\dd(hX^{-2}J)+G_\perp$

$\star_6G_\perp
=-hX^{-2}\dd\omega-2cX^{-3}F_\perp$

$B=-h^{-1}X^{-2}K\wedge J+B_\perp$

$\dd B_\perp
=G_\perp-hX^{-2}J\wedge\dd\omega$

$\boxed{\boldsymbol{G=0\quad\Longrightarrow\quad F\wedge F=0}}$

$\boxed{\boldsymbol{G=0\quad\not\!\!\Longrightarrow\quad F=0}}$
&
$\boxed{\boldsymbol{s\neq0:
		\dd(Xhs)=0,\quad\dd\varphi=0}}$

$\boxed{\boldsymbol{s=0:
		\theta=0,\tfrac{\pi}{2}}},
\quad
\varphi\ {\rm unconstrained}$

$J\lrcorner\dd\omega=-ch^{-1}X^{-1}\lambda$

$\boxed{\boldsymbol{(\dd\omega)^{(1,1)_0}
		=h^{-1}X^2\kappa-2ch^{-1}X^{-1}\tau}}$

$\begin{aligned}
	(\dd\omega)^{(2,0)+(0,2)}
	={}&-\frac{c}{4hX}
	\operatorname{Re}\!\left[
	\alpha\lrcorner\overline\Psi
	\right]
	+\frac{s}{2h^2X}
	\operatorname{Im}\!\left[
	p^{-1}f\lrcorner\overline\Psi
	\right]
\end{aligned}$

$\dd\left(F_\perp+Xhc\,\dd\omega\right)=0$

$\dd\left(G_\perp-hX^{-2}J\wedge\dd\omega\right)=0$

$(E_B)_{\perp}^{(2,2)}=0$ in eq.~\eqref{ugnpGEOM}

$s=0:\quad E_X=0$ in eq.~\eqref{ugnpscalarEOM}

&
\textbf{Non-/Integrable SU$(3)$}

$\mathcal W_1
=-\tfrac16X^{-1}p^{-1}\boldsymbol{s\lambda}$

$\mathcal W_2
=-2X^{-1}p^{-1}\boldsymbol{s\tau}$

\fbox{\textbf{Integrable iff}}

\qquad \quad $\boxed{\boldsymbol{s=0$ or $\lambda=\tau=0}}$

$\mathcal W_3
=-h^{-1}X^2(\rho+\overline\rho)$

$\mathcal W_4
=-\tfrac12\dd\log h
\boxed{+\frac{\boldsymbol{s}}{2X}\text{Im}(p\boldsymbol{\alpha)}}$

$\boxed{\boldsymbol{\mathcal W_5=\dd\log(h^{-1}X)}}$

For $s\neq0$,

$\mathcal W_5
=-2\dd\log h-2(c/s)\dd\theta$
\\
\hline
$\mathcal C_{\rm ug}^{\rm null}$

&

$g=0$

$F \Gamma_+ \chi^A = 0$

&

$\boxed{\boldsymbol{f=\lambda=\alpha=\zeta=\sigma=\beta=0}}$

${\boldsymbol{\eta=\pm1}},
\qquad
\eta=+1\implies\tau=0$

$q=5\eta hX^{-3}D
=-5\eta X^{-2}(\dd h)^{(1,0)}$

$F=\tau,\qquad\dd\tau=0$

$A=A_\perp,\qquad\dd A_\perp=\tau$

$A=0$ may be chosen at $\tau=0$

$G=-\eta\,\dd\left(
h^{-1}X^{-2}K\wedge J
\right)$

$B=-\eta h^{-1}X^{-2}K\wedge J+B_\perp$

$\dd B_\perp=0$

$\boxed{\boldsymbol{G=0\quad\Longrightarrow\quad F=0}}$

&

$\boxed{\boldsymbol{\dd\theta=0,\qquad s\dd\varphi=0}}$

$\boxed{\boldsymbol{Xh=\mathrm{const}}>0}$

$\boxed{\boldsymbol{
		(\dd\omega)^{(1,1)_0}
		=\eta h^{-1}X^2\kappa}}$

$(E_B)_{\perp}^{(2,2)}=0$ in eq.~\eqref{ugpGEOM}

&

\textbf{Integrable SU$(3)$}

$\mathcal W_1=\mathcal W_2=0$

$\mathcal W_3=
-\eta h^{-1}X^2
(\rho+\overline\rho)$

$\mathcal W_4=-\tfrac12\dd\log h$

$\boxed{\boldsymbol{
		\mathcal W_5=-2\dd\log h}
	=4\mathcal W_4}$
\\
\hline
\end{tabular}

\caption{
Solution classes in the ungauged theory.
The SU$(3)$ irreducible components of the fluxes and dilaton gradient are
defined in Appendix~\ref{su3projections}.
The spinor-angle parameters $s$, $c$ and $p$, are defined 
in Section~\ref{gstructure}. 
In the null classes, $\eta=\pm1$ labels the two spinor branches, 
as in eqs.~\eqref{etadef}, \eqref{eigenGamma5}, and tracks the
relative signs of $G$-dependent terms.
Generally, $\kappa$ and $\rho$ of $G$ 
remain unconstrained by the bilinear relations.
The fourth column lists the remaining Bianchi conditions and the
independent EOMs of Proposition~\ref{independentbosonic}, with
$E_B$ and $E_X$ defined in eqs.~\eqref{intresiduals}.
}
\label{tab:preliminary-classes-ug}
\end{sidewaystable}

\begin{sidewaystable}[ht]
\centering
\footnotesize
\setlength{\tabcolsep}{4.0pt} 
\renewcommand{\arraystretch}{1.5}
\begin{tabular}{
|
p{0.04\linewidth}| %
p{0.08\linewidth}| %
p{0.32\linewidth}|
p{0.26\linewidth}|
p{0.20\linewidth}| %
@{}}
\hline
\textbf{Class}
&
\textbf{Definition}
&
\textbf{Constraints on $F$, $A$ and $G$, $B$}
& 
\textbf{Constraints on $\theta$, $\varphi$, $X$, $h$, $\omega$ and $g_{\mu\nu}$}
&
\textbf{Structure and Torsion}
\\
\hline
$\mathcal C_{\rm g}^{\rm null}$

&

$g\neq0$

$F \Gamma_+ \chi^A = 0$

&

$\boxed{\boldsymbol{f=\lambda=\alpha=\zeta=\beta=0}}$

${\boldsymbol{\eta=\pm1}},
\qquad
\eta=+1\implies\tau=0$

$q=5\eta hX^{-3}D
=-5\eta X^{-2}(\dd h)^{(1,0)}$

$\boxed{\boldsymbol{
		\sigma=8\eta gX^{-3}\Delta(X)}}$

$F=\tau,\qquad\dd\tau=0$

$A=A_\perp,\qquad\dd A_\perp=\tau$

$A=0$ may be chosen at $\tau=0$

$G=-\eta\,\dd\left(
h^{-1}X^{-2}K\wedge J
\right)
+2\eta gX^{-1}J\wedge J$

$\eta B
=-h^{-1}X^{-2}K\wedge J
-\tfrac{X^2}{2g}K\wedge\dd\omega+B_\perp$

$\dd B_\perp
=-\tfrac{(Xh)^2}{2g}
\dd\omega\wedge\dd\omega
+2gX^{-1}J\wedge J$

$2\eta gB_\perp
=h^{-1}X^4\star_6(\iota_KG)
+A_\perp\wedge\tau
-(\dd C_{(2)})_\perp$

$\boxed{\boldsymbol{
		G=0\quad\Longrightarrow\quad F=0}}$

&

$\boxed{\boldsymbol{
		\dd\theta=0,\qquad s\dd\varphi=0}}$

$\boxed{\boldsymbol{Xh=\mathrm{const}}>0}$

$J\lrcorner\dd\omega
=-\frac{2g}{h}(2X+X^{-4})$

$\boxed{\boldsymbol{
		(\dd\omega)^{(1,1)_0}
		=\eta h^{-1}X^2\kappa}}$

$(E_B)_{\perp}^{(2,2)}=0$ in eq.~\eqref{gpGEOM}

&

\textbf{Integrable U$(3)$}

$\mathcal W_1=\mathcal W_2=0$

$\mathcal W_3
=-\eta h^{-1}X^2
(\rho+\overline\rho)$

$\mathcal W_4=-\tfrac12\dd\log h$

\fbox{\textbf{No intrinsic}
	$\boldsymbol{\mathcal W_5}$}

$\boldsymbol{
	\mathcal L_K\Psi
	=4ig\underbrace{(Xh)}_{\text{const}}\Psi\neq0}$

\\
\hline

$\mathcal C_{\rm g}^{\rm active}$

&

$g\neq0$

$F \Gamma_+ \chi^A \neq 0$

&

$\boxed{\boldsymbol{\zeta=u=v=0}}$

$\boldsymbol{c=\pm1}$,
\quad
$2c\lambda=8g\Delta(X)-X^3\sigma$

$cf=5hD-X^3q$

$\boxed{\boldsymbol{\alpha=-icX^3\beta}}$

$F=-ch^{-2}K\wedge\dd(Xh)+F_\perp$

$A=ch^{-1}XK+A_\perp$

$\dd A_\perp
=-\tau+cX^3\kappa
-\tfrac{2cg}{3}(2X^2+X^{-3})J$

$G=h^{-2}K\wedge\dd(hX^{-2}J)+G_\perp$

$\star_6G_\perp
=-hX^{-2}\dd\omega
-2cX^{-3}F_\perp
-2gX^{-6}J$

$\begin{aligned}
	B={}&-h^{-1}X^{-2}K\wedge J
	-\tfrac{c}{2g}h^{-2}K\wedge
	\dd(XhA_\perp)
	+B_\perp
\end{aligned}$

$\begin{aligned}
	\dd B_\perp
	={}&G_\perp-hX^{-2}J\wedge\dd\omega
	-\tfrac{c}{2g}\dd\omega\wedge
	\dd(XhA_\perp)
\end{aligned}$

$\boxed{\boldsymbol{
		G=0\quad\Longrightarrow\quad F\wedge F=0}}$

$\boxed{\boldsymbol{
		G=0\quad\not\!\!\Longrightarrow\quad F=0}}$

&

$\boxed{\boldsymbol{
		\theta=0,\tfrac{\pi}{2}}},
\qquad
\varphi\ \text{unconstrained}$

$\begin{aligned}
	J\lrcorner\dd\omega
	=-\frac1h\left[
	cX^{-1}\lambda
	+2g(2X+X^{-4})
	\right]
\end{aligned}$

$\boxed{\boldsymbol{
		(\dd\omega)^{(1,1)_0}
		=h^{-1}X^2\kappa
		-2ch^{-1}X^{-1}\tau}}$

$\boxed{\boldsymbol{\begin{aligned}
			(\dd\omega)^{(2,0)+(0,2)}
			=-\frac{c}{4hX}
			\operatorname{Re}\!\left[
			\alpha\lrcorner\overline\Psi
			\right]
\end{aligned}}}$

$\dd\left(F_\perp+cXh\,\dd\omega\right)=0$

$\dd\left(G_\perp-hX^{-2}J\wedge\dd\omega\right)=0$

$(E_B)_{\perp}^{(2,2)}=0$ in eq.~\eqref{gnpGEOM}

$E_X=0$ in eq.~\eqref{gnpscalarEOM}

&

\textbf{Integrable SU$(3)$}

$\mathcal W_1=\mathcal W_2=0$

$\mathcal W_3
=-h^{-1}X^2(\rho+\overline\rho)$

$\mathcal W_4=-\tfrac12\dd\log h$

$\boxed{\mathcal W_5
	=\dd\log(h^{-1}X)
	\boldsymbol{-4gA_\perp\lrcorner J}}$

\fbox{\textbf{For $A$, $B$ in Killing gauges}}

$\boxed{\boldsymbol{\mathcal L_K\Psi=0}}$

\\
\hline
\end{tabular}

\caption{
Solution classes in the gauged theory. Notation and EOM conventions
are as in Table~\ref{tab:preliminary-classes-ug}.
}
\label{tab:preliminary-classes-g}
\end{sidewaystable}
\clearpage
}

The active/null classification is not identical to separating $F$
into primitive and non-primitive parts with respect to the SU$(3)$
structure induced by the Killing spinor. Recalling the definition
of $F^{\rm np}$ in eq.~\eqref{nonprimflux}, we have
\be
\label{primitivecriterion}
F^{\rm np}=0
\quad\iff\quad
f=\lambda=\alpha=0
\quad\iff\quad
F=\tau\, .
\ee
By Proposition~\ref{uniquespinor} and
Corollary~\ref{nonprimspinor}, a non-primitive Cartan flux 
is possible only in the active classes.
Conversely, the two null Cartan-flux configurations are those of
eq.~\eqref{nullcartanflux}, namely $F=0$ on $\Ker\Gamma_-$ and
$F=\tau$ on $\Ker\Gamma_+$. We therefore have the strict
implications
\be
F^{\rm np}\neq0
\quad\Longrightarrow\quad
\text{Active Cartan flux}\, ,
\qquad
F=0
\quad\Longrightarrow\quad
\text{Null Cartan flux}\, .
\ee
These implications are sufficient to uniquely assign every solution proposed in 
the literature to one of the four classes, as we show in Section~\ref{revsol}.

The four classes above also directly provide the classification framework
for timelike supersymmetric solutions of the minimal U$(1)$-gauged theory 
when the maximal-supersymmetry coupling relation is imposed. 
The complete bosonic and fermionic truncation
established in Sections~\ref{bostheory}, \ref{spinors}, 
identifies the minimal SM doublet with 
the retained maximal SM spinors in $\Ker\Gamma_-$, as in
eqs.~\eqref{maxtominspinor}, \eqref{maxtominfermions}.
On this spinor subspace the Cartan flux operator is null precisely when $F=0$. 
Hence the ungauged and gauged active classes contain all minimal-theory
configurations with $F\neq0$, whether $F$ is primitive or
non-primitive, while the corresponding null classes are their
$F=0$ specializations. Accordingly, Tables
\ref{tab:preliminary-classes-ug}, \ref{tab:preliminary-classes-g}
apply directly to the minimal theory: in their null rows one retains
the $\eta=+1$, $\Ker\Gamma_-$ branch, for which $\tau=0$.
The complete KSEs and bosonic systems in all four classes then coincide
under the maximal-to-minimal map in eq.~\eqref{Dibitettomap}. 

The equal-charge maximal theory contains in addition the orthogonal USp$(4)$ 
spinor components that are removed by the fermionic truncation. 
As shown in eq.~\eqref{maxtominspinorcomplement}, for the fixed bosonic 
configuration $F=0$, the $\Ker\Gamma_+$ Killing-spinor system 
is then identified with the minimal one after a sign reversal of the maximal 
$3$-form potential. 
The $F=0$ null classes therefore admit a second
maximal-theory spinor origin, which can enhance the supersymmetry of
the embedded configuration in the same bosonic solution class. 
Indeed, by eqs.~\eqref{directsumofkers},
\eqref{projections},
\be
\Ker\Gamma_-=\operatorname{Im}P_+\, ,
\qquad
\Ker\Gamma_+=\operatorname{Im}P_-\, ,
\ee
where $P_\pm$ are the complementary rank-two projectors.
Killing spinors in these two subspaces are independent, and their
simultaneous presence enlarges the internal Killing-spinor space.
The absolute BPS fraction remains dependent on the independent
Spin$(1,6)$ Killing-spinor solutions and projections. Moreover, the
SU$(3)$ structure is defined by the chosen Killing spinor, so
complementary spinors need not induce the same $K$ and $J$. 
If they do, however, applying eq.~\eqref{diffT} with both $\Gamma_5$ 
eigenvalues implies $\iota_KG=0$. Hence $\rho=0$, and therefore
$\mathcal W_3\sim\operatorname{Re}\rho=0$.

There remains a minor peculiarity specific to the equal-charge maximal theory. 
From eq.~\eqref{nullcartanflux} a non-zero primitive field strength,
\be
F=\tau\neq0\, ,
\ee
belongs to an active class when the Killing spinor lies in
$\Ker\Gamma_-$, and is then included in the minimal-theory
classification described above. For a Killing spinor in
$\Ker\Gamma_+$, however, the same primitive field strength belongs
to a null class. In this case the variation of the independent
antisymmetric spin-$\tfrac12$ multiplet $\hat\lambda_-$ restricts
$F$ to its primitive component, while the $F$-dependent terms in
the symmetric spin-$\tfrac12$ and gravitino KSEs vanish. The
KSE-derived twist, intrinsic torsion and structure consequently
have the same form as for $F=0$, although the non-zero $\tau$
continues to enter the bosonic equations. Those integrated equations 
give $G=0\Longrightarrow F=0$ in both null classes, so this 
possibility would necessarily satisfy
\be
\chi^A\in\Ker\Gamma_+\, ,
\qquad
F=\tau\neq0\, ,
\qquad
G\neq0\, ,
\ee
where the Bianchi identity additionally requires $\dd\tau=0$.
This $\Ker\Gamma_+$ supersymmetric system is not mapped to the
minimal theory, and no known solution realizes it. 
Yet, the combined KSE and bosonic systems have not produced a contradiction,
and this unresolved possibility is indeed accommodated 
in the $\eta=-1$ case of the maximal null classes. 

The case $s=0$ is the only value of internal spinor orientation shared 
by all four solution classes, and hence the only common setting for a parallel
comparison of the intrinsic torsion of their bases. In this case,
$\mathcal W_1=\mathcal W_2=0$, $\mathcal W_3$ is controlled by
$\rho$, and
\be
\mathcal W_4=-\frac12\dd\log h\, .
\ee
The surviving torsion therefore admits the geometric refinements
\be
\dd h=0
\quad\Longleftrightarrow\quad
\text{Balanced}\, ;
\qquad
\rho=0
\quad\Longrightarrow\quad
\dd\!\left(h^{1/2}J\right)=0\, ,
\ee
where the second condition gives a globally conformally-K\"ahler
base. Imposing both $\rho=0$ and $\dd h=0$ makes the base K\"ahler.

Since the preceding torsion forms share the same structural form
throughout, the essential distinction among the four classes 
is solely carried by the fifth torsion form:
\be
\label{s0w5comparison}
\mathcal W_5
=
\begin{cases}
\dd\log(h^{-1}X)&
\qquad
\mathcal C_{\rm ug}^{\rm active}\, ,
\\
-2\dd\log h&
\qquad
\mathcal C_{\rm ug}^{\rm null}
\, ,
\\
\text{No } \mathcal W_5&
\qquad
\mathcal C_{\rm g}^{\rm null}\, ,
\\
\dd \log(h^{-1}X)-4gA_\perp\lrcorner J&
\qquad
\mathcal C_{\rm g}^{\rm active}\, .
\end{cases}
\ee
In the ungauged sector, the active and null expressions coincide
on the primitive specialization, for which the KSEs give $\dd(Xh)=0$.
In the gauged sector, the active class carries an SU$(3)$ structure 
with an intrinsic fifth torsion form that is gauge-dependent,
whereas the null class carries only a U$(3)$ structure. 

The fifth torsion form also controls the departure from a
conformally Calabi--Yau structure. 
Under a conformal rescaling $g_{ij}\to e^{2f}g_{ij}$, 
the SU$(3)$ normalization in eq.~\eqref{normal} requires $J\to e^{2f}J$ 
and $\Psi\to e^{3f}\Psi$. Eqs.~\eqref{dJ}, \eqref{dpsireal} then give
\be
\mathcal W_4\to\mathcal W_4+2\dd f\, ,
\qquad
\mathcal W_5\to\mathcal W_5+3\dd f\, .
\ee
Thus, a single conformal rescaling can eliminate both torsion forms when
\be
2\mathcal W_5-3\mathcal W_4=0\, .
\ee
Since $\mathcal W_4$ is common to all classes on the $s=0$ locus,
their respective conformally Calabi--Yau conditions are therefore
$\rho=0$, together with
\be
\label{s0conformalCY}
\begin{cases}
4 \dd\log X = \dd\log h&
\qquad
\mathcal C_{\rm ug}^{\rm active}\, ,
\\
\dd h=0&
\qquad
\mathcal C_{\rm ug}^{\rm null}
\, ,
\\
4 \dd\log X = \dd\log h + 16 g A_\perp\lrcorner J&
\qquad
\mathcal C_{\rm g}^{\rm active}
\, .
\end{cases}
\ee
The ungauged null class therefore admits no non-trivial conformal
rescaling to a Calabi--Yau structure: its torsion must simply
vanish. For the gauged null class, the corresponding question is not
intrinsic to its U$(3)$ structure. 
Further restrictions on the surviving torsion 
select such special geometries, which refine the four solution classes, 
rather than define additional ones.

\subsection{Classification of Solutions}
\label{solanalysis}

We now determine which bosonic EOMs are necessary 
to supplement the KSEs and Bianchi identities. 
In the equal-charge sector, the two Abelian EOMs coincide and
the second scalar EOM vanishes identically. 
We define the scalar, Abelian and $3$-form EOM residuals, respectively, by
\begin{subequations}
\label{intresiduals}
\begin{align}
\label{intscalarresidual}
E_X
&=5\nabla^2\log X+\frac12X^{-2}F_{\mu\nu}F^{\mu\nu}
-\frac1{24}X^4G_{\mu\nu\rho\sigma}G^{\mu\nu\rho\sigma}
-g^2X\mathcal V'(X)\, ,\\
\label{intvectorresidual}
E_A
&=X\star\left[\dd(X^{-2}\star F)-F\wedge G\right]\, ,\\
\label{intformresidual}
E_B
&=\dd(X^4\star G)+2gG-F\wedge F\, .
\end{align}
\end{subequations}
We denote by $\mathcal E_{\mu\nu}$ the trace-reversed Einstein residual
of eq.~\eqref{EinsEOM} in the same sector, with curvature term $R_{\mu\nu}$.

\begin{prop}[Independent Bosonic Equations]
\label{independentbosonic}
Consider a bosonic configuration in the equal-charge sector with $X,h>0$, 
satisfying the KSEs and Bianchi identities. 
The only independent bosonic EOMs that must be imposed are
\be
\label{independentEOM}
(E_{B})_{\perp}^{(2,2)}=0\, ,
\ee
and in the active classes, i.e.~$F\Gamma_+\chi^A\neq0$, when $s=0$, also
\be
E_X=0\, .
\ee
All Abelian and Einstein equations, and the longitudinal and transverse
$(3,1)+(1,3)$ components of the $3$-form equation, then follow.
The scalar equation also follows whenever it is not independently imposed.
\end{prop}

\begin{proof}
We commute the supercovariant derivatives defined from the gravitino KSE
in eq.~\eqref{Killinggravitineqcharge}, and contract the resulting identity
with Clifford matrices. Combining this with derivatives of the algebraic
KSEs in eqs.~\eqref{Killinghalfspinequal}, and using the KSEs 
and the Bianchi identities in eq.~\eqref{Bianchi}, gives
\begin{subequations}
\label{intspinorialidentities}
\begin{align}
\label{intscalaridentity}
0&=\left[\frac12E_X
-\frac1{12X^2}(\star E_B)_{\nu\rho\sigma}\gamma^{\nu\rho\sigma}\Gamma_5
+\frac i2(E_A)_\nu\gamma^\nu\Gamma_+\right]\epsilon\, ,\\
\label{intantisymidentity}
0&=(E_A)_\nu\gamma^\nu\Gamma_-\epsilon\, ,\\
\label{intEinsteinidentity}
0&=\left[\mathcal E_{\mu\nu}\gamma^\nu-\frac15E_X\gamma_\mu
+\frac1{12X^2}(\star E_B)_{\nu\rho\sigma}\gamma_\mu{}^{\nu\rho\sigma}\Gamma_5
-i(E_A)_\mu\Gamma_+\right]\epsilon\, .
\end{align}
\end{subequations}

\emph{Longitudinal $3$-form equation.}
Equations~\eqref{diffS3}, \eqref{diffK}, \eqref{diffT}, Killing
invariance and the Bianchi identity in eq.~\eqref{Bianchi}, give
\be
\label{intlongitudinalform}
\iota_K E_B=
\begin{cases}
2chX\,\dd F=0 & \eta=+1\, ,\\
0 & \eta=-1\, .
\end{cases}
\ee

\emph{Mixed Einstein equation and temporal components.}
In the frame of eq.~\eqref{Killingdual}, since $E_B$ is transverse,
its term in the $\mu=0$ component of
eq.~\eqref{intEinsteinidentity} vanishes. Inserting the SM pair in
eq.~\eqref{fullsmspinors} and using the chirality relations
\eqref{chiral6}, \eqref{chiralfix} gives
\be
\label{inttimeprojections}
\mathcal E_{0i}=\mathcal E_{i0}=0\, ,\qquad
\mathcal E_{00}=-\frac15E_X-c(E_A)_0\, ,\qquad s(E_A)_0=0\, .
\ee

\emph{Spatial Abelian equation.}
On $\Ker\Gamma_+$, eq.~\eqref{intantisymidentity} and
$\Gamma_-^2=1$ imply $(E_A)_\mu\gamma^\mu\epsilon=0$.
Squaring this action and taking its vector bilinear show that $E_A$
is null and orthogonal to the timelike $K$, hence $E_A=0$.
On the other null branch, $\Ker\Gamma_-$, we have $F=0$, so again $E_A=0$.
For the active classes, the spatial part of eq.~\eqref{intEinsteinidentity}
reverses base chirality except for the Abelian term. Its same-chirality
projection therefore gives $c(E_A)_i=0$.
At $c=0$, which may occur only in the ungauged active class, 
the triplet projection of eq.~\eqref{intscalaridentity} gives
\be
\label{intspatialtriplet}
(\star_6 E_B)^{(2,0)+(0,2)}_{jk}\gamma^{jk}\psi
=-2X^2sp(E_A)_j\gamma^j\psi^*\, .
\ee
Substitution into the spatial singlet projection of
eq.~\eqref{intEinsteinidentity}, using eqs.~\eqref{Cliff1+6},
\eqref{defJ}, yields
\be
\label{intspatialabelian}
\frac12\left[(E_A)_i+iJ_i{}^j(E_A)_j\right]=(E_A)_i
\quad\Longrightarrow\quad (E_A)_i=0\, ,
\ee
by reality. Thus the spatial Abelian equation holds in every class.

\emph{Transverse $(3,1)+(1,3)$ equation.}
With $(E_A)_i=0$, the triplet projection of eq.~\eqref{intscalaridentity}
and its complex conjugate, together with the Hodge identities in
eqs.~\eqref{hodgesu3prim}, \eqref{doublehodge}, give
\be
\label{intautomaticcomponents}
(\star_6 E_B)^{(2,0)+(0,2)}=0
\quad\Longleftrightarrow\quad (E_{B})_{\perp}^{(3,1)+(1,3)}=0\, .
\ee
Together with eq.~\eqref{intlongitudinalform}, the imposed
$(E_{B})_{\perp}^{(2,2)}=0$ therefore gives $E_B=0$.

\emph{Temporal Abelian and scalar equations.}
The remaining scalar projection of eq.~\eqref{intscalaridentity} now reads
\be
\label{intscalarrelation}
E_X=-c(E_A)_0\, .
\ee
In the null classes $(E_A)_0=0$ was already established.
In the ungauged active class with $s\neq0$, it follows from
eq.~\eqref{inttimeprojections}. For active $s=0$, $c=\pm1$, so the
additional condition $E_X=0$ implies $(E_A)_0=0$.
Thus $(E_A)_0=E_X=0$ in all cases.

\emph{Temporal and spatial Einstein equations.}
Equation~\eqref{inttimeprojections} now gives $\mathcal E_{00}=0$.
Since $E_X=E_A=E_B=0$ and $\mathcal E_{i0}=0$, the spatial part of
eq.~\eqref{intEinsteinidentity} reduces to
\be
\label{intspatialrelation}
\mathcal E_{ij}\gamma^j\epsilon=0
\quad\Longrightarrow\quad
\mathcal E_{ij}\mathcal E_i{}^j=0
\quad\Longrightarrow\quad \mathcal E_{ij}=0\, ,
\ee
where the middle equality follows by squaring, with no sum over $i$,
and the last uses the positive-definite base metric.
\end{proof}

An early timelike bilinear analysis of minimal seven-dimensional
supergravity was undertaken in \cite{Cariglia:2004qi}. 
The framework developed here extends this program by providing a complete
bilinear analysis and bosonic classification for the minimal U$(1)$-gauged 
theory, embedded within the broader equal-charge sector 
of the maximal U$(1)^2$-gauged theory.
The present work retains the full dependence of the geometric
bilinears and their algebraic and differential relations on the
internal SM orientation $(\theta,\varphi)$, 
without simplistically fixing that orientation.
In this work we provide the exhaustive algebraic system 
and its SU$(3)$-isotypic analysis, which together furnish a practical 
representation-theoretic diagnostic for the
irreducible components of the bosonic configuration $\{\dd X, F, G\}$.
The differential scalar system here also determines the restrictions
on this internal SM orientation dynamically.
In particular, Proposition~\ref{vanishingcscalar} proves, independently
of gauge choice, that the complex scalar bilinear vanishes for
$g,F\neq0$ in the minimal theory, and on the corresponding
$\chi^A\in\Ker\Gamma_-$ spinor subspace of the equal-charge maximal theory.

The differential analysis in this work treats the complete $3$-form relations, 
and provides Killing gauges for both gauge potentials.
In particular, the Killing gauge of the $3$-form potential is first 
established here, in conjunction with the self-duality condition.
It thereby captures the action of the gauged R-symmetry on the
holomorphic $3$-form.
In particular, in the gauged null class, or when $g\neq0$ and $F=0$, 
Proposition~\ref{enlargedstructure} shows $\Psi$ is U$(1)$-charged, enlarging 
the base structure from SU$(3)$ to U$(3)$.
We also extract the complete intrinsic torsion of the base and derive its 
integrability criteria and geometric refinements. 
Proposition~\ref{independentbosonic} then shows that, given the KSEs and
Bianchi identities, the independent bosonic EOMs reduce to the transverse
$(2,2)$ component of the $3$-form equation, supplemented by the scalar
equation only in the active classes when $s=0$.
This reduction is implemented explicitly in all four classes, with all
Abelian and Einstein equations, together with the remaining
components of the $3$-form equation, following from the resulting systems.

Integrating all of this analysis produces a framework 
for mapping all timelike supersymmetric solutions in the 
present sector into the four classes summarized in
Tables~\ref{tab:preliminary-classes-ug}, \ref{tab:preliminary-classes-g}.
This provides a direct diagnostic procedure for classifying a
proposed bosonic configuration, identifying its generic SU$(3)$ flux
components and geometric structure, and for singling out unexplored sectors.
The present framework can therefore be used both to construct new
supersymmetric solutions and to test proposed ones efficiently.

The following class analyses assemble the algebraic and differential
KSE constraints, the remaining transverse Bianchi conditions, and the
independent EOMs. The longitudinal Bianchi projections were already used in
Section~\ref{Killingresultsdiff} to establish Killing invariance of the fluxes. 
We also construct Killing gauges for both potentials and, in the gauged 
classes, verify compatibility with the integrated self-duality condition.

\subsubsection{Ungauged Active Cartan Flux}
\label{ugactiveclass} 

The class $\mathcal C_{\rm ug}^{\rm active}$ is defined by $g=0$ and
$F\Gamma_+\chi^A\neq0$, 
and therefore selects only the spinor branch $\chi^A\in\Ker\Gamma_-$. 
The algebraic KSE constraints give
\begin{subequations}
\label{ugnpalg}
\begin{align}
&u=0 \implies
\zeta=-ip^{-1}v\, ,
\qquad
X^3v=-hs\lambda\, ,
\qquad
X^3\sigma=-2c\lambda\, ,
\\
&f=5hcD-X^3cq+\frac12X^3hsp\,\beta\, ,
\qquad 5hsD-X^3sq
=-\frac{i}{2}hp\left(\alpha+icX^3\beta\right).
\end{align}
\end{subequations}
Thus, no SU$(3)$ component of either flux is forced to vanish
at this stage; in particular, the primitive components 
$\tau$, $\rho$ and $\kappa$ remain unconstrained.

The differential KSEs give
\be
\label{ugnpscalars}
\begin{cases}
s\neq0
& \qquad 
\dd(Xhs)=0\, ,
\qquad
\dd\varphi=0\, ,
\\
s=0
& \qquad
\theta=0,\frac{\pi}{2}\, ,
\qquad
\varphi\ \text{unconstrained}\, ,
\end{cases}
\ee
and combining them, the complete fluxes take the form
\begin{subequations}
\label{ugnpdiff}
\begin{align}
F
&=-h^{-2}K\wedge\dd(Xhc)+F_\perp\, ,
\qquad \qquad
\iota_KF =\dd(Xhc)\, ,
\\
G
& =h^{-2}K\wedge\dd(hX^{-2}J)+G_\perp\, ,
\qquad
\iota_KG =-\dd(hX^{-2}J)\, ,
\\
\star_6G_\perp
&=-hX^{-2}\dd\omega-2cX^{-3}F_\perp\, .
\end{align}
\end{subequations}
For $s\neq0$, these equations also imply
\be
\label{ugnptheta}
\iota_KF=-\frac{2Xh}{s}\dd\theta\, .
\ee
The SU$(3)$ components of the twist are fixed by
\begin{subequations}
\label{ugnptwist}
\begin{align}
J\lrcorner\dd\omega
&=-ch^{-1}X^{-1}\lambda\, ,
\\*
(\dd\omega)^{(1,1)_0}
&=h^{-1}X^2\kappa-2ch^{-1}X^{-1}\tau\, ,
\\*
(\dd\omega)^{(2,0)+(0,2)}
&=-\frac{c}{4hX}
\operatorname{Re}\!\left[
\alpha\lrcorner\overline\Psi
\right]
+\frac{s}{2h^2X}
\operatorname{Im}\!\left[
p^{-1}f\lrcorner\overline\Psi
\right].
\end{align}
\end{subequations}
The resulting torsion forms are displayed in 
Table~\ref{tab:preliminary-classes-ug}. 
In particular, integrability is fixed by
\be
\label{ugnpintegrability}
N_J \neq 0
\quad \iff \quad
s\neq0\, ,
\quad 
( \lambda, \tau ) \neq (0,0)\, .
\ee

We next impose the bosonic field equations.
The transverse Bianchi identities are
\be
\label{ugnpBianchi}
\dd\left(F_\perp+Xhc\,\dd\omega\right)=0\, ,
\qquad
\dd\left(G_\perp-hX^{-2}J\wedge\dd\omega\right)=0\, .
\ee
On the $s=0$ locus, these equations are equivalently obtained as the
$g\rightarrow0$ limit of the corresponding ones in eqs.~\eqref{gnpBianchi}
in the gauged active class.

The independent transverse $3$-form EOM is
\begin{align}
\label{ugnpGEOM}
\left[\dd\left[
X^4h^{-1}\star_6\dd(hX^{-2}J)
\right]\right]^{(2,2)}
={}&\Bigl[(Xh)^2\dd\omega\wedge\dd\omega
\nn\\*
&+2c(Xh)\dd\omega\wedge F_\perp
+F_\perp\wedge F_\perp\Bigr]^{(2,2)}\, .
\end{align}
For $s\neq0$, this is the only EOM to impose. For $s=0$, the
additional scalar equation is
\begin{align}
\label{ugnpscalarEOM}
h^{-1}(\nabla_6)_i(h\nabla_6^i\log X)
={}&-\frac15X^{-2}\left[
\lvert F_\perp\rvert^2
-h^{-2}\lvert\dd(Xhc)\rvert^2
\right]
\nn\\
&+\frac15X^4\left[
\lvert G_\perp\rvert^2
-h^{-2}\lvert\dd(hX^{-2}J)\rvert^2
\right]\, .
\end{align}

For $G=0$, the $3$-form EOM in eq.~\eqref{EOM3-form} directly implies
\be
G=0
\quad\Longrightarrow\quad
F\wedge F=0\, .
\ee

Finally, the Killing gauges may be chosen locally as
\be
\label{ugnpgauges}
A=ch^{-1}XK+A_\perp\, ,
\qquad
\dd A_\perp=F_\perp+Xhc\,\dd\omega\, ,
\ee
and
\be
B=-h^{-1}X^{-2}K\wedge J+B_\perp\, ,
\qquad
\dd B_\perp=G_\perp-hX^{-2}J\wedge\dd\omega\, .
\ee
Their integrability conditions are exactly those in
eqs.~\eqref{ugnpBianchi}, and impose no further restriction.

\subsubsection{Ungauged Null Cartan Flux}
\label{ugnullclass}

In the class $\mathcal C_{\rm ug}^{\rm null}$ 
we have $g=0$ and $F\Gamma_+\chi^A=0$. 
The two complementary configurations on the spinor branches 
are given in eq.~\eqref{nullcartanflux}, 
with the branch notation defined in eq.~\eqref{etadef}.
Their algebraic and differential constraints combine into
\begin{subequations}
\label{ugpalg}
\begin{align}
f& =\lambda=\alpha=0 \quad \implies \quad \zeta=\beta=0\, , 
\\ \label{ugpsingletalg}
\sigma & =8\eta gX^{-3}\Delta(X) \quad \implies \quad \sigma = 0\, ,
\\ q & =5\eta hX^{-3}D
= -5\eta X^{-2}(\dd h)^{(1,0)}\, ,
\end{align}
\end{subequations}
where it also holds that
\be
\label{Finnull}
F =\tau\, ,\quad \eta = +1 \quad \implies \quad \tau=0\, .
\ee
The surviving SU$(3)$ components of the fluxes are therefore 
$\tau,q,\rho,\kappa$, with $q$ fixed above, and we have
\begin{align}
\label{ugpfluxes}
G
={}&-h^{-2}K\wedge\left[
\frac12J\wedge(q+\overline q)
+\rho+\overline\rho
\right]
+J\wedge\kappa
=
- \eta\,\dd\left(h^{-1}X^{-2}K\wedge J\right),
\end{align}
where we used that the differential constraints give
\be
\label{ugpdiff1}
Xh=\mathrm{const}>0\, ,
\qquad
\dd\theta=0\, ,
\qquad
s\dd\varphi=0\, ,
\ee
and
\be
\label{ugpdiff2}
\dd\omega=\eta h^{-1}X^2\kappa\, ,
\qquad
\iota_KG=-\eta\,\dd(hX^{-2}J)\, .
\ee
As established in Section~\ref{intrinsictorsion}, the SU$(3)$
structure is integrable, and obeys
\be
\label{ugpweightedstructure}
\dd(hJ\wedge J)=0\, ,
\qquad
\dd(h^2\Psi)=0\, ,
\ee
or equivalently $\mathcal W_5=4\mathcal W_4$.

We next impose the bosonic field equations.
The Bianchi identity of $G$ is automatic from its exact expression
in eq.~\eqref{ugpfluxes}. The Bianchi identity of $F$ gives
\be
\label{ugpMaxwell}
\dd\tau=0\, .
\ee
The independent transverse $3$-form EOM is
\begin{align}
\label{ugpGEOM}
\left[\dd\left[
X^5\star_6(\iota_KG)
\right]\right]^{(2,2)}
={}&-(Xh)\left[
\eta X^6\kappa\wedge\kappa
+\tau\wedge\tau
\right].
\end{align}
The mixed Chern--Simons source is $F\wedge F=\tau\wedge\tau$,
which may be non-vanishing.
Equation~\eqref{ugpGEOM} nevertheless excludes a solution supported
only by $\tau$.
Indeed, $G=0$ implies $\iota_KG=\kappa=0$, and then
$\tau\wedge\tau=0$. Together with $\star_6\tau=-J\wedge\tau$,
this gives $\tau\wedge\star_6\tau=0$.
Hence $\lvert\tau\rvert^2=0$, and therefore $\tau=0$, so that
\be
\label{ugpnopureF}
G=0
\quad\Longrightarrow\quad
F=0\, .
\ee

The Killing gauges of the gauge potentials are
\be
\label{ugpgauges}
A=A_\perp\, ,
\qquad
\dd A_\perp=\tau\, ,
\qquad
B=-\eta h^{-1}X^{-2}K\wedge J+B_\perp\, ,
\qquad
\dd B_\perp=0\, .
\ee
At $\tau=0$, $A=0$ may be chosen locally. The closed form
$B_\perp$ labels an independent gauge representative and may, but
need not, be set to zero locally.

\subsubsection{Gauged Null Cartan Flux}
\label{gnullclass}

In the class $\mathcal C_{\rm g}^{\rm null}$ 
we retain $F\Gamma_+\chi^A=0$, and switch on $g$. The two configurations
remain those of eq.~\eqref{nullcartanflux}, with the notation in 
eq.~\eqref{etadef}.
The algebraic constraints are similar to those in eqs.~\eqref{ugpalg}, 
with eq.~\eqref{Finnull}, apart from eq.~\eqref{ugpsingletalg}, 
which does not vanish with gauging, so that in this class we have
\begin{align}
\label{gpalg}
\sigma=8\eta gX^{-3}\Delta(X).
\end{align}
Thus, the SU$(3)$ components $\tau,q,\sigma,\rho,\kappa$ may survive, 
with $q$ and $\sigma$ constrained above. 

The differential bilinear relations likewise agree with
eqs.~\eqref{ugpdiff1}, \eqref{ugpdiff2}, except for the twist, which becomes
\be
\label{gptwist}
\dd\omega
=\eta h^{-1}X^2\kappa
-\frac{2g}{3h}(2X+X^{-4})J\, .
\ee
The transverse and complete $G$ fluxes can then be written 
in terms of the geometry as
\begin{align}
\label{gpGforms}
G_\perp
&=\eta hX^{-2}\dd\omega\wedge J
+2\eta gX^{-1}J\wedge J\, ,
\\
G
&=-\eta\,\dd\left(h^{-1}X^{-2}K\wedge J\right)
+2\eta gX^{-1}J\wedge J\, .
\end{align}
The base admits an integrable U$(3)$ structure, 
namely with no intrinsic torsion form $\mathcal W_5$,
and the locally defined holomorphic volume form is not Killing invariant as
\be
\label{gpLiePsi}
\mathcal L_K\Psi=4ig(Xh)\Psi\neq0\, .
\ee

The Bianchi identity of $F$ gives eq.~\eqref{ugpMaxwell}, 
while that of $G$ is automatic since $\dd(X^{-1}J\wedge J)=0$ 
from eqs.~\eqref{ugpweightedstructure}, \eqref{ugpdiff1}.
The independent transverse $3$-form EOM is
\begin{align}
\label{gpGEOM}
\left[\dd\left[
X^5\star_6\dd(hX^{-2}J)
\right]\right]^{(2,2)}
={}&(Xh)\left[
(Xh)^2\dd\omega\wedge\dd\omega
-\tau\wedge\tau
-4g^2X^{-1}J\wedge J
\right].
\end{align}
Compared with the ungauged null class, this equation contains
additional sources from gauging.

As in the ungauged null class, there is no solution supported only by $\tau$,
and for $G=0$, the same argument gives $F=0$, and eq.~\eqref{gpalg} fixes $X=1$.
This gives the vacuum configuration
\be
\label{gpvacuum}
G=0
\quad\implies\quad
F=0\, ,
\quad X=1\, ,
\quad \dd\omega=-2gh^{-1}J\, ,
\quad R^{(6)}_{ij}=-8g^2g_{ij},
\ee
which is the supersymmetric AdS$_7$ vacuum in timelike-fibration form.

Finally, the simultaneous Killing gauges of the gauge potentials 
may be chosen as
\be
\label{gpgauges}
A=A_\perp\, ,
\qquad
\dd A_\perp=\tau\, ,
\qquad
\eta B
=-h^{-1}X^{-2}K\wedge J
-\frac{X^2}{2g}K\wedge\dd\omega+B_\perp\, ,
\ee
where $A=0$ may be chosen locally at $\tau=0$,
and compatibility with $G=\dd B$ gives
\be
\label{gpBperp}
\dd B_\perp
=-\frac{(Xh)^2}{2g}\dd\omega\wedge\dd\omega
+2gX^{-1}J\wedge J\, .
\ee
The longitudinal self-duality condition reproduces
eq.~\eqref{gptwist}, while its transverse component fixes
\be
\label{gpselfdualperp}
2\eta gB_\perp
=h^{-1}X^4\star_6(\iota_KG)
+A_\perp\wedge\tau
-(\dd C_{(2)})_\perp\, .
\ee
Thus, the Killing gauges and the integrated condition are mutually compatible.

\subsubsection{Gauged Active Cartan Flux} 
\label{gactiveclass}

Finally, the class $\mathcal C_{\rm g}^{\rm active}$ 
is defined by $g\neq0$ and $F\Gamma_+\chi^A\neq0$. 
It exclusively selects the spinor branch $\chi^A\in\Ker\Gamma_-$, 
and Proposition~\ref{vanishingcscalar} sets $s=0$, 
so that $c=\pm1$, and the internal spinor angle $\varphi$ is unconstrained.
The algebraic constraints are
\begin{subequations}
\label{gnpalg}
\begin{align}
&\zeta=u=v=0\, ,
\qquad
2c\lambda=8g\Delta(X)-X^3\sigma\, ,
\\
&cf=5hD-X^3q\, ,
\qquad
\alpha=-icX^3\beta\, .
\end{align}
\end{subequations}
Thus, just a single complex SU$(3)$ component, $\zeta$, vanishes
identically.

The differential relations give
\begin{subequations}
\label{gnpdiff}
\begin{align}
\label{gnpdiff1}
F
&=-ch^{-2}K\wedge\dd(Xh)+F_\perp\, ,
\\*
G
&=h^{-2}K\wedge\dd(hX^{-2}J)+G_\perp\, ,
\\*
X^3(q+\overline q)
&=-X\dd h+4h\dd X\, ,
\end{align}
\end{subequations}
together with
\be
\label{gnpGperp}%
\star_6G_\perp
=-hX^{-2}\dd\omega
-2cX^{-3}F_\perp
-2gX^{-6}J\, ,
\ee
with $F_\perp$ as in eq.~\eqref{f2su3}. 
The SU$(3)$ components of the twist are
\begin{subequations}
\label{gnptwist}
\begin{align}
J\lrcorner\dd\omega
&=-\frac1h\left[
cX^{-1}\lambda+2g(2X+X^{-4})
\right],
\\*
\label{gnptwist2}
(\dd\omega)^{(1,1)_0}
&=h^{-1}X^2\kappa-2ch^{-1}X^{-1}\tau\, ,
\\*
\label{gnptwist3}
(\dd\omega)^{(2,0)+(0,2)}
&=-\frac{c}{4hX}
\operatorname{Re}\!\left[
\alpha\lrcorner\overline\Psi
\right].
\end{align}
\end{subequations}
The base has an integrable SU$(3)$ structure, obeys $\dd(hJ\wedge J)=0$, 
and has the torsion forms displayed in Table~\ref{tab:preliminary-classes-g}.
In particular, $\mathcal W_5$ receives independent contributions 
from the scalar gradient, and the transverse gauge connection. 
In the Killing gauges used below, 
the complex 3-form of the structure $\Psi$ is Killing invariant.

We now impose the bosonic field equations.
The transverse Bianchi identities reduce to
\begin{subequations}
\label{gnpBianchi}
\begin{align}
0&=\dd\left[
\tau-cX^3\kappa
+\frac{2cg}{3}(2X^2+X^{-3})J
\right]\, ,
\\
0&=\dd\left[
\left(2gX^{-1}+\frac{c\lambda}{6X^3}\right)J\wedge J
+2cX^{-3}J\wedge\tau
\right]\, .
\end{align}
\end{subequations}
At $g\rightarrow0$, these equations reproduce the $s=0$
specialization of eqs.~\eqref{ugnpBianchi}.

The independent EOMs are the transverse $3$-form equation
\begin{align}
\label{gnpGEOM}
\left[\dd\left[
X^4h^{-1}\star_6\dd(hX^{-2}J)
\right]\right]^{(2,2)}
={}&\Bigl[(Xh)^2\dd\omega\wedge\dd\omega
+2c(Xh)\dd\omega\wedge F_\perp
+F_\perp\wedge F_\perp
\nn\\
&-\left(
4g^2X^{-1}+\frac{cg\lambda}{3X^3}
\right)J\wedge J
-4cgX^{-3}J\wedge\tau\Bigr]^{(2,2)}\, .
\end{align}
and the scalar equation
\begin{align}
\label{gnpscalarEOM}
h^{-1}(\nabla_6)_i(h\nabla_6^i\log X)
={}&-\frac15X^{-2}\left[
\lvert F_\perp\rvert^2
-h^{-2}\lvert\dd(Xh)\rvert^2
\right]
+\frac15X^4\left[
\lvert G_\perp\rvert^2
-h^{-2}\lvert\dd(hX^{-2}J)\rvert^2
\right]
\nn\\*
&-\frac45g^2
(2X^{-8}-6X^{-3}+4X^2)\, .
\end{align}

For $G=0$, the $3$-form EOM in eq.~\eqref{EOM3-form} implies 
\be
\label{noGentails}
G=0
\quad\Longrightarrow\quad
F\wedge F=0
\quad\Longrightarrow\quad
\begin{cases}
\dd(Xh)\wedge F_\perp=0\, ,\\
F_\perp\wedge F_\perp=0\, .
\end{cases}
\ee
Wherever $G=0$ and $\dd(Xh)\neq0$, the first condition makes $\dd(Xh)$ 
a local factor of $F_\perp$, consistently with the second condition.

Finally, the Killing gauge of the 1-form potential can be written as
\begin{subequations}
\label{gnpgauges1form}
\begin{align}
A
&=ch^{-1}XK+A_\perp\, ,
\\*
\dd A_\perp
&=F^{(1,1)}+cXh(\dd\omega)^{(1,1)}
\nn\\*
&=-\tau+cX^3\kappa
-\frac{2cg}{3}(2X^2+X^{-3})J\, .
\end{align}
\end{subequations}
Here the absence of a holomorphic component in $\dd A_\perp$
follows from Proposition~\ref{noholo1formwg}.
The Killing gauge of the 3-form potential can be written as
\begin{subequations}
\label{gnpgauges3form}
\begin{align}
B
&=-h^{-1}X^{-2}K\wedge J
-\frac{c}{2g}h^{-2}K\wedge\dd(XhA_\perp)
+B_\perp\, ,
\\
\label{gnpBperp}
\dd B_\perp
&=G_\perp-hX^{-2}J\wedge\dd\omega
-\frac{c}{2g}\dd\omega\wedge\dd(XhA_\perp)\, .
\end{align}
\end{subequations}
The latter equation follows from compatibility with $G=\dd B$.
With the auxiliary gauge choice in eq.~\eqref{C2Killinggauge},
the longitudinal self-duality condition reproduces
eq.~\eqref{gnpGperp}, while its transverse part fixes
\be
\label{gnpselfdualperp}
2gB_\perp
=-h^{-1}X^4\star_6\dd(hX^{-2}J)
+A_\perp\wedge F_\perp
-(\dd C_{(2)})_\perp\, .
\ee
Hence the Killing gauges are consistent with the complete bosonic solution.

\subsection{Relation to Solutions in Literature}
\label{revsol}

We now apply the complete classification framework of the previous section 
to examine supersymmetric solutions proposed within the present setup,
in which the vector bilinear of the Killing spinor is timelike.
Besides identifying their SU$(3)$-representation content, this
analysis illustrates how the classification provides efficient
diagnostics for testing such proposed solutions, 
and for organizing searches for new ones, 
without requiring a new full solution of the KSEs in every instance.
For configurations satisfying the KSEs and Bianchi identities,
Proposition~\ref{independentbosonic} then dramatically reduces 
the remaining bosonic checks in each of the four classes 
in our partition of solution space.

Tables~\ref{tab:known-gactive-solutions} and \ref{tab:minimal-solutions}
summarize the solution landscape within the present setup.
All families constructed directly in the maximal theory occupy
the gauged active class $\mathcal C_{\rm g}^{\rm active}$.
They range from magnetic fixed points with sparse SU$(3)$ content to
rotating black holes and warped AdS$_3$ geometries with increasingly
rich content of the $4$-form flux.
The running-dilaton families also include the static charged singularity
with electric flux and its analytic continuation 
yielding the magnetic AdS$_5$ flux-brane of \cite{Lu:2003iv}.
The equal-rotation black holes of \cite{Chong:2004dy} extend in two
distinct directions: three independent angular momenta in \cite{Chow:2007ts}
and Sasaki--Einstein $L^{p,q,r}$ horizon geometries in \cite{Bobev:2025xan}.
These constitute two distinct most general known black-hole endpoints
within the equal-charge sector.

The early search of \cite{Cariglia:2004qi} in the minimal theory 
proposed only isolated and sparse configurations, yielding no additional 
solutions satisfying the necessary requirements.
The Calabi--Yau, complex-Heisenberg and squashed AdS$_7$ configurations
therein display global pathologies, while the Hermitian
configuration admits no supersymmetric Cartan restriction.
To the best of our knowledge, the ungauged solution space remains essentially 
unexplored. 
Relaxing the narrow assumptions of that early search offers a promising route 
to finding solutions in this sector.

To the best of our knowledge, the additional solution families 
from the minimal theory that fall within the present setup appear in
\cite{Dibitetto:2017tve,Conti:2024rwd}.
These extend the picture to the $F=0$ gauged null class
$\mathcal C_{\rm g}^{\rm null}$.
Table~\ref{tab:minimal-solutions} contains their two geometrically distinct
families: the $\mathbb R^{1,2}\times S^3$ flow of \cite{Dibitetto:2017tve}
and the AdS$_3\times S^3$ Janus family of \cite{Conti:2024rwd},
which contains an earlier AdS$_3\times S^3$ flow of \cite{Dibitetto:2017tve}.
The vector-coupled SU$(2)$ configurations of \cite{Dibitetto:2017tve}
admit no supersymmetric Cartan restriction and thus provide no further
solutions within the present setup.
These two families of $F=0$ BPS flows also first permit Killing spinors 
with an internal SM orientation of $s\neq0$.

The generic black hole of \cite{Chow:2007ts} and the warped AdS$_3$
family denoted below by Cheung et al.~(2)~\cite{Cheung:2022ilc} each activate
every SU$(3)$ component permitted in the gauged active class,
together with all non-trivial torsion forms $\{\mathcal W_i\}_{i=3,4,5}$.
In a complementary manner, both flow families of 
\cite{Dibitetto:2017tve,Conti:2024rwd} activate every component permitted 
in the $F=0$ gauged-null class, together with both non-trivial torsion forms 
$\mathcal W_3$ and $\mathcal W_4$.
Across all known solutions $\mathcal W_1=\mathcal W_2=0$, 
so their bases’ almost-complex structures are integrable, with
SU$(3)$ and U$(3)$ structures, respectively, for each of the above noted classes.
In the gauged active class, the torsion form $\mathcal W_5$ carries additional
geometric information through the transverse gauge connection,
beyond the SU$(3)$ component content of the fluxes.

The static charged singularity of \cite{Cvetic:1999ne,Liu:1999ai}
and its analytic continuation yielding the AdS$_5$ flux-brane of \cite{Lu:2003iv}
have identical entries in Table~\ref{tab:known-gactive-solutions},
despite their distinct real Lorentzian geometries.
The equal-rotation and $L^{p,q,r}$ black holes of \cite{Chong:2004dy} and 
\cite{Bobev:2025xan}, respectively, also have identical entries
in Table~\ref{tab:known-gactive-solutions}, despite their different
local geometries and horizon topologies.
The black hole of \cite{Chow:2007ts} and Cheung et al.~(2)~\cite{Cheung:2022ilc}
likewise share the fully activated entry in that table.
The complementary two geometrically distinct flow families of 
\cite{Dibitetto:2017tve,Conti:2024rwd} also have identical
fully activated entries in Table~\ref{tab:minimal-solutions}.
The structure component and torsion data therefore furnish a low-resolution
geometric classification.
In particular, complete activation of the permitted SU$(3)$ content does not
exhaust the functional and geometric freedom within either class,
so classification resolves a substantially larger solution space
than the known families.

Two preliminary observations substantially streamline 
and expedite the analysis of solutions.
Let $\mathcal S_{\rm old}$ and $\mathcal S_{\rm new}$ denote two
bosonic solution families such that
$\mathcal S_{\rm old}\subset\mathcal S_{\rm new}$. 
Since supersymmetry is a property of the bosonic configuration, 
\be
\label{BPSinheritance}
\mathcal S_{\rm old}\subset\mathcal S_{\rm new}
\qquad\implies\qquad
\mathcal S_{\rm old}^{\rm BPS}
\subset
\mathcal S_{\rm new}^{\rm BPS}\, .
\ee
Consequently, an irreducible component that is non-vanishing on an older
BPS solution cannot vanish identically throughout the BPS subset of
a family that contains it.
The same non-vanishing conclusion holds for a regular limit
of a BPS family.
On the other hand, any irreducible component 
that vanishes identically in a general bosonic solution, 
remains absent under every parameter restriction, including its BPS limits.

The converse inference is not valid: a component that is
non-vanishing in a general non-supersymmetric solution need not
survive on its BPS limit. Such components must instead be inherited
from an already known BPS subfamily, or checked directly after the
proper BPS condition has been imposed. In practice, these two
observations determine most components for increasingly 
general bosonic configurations, leaving only a small number of 
projections from Appendix~\ref{su3projections} to be directly 
assessed and distinguished via progressively more refined pieces 
of the classification formalism.

\subsubsection{Maximal Theory}
\label{maxsol}

To the best of our knowledge, within the present setup
the solution families analysed below exhaust the distinct
Lorentzian supersymmetric families presently constructed directly in
maximal Cartan-gauged seven-dimensional supergravity and covered
within the equal-charge sector:
the wrapped-brane fixed points of \cite{Maldacena:2000mw,Gauntlett:2000ng},
the static charged singularity of \cite{Cvetic:1999ne,Liu:1999ai},
the AdS$_5$ flux-brane of \cite{Lu:2003iv},
the rotating charged black-hole families of
\cite{Chong:2004dy,Chow:2007ts,Bobev:2025xan},
and the two warped AdS$_3$ families of \cite{Cheung:2022ilc}.

The AdS$_5$ flux-brane of \cite{Lu:2003iv} and the warped AdS$_3$
families of \cite{Cheung:2022ilc} allow two independent Cartan charges.
Here we develop their equal-charge specializations.
The fully general rotating AdS$_7$ black holes of
\cite{Bobev:2023bxl} restore the second independent Cartan charge
with respect to \cite{Chow:2007ts},
and therefore extend beyond the equal-charge sector resolved here.
Their equal-charge Lorentzian specialization reduces to the
black-hole family of \cite{Chow:2007ts}, while their Euclidean
analytic continuations, further studied in \cite{Gaar:2026gmo},
lie outside the Lorentzian timelike analysis carried out here.

All solutions belong to the gauged active class
$\mathcal C_{\rm g}^{\rm active}$.
For the fixed point of \cite{Maldacena:2000mw} and the flux-brane of
\cite{Lu:2003iv}, this follows immediately from
$G=0$ and $F\neq0$, since eq.~\eqref{gpvacuum} would instead force $F=0$ in
the gauged null class.
The fixed point of \cite{Gauntlett:2000ng} has $\lambda\neq0$,
while the charged AdS$_7$ families of \cite{Cvetic:1999ne,Liu:1999ai}
and \cite{Chong:2004dy,Chow:2007ts,Bobev:2025xan},
and the warped AdS$_3$ families of \cite{Cheung:2022ilc}, have $f\neq0$.
Either component excludes the null class by eq.~\eqref{primitivecriterion}.
In accordance with eq.~\eqref{gnpalg}, the complex singlet
therefore vanishes in all solutions,
\be
\label{nozetahere}
\zeta=0\, .
\ee

\paragraph{Magnetic fixed points.}

The simplest configurations are the wrapped M5-brane fixed points of
\cite{Maldacena:2000mw,Gauntlett:2000ng}. They have constant dilaton
and purely transverse magnetic Cartan flux, but already differ in the
4-form sector: $G=0$ for the fixed point of
\cite{Maldacena:2000mw}, whereas $G\neq0$ for that of
\cite{Gauntlett:2000ng}. 
Their SU$(3)$ content can be identified through a few invariant tests.

For the equal-charge-compatible AdS$_5\times H^2$ fixed point of
\cite{Maldacena:2000mw}, the scalar value given there maps in our
conventions to
\be
\label{dilsol1}
X^5=\frac34
\quad \implies \quad
D=0\, .
\ee
Since $G=0$, all of its SU$(3)$ components vanish, and
$\alpha=0$ follows from eq.~\eqref{gnpalg}. The singlet and
fundamental relations in that equation then give
\be
\lambda \sim gX^{-3}\neq0\, ,
\qquad
f=0\, .
\ee
Were $\tau=0$, one would have
$F=(\lambda/3)J$, and hence
$F\wedge F=(\lambda^2/9)J\wedge J\neq0$, contradicting
eq.~\eqref{noGentails}. Therefore
\be
\label{MNdiagnostic}
\lambda,\tau\neq0\, ,
\qquad
f=\alpha=D=0\, ,
\qquad
G=0\, .
\ee

In particular, $\rho=q=0$, so eqs.~\eqref{w3},
\eqref{gnpdiff} give $\mathcal W_3=0$ and $\dd h=0$, and hence
$\mathcal W_4=0$. The remaining torsion component is consequently
\be
\label{w5ofMN}
\mathcal W_5=-4gA_\perp\lrcorner J\, .
\ee
The corresponding AdS$_5$ fixed point preserves eight of the
thirty-two supercharges, divided into four super-Poincar\'e and four
superconformal charges.

The K\"ahler four-cycle fixed point of \cite{Gauntlett:2000ng}
provides the next variation. 
The Killing-spinor projections imply that the almost-complex structure 
preserves separately the tangent bundle of the wrapped cycle $\Sigma_4$
and its orthogonal two-dimensional complement $N_2$, so that 
\be
J=J_4+J_2\, .
\ee
Thus, $J_4$ and $J_2$ denote the restrictions of $J$ to
$T\Sigma_4$ and $N_2$, respectively.
The Cartan field is proportional to the K\"ahler form of the
cycle,
\be
F_\perp=\mathcal FJ_4\, ,
\qquad
\mathcal F\neq0\, .
\ee
It is therefore of type $(1,1)$ and has no longitudinal component,
so $\alpha=f=0$. Moreover,
\be
\lambda=J\lrcorner F_\perp
=\mathcal F\,J_4\lrcorner J_4\neq0\, .
\ee
If $\tau$ vanished, one would have $F_\perp=(\lambda/3)(J_4+J_2)$, 
whose restriction to $N_2$ is $(\lambda/3)J_2\neq0$, whereas 
the actual field $F_\perp=\mathcal FJ_4$ has no legs along $N_2$.
Hence $\tau\neq0$.
The fixed-point scalar is again fixed to $X^5=3/4$, as in eq.~\eqref{dilsol1}, 
with $D=0$. 

Using the higher-form mapping to field strength established in 
eq.~\eqref{BtoGmap}, the non-vanishing 3-form 
$S_5\propto\operatorname{vol}(\mathrm{AdS}_3)$ 
of \cite{Gauntlett:2000ng} corresponds in our formulation to
\be
G\propto\operatorname{vol}(\Sigma_4)  
\propto J_4\wedge J_4\, .
\ee
Thus $G$ is transverse and purely of type $(2,2)$, giving
$q=\rho=\beta=0$, and consequently $\alpha=0$ consistently with
eq.~\eqref{gnpalg}. Its double contraction with $J$ is non-zero, so
$\sigma\neq0$. 
If $\kappa$ vanished, $G_\perp$ would be proportional to
\be
J\wedge J
=J_4\wedge J_4
+2J_4\wedge J_2\, ,
\ee
whereas a mixed term is absent from the actual 4-form. 
Therefore $\kappa\neq0$, and
\be
\label{GKWdiagnostic}
\lambda,\tau,\sigma,\kappa\neq0\, ,
\qquad
f=\alpha=D=q=\rho=\beta=0\, .
\ee

It follows again that $\mathcal W_3=\mathcal W_4=0$, while
$\mathcal W_5$ takes the form in eq.~\eqref{w5ofMN}. This
AdS$_3$ fixed point preserves four of the thirty-two supercharges,
divided into two super-Poincar\'e and two superconformal charges.

The two fixed points therefore activate the same singlet and adjoint
components $(\lambda,\tau)$ of $F$, while the K\"ahler four-cycle
fixed point additionally activates the singlet and primitive
$(2,2)$ components $(\sigma,\kappa)$ of $G$. 
Their distinct fixed-point geometries and preserved supersymmetries 
are thus resolved by a small number of classification diagnostics, 
without requiring an explicit SU$(3)$ basis.

\paragraph{Running dilaton in charged singularity.}

The static charged AdS$_7$ solutions of \cite{Cvetic:1999ne,Liu:1999ai}
provide the first substantial activation in SU$(3)$ content.  
On the equal-charge BPS branch,
\be
X=H^{-1/5}\, ,
\qquad
h=X^4\, ,
\qquad
G=0\, .
\ee
The vanishing of $G$, inserted in the singlet relation in eq.~\eqref{gnpalg}, 
fixes
\be
\label{CLlambda}
\lambda \sim g\,\Delta(X)\not\equiv0\, .
\ee
Moreover, eq.~\eqref{gnpdiff1} gives directly
\be
\label{CLelectric}
\iota_KF
\sim \dd(X^5)\neq0\, ,
\ee
so that $f\neq0$, while the running dilaton gives $D\neq0$.
Moreover, $G=0$ gives $\beta=0$, and hence $\alpha=0$ by eq.~\eqref{gnpalg}.  
Were $\tau=0$, one would have $F_\perp \sim\lambda J$, and therefore
$F_\perp\wedge F_\perp \sim \lambda^2J\wedge J\neq0$, in contradiction with 
eq.~\eqref{noGentails}, and eq.~\eqref{CLlambda}. Hence $\tau\not\equiv0$.
Thus, the complete SU$(3)$ content is
\be
\label{CLfingerprint}
\lambda,\tau,f,D\neq0\, ,
\qquad
\alpha=0\, ,
\qquad
G=0\, .
\ee

The primitive 3-form component $\rho$ therefore still vanishes,
so $\mathcal W_3=0$. Since $h=X^4$ and $\dd X\not\equiv0$,
\be
\label{w4singularity}
\mathcal W_4=-\frac12\dd\log h=-2\dd\log X\not\equiv0\, .
\ee
Since $G=0$ also gives $q=0$, the fifth torsion form is
\be
\mathcal W_5
\label{w5singularity}
=-\frac34\dd\log h-4gA_\perp\lrcorner J\, .
\ee

This is the first solution in the progression with a running dilaton. 
In this static BPS solution, however, the dilaton runs to a singular endpoint. 
Near the charged origin, the areal radius of the transverse
5-sphere behaves as
\be
\label{CLsingularity}
H\sim \frac{Q}{r^4}\, ,
\qquad
X\sim Q^{-1/5}r^{4/5}\longrightarrow0\, ,
\qquad
R_{S^5}\equiv rH^{1/5}
\sim Q^{1/5}r^{1/5}\longrightarrow0\, .
\ee
Thus the dilaton reaches the limit of its physical range, while the
transverse sphere collapses, rather than approaching a regular
finite-area horizon.  
The supersymmetric equal-charge solution is therefore a naked singularity, 
rather than a black hole as noted in \cite{Liu:1999ai}, 
preserving eight supercharges.

\paragraph{AdS$_5$ flux-brane.}

The magnetic AdS$_5$ flux-branes of \cite{Lu:2003iv} provide a further
running-dilaton family. The original construction allows two independent
Cartan charges and already considers their equal-charge specialization,
which we examine here. These solutions were obtained by analytically
continuing the static electric family of \cite{Cvetic:1999ne,Liu:1999ai}, 
as shown in eqs.~(3.3)--(3.5)
of \cite{Lu:2003iv}. Their AdS$_5$ slices are supported by magnetic
Cartan flux through the two transverse directions.

Using the same radial harmonic function $H$ as in eq.~\eqref{CLsingularity} 
and a constant normalization of a single SM Killing spinor, the equal-charge 
solution gives
\be
\label{LPVtimelike}
X=H^{-1/5}\, ,\qquad h=X^4\, ,\qquad G=0\, ,
\ee
and the Killing vector is timelike since
\be
K^2=-h^2=-X^8<0\, .
\ee 
Since $F\neq0$ and $G=0$, the solution belongs 
to $\mathcal C_{\rm g}^{\rm active}$ by eq.~\eqref{gpvacuum}.
Eq.~\eqref{gnpdiff1} gives
\be
\iota_KF\sim\dd(X^5)\not\equiv0\, ,
\qquad \Rightarrow \qquad f,D\not\equiv0\, .
\ee
The same singlet and wedge-product arguments used prior to 
eq.~\eqref{MNdiagnostic} give $\lambda,\tau\not\equiv0$, 
while $G=0$ implies $\beta=\alpha=0$.
The complete SU$(3)$ content therefore coincides with
eq.~\eqref{CLfingerprint}. The torsion forms likewise satisfy $\mathcal W_3=0$, 
as well as eqs.~\eqref{w4singularity}, \eqref{w5singularity} for the fourth and 
fifth torsion forms, respectively.

The magnetic flux-brane and the static electric solution consequently
have identical entries in Table~\ref{tab:known-gactive-solutions},
despite their different real Lorentzian geometries.
For nonzero charge, on their respective Lorentzian domains,
\be
\label{LPVdistinct}
\left.F_{\mu\nu}F^{\mu\nu}\right|_{\rm magnetic}>0\, ,
\qquad
\left.F_{\mu\nu}F^{\mu\nu}\right|_{\rm electric}<0\, .
\ee
Thus no real coordinate or gauge transformation identifies them.
As noted above, they are related via analytic continuation.
This pair of distinct solutions provides an example of the additional geometric 
freedom left by the SU$(3)$ component and torsion data recorded in the table.

The global behaviour of these two solutions also differs. 
The equal-charge parameter choice in eq.~(3.11) of \cite{Lu:2003iv} admits 
a smooth noncompact flux-brane completion. It interpolates between an 
AdS$_5\times H^2$ region and an asymptotically locally AdS$_7$ region with 
AdS$_5\times S^1$ conformal boundary. At the AdS$_5\times H^2$ end,
\be
\label{LPVfixedpointlimit}
X^5\longrightarrow\frac34\, ,
\qquad D,f\longrightarrow0\, ,
\ee
recovering the equal-charge fixed point of \cite{Maldacena:2000mw}
discussed above. The constant-scalar fixed point is therefore a limit
of this running-dilaton family. The generic charged flux-brane
preserves eight supercharges.

\paragraph{Warped AdS$_3 \times \Sigma_g$.}

The warped AdS$_3$ solution in eqs.~(4.1), (4.2) of \cite{Cheung:2022ilc}, 
denoted here by Cheung et al.~(1) in Table~\ref{tab:known-gactive-solutions} 
below, provides a running-dilaton family with nonzero $4$-form flux. 
Its construction uplifts the AdS$_3\times\Sigma_g$ near-horizon geometry 
of a five-dimensional black string, where $\Sigma_g$ is
a smooth compact Riemann surface of genus $g>1$.
The uplift introduces a two-dimensional fibre twisted over $\Sigma_g$,
forming the four-dimensional space transverse to AdS$_3$.
The local solution has two independent Cartan charges, whereas the compact 
global construction of \cite{Cheung:2022ilc} requires unequal magnetic flux 
numbers $p_1\neq p_2$ through the spindle fibre.
Using the radial coordinate $y$ in eq.~(4.1) and the charge parameters
$q_1,q_2$ in eq.~(2.5) therein, we examine
the equal-charge specialization $q_1=q_2$, for which
\be
\label{CFGSsigmascalar}
X^5=\frac{y^2}{y^2+q_1}\, .
\ee
Using the Killing-spinor ansatz in eqs.~(B.8)--(B.10) of
\cite{Cheung:2022ilc}, one can choose a single SM Killing spinor
whose warp factor obeys $h=X^4$, as in eq.~\eqref{LPVtimelike},
such that the vector bilinear is timelike by eq.~\eqref{Fierznorms}.

For generic nonzero charge, the dilaton gradient and longitudinal
Cartan flux have the same non-vanishing behaviour as in
eq.~\eqref{CLelectric}, giving $f,D\not\equiv0$.
The solution consequently belongs to $\mathcal C_{\rm g}^{\rm active}$
by eq.~\eqref{primitivecriterion}.
The relation between $h$ and $X$ also sets $q=0$ through
eq.~\eqref{gnpdiff}.
For the equal-charge global completion that we present below in 
eq.~\eqref{CFGSsigmaglobal}, we find that the local cusp scaling limit 
recovers the K\"ahler four-cycle fixed point with the SU$(3)$ content in 
eq.~\eqref{GKWdiagnostic}.
Applying the non-vanishing argument following 
eq.~\eqref{BPSinheritance} to this regular limit gives 
$\lambda,\tau,\sigma,\kappa\not\equiv0$ as inherited BPS components.

Only $\alpha,\beta,\rho$ remain to be determined.
After mapping the fields in eq.~(4.2) of \cite{Cheung:2022ilc}
to our conventions, direct projection gives
\be
\label{CFGSsigmatests}
\left(F_\perp\right)^{(2,0)+(0,2)}=0\, .
\ee
Thus $\alpha=0$ by eq.~\eqref{projalpha}, and
eq.~\eqref{gnpalg} then implies $\beta=0$.
For $\rho$, since $q=\zeta=0$, eq.~\eqref{longG} gives
$\rho+\overline\rho=\iota_KG$.
The $4$-form contains a nonzero term proportional to
$\operatorname{vol}(\mathrm{AdS}_3)\wedge\dd y$.
The timelike $K$ has a nonzero AdS$_3$ component, hence $\rho\not\equiv0$.

The complete SU$(3)$ content is consequently
\be
\label{CFGSsigmafingerprint}
\lambda,\tau,f,D,\sigma,\kappa,\rho\neq0\, ,
\qquad
\alpha=q=\beta=0\, .
\ee
The nonzero primitive component $\rho$ gives
$\mathcal W_3\not\equiv0$ through eq.~\eqref{w3}.
The fourth and fifth torsion forms are given by
eqs.~\eqref{w4singularity}, \eqref{w5singularity}, respectively.
Thus this family retains the Cartan-flux content of
eq.~\eqref{CLfingerprint} while activating $\sigma$, $\kappa$
and $\rho$ of the $4$-form.

After eq.~(2.11) of \cite{Cheung:2022ilc} the compact global construction 
requires restriction to unequal magnetic fluxes as $p_1p_2<0$, where
the compact global solutions of \cite{Cheung:2022ilc} have a spindle fibre.
Yet, section~2.2 states that the solution ansatz in eqs.~(2.13)--(2.15), 
is valid for arbitrary $q_1,q_2$ and accommodates noncompact geometries
with other parameter choices and coordinate ranges.
For the equal-charge fields, we indeed present a smooth noncompact
global completion not specified in that work. 
Using the fibre coordinate $\phi$ of eq.~(4.1) therein,
we take, in their normalization,
\be
\label{CFGSsigmaglobal}
q_1=q_2=\frac{27}{16}\, ,
\qquad
y>\frac94\, ,
\qquad
\phi\sim\phi+\frac{8\pi}{3}\, .
\ee
As $y\to9/4$, the geometry approaches a regular cusp at infinite
proper distance,
while $y\to\infty$ similarly approaches 
an asymptotically locally AdS$_7$ geometry.
The curvature and fluxes remain bounded at both ends.
With the chosen circle period, the bundle defined in eq.~(4.3) of 
\cite{Cheung:2022ilc} has a first Chern number of $g-1\in\mathbb Z$, 
ensuring that the circle fibre can be patched consistently over $\Sigma_g$.
The Killing spinor obtained from eq.~(B.8) therein is globally
well defined on this bundle.
Taking the universal cover of AdS$_3$ provides a global time
whose gradient remains timelike in the full metric, excluding
closed timelike curves.
These additional global checks establish an admissible equal-charge solution 
for inclusion in Table~\ref{tab:known-gactive-solutions}.
As noted above, for our noncompact completion in eq.~\eqref{CFGSsigmaglobal}, 
the local cusp scaling limit used in the inheritance argument gives the
AdS$_3\times\Sigma_g\times H^2$ representative of the fixed point
of \cite{Gauntlett:2000ng} in eq.~\eqref{GKWdiagnostic}.
As $y\to9/4$, the dilaton and its gradient have the limits in
eq.~\eqref{LPVfixedpointlimit}, with $\rho\to0$ as well.
The generic solution preserves four supercharges.

\paragraph{Black objects.}

Adding rotation to the static electric family activates
a substantially richer 4-form flux and gives rise to black-hole solutions.
The rotating charged black-hole family of \cite{Chong:2004dy}, with
three equal angular momenta and two independent electric charges,
also marks a methodological change in the present sequence: it was
constructed from a bosonic ansatz and verified against the bosonic field
equations, rather than derived from the KSEs.
The supersymmetric limit examined in \cite{Chong:2004dy}
has closed timelike curves.
The regular supersymmetric black hole instead was identified only in 
\cite{Cvetic:2005zi}, with its BPS condition 
and the additional horizon-regularity condition
given there in eqs.~(4.34), (4.38), respectively.
The component tests below use the latter supersymmetric branch of
\cite{Cvetic:2005zi} that admits regular black holes.

The static specialization of the solution family in \cite{Chong:2004dy} 
recovers the charged family of \cite{Cvetic:1999ne,Liu:1999ai}. 
Equation~\eqref{BPSinheritance}
therefore shows that the earlier BPS solution is contained in the
supersymmetric sector of the black-hole family of \cite{Chong:2004dy}. 
Consequently,
\be
\lambda,\tau,f,D\neq0\, ,
\ee
as inherited BPS components.

Direct projection of the general equal-rotation solution gives
\be
(\dd\omega)^{(2,0)+(0,2)}=0\, ,
\qquad
\left(\iota_KG\right)^{(2,1)_0+(1,2)_0}=0\, .
\ee
Their identically-vanishing values are consistent with the enhanced
equal-rotation symmetry, and remain unchanged upon restriction 
to the BPS branch. Equations~\eqref{gnptwist3}, \eqref{gnpalg}, then imply
\be
\alpha=\beta=\rho=0\, ,
\ee
and hence $\mathcal W_3=0$. 

The components that are non-vanishing in the general bosonic solution, 
but are not inherited from the static BPS solution, 
must instead be tested on the regular BPS branch.
Equations~\eqref{gnpdiff}, \eqref{gnpalg} and \eqref{gnptwist2}, 
provide the concise diagnostics, where direct projection 
on the regular BPS branch gives
\begin{align}
q+\overline q
&\sim J\lrcorner\iota_KG\not\equiv0\, ,
\\*
\sigma
&\sim J\lrcorner J\lrcorner G_\perp\not\equiv0\, ,
\\*
\kappa
&\sim
\Pi^{(1,1)_0}\!\left(J\lrcorner G_\perp\right)
\not\equiv0\, .
\end{align}
The inherited components also remain non-vanishing after imposing the 
horizon-regularity condition of \cite{Cvetic:2005zi}.
Thus the regular BPS black hole has
\be
\label{Chongfingerprint}
\lambda,\tau,f,D,\sigma,\kappa,q\neq0\, ,
\qquad
\alpha=\beta=\rho=0\, .
\ee
The warp factor is non-constant on the regular BPS branch, and
therefore $\mathcal W_4\not\equiv0$.
Thus, relative to the static charged solution, rotation additionally
activates $q$, $\sigma$ and $\kappa$. Using eq.~\eqref{gnpdiff},
the fifth torsion form is
\be
\label{w5withqon}
\mathcal W_5
=-\frac34\dd\log h
+\frac{X^2}{4h}(q+\overline q)
-4gA_\perp\lrcorner J\, .
\ee
The regular supersymmetric black hole preserves four of the
thirty-two supercharges.

The $L^{p,q,r}$ solutions of \cite{Bobev:2025xan} extend the
equal-rotation black holes of \cite{Chong:2004dy} by replacing
the $\mathbb{CP}^2$ base of the spherical-horizon fibration with
a more general four-dimensional K\"ahler--Einstein geometry.
This changes the local geometry, while smooth compact
$L^{p,q,r}$ horizons have topology $S^2\times S^3$.
The spherical-horizon solutions are recovered as a limit of
the local metric, with the proper choice of angular periods and identifications.
Let us highlight that this family was also constructed at the bosonic level.

We restrict here to the equal-charge Lorentzian solutions. 
The geometric replacement leaves the radial functions
and the BPS and radial regularity conditions unchanged 
with respect to \cite{Cvetic:2005zi}.
Equation~\eqref{BPSinheritance} therefore establishes the
generic non-vanishing components already identified in
eq.~\eqref{Chongfingerprint},
\be
\lambda,\tau,f,D,\sigma,\kappa,q\neq0\, .
\ee

Only the remaining components require separate projection tests.
Equation~\eqref{gnpGperp} determines $J=J_2+J_4$, where
$J_2$ and $J_4$ are the fundamental forms on the radial--fibre
plane and the four-dimensional K\"ahler--Einstein block of
the six-dimensional base, respectively.
The radial coefficients and K\"ahler fibre curvature in
eqs.~(4.1), (4.14a) of \cite{Bobev:2025xan} give
\be 
\dd\omega\in\operatorname{span}\{J_2,J_4\}\, , \qquad 
\dd J_2,\dd J_4\propto\dd r\wedge J_4\, ,
\ee 
where $r$ is radial.
Since $h,X$ are radial, eq.~\eqref{gnpdiff} then yields
\be
(\dd\omega)^{(2,0)+(0,2)}=0\, ,
\qquad
\iota_KG\propto\dd r\wedge J_4=\dd r\wedge J\, .
\ee
The second relation gives $\rho=0$, while the first,
together with eqs.~\eqref{gnptwist3}, \eqref{gnpalg},
gives $\alpha=\beta=0$.
Thus the $L^{p,q,r}$ solutions have the same 
non-vanishing SU$(3)$ components as in
eq.~\eqref{Chongfingerprint}, with
$\mathcal W_3=0$, $\mathcal W_4\not\equiv0$, and
$\mathcal W_5$ given by eq.~\eqref{w5withqon}.

A distinct extension of the solution family in \cite{Chong:2004dy} 
is the rotating charged black-hole solution of \cite{Chow:2007ts}, 
with three independent angular momenta. It provides an endpoint 
of the present progression in SU$(3)$-component activation.
Following \cite{Chong:2004dy}, it was constructed at the bosonic level, 
and its equal-rotation specialization recovers the former. 
Hence the components
\be
\lambda,\tau,f,D,\sigma,\kappa,q\neq0\, ,
\ee
are inherited by its supersymmetric sector.

Only $\alpha,\beta,\rho$ remain to be determined.
On the regular BPS black hole specified by
the BPS condition and the black-hole regularity condition in
eqs.~(3.31), (3.43) of \cite{Chow:2007ts}, only two residual projection tests 
are then required, which give
\be
\label{Chowtests}
\left(F_\perp\right)^{(2,0)+(0,2)}\not\equiv0\, ,
\qquad
\left(\iota_KG\right)^{(2,1)_0+(1,2)_0}\not\equiv0\, .
\ee
Thus $\alpha,\rho\not\equiv0$, and eq.~\eqref{gnpalg}
then implies $\beta\not\equiv0$.

The generic regular black hole of \cite{Chow:2007ts} consequently 
activates every SU$(3)$ component permitted in the gauged active class:
\be
\label{Chowfingerprint}
\lambda,\tau,f,\alpha,D,\sigma,\kappa,q,\beta,\rho\neq0\, ,
\qquad
\zeta=0\, .
\ee
Thus, no additional SU$(3)$ representation can be activated within
this equal-charge class. In particular,
\be
\mathcal W_3\sim\operatorname{Re}\rho\neq0\, .
\ee
The warp factor is generically non-constant, so $\mathcal W_4\not\equiv0$, 
while $\mathcal W_5$ retains the form in eq.~\eqref{w5withqon}.

This generic regular BPS black hole preserves four of the thirty-two
supercharges. The same local supersymmetric family also admits
smooth topological solitons as distinct global completions. These
generically preserve four supercharges, with special admissible
parameter choices preserving eight.

\paragraph{Warped AdS$_3$ with full SU$(3)$ activation.}

The warped AdS$_3$ family in eqs.~(3.7), (3.8) of \cite{Cheung:2022ilc},
denoted here by Cheung et al.~(2) in Table~\ref{tab:known-gactive-solutions}
below, provides another geometry with the same complete SU$(3)$-component
activation in this class.
It admits Cheung et al.~(1) as a local scaling limit even before
our equal-charge restriction.
Its construction uplifts a five-dimensional warped AdS$_3$ solution.
In the compact constructions of \cite{Cheung:2022ilc}, both families
have a spindle fibre, but the base is $\Sigma_g$ for Cheung et al.~(1)
and a spindle for Cheung et al.~(2).
The local solution has two independent Cartan charges, whereas its original
compact global construction requires unequal magnetic flux numbers
$p_1\neq p_2$ through the spindle fibre.
We again study instead its equal-charge specialization.

The dilaton is given by eq.~\eqref{CFGSsigmascalar}, with the same fibre
coordinate $y$ and charge parameters $q_1=q_2$.
Let $x$ and $a$ be the radial coordinate and parameter of the
five-dimensional solution in eqs.~(3.1)--(3.3) of \cite{Cheung:2022ilc}.
Using a suitable combination of the spinors in their eq.~(B.20),
we choose a single SM Killing spinor.
After a constant normalization of the Killing spinor, the warp factor is
\be
\label{CFGSspindlewarp}
h=X^4\left(1-\frac{a}{3x}\right)\, .
\ee
On the domain $x>a/3$ used below, $h>0$, so the Killing vector is timelike
by eq.~\eqref{Fierznorms}.
Both $X$ and $Xh$ are non-constant, and therefore the dilaton gradient and
eq.~\eqref{gnpdiff1} give $D,f\not\equiv0$.
The solution therefore belongs to $\mathcal C_{\rm g}^{\rm active}$
by eq.~\eqref{primitivecriterion}.
Substituting eq.~\eqref{CFGSspindlewarp} into eq.~\eqref{gnpdiff} gives
\be
\label{CFGSspindleq}
q+\overline q=-\frac{aX^2}{3x^2}\dd x\not\equiv0\, .
\ee
For $a=1$, the $x\to1$ scaling limit noted above is realized
at the cusp of the global completion that we present below
in eq.~\eqref{CFGSspindleglobal}.
Applying the non-vanishing argument following eq.~\eqref{BPSinheritance}
to this regular limit gives $\lambda,\tau,\sigma,\kappa,\rho\not\equiv0$
as inherited BPS components from eq.~\eqref{CFGSsigmafingerprint}.

Only $\alpha,\beta$ remain to be determined.
After mapping the fields in eq.~(3.8) of \cite{Cheung:2022ilc}
to our conventions, direct evaluation using eq.~\eqref{projalpha}
gives $\alpha\not\equiv0$.
Equation~\eqref{gnpalg} then implies $\beta\not\equiv0$.
Thus, relative to eq.~\eqref{CFGSsigmafingerprint}, this family
additionally activates $\alpha,q,\beta$ and has the complete SU$(3)$
content in eq.~\eqref{Chowfingerprint}.
The nonzero $\rho$ gives $\mathcal W_3\not\equiv0$ through
eq.~\eqref{w3}, while the non-constant $h$ gives
$\mathcal W_4\not\equiv0$ through eq.~\eqref{w4exactornot}.
The fifth torsion form is given by eq.~\eqref{w5withqon}.

As for the $\Sigma_g$ family, the compact spindle fibre of
\cite{Cheung:2022ilc} is incompatible with equal charges.
We instead present here a smooth noncompact global completion.
Using the angular coordinates $\psi,\phi$ of eq.~(3.7) therein,
we retain the charges, $y$ domain and $\phi$ period in
eq.~\eqref{CFGSsigmaglobal} and additionally take
\be
\label{CFGSspindleglobal}
a=1\, ,\qquad x>1\, ,\qquad \psi\sim\psi+2\pi\, .
\ee
For both families, we retain the local metric and matter fields of
\cite{Cheung:2022ilc}, specialized to equal charges.
The parameter choices, coordinate ranges and circle periods in
eqs.~\eqref{CFGSsigmaglobal}, \eqref{CFGSspindleglobal}
determine our noncompact global completions.
In their original unequal-charge compact constructions, both families have a
spindle fibre, over $\Sigma_g$ for Cheung et al.~(1) and over
another spindle for Cheung et al.~(2).
Our choices for the equal-charge completions give a noncompact fibre 
in both cases. The base is still the compact surface $\Sigma_g$ for (1),
whereas for (2) the chosen $(x,\psi)$ domain is noncompact.
The scaling limit recovers the local fields of (1), but does
not reproduce its compact base $\Sigma_g$.

The ends $x\to1$ and $y\to9/4$ are regular cusps at infinite proper
distance. The limits $x\to\infty$ and $y\to\infty$ also lie at infinite
proper distance, and the curvature and fluxes remain bounded at all ends.
The common warp factor has a positive lower bound, and the $y$-radial metric 
coefficient is independent of $x$, so approaching any radial end requires 
infinite proper length, even along paths on which $x$ and $y$ vary 
simultaneously.
Together with the compact angular directions, this establishes
geodesic completeness of the four-dimensional manifold transverse
to AdS$_3$.
The gauge potentials are globally defined, and the Killing spinor in
eq.~(B.20) of \cite{Cheung:2022ilc} is independent of $\psi$.
The spinor's phase $e^{-3i\phi/4}$ is single-valued with the period in
eq.~\eqref{CFGSsigmaglobal}, so periodic spin structures on both circles
give a globally defined spinor.
Taking the universal cover of AdS$_3$ supplies a global time with
timelike gradient in the full metric, excluding closed timelike curves.
On this completion, the warp factor in eq.~\eqref{CFGSspindlewarp}
remains strictly positive, and none of the additional SU$(3)$
components vanishes identically.
These additional global checks establish an admissible equal-charge
solution for inclusion in Table~\ref{tab:known-gactive-solutions}.
This warped AdS$_3$ geometry of \cite{Cheung:2022ilc} and the rotating black
hole of \cite{Chow:2007ts} therefore have identical entries in
Table~\ref{tab:known-gactive-solutions}.
We have not found our particular equal-charge global completions presented here
in the literature, and their holographic interpretation remains open.
The generic solution preserves four supercharges.

We note that the broader bosonic AdS$_3$ family of \cite{Faedo:2024upq}
recovers Cheung et al.~(2) in the scaling limit of Appendix~A.2 therein.
The seven-dimensional KSEs are not checked, and no additional 
new supersymmetric solution beyond this known limit is demonstrated.
We therefore assign no additional entry to this extended family in
Table~\ref{tab:known-gactive-solutions}.

\begin{table}[t]
\centering
\small
\setlength{\tabcolsep}{1.5pt}
\renewcommand{\arraystretch}{1.1}
\begin{tabular}{|l|c|c|c|c|c|c|c|c|c|c||c|c|c|}%
\hline
Solution
& $\lambda$
& $\tau$
& $f$
& $\alpha$
& $D$
& $\sigma$
& $\kappa$
& $q$
& $\rho$
& $\beta$
& $\mathcal W_3$
& $\mathcal W_4$
& $\mathcal W_5$
\\
\hline
\hline
Maldacena--Nu\~nez \cite{Maldacena:2000mw}
& $\checkmark$
& $\checkmark$
& $0$
& $0$
& $0$
& $0$
& $0$
& $0$
& $0$
& $0$
& $0$
& $0$
& $-4gA_\perp\lrcorner J$
\\
\hline
Gauntlett et al.~\cite{Gauntlett:2000ng}
& $\checkmark$
& $\checkmark$
& $0$
& $0$
& $0$
& $\checkmark$
& $\checkmark$
& $0$
& $0$
& $0$
& $0$
& $0$
& $-4gA_\perp\lrcorner J$
\\
\hline
Cveti\v{c} et al./Liu et al.~\cite{Cvetic:1999ne,Liu:1999ai}
& $\checkmark$
& $\checkmark$
& $\checkmark$
& $0$
& $\checkmark$
& $0$
& $0$
& $0$
& $0$
& $0$
& $0$
& $\checkmark$
& $-\frac34\dd\log h
-4gA_\perp\lrcorner J$
\\
\hline
L\"u et al.~\cite{Lu:2003iv}
& $\checkmark$
& $\checkmark$
& $\checkmark$
& $0$
& $\checkmark$
& $0$
& $0$
& $0$
& $0$
& $0$
& $0$
& $\checkmark$
& $-\frac34\dd\log h
-4gA_\perp\lrcorner J$
\\
\hline
Cheung et al.~(1)~\cite{Cheung:2022ilc} 
& $\checkmark$
& $\checkmark$
& $\checkmark$
& $0$
& $\checkmark$
& $\checkmark$
& $\checkmark$
& $0$
& $\checkmark$
& $0$
& $\checkmark$
& $\checkmark$
& $-\frac34\dd\log h
-4gA_\perp\lrcorner J$
\\
\hline
Chong et al.~\cite{Chong:2004dy}
& $\checkmark$
& $\checkmark$
& $\checkmark$
& $0$
& $\checkmark$
& $\checkmark$
& $\checkmark$
& $\checkmark$
& $0$
& $0$
& $0$
& $\checkmark$
& $-\frac34\dd\log h
+\frac{X^2}{2h}\operatorname{Re}q%
-4gA_\perp\lrcorner J$
\\
\hline
Bobev et al.~\cite{Bobev:2025xan}
& $\checkmark$
& $\checkmark$
& $\checkmark$
& $0$
& $\checkmark$
& $\checkmark$
& $\checkmark$
& $\checkmark$
& $0$
& $0$
& $0$
& $\checkmark$
& $-\frac34\dd\log h
+\frac{X^2}{2h}\operatorname{Re}q%
-4gA_\perp\lrcorner J$
\\
\hline
Chow~\cite{Chow:2007ts}
& $\checkmark$
& $\checkmark$
& $\checkmark$
& $\checkmark$
& $\checkmark$
& $\checkmark$
& $\checkmark$
& $\checkmark$
& $\checkmark$
& $\checkmark$
& $\checkmark$
& $\checkmark$
& $-\frac34\dd\log h
+\frac{X^2}{2h}\operatorname{Re}q%
-4gA_\perp\lrcorner J$
\\
\hline
Cheung et al.~(2)~\cite{Cheung:2022ilc} 
& $\checkmark$
& $\checkmark$
& $\checkmark$
& $\checkmark$
& $\checkmark$
& $\checkmark$
& $\checkmark$
& $\checkmark$
& $\checkmark$
& $\checkmark$
& $\checkmark$
& $\checkmark$
& $-\frac34\dd\log h
+\frac{X^2}{2h}\operatorname{Re}q%
-4gA_\perp\lrcorner J$
\\
\hline
\end{tabular}
\caption{
Progressive activation of the SU$(3)$ components of the fluxes and dilaton
gradient, as defined in Appendix~\ref{su3projections}, and of the base torsion
in the known supersymmetric solution families constructed directly in
the maximal Cartan-gauged theory, and covered within the equal-charge sector.
A check mark denotes a component or torsion form that is
generically non-vanishing.
All solutions belong to the gauged active class
$\mathcal C_{\rm g}^{\rm active}$, and thus have $\zeta=0$.
Cheung et al.~(1), (2) denote our equal-charge completions of the families in 
Sections~4, 3 of \cite{Cheung:2022ilc}, respectively, where their equal-charge 
variations included here involve the noncompact global completions that we 
presented and verified in eqs.~\eqref{CFGSsigmaglobal}, 
\eqref{CFGSspindleglobal}. 
This is since their original compact completions exclude equal charges.
The Bobev et al.~row refers to the equal-charge Sasaki--Einstein extension,
which locally generalizes \cite{Chong:2004dy}, with $L^{p,q,r}$ horizons.
The Chow row refers to the generic BPS black hole.
The same local supersymmetric family admits topological-soliton completions.
The flux-brane family of \cite{Lu:2003iv}, both warped AdS$_3$
families of \cite{Cheung:2022ilc}, and the more general black-hole
solutions of \cite{Bobev:2023bxl} admit two independent Cartan
charges and extend beyond the equal-charge sector resolved here.}
\label{tab:known-gactive-solutions}
\end{table}

The progression of these supersymmetric solutions is summarized in 
Table~\ref{tab:known-gactive-solutions}. 
Only the three distinguishing torsion forms
$\mathcal W_3$, $\mathcal W_4$ and $\mathcal W_5$, are recorded in the table:
All solution families belong to the class 
$\mathcal C_{\rm g}^{\rm active}$, 
where $s=0$, and hence $\mathcal W_1=\mathcal W_2=0$. 
Although $\mathcal W_3$ and $\mathcal W_4$ distinguish among solutions,
they add no independent information beyond the SU$(3)$ content. 
The genuinely additional geometric information is carried by $\mathcal W_5$
through its dependence on the transverse gauge connection
$A_\perp$, which is not determined by the SU$(3)$ content of the fluxes alone.

\subsubsection{Minimal Theory}
\label{minsol}

The minimal theory admits an SU$(2)$-gauged extension, which in turn admits 
a consistent truncation to the U$(1)$-gauged theory. 
As we established in Section~\ref{7dsugra}, and further discussed in the complete
classification framework of Section~\ref{solanalysis}, our formulation 
of the equal-charge maximal Cartan-gauged theory admits such a consistent 
truncation to the minimal theory.
Therefore we can also assess within our classification further timelike 
supersymmetric solutions that have been considered directly 
in the minimal theory.

\paragraph{Mapping to \cite{Cariglia:2004qi}.}

Before assessing the proposed solutions in the early search of  
\cite{Cariglia:2004qi}, we first establish its map to our conventions, 
noting that their odd-dimensional Clifford representation agrees with ours,
whereas their spacetime orientation is opposite to ours, 
so that $\star_{\rm C}=-\star$. This relative sign, together with their 
Clifford-duality identity in eq.~(A5), is included in the identifications below 
of the $3$- and $4$-form fluxes.
Accounting also for the scalar factors included in their definitions
of the fluxes, the map is therefore
\begin{align}
\label{Carigliamap}
e^{\phi_{\rm C}}&=X\, ,
&
A_{\rm C}^{3}&=A\, ,
&
e^{\phi_{\rm C}}F_{\rm C}^{3}&=F\, ,
\nn\\
A_{(3){\rm C}}&=-\frac12B\, ,
&
e^{-2\phi_{\rm C}}G_{\rm C}&=-\frac12G\, ,
&
g_{\rm C}&=-4g\, ,
&
h_{\rm C}&=-\frac14m\, ,
\end{align}
with an overall factor of two between the two actions. 
Upon imposing the maximal-supersymmetry coupling relation in 
eq.~\eqref{maxsusy}, equivalently $g_{\rm C}=16h_{\rm C}$, 
and identifying $\sigma^3=\Gamma_+$ for the 
Cartan-gauged reductions, this map agrees 
with our maximal to minimal truncation.

\paragraph{Ungauged sector.}

To the best of our knowledge, the work of \cite{Cariglia:2004qi} is 
the only one to have searched for supersymmetric solutions in 
the ungauged sector of either the minimal or maximal theory.
Yet, the search therein was highly restrictive. 
The dilaton and warp factor were fixed, while the internal
SM orientation---represented in their bilinears by the counterpart
of our $(\theta,\varphi)$---was also restricted. In the ungauged
theory this restriction freezes an otherwise independently specified
spinor-bilinear orientation.
In our formulation these restrictions amount to 
\be
\dd X=\dd h=0\, ,
\qquad
\dd\theta=0\, ,
\qquad
s\,\dd\varphi=0\, .
\ee
Therefore, according to our classification, 
the relevant SU$(3)$ content is already quite restricted,
\be
\label{identifiedsu3restrict}
D=f=q=\zeta=0\, .
\ee
Moreover, the intrinsic torsion of the base was 
restricted throughout to 
\be
\mathcal W_1 = \mathcal W_2 = \mathcal W_5 = 0\, ,
\ee 
while the remaining ans\"atze retained at most a few flux and twist components.
These restrictions probe only a thin part of the solution space 
permitted by our ungauged classes, and leave considerable room for 
constructing solutions with a running dilaton and richer SU$(3)$ content.

The Calabi--Yau configuration in eqs.~(42)--(44) of
\cite{Cariglia:2004qi} is already restricted to the Cartan, and satisfies
\be
F \propto J\, , %
\qquad
G \propto J\wedge J\, .
\ee
Thus, the configuration belongs 
to the class $\mathcal C_{\rm ug}^{\rm active}$, with
\be
c=1\, ,
\qquad
\lambda,\sigma\neq0\, ,
\ee
while all other SU$(3)$ flux components and intrinsic torsion forms vanish.
We verified that this solution satisfies the complete KSE conditions, 
Bianchi identities, and EOMs of the corresponding ungauged active class, and
that its base admits a torsion-free SU$(3)$ structure.  
Hence this is a valid local supersymmetric solution of both the minimal 
ungauged theory, and the corresponding equal-charge maximal truncation.
However, as already explicitly noted in \cite{Cariglia:2004qi}, the resulting 
geometry contains closed timelike curves, which constitute a global pathology. 
Therefore, this solution is not included in Table 
\ref{tab:minimal-solutions} as a valid supersymmetric solution. 

The configuration in \cite{Cariglia:2004qi} whose base is the
three-dimensional complex Heisenberg group has
\be
\mathcal W_1=\mathcal W_2=\mathcal W_4=\mathcal W_5=0\, ,
\qquad
\mathcal W_3\neq0\, .
\ee
We verified directly that the SU$(2)$ configuration satisfies its
KSEs, Bianchi identities, and the bosonic EOMs derived from the
action of \cite{Cariglia:2004qi} (throughout these checks,
the 4-form metric-trace coefficient inside the $1/3$ bracket
of the Einstein equation is $3/20$, correcting their printed $1/10$). 

Under the Cartan restriction, the off-Cartan form in eqs.~(47), (58), 
is removed, whereas the independent holomorphic twist
in eqs.~(61), (62), remains in their retained Cartan through eqs.~(46), (59). 
Applying eq.~\eqref{Carigliamap}, with eqs.~(45), (60),  
in \cite{Cariglia:2004qi}, its SU$(3)$ content is
\be
D=f=\lambda=\tau=\zeta=q=\sigma=0\, ,
\qquad
\rho,\kappa\neq0\, ,
\ee
with $\alpha,\beta$ either both vanishing or both non-vanishing.
In the latter case, the active-class KSEs further impose
$s=0$ and $c=1$.

Therefore, our classification places their Cartan-restricted configurations 
in either the ungauged null class, %
with the internal SM bilinear orientation constant but not fixed, and
\be
\label{actualtrivialF}
F=\alpha=0
\quad \implies 
\quad \beta=0\, ,
\qquad \rho,\kappa\neq0\, ,
\ee
or in the ungauged active class and thus
\be
F,G\neq0\, ,
\qquad
\alpha,\beta,\rho,\kappa\neq0\, .
\ee
We then verified that both complex-Heisenberg Cartan-restricted configurations 
satisfy the complete KSEs, Bianchi identities, and bosonic EOMs of their 
respective ungauged classes. According to our classification their 
base then carries a balanced SU$(3)$ structure. 

However, eqs.~(54)--(56), (63), (64) of \cite{Cariglia:2004qi}
yield, for either Cartan restriction, suitably oriented circles
of radius $R$ in the four Cartesian base directions, at fixed time
and fibre coordinates, whose angular tangent has squared norm
at most $R^2-R^4/4$ in their normalization.
These contractible circles are timelike for $R>2$, including on
the universal cover, so both configurations contain closed timelike curves, 
which constitute a global pathology, and are thus excluded from
Table~\ref{tab:minimal-solutions}.

Finally, although the Hermitian SU$(2)$ configuration lies outside our Cartan 
classification, we verified directly that it would satisfy its KSEs, 
Bianchi identities, and bosonic EOMs, if eqs.~(77), (78), of 
\cite{Cariglia:2004qi} were replaced by the common corrected condition
\be
(C_6-C_5)^2=5+8|\alpha_{\text{Carig}}|^2\, .
\ee
Under the U$(1)$ restriction, the explicit ansatz for the Hermitian family 
in eqs.~(65), (67)--(76), together with eq.~\eqref{Carigliamap}, gives
\be
D=f=\alpha=\beta=\zeta=q=0\, ,
\qquad
\rho\neq0\, .
\ee
The Cartan restriction leaves the explicit base geometry, defined by their 
eqs.~(68)--(70), unchanged. However, applying the Cartan restriction 
in their KSE condition, eq.~(66), requires $\mathcal W_4=0$, which is 
incompatible with the geometry, 
since their eq.~(73) leads to $\mathcal W_4=\tfrac12(e^6-e^5)\neq0$. 
Independently in our classification, the constant warp factor and $\alpha=0$ 
likewise imply $\mathcal W_4=0$ through our eq.~\eqref{w4gen}. 
Thus, the proposed Hermitian SU$(2)$ configuration admits no U$(1)$-restricted 
supersymmetric solution.

To conclude the ungauged sector, the Calabi--Yau and complex-Heisenberg 
configurations are excluded due to their global pathology,
while the Hermitian configuration admits no supersymmetric
U$(1)$ restriction.
Thus none of these configurations yielded a valid ungauged solution.
This empty outcome should not be mistaken for a general obstruction in the 
theory: there remain good prospects for finding solutions beyond the narrow 
assumptions of that search.
The rapid examination of the proposed solutions via our classification 
framework illustrates its utility, which reduces the assessment of proposed 
configurations to short irreducible-component and projected-EOM tests. 
Notably, this analysis exposes the ungauged solution-space that remains 
unexplored.

\paragraph{Gauged sector.}

The early search of \cite{Cariglia:2004qi} further approached a rather
restricted sector of the possible gauged solutions. 
The search therein specialized to the Cartan-gauged sector, 
constant dilaton and warp factor, and further imposed $\iota_KG=0$. 
In terms of the SU$(3)$ decomposition, these simplifying assumptions of 
\cite{Cariglia:2004qi} directly amount to 
\be
\label{Carigliarestrict}
\zeta = D = f = q = \rho = 0\, .
\ee
Thus the fundamental isotype was mostly deactivated from the outset,
and the remaining holomorphic pair $(\alpha,\beta)$ also vanishes
in the gauged configurations worked out explicitly.
Among these configurations, the search did not go beyond
realizing the SU$(3)$ content $(\lambda,\tau,\sigma,\kappa)\neq0$
by recovering the formerly known solution of \cite{Gauntlett:2000ng}.
The earlier charged solution of \cite{Cvetic:1999ne,Liu:1999ai},
and the magnetic flux-brane of \cite{Lu:2003iv}
which had already activated the running dilaton 
and hence the fundamental isotype, were absent from
both the discussion and bibliography of \cite{Cariglia:2004qi}.

We next examine the proposed squashed AdS$_7$ configuration
of \cite{Cariglia:2004qi}, which has
\be
F=0\, ,
\qquad
G \propto g J\wedge J 
\quad \implies \quad 
\sigma \neq 0\, , \quad \lambda=\alpha=\tau=\beta=\kappa=0\, ,
\ee
in addition to the initial restrictive assumption in 
eq.~\eqref{Carigliarestrict}. Its only activated SU$(3)$ flux
component is therefore the singlet $\sigma$.  
Although the gauge and mass couplings are independent in the
minimal theory, this particular solution imposed $g_{\rm C}=4h_{\rm C}$.
Under eq.~\eqref{Carigliamap}, its $X=1$ value therefore imposes $m=4g$.
Thus, it solves the two-coupling minimal theory, but does not satisfy
the maximal-theory supersymmetric relation $m=g$ in eq.~\eqref{maxsusy}.

To determine whether the same squashed ansatz admits a
maximal-theory embedding, we restore the constant scalar before
imposing the maximal-theory coupling relation. The resulting
condition is
\be
4h_{\rm C}X^{-4}=g_{\rm C}X
\quad\Longleftrightarrow\quad
m=4gX^5\, .
\ee
Imposing the coupling relation of the maximal theory in eq.~\eqref{maxsusy}, 
we then find that a squashed AdS$_7$ background can be obtained
when the dilaton is fixed at
\be
X^5=\frac14\, ,
\ee
which notably, is not a stationary point of the scalar potential,
the 4-form flux is set to
\be 
G = - 2 g X^{-1} J \wedge J\, , 
\ee
and the corresponding potential, twist and base curvature are
\be
B=\frac1{2X^2}(\dd t+\omega)\wedge J\, ,
\qquad
\dd\omega=-4gXJ\, ,
\qquad
R^{(6)}_{ij}=-16g^2X^2g_{ij}\, .
\ee
We verified that this maximally-embedded squashed solution satisfies %
the complete KSE conditions, Bianchi identities, EOMs, and self-duality 
condition of the corresponding gauged null class 
$\mathcal C_{\rm g}^{\rm null}$. 
Following our classification, this means that its base 
admits a torsion-free U$(3)$ structure.  

Nevertheless, using the radial coordinate $r$ and periodic
$S^3$ Euler angle $\psi$ in eqs.~(101), (102), (105) of
\cite{Cariglia:2004qi}, the metric gives
$g_{\psi\psi}\propto(\cosh r-1)(3-\cosh r)$,
with a positive proportionality factor.
The corresponding circles are contractible and timelike for
$\cosh r>3$; neither passage to the universal cover nor the
constant scalar rescaling above removes this obstruction.
The maximally embedded configuration is therefore globally unphysical, 
and thus excluded from Table~\ref{tab:minimal-solutions}.

\paragraph{BPS flows and Janus solutions.}

The empty outcome of the above early search contrasts sharply with the
rich solutions of minimal seven-dimensional supergravity constructed
in \cite{Dibitetto:2017tve,Conti:2024rwd}.
Rather than imposing highly restrictive geometric ans\"atze, the latter BPS 
flows were found directly by solving for KSE systems and explicit 
Killing-spinor ans\"atze. 
Moreover, their maximal-theory embeddings that lie in the $F=0$ sector of 
the gauged null class $\mathcal C_{\rm g}^{\rm null}$
comprise two geometrically distinct families that realize the complete SU$(3)$ 
content permitted in this class, as we now show.

The two coupling deformations $g_{\rm D}$ and $h_{\rm D}$ remain independent
in the action and KSEs of \cite{Dibitetto:2017tve}, whereas 
\cite{Conti:2024rwd}, which uses the same field normalizations, and, 
up to redefinitions of Clifford-rotation angle and curvature parameters,
the same Killing-spinor ansatz and KSEs, imposes
\be
h_{\rm Co}=\frac{g_{\rm Co}}{2\sqrt{2}}\, ,
\ee
from the outset. Under the map in eq.~\eqref{Dibitettomap}, this is just the 
maximal-supersymmetry relation in eqs.~\eqref{mintomax}, \eqref{maxsusy}. 
Thus the work of \cite{Conti:2024rwd}, 
which further sets all vector fields to zero,  
is automatically embedded into the equal-charge maximal theory.
Throughout this analysis we use the corrected version of the spin-$\tfrac12$ KSE 
printed in eq.~(2.7) of \cite{Dibitetto:2017tve}, already established in 
Section~\ref{spinors}.

In these works, $\theta_{\rm D}$ and $\theta_{\rm Co}=2\theta_{\rm D}$ are the
angles multiplying $\gamma^{012}$ in the respective Killing-spinor
ans\"atze. They parametrize a spacetime Clifford rotation and should
not be confused with the internal SM angle $\theta$ defined here in
Section~\ref{gstructure}. The latter is carried by the SU$(2)$
orientation of their constant SM spinor underlying their Dirac notation.
When the Yang--Mills fields vanish, and no SU$(2)$ SM-orientation
is imposed, the solutions of \cite{Dibitetto:2017tve,Conti:2024rwd} 
thus admit its full constant range, including $s\neq0$. 
These works therefore provide explicit non-trivial gauged-null solutions 
with $s\neq0$ Killing spinors, 
although in this case different internal SM orientations 
do not define distinct bosonic configurations.

We first consider the configurations with vanishing Yang--Mills fields, 
$F_{(2){\rm D}}^i=0$, and consider only solutions that have 
a maximal-theory embedding.
Moreover, recall that in this work we focus on timelike solutions, so we note 
that the causal character of the solutions can be read off directly
from the Clifford-rotation angle in eq.~(4.2) of \cite{Dibitetto:2017tve}: 
For a generic allowed choice of the constant spinor,
the norm of its vector bilinear is proportional to $-\sin^2(2\theta_{\rm D})$.
Furthermore, the domain wall in eq.~(3.5), for example, has 
$F_{(2){\rm D}}^i=F_{(4){\rm D}}=0$, but $\dd X_{\rm D}\neq0$. 
Proposition~\ref{codilaton} therefore immediately rules out 
a corresponding timelike Killing vector, 
without even examining explicit spinor bilinears.

\begin{table}[t]
\centering
\small
\setlength{\tabcolsep}{3.5pt}
\renewcommand{\arraystretch}{1.1}
\begin{tabular}{|l||%
		c|c|c|c|c||c|c|c|}
\hline
Solution %
& $D$
& $\sigma$
& $\kappa$
& $q$
& $\rho$
& $\mathcal W_3$
& $\mathcal W_4$
& $\mathcal W_5$
\\
\hline
\hline
Dibitetto et al.~(1)~\cite{Dibitetto:2017tve}
& $\checkmark$
& $\checkmark$
& $\checkmark$
& $\checkmark$
& $\checkmark$
& $\checkmark$
& $\checkmark$
& n/a
\\
\hline
Conti et al.~$\supset$ Dibitetto et al.~(2) 
\cite{Conti:2024rwd,Dibitetto:2017tve}
& $\checkmark$
& $\checkmark$
& $\checkmark$
& $\checkmark$
& $\checkmark$
& $\checkmark$
& $\checkmark$
& n/a
\\
\hline
\end{tabular}
\caption{
SU$(3)$ components of the fluxes and dilaton gradient, as defined in 
Appendix~\ref{su3projections}, and of the base torsion in the additional 
supersymmetric solution families from minimal seven-dimensional supergravity. 
A check mark denotes a generic non-vanishing component or torsion form.
Both lie in the $F=0$ sector of the gauged null class 
$\mathcal C_{\rm g}^{\rm null}$, and thus $f=\lambda=\tau=\zeta=\alpha=\beta=0$.
The fifth torsion form is not applicable in the gauged null class,
where the geometric structure is U$(3)$.
Dibitetto et al.~(1) and (2) denote the solutions of Sections~4.2 and
4.4 of \cite{Dibitetto:2017tve}, respectively.
The maximally-embedded family (2) of \cite{Dibitetto:2017tve} is entirely 
contained in the general solution in eq.~(3.11) of \cite{Conti:2024rwd}.
The two rows have identical entries, although they describe 
distinct flows on backgrounds $\mathbb R^{1,2}\times S^3$ 
and AdS$_3\times S^3$, respectively.}
\label{tab:minimal-solutions}
\end{table}

The BPS flow in eq.~(4.15) of \cite{Dibitetto:2017tve}, 
denoted here by Dibitetto et al.~(1) in Table~\ref{tab:minimal-solutions} below, 
has $\theta_{\rm D}=\arctan(\varrho^{-2})$, and is timelike at finite radius
(here the radial metric coefficient in their eq.~(4.15)
includes an additional factor of $16$, as required by their
radial gauge condition (4.14)).
The couplings are independent in the general minimal solution, while
the choice made below eq.~(4.16), $h_{\rm D}=g_{\rm D}/(2\sqrt{2})$, 
selects its maximal-theory specialization. 
In this specialization $X_{\rm D}\rightarrow1$ as 
$r=\log\varrho\rightarrow+\infty$, whereas $X_{\rm D}\rightarrow0$ towards 
$r\rightarrow0$. 
Since $g\neq0$ and $F=0$, the embedded solution belongs to the class 
$\mathcal C_{\rm g}^{\rm null}$. 
Therefore, due to its running scalar, our classification readily gives, 
generically
\be
f=\lambda=\tau=\alpha=\zeta=\beta=0\, ,
\qquad
D,q,\sigma\neq0\, .
\ee
Then, a direct projection of its dyonic 4-form flux further shows 
that generically,
\be
\rho,\kappa\neq0\, .
\ee
The gauged-null classification then gives an integrable U$(3)$
structure with
\be
\mathcal W_1=\mathcal W_2=0\, ,
\qquad
\mathcal W_3\neq0\, ,
\qquad
\mathcal W_4\neq0\, ,
\ee
where the product of timelike warp factor and dilaton, $Xh$, is constant.

Next, the timelike flow in eq.~(4.29) of \cite{Dibitetto:2017tve}, 
denoted here by Dibitetto et al.~(2) in Table~\ref{tab:minimal-solutions} below, 
similarly embeds after imposing $h_{\rm D}=g_{\rm D}/(2\sqrt{2})$. 
This embedding is fully contained in the more general Janus solutions 
in eq.~(3.11) of \cite{Conti:2024rwd}. Indeed, the latter work identifies 
the former solution below eq.~(3.5) by the restrictions $C_1=1$, 
$C_2=5g_{\rm Co}/\sqrt{2}$, 
with $C_X$ fixed by the remaining parameter of \cite{Dibitetto:2017tve}. 
Thus, the work of \cite{Conti:2024rwd} genuinely enlarges the solution family, 
and the whole maximally-embedded Dibitetto et al.~(2) solution family 
is nested in it.
With $\theta_{\rm Co}=2\theta_{\rm D}$, its vector-bilinear norm is
proportional to $-\sin^2\theta_{\rm Co}$, so the Janus family is timelike
throughout $0<\theta_{\rm Co}<\pi/2$.
Moreover, the regularity condition in eq.~(3.13) selects the smooth 
one-parameter subfamily in eq.~(3.17), with $-1<\lambda_{\rm Co}<1$, 
and the AdS$_7$ vacuum recovered at $\lambda_{\rm Co}=0$. 

The Janus solution family of \cite{Conti:2024rwd} realizes the
complete SU$(3)$ content permitted in the $F=0$ gauged-null class;
indeed, the contained Dibitetto et al.~(2) family already activates
every component allowed by the corresponding constraints.
Since it is a BPS subfamily, eq.~\eqref{BPSinheritance} prevents any of these 
components from vanishing identically throughout the general Janus family, 
while the common $F=0$ gauged-null constraints preserve all components that
vanish identically. 
The complete qualitative SU$(3)$ signature of the Janus family therefore
follows without requiring a new projection of flux components in 
the general family.
Accordingly, Table~\ref{tab:minimal-solutions} assigns a single row 
to the Janus family containing Dibitetto et al.~(2), since
the two families have identical entries.

Finally, the vector-coupled configurations in Sections~5.2 and 5.3 of
\cite{Dibitetto:2017tve} employ the full SU$(2)$ ansatz in their eq.~(5.1).
Under the direct Cartan restriction in eq.~\eqref{Cartanrestrict},
with the remaining fields fixed, 
the sum of the orthonormal timelike and transverse-radial Einstein equations
loses two positive Yang--Mills contributions;
see eqs.~(2.9), (5.1) of \cite{Dibitetto:2017tve}.
Since the common gauge profile runs, neither configuration 
admits the Cartan restriction as a solution of the field equations.
Alternatively, setting the complete Yang--Mills sector to zero reduces the 
ans\"atze of their Sections~5.2 and 5.3 to the vector-free solution families of 
their Sections~4.2 and 4.4, respectively, which have already been discussed 
above. No additional relevant solutions therefore arise from these SU$(2)$ 
configurations.

\section{Conclusions}
\label{lafin}

In this work we considered maximal Cartan-gauged seven-dimensional 
supergravity and its minimal theory truncation via the equal-charge sector of 
the maximal theory, for which we developed a geometric classification of their 
timelike supersymmetric solutions.
We first established the complete formulation of the maximal Cartan-gauged 
theory in terms of the gauge-invariant $4$-form field strength, including 
both the bosonic action and supersymmetry variations, directly from 
the original maximal SO$(5)$-gauged theory \cite{Pernici:1984xx}. 
The original theory, the Cartan KSEs of \cite{Liu:1999ai}, and 
a matching bosonic action provided later in \cite{Cheung:2022ilc} 
use the explicit $3$-form potential, as detailed in Appendix~\ref{addendumtov1}.
The field-strength formulation established here requires 
taking into account a topologically massive $3$-form that is self-dual, 
a distinctive feature of seven-dimensional gauged supergravity.
The resulting formulation, including a well-defined self-duality condition in 
the gauged theory, has a smooth ungauged limit and provides a single consistent 
starting point for both the gauged and ungauged theories.

Moreover, this formulation was achieved by requiring a consistent truncation 
from the maximal U$(1)^2$-gauged theory to the minimal U$(1)$-gauged one, 
which we established at both the bosonic and fermionic levels, 
using the equal-charge limit. 
This consistent truncation uniquely fixed 
the action and KSEs, especially the coefficient $-m$ of the topological mass 
term, and the corresponding higher-form mapping from potential to field 
strength, while also yielding the maximal-supersymmetry coupling relation $m=g$.
The limit of equal bosonic fields within the maximal theory 
must be accompanied by restricting the maximal spinor USp$(4)$ components 
to form the minimal SM doublet, while the antisymmetric spin-$\tfrac12$ 
combination and the orthogonal USp$(4)$ subspace are removed. 
The complementary USp$(4)$ components nevertheless give an additional 
maximal-theory spinor branch, %
and can therefore enhance the supersymmetry of a maximally-embedded 
minimal bosonic configuration when $F=0$. 
This complete truncation makes the same geometric classification directly 
applicable to the minimal U$(1)$-gauged theory, while retaining the additional 
spinorial possibilities of the maximal theory.

Starting from a single SM Killing spinor, we constructed its complete 
set of non-vanishing bilinears, and expressed them in terms of a timelike 
fibration and the SU$(3)$ structure $(J,\Psi)$ on the six-dimensional base. 
We kept the internal SM spinor orientation %
general, rather than fixing it at the outset, which was important for deriving 
the complete algebraic and differential systems, and determining the 
restrictions they impose. 
Resulting general consequences include that
a running dilaton requires a non-vanishing flux, 
the second stationary point of the scalar potential is not supersymmetric, 
and the complex scalar bilinear identically vanishes, 
independently of gauge choice, 
whenever $g,F\neq0$ (and the maximal Killing spinor lies in 
$\Ker\Gamma_-$). 

The SU$(3)$-isotypic decomposition decouples the algebraic KSE
constraints into independent systems for the irreducible components
of $\dd X$, $F$ and $G$.
The differential bilinear relations further show that in addition to preserving 
the full bosonic solution, the Killing vector preserves the internal SM 
orientation angle, and the almost-Hermitian structure of the transverse base. 
Compatible Killing gauges can also be chosen for both gauge potentials, such 
that they are also Killing-invariant. 
When there is gauging, the Killing gauge of the 3-form potential is constructed 
in conjunction with its self-duality condition and the related gauge freedom. 
In the gauged null Cartan-flux configurations, the structure form $\Psi$ is not 
Killing-invariant, but instead undergoes a U$(1)$ R-symmetry phase rotation 
along the Killing direction. 
The base structure is consequently enlarged from SU$(3)$ to U$(3)$, with the 
corresponding phase rotation geometrized by the $K\wedge\dd K$ twist term in the 
$3$-form potential.

Completing this geometric bilinear analysis, we determined all five 
intrinsic-torsion forms, wherever the base retains an SU$(3)$ structure.
Non-integrability with either the first or second torsion form non-vanishing 
can occur only in the ungauged sector with an active Cartan flux. 
The third torsion form is controlled by the single flux component in the 
highest-degree sector of the SU$(3)$ decomposition, 
and the fourth torsion form is necessarily exact when there is gauging. 
At the common SM spinor-angle parameter value $s=0$, 
the essential distinction among the resulting geometries is carried solely by 
the fifth torsion form or its absence. 
In the class of gauged active Cartan flux this torsion form contains the 
transverse gauge connection, while in the gauged null Cartan flux class, 
there is no intrinsic fifth torsion form, as only a U$(3)$ structure is 
preserved. 
This complete torsion analysis, together with the resulting base-structure 
refinements, provides a comprehensive geometric picture, also for the minimal 
theory.

The resulting bosonic solution space separates exhaustively according to two 
independent distinctions: 
the presence or absence of gauging, and whether $F\Gamma_+$ acts non-trivially 
on the Killing spinor. 
This yields a four-class partition, 
see Tables~\ref{tab:preliminary-classes-ug}, \ref{tab:preliminary-classes-g}, 
where each class specifies the permitted irreducible flux components, scalar 
and internal spinor-angle restrictions, twist and intrinsic torsion, 
Killing-gauge potentials, and remaining bosonic equations. 
We established that, given the KSEs and Bianchi identities, 
the independent bosonic EOMs reduce to the transverse $(2,2)$-form component 
of the $3$-form equation in every class, supplemented by the scalar equation 
only in the active classes when $s=0$.
We then integrated the algebraic and differential KSE constraints
with the full bosonic Bianchi identities and these independent EOMs,
Killing gauges of potentials and, when there is gauging, also the
self-duality condition.
The same four-class framework applies directly to every solution of the
minimal U$(1)$-gauged theory for which the vector bilinear of the chosen
Killing spinor is timelike, when the maximal-supersymmetry coupling relation
is imposed.

Applying this classification framework as a diagnostic to the known
and proposed solutions in the literature reveals three complementary
solution frontiers among families constructed directly in the maximal
and minimal theories, summarized in
Tables~\ref{tab:known-gactive-solutions}, \ref{tab:minimal-solutions}.
On the black-hole side, in the maximal theory, the general rotating AdS$_7$
solutions of \cite{Bobev:2023bxl,Bobev:2025xan} restore the second Cartan
charge and hence extend beyond the equal-charge sector resolved here.
The equal-charge Lorentzian specialization of \cite{Bobev:2023bxl}
reduces to the family with three independent angular momenta of
\cite{Chow:2007ts}, while the $L^{p,q,r}$ family of
\cite{Bobev:2025xan} retains its distinct horizon geometry when the
charges are equal.
Within the present analysis, these two distinct families constitute the
black-hole frontier of the progression from a charged singularity to charged
and rotating geometries. The generic BPS black hole of
\cite{Chow:2007ts} activates every SU$(3)$ component permitted in
the gauged active class.

A second frontier in the maximal theory is provided by the two warped
AdS$_3$ families of \cite{Cheung:2022ilc}, also in the gauged active class.
We analyse their equal-charge specializations and present smooth
noncompact global completions for both families.
One of these families shares the fully activated SU$(3)$ component
and torsion pattern of the generic black hole of \cite{Chow:2007ts}.
The third frontier lies in the complementary gauged null Cartan-flux class
and originates from the minimal theory.
The maximally-embedded BPS flows and Janus solutions of
\cite{Dibitetto:2017tve,Conti:2024rwd} realize the complete
permitted SU$(3)$ flux content with $F=0$.
They provide non-trivial solutions permitting $s\neq0$ Killing-spinor
orientations, and constitute two geometrically distinct flows
with identical qualitative flux and torsion data, also demonstrating
the further geometric and functional freedom within such resolution.
The early search in \cite{Cariglia:2004qi} provided no additional
solutions satisfying the necessary requirements, owing to global
pathologies, or the absence of a supersymmetric Cartan restriction.

More broadly, the known configurations occupy only a small portion of the 
full solution space that is exposed by the classification, 
and the resulting diagnostics can guide the formulation of new ans\"atze in 
searches for further supersymmetric solutions.
Several concrete directions of future investigation now follow. 
Firstly, the ungauged classes remain essentially unexplored, 
particularly when the dilaton is taken to run 
and the complete SU$(3)$ content is permitted. 
Within the equal-charge maximal theory, the null Cartan-flux configuration with 
$\chi^A\in\Ker\Gamma_+$, closed primitive flux $F=\tau\neq0$ and necessarily 
$G\neq0$, is admitted by the complete system, 
but is not realized by any known solution. 
Fully extending the analysis beyond equal charges would allow the flux-branes 
of \cite{Lu:2003iv} and the warped AdS$_3$ families of \cite{Cheung:2022ilc} 
to be fully classified, alongside the full families of recently discovered 
black holes constructed in \cite{Bobev:2023bxl,Bobev:2025xan}.
An efficient extension is already provided in Section~4 of \cite{Gaar:2026gmo}, 
where the equal-charge core is retained and the additional unequal-charge 
supersymmetric input is algebraically confined.
The underlying KSE reorganization substantially reduces the additional analysis 
needed to classify these general families.

Treating the null Killing-vector case would further develop existing work,
ranging from the early solutions of \cite{Gauntlett:2000ng}
and the classification of \cite{MacConamhna:2004fb} to the more recent
works of \cite{Dibitetto:2017tve,Conti:2024rwd,Conti:2025wyj,Conti:2026row},
which contain additional null solutions.
Furthermore, the present complete field-strength formulation of the action and 
KSEs, together with the bilinears, gauge potentials and intrinsic geometry,
provides the critical input for geometric localization in odd-dimensional
supergravity.
Its separation of the gauge-invariant part of the action from the
Chern--Simons terms already underlies the relative equivariant
localization of seven-dimensional supergravity in \cite{Gaar:2026gmo}.
The framework established here therefore provides 
a systematic geometric toolbox for constructing new solutions 
and developing more general geometric and holographic applications.

\paragraph{Note added.}
After the first version of this work appeared on arXiv,
\cite{Cheung:2022ilc} was brought to our attention.
In this second version we clarified in Appendix~\ref{addendumtov1} its relation 
to our formulation in Section~\ref{7dsugra}, 
and extended our analysis of known solutions in Section~\ref{maxsol} and 
Table~\ref{tab:known-gactive-solutions} to include its two relevant solution
families, as well as an earlier one of \cite{Lu:2003iv}.

\begin{acknowledgments}
I thank Tom\'as Ort\'in for informative correspondence 
at an early stage of this work.
I am profoundly grateful to James Sparks for standing by me throughout my work 
on this project. 
His exceptionally generous support and genuine interest were instrumental 
to its completion.
ML has been supported by the Science and Technology Facilities Council 
(STFC) Rutherford Grant ST/V003895 \textit{``Harnessing QFT for 
Gravity''}, and by the Mathematical Institute, University of Oxford.
\end{acknowledgments}

\appendix

\section{Appendices}

\counterwithin{equation}{subsection}

\subsection{Clifford Algebra in Dimension 7}
\label{7Cliff}

In seven-dimensional spacetime the Dirac spinors carry 
$2^{\lfloor7/2\rfloor} = 8$ 
complex components. 
The $8\times8$ gamma matrices, $\{\gamma_{\alpha}\}_{\alpha=0,\ldots,6}$, 
which span the Clifford algebra with a Lorentzian signature, 
Cliff$(1,6)$, satisfy
\be
\label{Cliff1+6}
\{ \gamma^\alpha, \gamma^\beta \} = 2 \eta^{\alpha\beta} \mathds{1}_8 \, ,
\ee
with the tangent metric $\eta_{\alpha\beta} = \diag (-1, + \mathds{1}_6)$, 
such that $(\gamma^0)^{-1} = - \gamma^{0} = \gamma _{0}$,
and their conjugation properties are defined via
\begin{subequations}
\begin{align}
(\gamma^{\alpha})^\dag & = \gamma^0 \gamma^\alpha \gamma^0\, ,  
\\*
\label{complexconj} 
(\gamma^\alpha)^* & = B \gamma^\alpha B^{-1}\, , 
\\*
(\gamma^\alpha)^T & = - C \gamma^\alpha C^{-1}\, , 
\end{align}
\end{subequations}
with $\gamma^0$, $B$, and $C$, the matrices of hermitian, complex, 
and charge conjugation, respectively, where the latter two 
are always unitary, and satisfy
\begin{subequations}
\begin{align}	
C & = C^T \, , \\
B & = - i C \gamma^0 \, , 
\end{align}
\end{subequations}
such that $(\gamma^{0})^\dag = - \gamma^{0}$,
$(\gamma^{i})^\dag = \gamma^{i}$,
$C^{-1} = C^\dag = C^*$, and
$B^{-1} = B^\dag = - B^*$, 
that is $B = - B^T$.

The charge and Majorana conjugates, and the Majorana conjugate of 
the charge conjugate of any spinor are defined, respectively, via
\begin{subequations}
\begin{align}
\label{chargeconj}	
\epsilon^c & \equiv B^{-1} \epsilon^* 
= -i \gamma^0 C^{-1} \epsilon^* 
\, ,\\*
\overline{\epsilon} & \equiv \epsilon^T C\, ,\\*
\overline{\epsilon^c} & \equiv (\epsilon^c)^T C 
= - i \epsilon^\dag \gamma_0 \, .
\end{align}
\end{subequations}

We can then choose a complex antisymmetric representation of 
the gamma matrices, where the charge conjugation matrix is taken as 
\be
\label{Crepresent}
C = \mathds{1}_8\, , 
\ee 
such that 
\be 
\label{Majconj}
(\gamma^\alpha)^T = - \gamma^\alpha\, ,
\qquad \overline{\epsilon} = \epsilon^T\, .
\ee

Furthermore, in odd dimensions there are two inequivalent irreducible 
representations of the algebra to choose between 
as an independent convention, which we fix here as 
\be
\label{7dchirality}
\gamma_{0123456} \equiv 
\gamma_{0}\gamma_{1}\gamma_{2}\gamma_{3}\gamma_{4}\gamma_{5}\gamma_{6} 
= - \gamma^{0123456} = - \mathds{1}_8\, .
\ee

\subsection{Euclidean Clifford Algebra in Dimension 5}
\label{5ECliff}

Due to the USp$(4)$ R-symmetry, the Euclidean matrices in dimension $5$, 
which span the Clifford algebra, Cliff$(5)$, are also featured in our theory. 

We choose the gamma matrices, $\Gamma_a$, 
satisfying $\{ \Gamma_a, \Gamma_b \} = 2 \delta_{ab} \mathds{1}_4$, as 
follows:
\be
\Gamma_a \equiv \left\{\begin{array}{ll}
	\sigma_a \otimes \mathds{1}_2 \, , \quad & a=1,2 \\
	\sigma_3 \otimes \sigma_{a-2} \, , \quad & a=3,4,5 
\end{array} \right. \, , 
\ee
with the Pauli matrices, $\sigma_a$, defined here as  
\be
\label{pauli}
\sigma_0 \equiv -\mathds{1}_2\, , \quad 
\sigma_1 \equiv \left( \begin{array}{rr}
	0 & 1  \\
	1 & 0  \end{array} \right), \quad
\sigma_2 \equiv \left( \begin{array}{rr}
	0 & -i \\
	i & 0  \end{array} \right), \quad
\sigma_3 \equiv \left( \begin{array}{rr}
	1 & 0  \\
	0 & -1  \end{array} \right). 
\ee

Specifically in the U$(1)^2$-gauged sector we only need 
the following matrices of the Cartan subalgebra:
\begin{subequations}
\label{cliff5}	
\begin{align}
i \, \Gamma_{\hat{1}} \equiv \Gamma_{12} & 
= i \, \diag (+1,+1,-1,-1)\, ,  \\*
i \, \Gamma_{\hat{2}} \equiv \Gamma_{34} & 
= i \, \diag (+1,-1,+1,-1)\, , \\*
\Gamma_5 = - \Gamma_{1234} = \Gamma_{\hat{1}} & \Gamma_{\hat{2}}  
= \diag (+1,-1,-1,+1)\, ,
\end{align}
\end{subequations}
where we have 
\be
\Gamma_{12345} = -\mathds{1}_4\, .
\ee

Let us record the product table of this Abelian subgroup -- the 
Klein $4$-group $Z_2 \times Z_2$:
\begin{align}
\Gamma_{\hat{1}}^2 = \Gamma_{\hat{2}}^2 = 
\Gamma_{5}^2 = \mathds{1}_4\, ,\quad
\Gamma_{\hat{1}} \Gamma_{\hat{2}} = \Gamma_5\, ,\quad
\Gamma_{\hat{1}} \Gamma_{5} = \Gamma_{\hat{2}}\, ,\quad 
\Gamma_{\hat{2}} \Gamma_{5} = \Gamma_{\hat{1}}\, .
\end{align}
We then define the two U$(1)$ generators to form the Cartan basis as
\begin{align}
\label{cartangen}	
\Gamma_{+} \equiv \frac{1}{2} 
\left( \Gamma_{\hat{1}} + \Gamma_{\hat{2}} \right)
= \diag (1,0,0,-1)\, , \qquad
\Gamma_{-} \equiv \frac{1}{2} 
\left( \Gamma_{\hat{1}} -\Gamma_{\hat{2}} \right)
= \diag (0,1,-1,0)\, ,
\end{align}
as they satisfy 
\begin{align}
\Gamma_{+} \,\Gamma_{-} = 0\, .
\end{align}

Thus the two $\frac{1}{2}$-supersymmetry projections emerge as 
\begin{align}
\label{projections}	
P_{\pm} \equiv  \frac{1}{2} \left( \mathds{1}_4 \pm \Gamma_{5} \right)
= \Gamma_{\pm}^2\, , 
\end{align}
or explicitly
\be
P_+ = \diag (1,0,0,1) \, , 
\qquad 
P_- = \diag (0,1,1,0)\, .
\ee

\subsection{Conventions and Formulae for Bilinears}
\label{bilinearformulae}

In order to determine which bilinears in Section~\ref{bilinear} 
are non-vanishing, and analyse them, 
the Majorana flip relation is required, which reads
\be 
\label{Majoranaflip}
\bar{\lambda} \, \gamma_{\mu_1\cdots\mu_r} %
\chi
= (-1)^{r(r+1)/2} %
\, \bar{\chi} \, \gamma_{\mu_1\cdots\mu_r} %
\lambda\, ,
\ee
where we have an extra overall minus sign with respect to 
\cite{Freedman:2012zz}, due to the commuting spinors used 
throughout the present work.

In order to determine the reality properties of the bilinears, 
we use the formula for complex conjugation of spinor bilinears 
with an arbitrary matrix $M$, which reads:
\be 
\label{Lorentzcc}
( \bar{\lambda} M \chi )^* \equiv ( \bar{\lambda} M \chi )^c
= - \overline{\lambda^c} M^c \chi^c\, , 
\qquad
M^c \equiv B^{-1} M^* B\, .
\ee

To get the generic geometric bilinear relations, we apply the 
Fierz identity in $d=7$, which reads: 
\be 
\label{Fierz}
\bar{\epsilon}_1 \epsilon_2 \bar{\epsilon}_3 \epsilon_4 
= \frac{1}{8} \bigg[ \bar{\epsilon}_1 \epsilon_4 
\bar{\epsilon}_3 \epsilon_2
+ \bar{\epsilon}_1 \gamma_\mu \epsilon_4 
\bar{\epsilon}_3 \gamma^\mu \epsilon_2 
- \frac{1}{2} \bar{\epsilon}_1 \gamma_{\mu\nu} \epsilon_4 
\bar{\epsilon}_3 \gamma^{\mu\nu}\epsilon_2 
- \frac{1}{6} \bar{\epsilon}_1 \gamma_{\mu\nu\rho} \epsilon_4 
\bar{\epsilon}_3 \gamma^{\mu\nu\rho}\epsilon_2 \bigg]\, .
\ee 
In the above applications of the Fierz identity 
we also needed the following identity:
\be
\label{gammasandwich}
\gamma_{\alpha} \gamma^{\mu_1\cdots\mu_r}\gamma^{\alpha} 
= (-1)^r (d-2r) \gamma^{\mu_1\cdots\mu_r}\, ,
\ee
and its extension which we derived as
\be
\label{gammadblsandwich}
\gamma_{\alpha\beta} \gamma^{\mu_1\cdots\mu_r}\gamma^{\alpha\beta} 
= \left[ d -(d-2r)^2\right] \gamma^{\mu_1\cdots\mu_r}\, .
\ee

Furthermore, in the above applications of the Fierz identity we also used
the duality relation between Clifford matrices of ranks $r$ and $d-r$
in odd dimensions, $d=2m+1$. The overall sign in this relation depends on
the choice of irreducible Clifford representation, which was fixed 
in eq.~\eqref{7dchirality}, while the totally antisymmetric tensor is 
fixed independently by the orientation convention, so that we write
\be
\gamma_{\mu_1\cdots\mu_r}
=
\sigma_\gamma\, i^{m+1}
\frac{1}{(d-r)!}
\epsilon_{\mu_1\cdots\mu_r\mu_{r+1}\cdots\mu_d}
\gamma^{\mu_d\cdots\mu_{r+1}}\, ,
\qquad
\sigma_\gamma=\pm1\, .
\ee
With the representation convention fixed in Appendix~\ref{7Cliff}, and
the orientation convention
\be
\label{orient}
\epsilon_{0123456}
=
-\epsilon^{0123456}
=
1\, ,
\ee
one has $\sigma_\gamma=-1$. Therefore, in $d=7$, this gives
\be
\label{dualclifford7}
\gamma_{\mu_1\cdots\mu_r}
=
\frac{(-1)^{\lfloor r/2\rfloor}}{(7-r)!}
\epsilon_{\mu_1\cdots\mu_r\nu_1\cdots\nu_{7-r}}
\gamma^{\nu_1\cdots\nu_{7-r}}\, .
\ee

The above dual of Clifford matrices in eq.~\eqref{dualclifford7} 
is closely related to the Hodge dual of forms, once the Clifford 
matrices are put into bilinear forms. 
The Hodge-star operator for any complex p-forms $\omega$, $\eta$ 
is defined in this work as 
\be
\label{starmathdef}
\omega \wedge \overline{ \star \eta }
= \langle \omega, \eta \rangle \text{Vol}_d\, , 
\ee
or equivalently, in components 
\be
\label{stardef}
(\star \omega)_{\mu_{1}\cdots\mu_{d-p}} 
= \frac{1}{p!} \omega^{\nu_1\cdots\nu_p} 
\epsilon_{\nu_1\cdots\nu_p\mu_{1}\cdots\mu_{d-p}}\, .
\ee

The convention that we use for the interior product 
of a vector $V$ and $\omega$ reads:
\be
\label{intprod}
(\iota_V \omega)_{\mu_1\cdots\mu_{p-1}} 
\equiv V^\nu \omega_{\nu\mu_1\cdots\mu_{p-1}}\, ,
\ee
and for the more general contraction with a 
tensor $T$ of rank $q\le p$
\be 
\label{formcontract}
(T \lrcorner \omega)_{\mu_1\cdots\mu_{p-q}} 
\equiv \frac{1}{q!} T^{\nu_1\cdots\nu_q} 
\omega_{\nu_1\cdots\nu_{q}\mu_1\cdots\mu_{p-q}}\, .
\ee

For a real transverse $p$-form $\zeta$ on the six-dimensional base, 
we use the norm
\be
\label{6dnormconvention}
\lvert\zeta\rvert^2
\equiv
\frac1{p!}\zeta_{i_1\ldots i_p} \zeta^{i_1\ldots i_p}\, ,
\ee
with indices raised using the base metric $g_{ij}$, 
as defined in eq.~\eqref{genmetric}.

Finally, the Pauli matrices defined in eq.~\eqref{pauli} 
satisfy the product rule:
\be
\label{Pauliprod}
\sigma_a \sigma_b  = \delta_{ab} \mathds{1}_2 + i \epsilon_{abc} 
\sigma_c\, . 
\ee

\subsection{Euclidean Clifford Algebra in Dimension 6}
\label{6ECliff}

In dimension $d=6$ we take the same conventions of charge conjugation 
as those in dimension $d=7$ in Appendix~\ref{7Cliff}, 
so the focus here is on switching to a Euclidean signature. 
The matrices in the Euclidean Clifford algebra, Cliff$(6)$, satisfy
\be
\{ \gamma^a, \gamma^b \} = 2 \delta^{ab} \mathds{1}_8 \, ,
\ee
with the Euclidean gamma matrices related to the Lorentzian ones via
\begin{subequations}
\begin{align}
\label{Euclidgamma0}	
\{\gamma^{a}\}_{a=1,\ldots,6} & = \{\gamma^{\mu}\}_{\mu=1,\ldots,6}\, ,
\end{align}
\end{subequations}
where they are all Hermitian. 
As noted, in the Euclidean signature charge conjugation remains 
the same:
\begin{subequations}
\begin{align}
C_E & = C =C^T\, , \\
(\gamma^a)^T & = - C \gamma^a C^{-1}\, , 
\end{align}	
\end{subequations}
whereas hermitian and complex conjugation are modified as follows:
\begin{subequations}
\begin{align}
(\gamma^a)^\dag & = \gamma^a \, , \\
\label{complexconjEuclid}
\implies (\gamma^a)^* & = -B_E \gamma^a B_E^{-1}\, ,\\
\implies B_E & = C \, , 
\end{align}
\end{subequations}
that is $B_E = B_E^T$ in the Euclidean signature, and since $B_E$ is 
always unitary, it now holds that $B_E^*B_E = \mathds{1}_8$, which 
allows for Majorana spinors, in contrast to the Lorentzian case with 
the same conventions of charge conjugation. 
The charge conjugate, and its Majorana conjugate, are now defined via
\begin{subequations}
\label{Euclidcharge}	
\begin{align}	
\label{chargeEuclid}	
\epsilon^c & \equiv B_E^{-1} \epsilon^* = C^{-1} \epsilon^* \, , \\
\label{barofchargeEuclid}
\overline{\epsilon^c} & \equiv (\epsilon^c)^T C = \epsilon^\dag \, .
\end{align}
\end{subequations}

\subsection{SU\texorpdfstring{$(3)$}{(3)} Structure}
\label{gstructure6d}

First, note that on the six-dimensional Riemannian base manifold we 
switch to spinors in the corresponding Euclidean Clifford algebra, 
Cliff$(6)$, as detailed in Appendix~\ref{6ECliff} above.

Then, there exists a global spinor of positive chirality, $\psi$, 
along with its anti-chiral counterpart, $\psi^*$, 
such that their chiral projections by $\gamma_*$, 
which is defined in eq.~\eqref{chiral6}, read
\begin{align}
\label{chiralfix}	
\gamma_*\psi & = i\psi\, , 
\qquad 
\gamma_*\psi^* = -i\psi^*\, ,
\end{align}	
with a unit norm, namely
\begin{align}
\overline{\psi^c}\psi & 
= \psi^\dag \psi = 1\, ,\\
\implies \quad & 
\overline{\psi} \psi^c = \psi^T\psi^* = 1,
\end{align}
where as in Appendix~\ref{7Cliff}, $C=\mathds{1}_8$, so that 
$\overline{\psi}=\psi^T$, and $\psi^c=\psi^*$. 

This spinor defines an SU$(3)$ structure, 
see also \cite{Gauntlett:2003cy,Lust:2004ig}, 
via a pair of fundamental forms, $(J, \Psi)$, a real 2-form, $J$, 
and a complex 3-form, $\Psi$, which are defined as follows:
\begin{subequations}
\label{structuredef}	
\begin{align}
\label{defJ}	
J_{ij} & \equiv i \overline{\psi^*} \gamma_{ij} \psi 
= - i \psi^T \gamma_{ij} \psi^*\, ,\\
\label{defPsi}
\Psi_{ijk} & \equiv \overline{\psi^*} \gamma_{ijk} \psi^*\, , 
\qquad 
\overline{\Psi}_{ijk} \equiv - \overline{\psi} \gamma_{ijk} \psi\, ,
\end{align}	
\end{subequations}
where $J$ defines an almost-complex structure, namely
\begin{align} 
\label{defjsq}	
J^{i}{}_{j} \, J^{j}{}_{k} & = - \delta^{i}{}_{k}\, ,
\end{align}
see also Appendix~\ref{almostcalculus} below for algebraic 
and differential almost-complex calculus.

This pair of forms satisfies \textit{algebraic} compatibility and 
normalization conditions. 
Since $\Psi$ is of type $(3,0)$ with respect to $J$, it is primitive.
Consequently, the compatibility condition between the defining forms
can be written as
\be 
\label{compatible}
J \lrcorner \Psi=0 \quad \iff \quad J \wedge \Psi = 0\, , 
\ee
where the equivalence is easily shown using 
the general identity in eq.~\eqref{Hodgeidentity}, and 
eq.~\eqref{su3hodge}. 

\textit{At each point} the basis can always 
be chosen, so that $J$ and $\Psi$ take the standard form
\begin{align}
\label{stdbasis}	
J & = e^{12} + e^{34} + e^{56}\, ,\\
\label{holo3formstd}
\Psi & = (e^1 + ie^2) \wedge (e^3 + ie^4) \wedge (e^5 + ie^6)\, ,
\end{align}	 
from which it is easy to verify the \textit{algebraic} normalization 
conditions that read
\be
\label{normal}
\star_6 1 \equiv \text{Vol}_6 = \frac{1}{6} \,J \wedge J \wedge J 
= \frac{i}{8} \,\Psi \wedge  \overline \Psi\, . 
\ee
This is consistent with the normalization contractions
\be
\label{normalize}
J^{ij}J_{ij}=6\, , \qquad \Psi^{ijk}\overline{\Psi}_{ijk}=48\, ,
\ee 
which are also consistent with the geometric relations of the 
seven-dimensional bilinears that we found in 
eq.~\eqref{Fierznorms}.

It is easy to see from the definition in eq.~\eqref{starmathdef} 
and from eq.~\eqref{normal} that the Hodge duals 
of $J$ and $\Psi$ are given by
\begin{align}
\label{su3hodge}
\star_6 J = \frac12 J \wedge J\, ,
\qquad
\star_6 \Psi = -i\,\Psi
\, . 
\end{align}
The complex 3-form can be written in terms of two real 3-forms, 
defined as
\be 
\label{gettingreal}
\Psi = \Psi_+ + i \Psi_-\, ,
\ee
so that the Hodge duals of its real and imaginary parts satisfy
\be 
\label{PsiPMhodge}
\star_6 \Psi_+ = + \Psi_-\, , 
\qquad \star_6 \Psi_- = - \Psi_+\, .
\ee

Finally, let us highlight that the almost-complex structure $J$ 
of a generic SU$(3)$ structure is not necessarily integrable,
see also Appendix~\ref{almostcalculus}, 
and could in fact be non-integrable in the theory studied 
in the present work. 
We note that the term ``holomorphic'' for the complex 3-form is considerably 
ambiguous in the literature, as it is commonly assumed synonymous with the 
integrability of $J$, or takes on various meanings across different works.
Whenever the term ``holomorphic'' is used in this work for a form or one of its 
components, it refers only to its algebraic $(p,0)$ type with respect to $J$.

\subsection{Intrinsic Torsion of SU\texorpdfstring{$(3)$}{(3)} Structure}
\label{torsiondef}

A generic SU$(3)$ structure is classified according to its intrinsic 
torsion, which decomposes into five torsion forms, 
$\{{\mathcal{W}}_i\}_{i=1,\ldots,5}$, 
with the first two being complex. 
The torsion forms measure how much each of the two fundamental 
structure forms deviates from being closed 
\cite{Lust:2004ig,chiossi2002intrinsic,Cabrera:2005}. 
Thus, the torsion forms are defined via the exterior derivatives 
of the fundamental structure forms, 
decomposed into SU$(3)$ irreducible components 
(see also Appendix~\ref{decomposeformstosu3irreps} below),
as follows:
\begin{subequations}
\label{torsionforms}		
\begin{align}
\label{dJ}
\dd J & = \frac{3}{2} \,\text{Im}\left( \overline{\mathcal{W}_1} \Psi 
\right) 
+ \mathcal{W}_3 + \mathcal{W}_4 \wedge J\, ,\\
\label{dpsi}
\dd \Psi & = \mathcal{W}_1 J \wedge J + \mathcal{W}_2 \wedge J 
+ \overline{\mathcal{W}_5^{\mathbb{C}}} \wedge \Psi\, , 
\end{align}	
\end{subequations}
where $\mathcal{W}_1$, $\mathcal{W}_2$, can be further decomposed 
into independent real and imaginary parts, $\mathcal{W}_1^\pm$, 
$\mathcal{W}_2^\pm$, and such that 
\begin{subequations}
\label{torsionsu3irreps}
\begin{align}
\label{dpsi22}	
(\dd \Psi)^{(2,2)} & = \mathcal{W}_1 J \wedge J 
+ \mathcal{W}_2 \wedge J\, , 
\qquad 
0 = \mathcal{W}_2 \wedge J \wedge J\, ,\\
\mathcal{W}_3 & = (\dd J)^{((2,1)+(1,2))_0}\, , 
\qquad \qquad \,\,
0 = \mathcal{W}_3 \wedge J \, ,\\
\mathcal{W}_4 & = \frac12 J\lrcorner \dd J\, , \\  
\label{dpsi31}
(\dd \Psi)^{(3,1)} & = 
\overline{\mathcal{W}_5^{\mathbb{C}}} \wedge \Psi\, , 
\end{align}
\end{subequations}
where eq.~\eqref{dpsi22} arises within differential 
almost-complex calculus, see Appendix~\ref{almostcalculus} below.

The complexified torsion form
$\overline{\mathcal{W}_5^{\mathbb{C}}}$ is a complex mapping 
of the real form $\mathcal{W}_5$ via its projection onto 
its $(0,1)$ component, which reads
\be
\label{w5complexifiedconj}
\overline{\mathcal W_5^{\mathbb C}}{}_i
\equiv
(\Pi^{(0,1)})_i{}^j(\mathcal W_5)_j
=
\frac12
\left(\delta_i{}^j-iJ^j{}_i\right)
(\mathcal W_5)_j\, ,
\ee
see also Appendix~\ref{almostcalculus} below for a recap of 
almost complex calculus, and projectors.
We use the convention
\be
\mathcal W_a=\mathcal W_a^++i\mathcal W_a^-\, ,
\qquad a=1,2\, ,
\ee
where $\mathcal W_a^\pm$ are real.
Written in a real basis, the exterior derivatives of the real 
and imaginary parts of the fundamental 3-form of the structure then read	
\begin{align}
\label{dpsireal}	
d \Psi_\pm & = + \mathcal{W}_1^{\pm} J \wedge J 
+ \mathcal{W}_2^{\pm} \wedge J 
+ \mathcal{W}_5 \wedge \Psi_{\pm}\, ,
\end{align}	
where the algebraic identity in eq.~\eqref{almostidentity} is
used to recast the last term from $J \mathcal{W}_5 \wedge \Psi_{\pm}$. 

The $\mathcal W_5$ torsion form distinguishes the SU$(3)$ structure from
its underlying U$(3)$ structure. For an SU$(3)$ structure, the
$(3,1)$ component of $\dd \Psi$ must be expressible 
as in eqs.~\eqref{dpsi31}, \eqref{w5complexifiedconj},
for a globally-defined real one-form $\mathcal W_5$, and a globally-fixed 
normalized form $\Psi$. If such a global identification is not
possible, while the almost-Hermitian pair $(g,J)$ remains globally
defined, the structure is enlarged to U$(3)$, 
see Appendix~\ref{gstructure6dwg} below. 
The latter structure allows for phase transformations among local choices 
of the complex 3-form $\Psi$.

Finally, the compatibility and normalization conditions in
eqs.~\eqref{compatible}, \eqref{normal}, or equivalently the defining relations 
in eqs.~\eqref{dJ}, \eqref{dpsi}, show that the first torsion form can be
defined from either $\dd J$ or $\dd\Psi$ as
\be
\label{equivw1}
\mathcal{W}_1 
\equiv -\frac{4i}{3|\Psi|^2} \Psi\lrcorner \dd J 
\equiv \frac{1}{4|J|^2} J\lrcorner J\lrcorner \dd \Psi 
\, ,
\ee
where the norms of $J$ and $\Psi$, using the norm convention in 
eq.~\eqref{6dnormconvention}, and the contractions in eqs.~\eqref{normalize}, 
are $|J|^2=3$, $|\Psi|^2=8$.

\subsection{U\texorpdfstring{$(3)$}{(3)} Structure}
\label{gstructure6dwg}  

We briefly present here a generic U$(3)$ structure, 
and highlight how it differs from a generic SU$(3)$ structure.

A U$(3)$ structure on a six-dimensional Riemannian manifold consists
of an almost-complex structure $J$ compatible with the metric,
or equivalently an almost-Hermitian structure $(g,J)$. No
integrability condition on $J$ is assumed.

An SU$(3)$ structure requires, in addition, a globally-defined
nowhere-vanishing normalized complex 3-form $\Psi$ of type $(3,0)$.
A U$(3)$ structure does not include a globally fixed choice of this
complex volume form. Local representatives may instead be related on
overlaps by phase transformations,
\be
\Psi_\alpha=e^{i\lambda_{\alpha\beta}}\Psi_\beta\, .
\ee

The intrinsic torsion of a U$(3)$ structure decomposes into four forms
$\{{\mathcal{W}}_i\}_{i=1,\ldots,4}$ \cite{Gray:1980}. 
The exterior derivative $\dd J$
contains the torsion forms $\mathcal W_1,\mathcal W_3,\mathcal W_4$, 
as in eq.~\eqref{dJ}, 
whereas $\mathcal W_1$ and $\mathcal W_2$ together encode the
Nijenhuis tensor, $N_J$. Schematically,
\be
\label{nijenhuisscheme}
[N_J]\simeq\mathcal W_1\oplus\mathcal W_2\, ,
\ee
see also Appendix~\ref{almostcalculus} below. 
In particular, $\mathcal W_2$ is not detected by $\dd J$; rather, it is
the component of the Nijenhuis tensor that remains after
$\mathcal W_1$ in $\dd J$ is accounted for.

The additional SU$(3)$ torsion form $\mathcal W_5$ measures the
variation of the chosen complex volume form $\Psi$. It is therefore
not an intrinsic-torsion component of the underlying U$(3)$
structure.

\subsection{Almost Complex Calculus}
\label{almostcalculus}

We begin by fixing the action of the almost-complex structure directly
from the six-dimensional bilinear definitions in Appendix~\ref{gstructure6d}.

Recall that the almost-complex endomorphism, satisfying eq.~\eqref{defjsq}, 
is obtained by raising the first index of the real 2-form $J$, 
\be
J^i{}_j \equiv g^{ik}J_{kj}\, .
\ee
The SU$(3)$ structure defined by the spinor bilinears determines an
almost-complex structure compatible with the base metric, so that
\be
g_{mn}J^m{}_iJ^n{}_j=g_{ij}\, .
\ee
Since the base metric is Hermitian with respect to $J$, and
$J_{ij}=-J_{ji}$, raising the second index instead gives
\be
\label{Jindexrelation}
J_i{}^j\equiv J_{ik}g^{kj}=-J^j{}_i\, .
\ee

The relative action of $J$ on the complex 3-form is not an
independent almost-complex convention in the present construction.
It is fixed by the Clifford-bilinear definitions in eqs.~\eqref{structuredef},
together with the chiral projection in eq.~\eqref{chiralfix},
and the six-dimensional Clifford and charge-conjugation conventions
fixed in Appendices~\ref{6ECliff}, \ref{gstructure6d}.
Indeed, substituting the bilinear definitions and applying the
seven-dimensional Fierz identity in eq.~\eqref{Fierz} to the base
spinors, with all tensor indices restricted to the six-dimensional
base, gives
\begin{align}
\label{FierzJactionPsi}
J_i{}^m\Psi_{mjk}
&=
i\bigl(\overline{\psi^*}\gamma_i{}^m\psi\bigr)
\bigl(\overline{\psi^*}\gamma_{mjk}\psi^*\bigr)
=
+i\Psi_{ijk}\, ,
\end{align}
where the corresponding relation for $\overline\Psi$ follows by complex
conjugation.
By antisymmetry of $\Psi$, the same eigenvalue occurs when $J_i{}^m$
acts on either of its other two covariant indices
\be
\label{JlowerleftPsi}
J_i{}^m\Psi_{mjk}
=
J_j{}^m\Psi_{imk}
=
J_k{}^m\Psi_{ijm}
=
+i\Psi_{ijk}\, .
\ee
Using eq.~\eqref{Jindexrelation}, the corresponding identities with
the indices of $J$ in the opposite positions are
\be
\label{JlowerrightPsi}
J^m{}_i\Psi_{mjk}
=
J^m{}_j\Psi_{imk}
=
J^m{}_k\Psi_{ijm}
=
-i\Psi_{ijk}\, ,
\ee
with complex-conjugate relations for $\overline\Psi$.

We shall use the latter action, which is the pullback of the
almost-complex endomorphism to a covariant index. For a covariant
$k$-tensor $\omega$, define its action on the $r$th slot by
\be
\label{Jsingleindexaction}
\bigl(J_{(r)}\omega\bigr)_{a_1\ldots a_r\ldots a_k}
\equiv
J^b{}_{a_r}
\omega_{a_1\ldots b\ldots a_k}\, .
\ee
The tensor-product action on all covariant indices will be denoted by
$\mathcal J$:
\be
\label{Jendo}
(\mathcal J\omega)_{a_1\ldots a_k}
\equiv
\bigl(J_{(1)}\cdots J_{(k)}\omega\bigr)_{a_1\ldots a_k}
=
J^{b_1}{}_{a_1}\cdots J^{b_k}{}_{a_k}
\omega_{b_1\ldots b_k}\, .
\ee
These two operations must not be confused: $J_{(r)}$ acts on one
specified index, whereas $\mathcal J$ acts multiplicatively on every
covariant index.

Equations \eqref{JlowerrightPsi}, dictated by the spinor bilinears,
fix a $(1,0)$ covector to have eigenvalue $-i$ under the pullback
action $J^j{}_i$. Therefore, the all-index action on a $(p,q)$-form 
is given by
\be
\label{Jaction}
\mathcal J\omega^{(p,q)}
=
(-i)^p(+i)^q\omega^{(p,q)}
=
i^{\,q-p}\omega^{(p,q)}\, .
\ee
From this, it easily follows that any real $k$-form 
$\omega = \overline{\omega}$ decomposes into $(p,q)$-type components 
satisfying the reality condition
\be
\label{almostreal}
\overline{\omega^{(p,q)}} = \omega^{(q,p)}\, .
\ee
In particular, since $\Psi$ is of type $(3,0)$, it satisfies
\be
\label{Jonpsi}
(\mathcal J\Psi)_{ijk}
= +i\Psi_{ijk}\, .
\ee
Notice that the single-index and all-index eigenvalues are different:
each individual covariant index of $\Psi$ has eigenvalue $-i$, whereas
the product action on its three indices has eigenvalue $+i$.

With the $J$ action in eq.~\eqref{Jaction}, one can always define projection 
operators $\Pi^{(p,q)}$, that extract the $(p,q)$-part of any complex form.
Therefore, the corresponding projectors on complexified covectors are 
\be
\label{1formprojector}
(\Pi^{(1,0)})_i{}^j
=
\frac12\left(\delta_i{}^j+iJ^j{}_i\right)\, ,
\qquad
(\Pi^{(0,1)})_i{}^j
=
\frac12\left(\delta_i{}^j-iJ^j{}_i\right)\, .
\ee
Equivalently, using $J_i{}^j=-J^j{}_i$,
\be
(\Pi^{(1,0)})_i{}^j
=
\frac12\left(\delta_i{}^j-iJ_i{}^j\right)\, ,
\qquad
(\Pi^{(0,1)})_i{}^j
=
\frac12\left(\delta_i{}^j+iJ_i{}^j\right)\, .
\ee
When $J$ is integrable, these algebraic type projections coincide
locally with the holomorphic and anti-holomorphic cotangent
decomposition.

Since $\Psi$ is of type $(3,0)$, each of its three covariant indices belongs 
to the $(1,0)$ cotangent subspace.
Accordingly, applying the $(1,0)$ projector to any one of its indices, 
using eq.~\eqref{JlowerrightPsi}, leaves $\Psi$ unchanged:
\be
\label{Psislotprojectors}
(\Pi^{(1,0)})_i{}^m\Psi_{mjk}
=(\Pi^{(1,0)})_j{}^m\Psi_{imk}
=(\Pi^{(1,0)})_k{}^m\Psi_{ijm}
=\Psi_{ijk}\, .
\ee
Equivalently, applying the complementary $(0,1)$ projector to any
one of its indices, gives zero:
\begin{align}
(\Pi^{(0,1)})_i{}^m\Psi_{mjk}
&=0\, ,\qquad
(\Pi^{(0,1)})_j{}^m\Psi_{imk}
=0\, ,\qquad
(\Pi^{(0,1)})_k{}^m\Psi_{ijm}
=0\, .
\end{align}
These single-index projection conditions imply the all-index
eigenvalue in eq.~\eqref{Jonpsi}. 

Raising a covariant index with the Hermitian metric reverses its
$J$-eigenvalue. Indeed, raising the first index in
eq.~\eqref{JlowerrightPsi} and using
eq.~\eqref{Jindexrelation}, gives
\be
\label{JraisedPsi}
J^i{}_m\Psi^m{}_{jk}
=
+i\Psi^i{}_{jk}\, .
\ee
Consequently,
\begin{align}
\frac12(\delta^i{}_m-iJ^i{}_m)\Psi^m{}_{jk}
&=\Psi^i{}_{jk}\, ,
&
\frac12(\delta^i{}_m+iJ^i{}_m)\Psi^m{}_{jk}
&=0
\, .
\end{align}
If all three indices are raised, the corresponding tensor-product
action is
\be
J^i{}_aJ^j{}_bJ^k{}_c\Psi^{abc}
=-i\Psi^{ijk}
\, .
\ee

The contraction of two indices of $\overline{\Psi}$ with $\Psi$ 
leaves a mixed tensor proportional to the  $(1,0)$ projector. 
Using the normalization of $\Psi$ in eq.~\eqref{normalize}, 
together with the projector defined in eq.~\eqref{1formprojector}, 
one obtains
\begin{subequations}
\label{psiprojectcontract}
\begin{align}
\Psi_{kmn}\overline\Psi^{mnl}
&=
16(\Pi^{(1,0)})_k{}^l
=
8(\delta_k{}^l-iJ_k{}^l)\, ,
\\
\Psi_{kmn}\Psi^{mnl}
&=0\, ,
\qquad
\overline\Psi_{kmn}\overline\Psi^{mnl}=0\, .
\end{align}
\end{subequations}

With eq.~\eqref{gettingreal} and the all-index action in
eq.~\eqref{Jonpsi}, we obtain
\be
\label{Jonpsialg}
\mathcal J\Psi_+=-\Psi_-,
\qquad
\mathcal J\Psi_-=\Psi_+\, .
\ee

Let $\omega_{(1)}$ be a real 1-form. Then one finds
\be
\label{almostidentity}
\mathcal J\omega_{(1)}\wedge\Psi_\pm
=
\omega_{(1)}\wedge\mathcal J\Psi_\pm
=
\mp\omega_{(1)}\wedge\Psi_\mp\, ,
\ee
since $\Psi$ is of type $(3,0)$, and the base has complex dimension three, 
so the $(1,0)$ component of $\omega_{(1)}$ wedges trivially with $\Psi$, and 
only its $(0,1)$ component contributes.

On an almost-complex manifold the structure $J$ need not be integrable, 
in which case differential calculus is modified 
\cite{CIRICI2021107970}. 
The exterior derivative does not simply decompose into the Dolbeault 
operators, as it does on complex manifolds, 
i.e.~$\dd = \partial + \overline{\partial}$. 
Rather, it decomposes into four parts, due to the non-vanishing Nijenhuis 
tensor. Acting on a $(p,q)$-type form $\omega$, the exterior derivative 
decomposes as
\be 
\dd \omega = \partial \omega + \overline{\partial} \omega 
+ \mu \omega + \overline{\mu} \omega\, , 
\ee 
where $\partial$, $\overline{\partial}$, are the bidegree $(1,0)$ and $(0,1)$ 
components of $\dd$, which are the holomorphic, and anti-holomorphic derivatives, 
respectively, when $J$ is integrable,
and $\mu$ is a $(2,-1)$ operator
determined by the Nijenhuis tensor, such that
\begin{align}
\label{almostdiff}	
\mu & : \omega^{(p,q)} \rightarrow \omega^{(p+2, q-1)}\, ,\\*
\overline{\mu} & : \omega^{(p,q)} \rightarrow \omega^{(p-1, q+2)}\, ,
\end{align}
and $\mu=\overline{\mu}=0$ when the Nijenhuis tensor vanishes.

As a result, the exterior derivative of a $(3,0)$-type form generally
contains four components of types:
$(4,0)$, $(3,1)$, $(5,-1)$, $(2,2)$. 
Yet, $(5,-1)$ does not exist by definition, 
and $(4,0)$ does not exist in six real (or three complex) dimensions.
Therefore, the exterior derivative of the complex 3-form $\Psi$ in 
eq.~\eqref{dpsi} contains forms of two types. 
The terms with the complex torsion forms 
$\mathcal{W}_1$, $\mathcal{W}_2$, 
correspond to the $(2,2)$-type component, 
which amounts to the Nijenhuis tensor:
\be 
\label{nijenhuis}
[N_J]\simeq\mathcal W_1\oplus\mathcal W_2\, ,
\ee
whereas the $(3,1)$-type component of the exterior derivative 
contains the torsion form $\mathcal{W}_5$, when it is defined. 
For this reason, the real torsion form $\mathcal{W}_5$ 
is complexified via a projection onto its $(1,0)$-type component, 
with its conjugate of $(0,1)$ type, as in eq.~\eqref{almostreal}.
Thus, in complex dimension three,
\be
\dd \Psi=\overline\mu\,\Psi+\overline\partial\Psi\, ,
\ee
with
\be
\overline\mu\,\Psi
= \mathcal W_1J\wedge J+\mathcal W_2\wedge J\, ,
\qquad
\overline\partial\Psi
= \overline{\mathcal W_5^{\mathbb C}}\wedge\Psi\, .
\ee
The first equation identifies the
$\mathcal W_1\oplus\mathcal W_2$ component of the Nijenhuis tensor.
In particular, $J$ is integrable precisely when
$\mathcal W_1=\mathcal W_2=0$.

\subsection{Conjugates of Killing Spinors}
\label{KSEforMCCC}

We provide here the necessary spinor equations satisfied by the charge and 
Majorana conjugates of the Killing spinors, as inferred from 
eqs.~\eqref{Killinggravitin} and \eqref{Killinghalfspinexchange}, emphasizing 
the resulting modifications relative to the original KSEs. 
To that end, we use the conjugation conventions detailed in 
Appendix~\ref{7Cliff}, 
together with the real Cartan basis underlying eq.~\eqref{realGammas}, 
from the relevant Cliff$(5)$ matrices, as defined in Appendix~\ref{5ECliff}. 
Hence complex conjugation and transposition act non-trivially only on the 
explicit factors of $i$, and on the Cliff$(1,6)$ matrices.

For the charge conjugate of the Killing spinor we obtain
\begin{align} 
\label{Killinggravitincharge}	
\nabla_\mu \epsilon^c = \bigg[ & 
  ig \left( A_\mu^1 \Gamma_{\hat{1}} 
+ A_\mu^2 \Gamma_{\hat{2}}\right)
+ \frac{i}{4} \gamma^{\nu} \left( X_1^{-1} F_{\mu\nu}^1 \Gamma_{\hat{1}} 
+ X_2^{-1} F_{\mu\nu}^2 \Gamma_{\hat{2}} \right)
- \frac{m}{2} \left( X_1 X_2 \right)^{-2} \gamma_\mu 
\nn\\
& 
- \frac{1}{4} \gamma_\mu \gamma^\nu 
\left( X_1^{-1} \partial_\nu X_1 + X_2^{-1} \partial_\nu X_2 \right)
+ \frac{X_1 X_2}{96} \gamma_{\mu}{}^{\nu\lambda\sigma\rho} 
G_{{\nu\lambda\sigma\rho}} \Gamma_5 \bigg] \epsilon^c \, , 
\end{align}
and 
\begin{subequations}
\label{Killinghalfspinexchangecharge}
\begin{align} 
\label{Killingfermionsymcharge}
0 = \bigg[ 
& 
\frac{m}{2} \left( X_1 + X_2 - 2 \left( X_1 X_2 \right)^{-2} \right)
- \frac{5}{8} \gamma^\mu 
\left( X_1^{-1} \partial_\mu X_1 
+ X_2^{-1} \partial_\mu X_2 \right)
\nn\\
& + \frac{i}{16} \gamma^{\mu\nu} 
\left( X_1^{-1} F_{\mu\nu}^1 \Gamma_{\hat{1}} 
+ X_2^{-1} F_{\mu\nu}^2 \Gamma_{\hat{2}} \right)
+ \frac{X_1 X_2}{96} \gamma^{\mu\nu\lambda\sigma} 
G_{\mu\nu\lambda\sigma}  \Gamma_5 \bigg] \epsilon^c \, , \\
\label{Killingfermionantisymcharge} %
0 = \bigg[ 
& 
\frac{m}{2} \left( X_1 - X_2 \right) 
- \frac{1}{8} \gamma^\mu 
\left( X_1^{-1} \partial_\mu X_1 - X_2^{-1} \partial_\mu X_2 \right)
\nn\\
&
+ \frac{i}{16} \gamma^{\mu\nu} \left( X_1^{-1} F_{\mu\nu}^1 \Gamma_{\hat{1}} 
- X_2^{-1} F_{\mu\nu}^2 \Gamma_{\hat{2}} \right)
\bigg] \epsilon^c \, . 
\end{align}
\end{subequations} 

For the Majorana conjugate of the Killing spinor, we only need then 
the differential equation
\begin{align} 
\label{KillinggravitinMajconj}	
\nabla_\mu \overline{\epsilon} = \overline{\epsilon} \bigg[ 
& - ig \left( A_\mu^1 \Gamma_{\hat{1}} 
+ A_\mu^2 \Gamma_{\hat{2}}\right)
+ \frac{i}{4} \gamma^{\nu} \left( X_1^{-1} F_{\mu\nu}^1 \Gamma_{\hat{1}} 
+ X_2^{-1} F_{\mu\nu}^2 \Gamma_{\hat{2}} \right)
+ \frac{m}{2} \left( X_1 X_2 \right)^{-2} \gamma_\mu 
\nn\\
& 
- \frac{1}{4} \gamma^\nu \gamma_\mu  
\left( X_1^{-1} \partial_\nu X_1 + X_2^{-1} \partial_\nu X_2 \right)
- \frac{X_1 X_2}{96} \gamma_{\mu}{}^{\nu\lambda\sigma\rho} 
G_{{\nu\lambda\sigma\rho}} \Gamma_5 \bigg] \, .%
\end{align}

\subsection{Higher-Rank Algebraic Bilinear Relations}
\label{algbilrel}

We complete here the algebraic bilinear relations presented in
Section~\ref{Killingresultsalg} by providing the required rank-$4$
relation with the $\Gamma_+$ generator, and the rank-$3$ relations with the
$\Gamma_-$ generator, which arise from the variations of spin-$\tfrac{1}{2}$ 
fermions in eqs.~\eqref{Killinghalfspinequal}. 
No assumption on the complex scalar bilinear, such as
$s^{11}=0$ in eq.~\eqref{cscalarto6d}, which fixes the internal
spinor angle, is imposed in deriving these higher-rank relations.
In these relations we use the dilaton-potential dependent scalar, 
$\Delta(X)$, defined in eq.~\eqref{DeltaX}.

The following complex algebraic relation at rank 4, 
which contains the generator $\Gamma_+$, 
is required to simplify the differential relation of 
the complex 3-form bilinear in eq.~\eqref{diffu11}:
\begin{align}
\label{br4cg}
0 & = \bigg[ 4 g \Delta(X)
\left( \star u^{11} \right)_{\mu\nu\sigma\rho}
+ 20 \dd X_{[\mu} u^{11}_{\nu\sigma\rho]}
- 4 i \Bigl( F^{\alpha}{}_{[\mu} 
\left( \star u^{11} \right)_{\nu\sigma\rho]\alpha} \Bigr) 
\Gamma_{+} 
\nn\\
& \qquad \qquad \qquad \qquad \qquad 
+ X^3 \Bigl( s^{11} G_{\mu\nu\sigma\rho}
-6 \left( u^{11} \right)^{\alpha}{}_{[\mu\nu} 
\left( \star G \right)_{\sigma\rho]\alpha} \Bigr) 
\Gamma_{5} \bigg] \chi\, .
\end{align}	

As to the algebraic bilinear relations with the generator $\Gamma_-$, 
in addition to eqs.~\eqref{antibr0}, \eqref{antibr1}, \eqref{antibr2}
at ranks 0, 1, 2, respectively, 
we obtain at rank 3:
\begin{subequations}
\label{antibr3}
\begin{align}
0 & = F^{\alpha}{}_{[\mu} (u^{11})_{\nu\sigma]\alpha}
\,\Gamma_- \chi	\, ,\\
0 & = F^{\alpha}{}_{[\mu} (U_3)_{\nu\sigma]\alpha}
\,\Gamma_- \chi	\, ,\\
0 & = \left[ 3 K_{[\mu} F_{\nu\sigma]} 
- \left( F \lrcorner \star T \right)_{\mu\nu\sigma} \right] 
\Gamma_- \chi\, .
\end{align}	
\end{subequations}

\subsection{Decomposing 7d Forms into 
SU\texorpdfstring{$(3)$}{(3)} Irreps}
\label{decomposeformstosu3irreps}

\paragraph{Real 1-form.}

On the six-dimensional base any complexified 1-form decomposes as
\be
\Lambda^1_{(6d)} \otimes \mathbb C 
= \Lambda^{1,0} \oplus \Lambda^{0,1}\, , 
\ee
and thus due to eq.~\eqref{almostreal} any real 1-form $\alpha$ splits as
\be
\alpha = \alpha^{(1,0)} + \overline{\alpha^{(1,0)}}\, .
\ee
For the real 1-form $\dd X$ we then have
\be
\dd X = \gamma\,K + (\dd X)_{\perp}\, ,
\qquad 
\gamma \equiv -\frac{1}{h^{2}}\,\iota_{K} \dd X\, ,
\qquad 
\iota_{K}(\dd X)_{\perp} = 0\, ,
\ee
where due to the algebraic bilinear relation in eq.~\eqref{br0r}, 
we have $\gamma=0$, and thus
\be 
\dd X = (\dd X)_{\perp}^{(1,0)+(0,1)} 
\in \mathbf 3\oplus \bar{\mathbf 3}\, ,
\ee
so that the decomposition of $\dd X$ reads
\be 
\label{dxsu3}
\dd X = D + \overline D\, ,
\qquad 
D \equiv \Pi^{(1,0)}(\dd X)\, ,
\ee 
where we defined $D$ as the projected SU(3)-irreducible component of $\dd X$.

\paragraph{Real 2-form.}

Any complexified 2-form on the six-dimensional base decomposes as
\be
\Lambda^2_{(6d)} \otimes \mathbb C
= \Lambda^{2,0} \oplus \Lambda^{1,1} \oplus \Lambda^{0,2}\, ,
\qquad
\Lambda^{1,1} 
= \Lambda^{0} J \oplus \Lambda^{1,1}_0\, .
\ee
For the real 2-form $F$ we then have
\be
F = F_{\parallel} + F_{\perp}\, ,
\qquad
F_{\parallel}\equiv -\frac1{h^2}\,K\wedge \iota_KF\, ,
\qquad
\iota_K F_{\perp}=0\, .
\ee
Decomposing the transverse 2-form $F_\perp$
into SU$(3)$-irreducible components, we then get
\be
F_{\perp} 
= \frac13 \bigl(J\lrcorner F\bigr)\,J
+ F^{(1,1)_0}
+ F^{(2,0)+(0,2)}\, ,
\ee
with
\be
J \lrcorner F \in \mathbf 1_{\mathbb{R}}\, ,
\qquad
F^{(2,0)+(0,2)}\in \mathbf 3\oplus \bar{\mathbf 3}\, ,
\qquad
F^{(1,1)_0}\in \mathbf 8\, ,
\ee
where the real singlet $\mathbf1_{\mathbb{R}}$ is 
the component proportional to $J$, 
$\mathbf3\oplus\bar{\mathbf3}$ is the real $(2,0)+(0,2)$ 
fundamental component, and $\mathbf8$ is the adjoint one, 
realized here as the remaining primitive $(1,1)$ form.
The full seven-dimensional SU$(3)$-representation content thus reads
\be
\Lambda^2_{(7d)}
\cong
(\mathbf 3 \oplus \bar{\mathbf 3})_{K \wedge \Lambda^1_{\perp}}
\oplus
\mathbf 1_{J}
\oplus
\mathbf 8_{(1,1)_0}
\oplus
(\mathbf 3 \oplus \bar{\mathbf 3})_{(2,0)+(0,2)}\, ,
\ee
where we recall that the form dual to the Killing vector, 
$K$, is an SU$(3)$ singlet.

Thus the full seven-dimensional SU$(3)$ decomposition of $F$ reads
\be
\label{f2su3}
F = - \frac1{h^2}\,K\wedge (f + \overline f)
+ \frac13 \lambda \,J
+ \tau 
+ \frac18\left( \overline \alpha \lrcorner \Psi 
+ \alpha \lrcorner \overline \Psi \right)\, ,
\ee
where we defined the following projected SU$(3)$-irreducible components:
\be
f \equiv \Pi^{(1,0)}(\iota_K F)\, ,
\quad
\lambda \equiv J \lrcorner F\, ,
\quad
\alpha \equiv F \lrcorner \Psi\, ,
\quad 
\quad \tau \equiv F^{(1,1)_0}\, ,
\quad 
\iota_K \tau = J \lrcorner \tau 
= \tau \lrcorner \Psi = 0\, ,
\ee
and we derived the component $F^{(0,2)}$ 
with its proper numerical coefficient via
\be
\label{fixcoeff} 
F^{(0,2)}_{mn} = C (\alpha \lrcorner \overline \Psi)_{mn} 
= C \alpha^p \overline \Psi_{pmn}
\quad \implies \quad
\frac12 \left(\alpha \lrcorner \overline\Psi\right)_{mn} 
\Psi^{mn}{}_{q}
= 8 \alpha_q 
\quad \implies \quad
C = \frac18\, ,
\ee
where in the second step we used eq.~\eqref{psiprojectcontract}.
Finally, in the SU$(3)$ decomposition the primitive component 
$F^{(1,1)_0}$ is characterized as the common kernel of all 
applicable SU$(3)$-equivariant contraction maps. Thus, it is the 
part of $F$ left after removing all of the SU$(3)$-irreducible 
components, which are detected by the projection maps:
\be
F \mapsto \Pi^{(1,0)}(\iota_K F)\, , 
\qquad F \mapsto J \lrcorner F\, ,
\qquad F \mapsto F \lrcorner \Psi\, ,
\ee
that yield the other components in eq.~\eqref{f2su3}. 
All of the projection maps, which are required to decompose 
$\dd X$, $F$, and $G$, are collected and further discussed 
in Appendix~\ref{su3projections}.

\paragraph{Real 3- and 4-forms.}

Any complexified 3-form on the six-dimensional base decomposes as
\begin{align}
\Lambda_{(6d)}^3 \otimes \mathbb C
= \Lambda^{3,0} \oplus \Lambda^{2,1} 
\oplus \Lambda^{1,2} \oplus \Lambda^{0,3}\, , 
\qquad 
\Lambda^{3,0} 
= \Lambda^{0}_{\mathbb C} \Psi\, ,
\quad
\Lambda^{2,1} 
= J \wedge \Lambda^{1,0} \oplus \Lambda^{2,1}_0\, ,
\end{align}
and any complexified 4-form as
\begin{align}
\Lambda_{(6d)}^4 \otimes \mathbb C
= \Lambda^{3,1} \oplus \Lambda^{2,2} \oplus \Lambda^{1,3}\, , 
\qquad 
\Lambda^{3,1} 
= \Psi \wedge \Lambda^{0,1}\, ,
\quad
\Lambda^{2,2} 
= \Lambda^{0} J \wedge J  
\oplus J \wedge \Lambda^{1,1}_0.
\end{align}
For the real 4-form $G$ we then have
\begin{align}
\label{f47dto6d}	
G & = -\frac1{h^2} K \wedge \iota_K G + G_{\perp}\, ,
\qquad
\iota_K G_{\perp} = 0\, .
\end{align}

Decomposing the transverse 3-form $\iota_K G$, which is the 
longitudinal component of $G$, into SU$(3)$-irreducible components, 
we then get
\begin{align}
\label{longG}
\iota_K G & =
\frac18 \bigl( \overline{\zeta}\Psi 
+ \zeta \overline \Psi \bigr)^{(3,0)+(0,3)}
+ \frac12 J \wedge \left( q + \overline{q} \right)^{(1,0)+(0,1)} 
+ \left( \rho + \overline{\rho} \right)^{(2,1)_0+(1,2)_0}\, ,
\end{align}
with
\begin{align}
\zeta & \in \mathbf{1}_{\mathbb C}\, , 
\qquad
 q + \overline{q} \in \mathbf{3} \oplus \mathbf{\bar{3}}\, ,
\qquad
\rho + \overline{\rho} \in \mathbf{6} \oplus \mathbf{\bar{6}}\, ,
\end{align}	
where $\zeta$, $\bar{\zeta}$, are the complex singlets, which are  
coefficients of the components proportional to $\overline \Psi$, $\Psi$, 
and we defined the following projected SU$(3)$-irreducible components:
\be 
\zeta \equiv \Psi \lrcorner \iota_K G\, ,
\qquad 
q \equiv \Pi^{(1,0)}\bigl(J\lrcorner \iota_K G\bigr)\, ,
\qquad 
\rho \equiv \bigl(\iota_K G \bigr)^{(2,1)_0}\, ,
\quad 
J \lrcorner \rho = \Psi \lrcorner \rho = 0\, ,
\ee
so that $\rho$ is the primitive component,
which remains after removing the contractions with $J$ and $\Psi$.

Next, decomposing the transverse $4$-form of $G$ into 
SU$(3)$ irreducible components, we then get
\begin{align}
\label{perpG}
G_{\perp} & =
\frac1{12} \sigma J \wedge J + J \wedge \kappa^{(1,1)_0} 
+ \frac18 \bigl( \Psi \wedge \overline{\beta} 
+ \overline \Psi \wedge \beta \bigr)^{(3,1)+(1,3)}\, ,
\end{align}
with 
\begin{align}
\sigma \in \mathbf{1}_{\mathbb R}\, ,
\qquad
\beta + \overline{\beta} \in \mathbf{3} \oplus \mathbf{\bar{3}}\, ,
\qquad 
\kappa \in \mathbf{8}\, ,
\end{align}
where $\sigma$ is the singlet of a double contraction with $J$, 
$\beta$ is the form extracted by contracting with $\Psi$, 
and thus we defined the following projected SU$(3)$-irreducible 
components:
\be 
\sigma \equiv J\lrcorner J\lrcorner G\, ,
\qquad
\beta \equiv \Psi\lrcorner G_{\perp}\, ,
\quad 
\kappa \equiv (J \lrcorner G_\perp)^{(1,1)_0}\, ,
\quad 
J \lrcorner \kappa =0\, ,
\ee
so that $\kappa$ is the primitive component, which remains after 
contracting once with $J$, and removing the singlet term. 

Similarly to the derivation in eq.~\eqref{fixcoeff}, 
the numerical coefficients in the two decompositions of $G$ 
in eqs.~\eqref{longG}, \eqref{perpG}, were fixed according 
to our contraction convention in eq.~\eqref{formcontract}, 
the SU$(3)$ normalisation in eq.~\eqref{normalize}, and using the 
SU$(3)$ contraction identities in Appendix~\ref{contractidentities} 
in eqs.~\eqref{1formid1}, \eqref{2formid1}, \eqref{2formid2}.

The full seven-dimensional SU$(3)$-representation content thus reads
\begin{align}
\Lambda^4_{(7d)} & \cong
\bigl[ (\mathbf 1 \oplus \bar{\mathbf 1})_{\overline{\Psi} 
\oplus \Psi} 
\oplus ( \mathbf 3 \oplus \bar{\mathbf 3})_{J \wedge \Lambda^1_\perp}
\oplus (\mathbf 6\oplus \bar{\mathbf 6}) \bigr]_{K \wedge \Lambda^3_\perp}
\oplus \bigl[ \mathbf 1_{J \wedge J} \oplus \mathbf 8 
\oplus (\mathbf 3 \oplus \bar{\mathbf 3})\bigr].
\end{align}
Thus the full seven-dimensional SU$(3)$ decomposition of $G$ reads
\begin{align}
\label{f4su3}
G & = - \frac1{h^2} K \wedge 
\left( \frac{1}{8} \bigl( \overline{\zeta} \Psi 
+ \zeta \overline \Psi \bigr)
+ \frac12 J \wedge ( q + \overline q ) 
+ \rho + \overline\rho \right) 
\nn \\
& \qquad \qquad 
+ \frac1{12} \sigma J \wedge J + J \wedge \kappa 
+ \frac18 \bigl( \Psi \wedge \overline \beta 
+ \overline \Psi \wedge \beta \bigr)\, .
\end{align}

\subsection{Projection Maps into SU(3) Irreps}
\label{su3projections}

There are five elementary contraction maps with the three bilinear forms, 
$K$, $J$, and $\Psi$, which provide the basic building blocks for 
the projection maps that are used in the SU$(3)$ decomposition detailed in 
Appendix~\ref{decomposeformstosu3irreps}.
These elementary maps for any rank-compatible p-form $A$ read:
\begin{subequations}
\begin{align}
\label{realmaps}	
A \mapsto \iota_K A\, ,\\
\label{contractwJ}
A \mapsto A \lrcorner J\, ,
\qquad 
A \mapsto J \lrcorner A\, ,\\
A \mapsto A \lrcorner \Psi\, ,
\qquad
A \mapsto \Psi \lrcorner A\, ,
\end{align}
\end{subequations}
or more compactly:
\begin{align}
\label{5basicmaps}	
\iota_K,
\qquad
(\cdot)\lrcorner J\, ,
\qquad
J\lrcorner(\cdot)\, ,
\qquad
(\cdot)\lrcorner \Psi\, ,
\qquad
\Psi\lrcorner(\cdot)\, ,
\end{align}
where note that the last two maps are complex-valued 
on real forms, 
and on a transverse 1-form $\alpha$,
the basic map in eq.~\eqref{contractwJ} is used as follows:
\be
\label{projector}
\Pi^{(1,0)} 
\equiv
\frac12\bigl(\mathds{1}
+i(\cdot\lrcorner J)\bigr),
\qquad
\mathcal J\alpha=\alpha\lrcorner J,
\qquad
(\mathcal J\alpha)_i=J^j{}_i\alpha_j\, ,
\ee
which is equivalent to the $(1,0)$ projector in 
eq.~\eqref{1formprojector}.
The eleven independent SU$(3)$-irreducible components, appearing 
in $\dd X$, $F$ and $G$, are then extracted 
by various compositions of the elementary contractions, also with the 
holomorphic type projections, as in eq.~\eqref{projector}, 
and the primitive extractions described below.

We recall that in our sector of the theory there are eight 
non-primitive SU$(3)$-irreducible components, 
which can be directly detected via the proper projection maps, 
defined by
\begin{subequations}
\label{8nonprim}
\begin{align}
\label{projD}		
D &\equiv \Pi^{(1,0)}(\dd X)\, ,\\*
\label{projf}
f &\equiv \Pi^{(1,0)}(\iota_K F)\, ,\\*
\label{projlambda}
\lambda &\equiv J\lrcorner F\, ,\\*
\label{projalpha}
\alpha &\equiv F\lrcorner \Psi\, ,\\*
\label{projchi}
\zeta &\equiv \Psi \lrcorner \iota_K G\, ,\\*
\label{projq}
q &\equiv \Pi^{(1,0)}\bigl(J\lrcorner \iota_K G\bigr)\, ,\\*
\label{projsigma}
\sigma &\equiv J\lrcorner J\lrcorner G\, ,\\*
\label{projbeta}
\beta &\equiv \Psi\lrcorner G_{\perp}\, .
\end{align}
\end{subequations}
The component $\zeta$ is a complex scalar. 
The components $\alpha$ and $\beta$, obtained by contraction with
$\Psi$, are automatically of type $(1,0)$ and therefore require no
additional application of the projector in eq.~\eqref{projector}.

In addition to the eight non-primitive SU$(3)$-irreducible components 
in eqs.~\eqref{8nonprim}, there are three more SU$(3)$-irreducible 
components, denoted as $\tau$, $\rho$, $\kappa$, which are 
primitive ones, namely they belong to the kernels 
of all elementary contraction maps in eq.~\eqref{5basicmaps}. 
We proceed to spell out these three primitive components.

In $F$ there is a single primitive SU$(3)$-irreducible component, 
which is defined as
\be
\label{tau}
\tau \equiv F^{(1,1)_0} 
\equiv F + \frac1{h^2} K \wedge ( f + \overline {f} ) 
- \frac13 \lambda J 
- \frac18 \bigl(\overline \alpha \lrcorner \Psi 
+ \alpha \lrcorner \overline \Psi \bigr)\, ,
\ee
and satisfies the kernel conditions
\be
\iota_K \tau = 
J \lrcorner \tau = 
\tau \lrcorner \Psi = 0\, ,
\ee
so that all in all $F$ decomposes into four independent 
SU$(3)$-irreducible components as in eq.~\eqref{f2su3}.

In $G$ there are two primitive SU$(3)$-irreducible components,  
which are defined as
\begin{align}
\label{rho}	
\rho & \equiv \left(\iota_K G\right)^{(2,1)_0} 
\equiv \Pi^{(2,1)_0} \Bigl[ \iota_K G 
- \frac{1}{8} \bigl( \overline{\zeta} \Psi 
+ \zeta \overline \Psi \bigr) 
- \frac12 J \wedge ( q + \overline q ) \Bigr]\, ,
\\
\label{kappa}
\kappa & \equiv \left(J \lrcorner G\right)^{(1,1)_0} 
\equiv J \lrcorner \Bigl[ G + \frac1{h^2} K \wedge \iota_K G
- \frac1{12} \sigma J \wedge J 
- \frac18 \bigl( \Psi \wedge \overline \beta 
+ \overline \Psi \wedge \beta \bigr) \Bigr]\, ,
\end{align}
and satisfy the kernel conditions
\begin{align}
\iota_K \rho = 
J \lrcorner \rho = 
\Psi \lrcorner \rho = 0\, ,
\\
\iota_K \kappa =
J \lrcorner \kappa = 
\kappa \lrcorner \Psi = 0\, , 
\end{align}
so that all in all $G$ decomposes into six independent 
SU$(3)$-irreducible components as in eq.~\eqref{f4su3}.

\subsection{Contraction Identities in 
	SU\texorpdfstring{$(3)$}{(3)}}
\label{contractidentities}

To start with,  we recall the convention of wedge product 
that is used in this work, which is the standard component 
or ``sum-of-permutations'', defined as
\be 
\label{wedgeconv}
(A \wedge B)_{i_1 \dots i_{k+l}} 
= \frac{(k+l)!}{k!l!} A_{[i_1 \dots i_k} B_{i_{k+1} \dots i_{k+l}]}\, .
\ee

Let us record the following useful identity for a transverse 2-form $\eta$:
\be 
\label{2formid1}	
J \lrcorner \left( J \wedge \eta \right) 
= \eta + \left( J \lrcorner \eta \right) J\, .
\ee

To evaluate projections of the singlet isotypic component 
in Section~\ref{su3singlet} and Appendix~\ref{higherproj},
we obtained the following useful identities. 
They hold for forms on the six-dimensional base $A_p$ and $B_q$ 
of degrees $p\leq q$, respectively, and a 1-form $K$, dual 
to the timelike Killing vector in a Lorentzian seven-dimensional manifold, 
as discussed in Section~\ref{gstructure}:
\begin{subequations}
\begin{align}
\label{Kwedgecontraction}
\left(K\wedge A_p\right)\lrcorner\left(K\wedge B_q\right)
&=K^2\left(A_p\lrcorner B_q\right)\, ,
\\*
\label{Kcontractionidentity}
A_p\lrcorner\left(K\wedge B_q\right)
&=
(-1)^p K\wedge\left(A_p\lrcorner B_q\right)\, ,
\\*
\label{Kwedge3formcontraction}
\left(K\wedge A_3\right)^{abc}{}_{[i}
\left(K\wedge B_3\right)_{j]abc}
&=
-3K^2 A^{ab}{}_{[i}B_{j]ab}\, ,
\\*
\label{JPsiPsibarcontraction}
J^{ij}\Psi^{ab}{}_{[i}\overline{\Psi}_{j]ab}
&=-48i\, ,
\\*
\label{KwedgeA2B4contraction}
K^\mu
\left(K\wedge A_2\right)^{ab}{}_{[\mu}
(B_4)_{\nu\rho]ab}
&=
\frac{1}{3}K^2 A^{ab}(B_4)_{\nu\rho ab}\, ,
\\*
\label{A3KwedgeB3contraction}
K^\mu
(A_3)^{ab}{}_{[\mu}
\left(K\wedge B_3\right)_{\nu\rho]ab}
&=
-\frac{2}{3}K^2
(A_3)^{ab}{}_{[\nu}(B_3)_{\rho]ab}\, , 
\\*
\label{KwedgeA2KwedgeB3contraction}
\left(K\wedge A_2\right)^{ab}{}_{[i}
\left(K\wedge B_3\right)_{jk]ab}
&=
2K^2 A^a{}_{[i}(B_3)_{jk]a}\, ,
\\*
\label{PsiJPsiBarcontraction}
\Psi\lrcorner
\left(
J^\alpha{}_{[i}\overline{\Psi}_{jk]\alpha}
\right)
&=
8i\, ,
\\*
\label{PsiJJcontraction}
\Psi^{ab}{}_{[i}
\left(J\wedge J\right)_{jk]ab}
&=
4\Psi_{ijk}\, .
\end{align}
\end{subequations}

To evaluate projections of the fundamental isotypic component
in Section~\ref{su3fundamental} and Appendix~\ref{higherproj},
we obtained the following identities, where here $\alpha$ denotes 
a transverse 1-form, and $\eta$ denotes a transverse 1-form or 2-form: 
\begin{subequations}
\begin{align} 
\label{1formid1}
J \lrcorner \left( J \wedge \alpha \right) 
& = 2 \alpha\, ,\\*
\label{2formid2}
\Psi \lrcorner 
\left[ \overline\Psi\wedge \left( \eta + \overline{\eta} \right) \right]
& = 8\eta\, ,\\*
\label{1formid2}
\left( J \wedge \alpha \right) \lrcorner \left( J \wedge J \right) 
& = 4 \alpha \lrcorner J \, , \\*
K^\mu A_4^{abc}{}_{[\mu}
\left(K\wedge B_3\right)_{\nu]abc}
& =
-\frac12K^2 A_4^{abc}{}_\nu(B_3)_{abc}\, ,\\*
K^\mu\left(K\wedge A_3\right)^{abc}{}_{[\mu}(B_4)_{\nu]abc}
& = -\frac12K^2 A_3^{abc}(B_4)_{\nu abc}\, ,\\*
\left(\alpha\lrcorner\overline{\Psi}\right)
\lrcorner\left(J\wedge J\right)
& = 2\left(\alpha\lrcorner\overline{\Psi}\right)\, ,\\*
\label{fundid7}
\alpha \wedge \overline{\Psi} 
& = -i \left( \alpha \lrcorner \overline{\Psi} \right) \wedge J\, , \\*
\label{fundid8}
\left[ J\lrcorner\left(\overline{\Psi}\wedge\alpha\right)
\right]\lrcorner\Psi
& =  i \Psi \lrcorner \left( \overline{\Psi} \wedge \alpha \right) %
\, ,\\*
\alpha \lrcorner \left( J \wedge J \right) 
& = 2J \wedge \left( \alpha \lrcorner J \right) %
\, , \\*
J\lrcorner \left[ 3\Psi^{ab}{}_{[i}
\left( \overline{\Psi}\wedge\alpha\right)_{jk]ab} \right]
& = - 4 i \Psi \lrcorner \left( \overline{\Psi} \wedge \alpha \right) 
\, . 
\end{align}
\end{subequations}

Further identities were derived to evaluate projections 
of the primitive isotypic components in Section~\ref{su3primitive}. 
Using the $J$ endomorphism action defined in eq.~\eqref{Jendo}, 
we obtained the following identities for any transverse 
primitive $(1,1)_0$ form $P$:
\begin{subequations}
\label{primids}
\begin{align}
\label{primitivePJJsq}
P\lrcorner (J\wedge J) & = -2P\, ,
\\
\label{primitiveJJP}
J\lrcorner (J\wedge P) & = P\, ,
\\ 
\label{primitiveJPantisym}
J^{\alpha}{}_{[i} P_{j]\alpha} & = 0\, ,
\\
\label{primitiveJJsqJPantisym}
(J\wedge J)^{abc}{}_{[i} (J\wedge P)_{j]abc} & = 0\, .
\end{align}
\end{subequations}
Finally, %
together with eq.~\eqref{defjsq}, we obtained the following identity 
for any transverse primitive $(2,1)_0$ form $\rho$:
\be
\label{rhoJinsert}
J^\alpha{}_{[i}\rho_{jk]\alpha} = -\frac{i}{3}\rho_{ijk}\, .  
\ee

\subsection{7d to 6d Hodge-Dual Identities}
\label{hodgeidentities}

We first recall the following general algebraic identity 
for forms $A$ and $B$ of degrees
$k\leq l$, respectively, on a $d$-dimensional manifold:
\begin{equation}
\label{Hodgeidentity}	
A \lrcorner B = (-1)^{s+(l-k)(d-l)} \star(A\wedge\star B)\, ,
\end{equation}
with $s=0,1$ for a Euclidean or Lorentzian signature (mod 2), respectively.

We also use the double Hodge in a general dimension $d$ and both signatures $s$
\be
\label{doublehodge}
\star_d \star_d A = (-1)^{s+k(d-k)} A
\quad
\implies
\quad
\star_6 \star_6 A = (-1)^{k(6-k)} A\, ,
\quad
\star \star A = (-1)^{1+k(7-k)} A\, .
\ee

It is useful to record here the general Weil formula 
for the Hodge dual of a primitive $(p,q)$-type $k$-form $\alpha$ 
in an SU$(3)$ structure (in complex dimension 3) \cite{voisin2002},
which reads 
\be
\label{hodgesu3prim}
\star_6 \alpha = 
(-1)^{\frac12k(k+1)}i^{p-q} \frac{1}{(3-k)!} 
J^{3-k} \wedge \alpha\, .
\ee
Here $k=p+q\leq3$, $J^{3-k}$ denotes the corresponding wedge
power of $J$, and primitivity means $J\lrcorner\alpha=0$.
This formula is pointwise algebraic and does not require the
almost-complex structure to be integrable.

Using this Weil formula %
for any primitive $(1,1)_0$ form $P$, 
we obtain the Hodge dual on the base as
\begin{align}
\label{primHodgeP}
\star_6 P=-J\wedge P .
\end{align}
Then, using eqs.~\eqref{primHodgeP} and \eqref{doublehodge}, we get
\begin{align}
\label{primHodgeJP}
\star_6(J\wedge P)=-P .
\end{align}
Using the Weil formula %
for any primitive $(2,1)_0$ form $\rho$, 
we obtain the Hodge dual on the base as
\be
\label{primHodgerho}
\star_6\rho=i\rho\, .
\ee

For applications to the seven-dimensional bilinears, we require the $K$-split 
relation between the seven- and six-dimensional Hodge duals on the base. 
Using the seven-dimensional orientation convention in eq.~\eqref{orient} 
and the relation between $K$ and the orthonormal frame in 
eq.~\eqref{Killingdual}, the key identities are
\be
\label{hodgedualsplit}
\star \alpha = (-1)^{p+1} \frac1h K \wedge \star_6\alpha\, ,
\qquad 
\star ( K \wedge \alpha ) = h \star_6\alpha\, ,
\ee
where $\alpha$ is a transverse base p-form.
Using the Hodge duals of the SU$(3)$ structure on 
the six-dimensional base in eqs.~\eqref{normal}, \eqref{su3hodge}, 
we then get for the Hodge duals of the seven-dimensional bilinears:
\begin{align}
\star J &= -\frac1{2h} K \wedge J\wedge J\, ,
\quad  
\star( K \wedge J ) = \frac{h}{2} J \wedge J\, ,
\quad 
\star K = \frac{h}{6} J \wedge J \wedge J\, , 
\\
\star\Psi & = -\frac{i}{h}\,K\wedge\Psi\, ,
\qquad 
\star \overline \Psi = \frac{i}{h} K \wedge \overline \Psi\, ,
\quad 
\star 1 = - \frac{1}{6h} K \wedge J \wedge J \wedge J\, ,
\end{align}
and for the general 3-form bilinears in eqs.~\eqref{generic3forms}, 
using the internal SM angle parameters defined in eqs.~\eqref{genericscalars}, 
we get
\begin{align}	
\label{hodgeu11gen}
\star u^{11} & = - \frac{h s p}{2} J\wedge J 
- \frac{i}{2} K \wedge \Bigl( (1-c) p^2 \Psi + (1+c) \overline \Psi\Bigr)\, ,\\
\label{hodgeu3gen}
\star U_3 & = \frac{h c}{2}\,J\wedge J
+ \frac{is}{2} K \wedge \bigl( p \Psi - p^{-1} \overline \Psi \bigr)\, .
\end{align}

\subsection{Higher-Rank Projections by Isotypes}
\label{higherproj}

Using the contraction and Hodge-dual identities in
Appendices~\ref{contractidentities}, \ref{hodgeidentities}, we
collect here the remaining singlet and fundamental projections of
the algebraic relations: the independent lower-rank projections
were given in Sections~\ref{su3singlet}, \ref{su3fundamental};
the additional equations below provide consistency checks and
impose no further independent constraints.

\paragraph{Singlet projections.}

We obtain the following additional equations for the singlet
isotypic component from the bilinear relations at ranks 1 to 3:
\begin{subequations}
\label{high0formproj}
\begin{align}
\text{Eq.~\eqref{br1c}$_{[0]}$:} 
\quad 0 & = i h s \lambda + X^3 \left(c u + i v \right)\, , 
\\	%
\text{Eq.~\eqref{br1r}$_{[0]}$:} 
\quad 0 & = X^3 s u\, , 
\\ %
\text{Eq.~\eqref{br1rg}$_{[0]}$:} 
\quad 0 & = 8 g \Delta - 2 c \lambda - X^3 \sigma\, , 
\\ %
\text{Eqs.~\eqref{br1}$_{[0]}$} 
& = J \lrcorner \text{Eqs.~\eqref{br2}$_{[ij]}$}\, , 
\\ 
\text{Eq.~\eqref{br0cg}} 
&  = J \lrcorner \text{Eq.~\eqref{br3cg}$_{[0ij]}$}\, , 
\\
\text{Eq.~\eqref{br0rg}} 
&  = J \lrcorner \text{Eq.~\eqref{br3rg}$_{[0ij]}$}\, , 
\\
J \lrcorner \text{Eq.~\eqref{br3r}$_{[0ij]}$:} 
\quad 0 & = 0\, ,
\\ %
\Psi \lrcorner \text{Eq.~\eqref{br3cg}$_{[ijk]}$:} 
\quad 0 & = \left( c + 1 \right) h 
\left( 8 g \Delta - 2 \lambda - X^3 \sigma \right) 
- 2 X^3 s \left( v + i u \right)\, ,
\\ %
\overline \Psi \lrcorner \text{Eq.~\eqref{br3cg}$_{[ijk]}$:} 
\quad 0 & = \left( c - 1 \right) h 
\left( 8 g \Delta + 2 \lambda - X^3 \sigma \right) 
- 2 X^3 s \left( v - i u \right)\, ,
\\ %
\Psi \lrcorner \text{Eq.~\eqref{br3rg}$_{[ijk]}$:} 
\quad 0 & = 8 h s g \Delta - X^3 h s \sigma + 2 X^3 c \left( v + i u \right)\, ,
\\ %
\Psi \lrcorner \text{Eq.~\eqref{br3r}$_{[ijk]}$:} 
\quad 0 & = i h s \lambda - X^3 \left( u - i v\right)\, .
\end{align}
\end{subequations}

\paragraph{Fundamental projections.}

We obtain the following additional equations for the fundamental 
isotypic component from the bilinear relations at ranks 1 to 3:
\begin{subequations}
\label{high1formproj}
\begin{align}	
\text{Eq.~\eqref{br1c}$_{[i]}$:} 
\quad 0 & = 5 h s D - s f - X^3 s q 
+ \frac{i}{2} h p \left( 1 - c \right) 
\left( \alpha - i X^3 \beta \right)\, ,
\\ %
\text{Eq.~\eqref{br1rg}$_{[i]}$:} 
\quad 0 & = 5 h D - c f - X^3 q 
+ \frac{i}{2} h s p \alpha\, ,
\\%
\text{Eqs.~\eqref{br1}$_{[i]}$} 
& = \text{Eqs.~\eqref{br2}$_{[0i]}$} 
= J \lrcorner \text{Eqs.~\eqref{br3}$_{[ijk]}$}\, ,\\	
\text{Eq.~\eqref{br2c}$_{[ij]}\lrcorner \Psi$:} 
\quad 0 & = \left( 1 + c \right) \left( 5 h D - f - X^3 q \right) 
+ \frac{i}{2} h s p \left( \alpha - i X^3 \beta \right)\, ,
\\ %
\label{prim2formbr2rgproj}
\text{Eq.~\eqref{br2rg}$_{[ij]}\lrcorner \Psi$:} 
\quad 0 & = 2 s f + i h p \left( c \alpha + i X^3 \beta \right)\, , 
\\ %
\text{Eqs.~\eqref{br2}$_{[ij]}\lrcorner \Psi$} 
& = \text{Eqs.~\eqref{br3}$_{[0ij]}\lrcorner \Psi$}\, .
\end{align}
\end{subequations}
\smallskip

\subsection{Positive Topological Mass Term}
\label{altcompletion}

Here we construct and examine possible supersymmetric completions 
to an alternative bosonic action, which differs from ours in 
eq.~\eqref{bosaction}, only in its sign of the topological mass term being 
positive, and is used in the bosonic works of
\cite{Chow:2007ts,Bobev:2023bxl,Bobev:2025xan}. 
We test the consistency and validity of these possible completions 
via their consistent truncation to the minimal theory 
and an integrability analysis corresponding to those which we carried 
out in Sections~\ref{spinors} and \ref{solanalysis}, respectively.
With the maximal-coupling relation imposed, the alternative Lagrangian reads
\be
\label{altpositiveaction}
\mathcal L^+=\mathcal L_{\rm GI}+gB\wedge G+B\wedge F_1\wedge F_2\, ,
\ee
where $\mathcal L_{\rm GI}$ contains the kinetic terms and scalar potential of
eq.~\eqref{bosaction}. Following Section~\ref{spinors} we first map the 
bosonic action of \cite{Gauntlett:2000ng}, allowing for undetermined convention 
signs. We then match its Cartan-reduced KSEs, and their minimal-theory 
truncation to \cite{Dibitetto:2017tve}. 
Unlike the case for our action in eq.~\eqref{bosaction}, 
this mapping is found to be non-unique, as it leaves two possible sign 
assignments, which turn out to be inconsistent with one another.
We further examine these two candidate completions by testing their 
integrability in common spacetime, Hodge and Clifford conventions, 
which are taken to be those used throughout this work.

\paragraph{Mapping KSEs.}
After the mapping of bosonic actions, 
absolute field normalizations are the same as those of Section~\ref{spinors}.
The remaining undetermined signs are parametrized by $\sigma_g$,
which multiplies both couplings in \cite{Gauntlett:2000ng}, and
$\sigma_G$, which multiplies $G$ in the linear higher-form mapping
in eq.~\eqref{BtoGmap}.
Each candidate completion is obtained from the KSEs in
eqs.~\eqref{Killinggravitin}, \eqref{Killingspinhalf} by
\be
\label{altKSEmap}
(g,m,G)\longmapsto(\sigma_g g,\sigma_g g,\sigma_G G)\, ,
\ee
while the scalar and vector fields, internal Cartan matrices
and spinors are unchanged.
The complete relative sign mapping that we find is
\begin{equation}
\label{fullKSEmap}	
(\sigma_g,\sigma_G)=
\begin{cases}
(+1,+1) & \quad \text{Our KSEs},\\
(+1,-1) & \quad \text{Candidate 1},\\
(-1,+1) & \quad \text{Candidate 2},
\end{cases}
\end{equation}
so that candidate~1 reverses explicit $G$ terms with respect to our KSEs, 
whereas candidate~2 reverses the signs of gauge connection and the mass/scalar 
shifts in the KSEs. 
The simultaneous reversal of Cartan generators and vector potentials 
remains a free convention choice.

\paragraph{Mapping integrability.}
In Proposition~\ref{independentbosonic}, we used the spinorial
relations in eqs.~\eqref{intspinorialidentities}. 
Applying the substitutions in eq.~\eqref{altKSEmap} to these relations 
gives their counterparts for the two candidate KSE systems.
To identify the bosonic EOM residuals entering these relations,
we apply the same substitutions to our action in
eq.~\eqref{bosaction}, denoted here by $\mathcal L^-$:
\begin{equation}
\label{alttransportaction}
\mathcal L^-_{\sigma}
=
\mathcal L_{\rm GI}
-\sigma_g gB\wedge G
+\sigma_G B\wedge F_1\wedge F_2\, .
\end{equation}
We compare the resulting EOM residuals with those of the
positive-sign action $\mathcal L^+$ in eq.~\eqref{altpositiveaction},
to determine whether compatibility with the candidate KSEs imposes
additional restrictions to the possible solutions. 
We carry out this comparison in the equal-charge timelike sector of 
Proposition~\ref{independentbosonic}.

The positive-sign action has the same Abelian EOM residual as 
eq.~\eqref{intvectorresidual}, denoted by $E_A^+$, whereas its 3-form EOM 
residual is modified from ours in eq.~\eqref{intformresidual} to
\be
\label{altformresidual}
E_B^+=\dd(X^4\star G)-2gG-F\wedge F\, .
\ee
The positive-sign action has the same scalar and Einstein equations residuals. 
The normalized residuals associated with eq.~\eqref{alttransportaction} are
\begin{subequations}
\label{altresiduals}
\begin{align}
\mathcal A_{\sigma}
&\equiv E_A^++(1-\sigma_G)X\star(F\wedge G)\, ,\\*
\mathcal B_{\sigma}
& \equiv E_B^++2(\sigma_g+1)gG+(1-\sigma_G)F\wedge F\, .
\end{align}
\end{subequations}
The complete spinorial identities in eqs.~\eqref{intspinorialidentities} 
therefore hold with the replacements
\be
\label{altidentitymap}
(E_A,E_B)\longmapsto(\mathcal A_{\sigma},\sigma_G\mathcal B_{\sigma})\, .
\ee
The differential bilinears and Bianchi identities likewise transport
eq.~\eqref{intlongitudinalform} to
\be
\label{altlongitudinal}
\iota_K\mathcal B_{\sigma}=0\, .
\ee

\paragraph{Candidate completion 1.}
For the sign assignment in eq.~\eqref{fullKSEmap},
both topological terms of $\mathcal L^-_{\sigma}$ in
eq.~\eqref{alttransportaction} have opposite signs to those of
$\mathcal L^+$ in eq.~\eqref{altpositiveaction}.
The spinorial relations in eqs.~\eqref{intspinorialidentities}, with the 
replacements in eq.~\eqref{altidentitymap}, together with
eq.~\eqref{altlongitudinal}, force both residuals in
eqs.~\eqref{altresiduals} to vanish by the same arguments used in 
Proposition~\ref{independentbosonic}.
With the sign assignment in eq.~\eqref{fullKSEmap}, 
this imposes the additional restrictions
\begin{subequations}
\label{altextra}
\begin{align}
F\wedge G&=0\, ,\\*
2gG+F\wedge F&=0\, .
\end{align}
\end{subequations}
These additional restrictions exclude verified known timelike
supersymmetric solutions. The maximally embedded flow in
eq.~(4.15) of \cite{Dibitetto:2017tve}, discussed in Section~\ref{minsol},
has vanishing Cartan flux and nonzero 4-form flux at nonzero gauging. Yet, 
eqs.~\eqref{altextra} would force its 4-form to vanish. 
This candidate completion therefore fails to preserve the timelike 
supersymmetric solution space.

\paragraph{Candidate completion 2.}
For the sign assignment in eq.~\eqref{fullKSEmap}, 
eq.~\eqref{alttransportaction} is exactly the positive-sign action, and both 
residuals in eqs.~\eqref{altresiduals} coincide with its actual EOM residuals.
Proposition~\ref{independentbosonic} consequently applies with the
same independent-equation results: one must impose only
\be
\label{altindependentEOM}
(E_B^+)_{\perp}^{(2,2)}=0\, ,
\ee
and additionally the scalar equation in the active classes when the complex 
scalar bilinear vanishes.  
In fact, the complete unequal-charge action and KSEs are simply related
to ours by reversing the signed coupling while keeping all bosonic
fields and Killing spinors fixed. Thus, this mapping preserves
the supersymmetric solution space.

To conclude, the non-unique mapping that arises from the action with a positive 
topological mass term thus yields two candidate completions that are 
inconsistent with one another, as they have different integrability 
implications. 
Only one candidate qualifies as a supersymmetric completion that is consistent 
with our formulation of the theory and its known solutions.

\subsection{Addendum}
\label{addendumtov1}

After the first version of this work appeared on arXiv,
\cite{Cheung:2022ilc} was brought to our attention.
We clarify here its relation to the formulation in Section~\ref{7dsugra}
and summarize the additions to our solution survey in Section~\ref{maxsol}.

Appendix~A.2 of \cite{Cheung:2022ilc} provides a bosonic action 
for the maximal U$(1)^2$-gauged theory consistent with the Cartan KSEs 
that are already given in \cite{Liu:1999ai}, with both its action and KSEs 
in terms of the explicit $3$-form potential. 
A higher-form mapping from potential to field strength was not provided in 
\cite{Cheung:2022ilc}.
We recall that our Cartan-reduced KSEs likewise agree with \cite{Liu:1999ai}
before our higher-form mapping from potential to field strength in 
Section~\ref{spinors}.
This higher-form conversion is the critical element in making the limit to 
ungauged theory manifestly smooth, and enabling the localization 
via a relative equivariant pair \cite{Gaar:2026gmo}. 

We provide here a field dictionary relating the action and KSEs in  
\cite{Cheung:2022ilc} to our field-strength formulation.
For a nonzero coupling $g_c$ in \cite{Cheung:2022ilc}, 
with the maximal coupling relation in eq.~\eqref{maxsusy} imposed, 
$g_c=2g=2m$, a dictionary relating the bosonic action to ours is
\begin{subequations}
\begin{align}
e^{2\lambda_I}=X_I\, ,
\qquad A^{12}_{(1)}&=A_1\, ,\qquad A^{34}_{(1)}=A_2\, ,\\*
S^5_{(3)}&=(X_1X_2)^2\star G,
\end{align}
\end{subequations}
where $\lambda_I$, $A^{12}_{(1)}$, $A^{34}_{(1)}$ and $S^5_{(3)}$ 
denote the fields of \cite{Cheung:2022ilc}.

Appendix~B of \cite{Cheung:2022ilc} also notes a consistency check
with the U$(1)$-gauged sector of minimal supergravity, without detailing 
the corresponding maximal-to-minimal field and fermion mapping. 
This comparison requires the higher-form conversion from 
the explicit $3$-form potential to the $4$-form field strength 
used in the minimal theory, 
and for the work carried out in the present paper, 
also the seven-dimensional maximal-to-minimal fermion dictionary. 
These ingredients are not provided there and are developed in 
Section~\ref{spinors}.

The additions to known solutions in Section~\ref{maxsol} and
Table~\ref{tab:known-gactive-solutions} comprise the warped AdS$_5$
flux-brane originally constructed in \cite{Lu:2003iv} and the two
warped AdS$_3$ configurations constructed in \cite{Cheung:2022ilc}.
We analyse their equal-charge specializations using a single SM
Killing spinor whose vector bilinear is timelike.
For the two warped AdS$_3$ families, we present smooth noncompact
global completions in eqs.~\eqref{CFGSsigmaglobal},
\eqref{CFGSspindleglobal}, ensuring their admissibility within
our solution survey.
Similar to all other solutions, which were already included in our survey 
summarized in Table~\ref{tab:known-gactive-solutions}, they all belong 
to the gauged active class $\mathcal C_{\rm g}^{\rm active}$.

\bibliographystyle{JHEPmod}
\bibliography{gravlocal}

\end{document}